\documentclass[12pt,reqno]{amsart}
\usepackage{amsaddr}

\usepackage[margin=0.9in]{geometry}
\usepackage{amsmath,amsthm,amssymb}
\usepackage{mathrsfs} 
\usepackage{enumitem}
\usepackage{pgf,tikz,pgfplots}
\usetikzlibrary{arrows}
\usepackage[unicode]{hyperref}
\hypersetup{colorlinks=true, linkcolor=blue, citecolor=blue, filecolor=blue, urlcolor=blue}
\usepackage{comment}
\usepackage{chngcntr}
\usepackage{mathabx}
\usepackage{bm}
\usepackage{etoolbox}
\usepackage{enumitem}
\usepackage{subcaption}
\usepackage{floatrow}
\usepackage{float}
\usepackage{xr}
\newfloatcommand{capbtabbox}{table}[][\FBwidth] 

\usepackage{blindtext}
\usepackage[comma,sort&compress, numbers]{natbib}
\usepackage{pdfpages}
\usepackage{tabularray}
\usepackage{multirow}
\numberwithin{equation}{section}

\newtheorem{theorem}{Theorem}[section]

\newtheorem{corollary}{Corollary}[section]

\newtheorem{lemma}{Lemma}[section] 

\newtheorem{proposition}{Proposition}[section]

\newtheorem{assumption}{Assumption}[section]

\theoremstyle{definition} 
\newtheorem{definition}{Definition}[section]
\newtheorem{remark}{Remark}[section]

\DeclareMathOperator{\Tr}{tr} 
\DeclareMathOperator{\tr}{tr} 
\newcommand{\sfK}{{\mathsf K}}
\newcommand{\sfL}{{\mathsf L}} 
\newcommand{\sfH}{{\mathsf H}}
\newcommand{\sfG}{{\mathsf G}}
\newcommand{\sfKmatrix}{{\hat{\bm{\sfK}}^\circ}}
\newcommand{\emmd}{\mathrm{MMD}^2}
\newcommand{\mmd}{\mathrm{MMD}}
\renewcommand{\P}{\mathrm{pr}}
\newcommand{\E}{\mathrm{E}}
\newcommand{\ra}{\rightarrow}
\newcommand{\lf}{\lfloor}
\newcommand{\rf}{\rfloor} 

\newcommand{\R}{\mathbb{R}}

\usepackage{mathrsfs}
\newcommand{\sB}{\mathscr{B}}
\newcommand{\sE}{\mathscr{E}}
\newcommand{\sX}{\mathscr{X}}
\newcommand{\sY}{\mathscr{Y}}

\newcommand{\cB}{\mathcal{B}}
\newcommand{\cE}{\mathcal{E}}
\newcommand{\cF}{\mathcal{F}}
\newcommand{\cH}{\mathcal{H}}
\newcommand{\cK}{\mathcal{K}}
\newcommand{\cN}{\mathcal{N}}
\newcommand{\cQ}{\mathcal{Q}}
\newcommand{\cR}{\mathcal{R}}
\newcommand{\cW}{\mathcal{W}}
\newcommand{\cX}{\mathcal{X}}
\newcommand{\cY}{\mathcal{Y}}
\newcommand{\cZ}{\mathcal{Z}}

\DeclareMathOperator{\Cov}{cov}
\DeclareMathOperator{\Var}{var}

\newcommand{\pto}{\stackrel{P}{\to}}
\newcommand{\dto}{\stackrel{D}{\to}} 
\newcommand{\asto}{\stackrel{a.s.}{\to}}

\allowdisplaybreaks
\begin{document}

\title[Boosting the Power of Kernel Two-Sample Tests]{\large Boosting the Power of Kernel Two-Sample Tests}
\author[Chatterjee and Bhattacharya]{Anirban Chatterjee and Bhaswar B. Bhattacharya} 
\address{Department of Statistics and Data Science, The Wharton School,\\University of Pennsylvania, Philadelphia,\\ Pennsylvania 19104, United States\\
\texttt{anirbanc@wharton.upenn.edu, bhaswar@wharton.upenn.edu}}

\begin{abstract}
    The kernel two-sample test based on the maximum mean discrepancy (MMD) is one of the most popular methods for detecting differences between two distributions over general metric spaces. In this paper we propose a method to boost the power of the kernel test by combining MMD estimates over multiple kernels using their Mahalanobis distance. We derive the asymptotic null distribution of the proposed test statistic and use a multiplier bootstrap approach to efficiently compute the rejection region. 
    The resulting test is universally consistent and, since it is obtained by aggregating over a collection of kernels/bandwidths, is more powerful in detecting a wide range of alternatives in finite samples. We also derive the distribution of the test statistic for both fixed and local contiguous alternatives. The latter, in  particular, implies that the proposed test is statistically efficient, that is, it has non-trivial asymptotic (Pitman) efficiency. \textcolor{black}{ The consistency properties of the Mahalanobis and other natural aggregation methods are also explored when the number of kernels is allowed to grow with the sample size.} Extensive numerical experiments are performed on both synthetic and real-world datasets to illustrate the efficacy of the proposed method over single kernel tests. 
    \textcolor{black}{The computational complexity of the proposed method is also studied, both theoretically and in simulations. Our asymptotic results rely on deriving the joint distribution of MMD estimates using the framework of multiple stochastic integrals, which is more broadly useful, specifically, in understanding the efficiency properties of recently proposed adaptive MMD tests based on kernel aggregation and also in developing more computationally efficient (linear time) tests that combine multiple kernels.  We conclude with an application of the Mahalanobis aggregation method for kernels with diverging scaling parameters. } 
    \end{abstract}


    \maketitle
    
    \section{Introduction}
    
    Given two probability distributions $P$ and $Q$ on a separable metric space $\mathcal X$, the two-sample problem is to test the hypothesis: 
    \begin{align}\label{eq:H01PQ}
        H_{0}: P=Q \quad \text{ versus } \quad H_{1}: P \neq Q , 
    \end{align} 
    based on i.i.d. samples 
    $\sX_m := \{X_1, X_2, \ldots, X_m\} \text{ and } \sY_n := \{Y_1, Y_2, \ldots, Y_n\}$ 
    from the distributions $P$ and $Q$, respectively. This is a classical problem that has  been extensively studied, especially in the parametric regime,
    where the data is assumed to have certain low-dimensional functional forms. However, parametric methods often perform poorly for misspecified models,  especially when the number of nuisance parameters is large, and for non-Euclidean data. This necessitates the development of non-parametric methods, which make minimal distributional assumptions on the data, but remain powerful for a wide class of alternatives.

    For univariate data, there are several well-known nonparametric tests
    such as the two-sample Kolmogorov--Smirnoff (KS) maximum deviation test
    \citep{Smirnov1948}, the Wald-Wolfwotiz runs test \citep{wald1940}, the
    rank-sum test \citep{mann1947,wilcoxon1947}, and the Cram\'er-von Mises test \citep{anderson1962distribution}. Efforts to generalize these univariate methods to higher dimensions date back to \citet{weiss1960two} and \citet{bickel1969distribution}. Thereafter, several nonparametric methods for multivariate two-sample testing have been proposed over the years. These include tests based on geometric graphs \citep{friedman1979multivariate,henze1984number,schilling1986multivariate,hall2002permutation,rosenbaum,multivariateruntest,chen2017new,bhattacharya2019}, tests based on data-depth \citep{liu1993},  the energy distance test \citep{szekely2003statistics,szekely2004testing,baringhaus2004new,aslan2005new,szekely2013energy}, kernel maximum mean discrepancy (MMD) tests  \citep{gretton2009fast,gretton2012kernel,gretton2012optimal,sejdinovic2013equivalence,chwialkowski2015fast,ramdas2015,ramdas2017wasserstein,song2020generalized,zhang2022testing,permutation2022}, ball divergence \citep{balldivergencenonparametric2018,balldivergence2022}, projection-averaging \citep{kim2020robust}, classifier-based tests \citep{lopez2017revisiting,kim2021classification}, among others. 
    Recently, a distribution-free version of the energy distance test has been proposed by \citet{deb2021multivariate} using the emerging theory of multivariate ranks based on optimal transport.

    Among the aforementioned methods kernel-based tests have emerged as  a powerful technique for detecting distributional differences on general domains. 
    The basic idea is to quantify the discrepancy between the two distributions $P$ and $Q$ in terms of the largest difference in expectation between $f(X)$ and $f(Y)$, for $X \sim P$ and $Y \sim Q$, over functions $f$ in  the unit ball of a reproducing kernel Hilbert space (RKHS) defined on $\cX$. This is called the {\it maximum mean discrepancy} (MMD) between the distributions $P$ and $Q$ (see \eqref{eq:FPQ} for the precise definition), which can be conveniently estimated from the data in terms of the pairwise kernel dissimilarities (see Section \ref{sec:Kdefiniton} for details). For characteristic kernels (see Assumption \ref{assumption:K}), a useful property of the MMD is that it takes value zero if and only if the distributions $P$ and $Q$ are the same. Consequently, the test which rejects $H_0$ for large values of the estimated MMD is universally consistent (the power of the test converges to 1 as the sample size increases) for the hypothesis \eqref{eq:H01PQ} (see~\citet{gretton2012kernel} for further details).

    Although the kernel two-sample test is widely used and has found numerous applications, it often performs poorly for high-dimensional problems \citep{ramdas2015} and its empirical performance depends heavily on the choice of the kernel. Kernels are usually parametrized by their bandwidths, and the most popular strategy for choosing the kernel bandwidth is the {\it median heuristic}, where the bandwidth is chosen to be the median of the pairwise distances of the pooled sample \citep{gretton2012kernel}. Despite its popularity there is limited understanding of the median heuristic and empirical results demonstrate that the median heuristic performs poorly when differences between the 2 distributions occur at a scale that differs significantly from the median of the interpoint distances. Another approach is to split the data and estimate the kernel by maximizing an approximate empirical power on the held-out data \citep{gretton2012optimal,liu2020learning}. This, however, can lead to loss in power for smaller sample sizes.

    In this paper we propose a strategy for augmenting the power of the classical (single) kernel two-sample test by borrowing strengths from multiple kernels. Specifically, we propose a new test statistic which combines MMD estimates from $r \geq 1$ kernels using their sample Mahalanobis distance. The advantage of aggregating across a collection of kernels/bandwidths is that the test can simultaneously deal with cases which require both small and large bandwidths, and, hence, detect both global and local differences more effectively. We illustrate the effectiveness of our method through a wide range of results, that includes a holistic study of its asymptotic properties, finite-sample and real-data performance, computational complexity, and comparison with other aggregation methods. 
    
    \begin{itemize} 
    
    \item {\it Asymptotic Properties}: We derive the joint distribution of the vector of MMD estimates under $H_0$, which can be described using bivariate stochastic integrals, and, as a consequence, derive the asymptotic distribution of the Mahalanobis aggregated MMD (MMMD) statistic under $H_0$ (Section \ref{sec:H0asymptotic}).  Moreover, using the kernel Gram matrix representation we develop a multiplier bootstrap approach that allows us to efficiently compute the rejection threshold for the MMMD statistic and show that  the resulting test is universally consistent (Section \ref{sec:H0implementation}). Next, we derive the distribution of the proposed test against local alternatives in the well-known contamination model (Section \ref{sec:asymptoticpower}).  In Appendix \ref{sec:H1asymptotic} of the supplementary material we derive the joint distribution of MMD estimates and, consequently, that of the MMMD statistic, under the alternative. 
    
    \item {\it Finite-Sample Performance}: In Section \ref{sec:experiments} and Appendix \ref{sec:experimentsadditional} of the supplementary materials, we perform extensive simulations to compare our MMMD based test with various single kernel MMD tests (with bandwidths chosen based on the median heuristic). The experiments show that the MMMD method outperforms the single kernel tests and also the graph-based Friedman-Rafsky test \citep{friedman1979multivariate} across a range alternatives and dimensions, showcasing the efficacy of our aggregation method. \textcolor{black}{ To further reinforce the benefits of our aggregation scheme we also compare the MMMD test with bandwidth optimized single-kernel tests (as in \citet{gretton2012optimal,liu2020learning}) and with $p$-value combination methods (see Section \ref{sec:bandwithoptimizedcombination}). }
    
    \item {\it Computational and Statistical Trade-offs}: \textcolor{black}{ In Appendix \ref{sec:computationpf} we analyze the time complexity of the MMMD tests and also report the trade-off between power and computation time of the MMMD and the kernel single kernel MMD tests in simulations. We also implement our Mahalanobis aggregation strategy for the linear time statistic \cite[Section 6]{gretton2012kernel}, derive the corresponding asymptotic theory, and report its finite-sample performance (see Appendix \ref{sec:linear}). } \textcolor{black}{ The multiplier bootstrap also emerges as the more computationally efficient option than the permutation test for calibrating the MMMD statistic (see Section \ref{sec:permutation} for details). }
    
    \item {\it Real-Data Applications}: In Section \ref{sec:data} we apply the proposed method to compare images of digits in the noisy MNIST dataset. The MMMD effectively distinguishes different digits for significantly more noisy images compared to its single kernel counterparts, again illustrating the advantage of using multiple kernels. 
    
    \item \textcolor{black}{ {\it Increasing Number of Kernels}: In Section \ref{sec:rkernels} we investigate the behavior of the Mahalanobis and other aggregation strategies when the number of kernels is allowed to grow with the sample size. Specifically, we derive consistent tests based on Mahalanobis, as well as maximum and $L_2$ type aggregations, in the growing $r$ regime. }
    
    \end{itemize}

    Our results on the joint distribution for multiple kernels are also more broadly useful in understanding the asymptotic properties of general aggregation strategies. To demonstrate this we present two applications. 
    
    \begin{itemize} 
    
    \item In Section \ref{sec:broaderscope} we propose an asymptotic implementation of the adaptive MMD test recently proposed in \citet{schrab2021mmd}, and derive its asymptotic local power. Numerical results comparing the MMMD method and the aforementioned adaptive test are also reported in the supplementary materials. 
    
    \item \textcolor{black}{ Moreover, in Section \ref{sec:kernelbandwidth} we derive the asymptotic distribution of the MMMD statistic for kernels with bandwidths depending on the sample size. Specifically, we show that when the scaling parameters are chosen proportional to the optimal bandwidth (as in \cite{kernelnonparametricsmoothalternatives,schrab2021mmd}), then the vector of MMD estimates has a multivariate normal distribution under the null. Using this we construct a test which aggregates multiple kernels with a chi-squared distribution under $H_0$. } 
    \end{itemize} 
    
    The codes for all the experiments can be found in \href{https://github.com/anirbanc96/MMMD-boost-kernel-two-sample}{https://github.com/anirbanc96/MMMD-boost-kernel-two-sample}.

    \section{Kernel Maximum Mean Discrepancy and Mahalanobis Aggregation}

    We begin by recalling the fundamentals of the kernel two-sample test as introduced in \citet{gretton2012kernel} in Section \ref{sec:Kdefiniton}.  Then in Section \ref{sec:multiplekernels} we describe our proposed test statistic obtained by combining multiple kernels.

    \subsection{Kernel Maximum Mean Discrepancy}\label{sec:Kdefiniton} 
    
    Suppose $\mathcal X$ is a separable metric space and $\mathscr B(\mathcal X)$ is the sigma-algebra generated by the open sets of $\mathcal X$. Denote by $\mathcal P(\mathcal X)$ the collection of all probability distributions on $(\mathcal X, \mathscr B(\mathcal X))$. Suppose $P, Q \in \mathcal  P(X)$ and $X \sim P$ and $Y \sim Q$ be random variables distributed as $P$ and $Q$, respectively. Throughout we will assume that $P$ and $Q$ are {\it non-atomic}. 
    The maximum mean discrepancy (MMD) between $P$ and $Q$ is defined as 
    \begin{align}\label{eq:FPQ}
        \mathrm{MMD}\left[\mathcal{F}, P, Q \right] = \sup_{f\in\mathcal{F}}\left\{ \E_{X \sim P} [ f( X ) ] -\E_{Y \sim Q} [f( Y ) ] \right \} ,
    \end{align} 
    where $\mathcal{F}$ is the unit ball of a reproducing kernel Hilbert space (RKHS) $\mathcal H$ defined on $\mathcal X$ \citep{aronszajn1950theory}.  Since $\mathcal H$ is an RKHS, for every $x \in \mathcal X$ the evaluation map operator $\eta_x: \mathcal H \rightarrow \mathbb R$ given $\eta_x(f) = f(x)$ is continuous. Thus, by the Riesz representation theorem \cite[Theorem II.4]{reedsimon}  for each $x \in \mathcal X$ there is a feature mapping $\psi_x \in \mathcal H$ such that $f(x) = \langle f, \psi_x  \rangle_\mathcal H$, for every $f \in \mathcal H$,  where $\langle\cdot,\cdot\rangle_{\mathcal{H}}$ is the inner product in $\mathcal{H}$. The feature mapping takes the canonical form $\psi_x(\cdot) = \sfK (x, \cdot)$, where $\sfK : \mathcal X \times \mathcal X \rightarrow \mathbb R$ is a positive definite kernel.  This, in particular, implies that $\sfK (x,y) = \langle \psi(x),\psi(y)\rangle_{\mathcal{H}}$. Extending the notion of feature map, an element $\mu_{P}\in\mathcal{H}$ is defined to be the \textit{mean embedding} of $P\in\mathcal{P}\left(\mathcal{X}\right)$ if 
    \begin{align}\label{eq:fXP}
    \langle f,\mu_{P}\rangle_{\mathcal{H}} = \E_{X\sim P}[f(X)] , 
    \end{align} 
    for all $f\in \mathcal{H}$.
    By the canonical form of the feature map it follows that 
    \begin{align}\label{eq:XP}
        \mu_{P}(t) : = \int_{\mathcal{X}} \psi_t(x) \mathrm d P(x) = \E_{X \sim P} [\psi_t(X)] =  \E_{X \sim P}  [\sfK(t, X)] . 
    \end{align} 
    Throughout we will make the following assumption: 
    
    \begin{assumption}\label{assumption:K} {\em The kernel $\sfK : \mathcal X \times \mathcal X \rightarrow \mathbb R$ satisfies the following: 
    \begin{itemize} 
    
    \item[(1)]  $\E_{X \sim P}[\sfK (X,X)^{\frac{1}{2}}]<\infty$ and $\E_{Y \sim Q} [\sfK (Y, Y)^{\frac{1}{2}} ]<\infty$. 
    
    \item[(2)]  $\sfK$ is \textit{characteristic}, that is, the mean embedding $\mu:\mathcal{P}(\mathcal{X})\rightarrow\mathcal{H}$ is a one-to-one (injective) function. 
    \end{itemize} }
    \end{assumption} 
    
    Assumption \ref{assumption:K} ensures that $\mu_{P}, \mu_{Q}\in\mathcal{H}$ (see \citet[Lemma 3]{gretton2012kernel} and \citet[Lemma 2.1]{park2020measure}) and MMD defines a metric on $\mathcal{P}(\mathcal{X})$. Then the MMD can be expressed as the distance between mean embeddings in $\mathcal H$ (see \citet[Lemma 4]{gretton2012kernel}): 
    \begin{align}\label{eq:KPQ}
        \mathrm{MMD}^2\left[\mathcal{F},P,Q\right] = \left\|\mu_{P}-\mu_{Q}\right\|_{\mathcal{H}}^2 , 
    \end{align}
    where $\|\cdot\|_{\mathcal{H}}$ is the norm corresponding to the inner product $\langle\cdot,\cdot\rangle_{\mathcal{H}}$. 
    This implies $ \mathrm{MMD}^2\left[\mathcal{F},P,Q\right]  = 0$ if and only if $P=Q$. Expanding the square in \eqref{eq:KPQ} and using the representation in \eqref{eq:XP} it follows that (see \citet[Lemma 6]{gretton2012kernel} for details) 
    \begin{align*}
        \mathrm{MMD}^2\left[\mathcal{F}, P, Q\right] = \E_{X, X' \sim P} [ \sfK (X, X') ]  +  \E_{Y, Y' \sim Q} [\sfK (Y, Y')]    - 2\E_{X \sim P, Y \sim Q} [\sfK (X, Y)] .
    \end{align*} 
    Therefore, based on i.i.d. observations $\sX_m := \{X_1, X_2, \ldots, X_m\} \text{ and } \sY_n := \{Y_1, Y_2, \ldots, Y_n\}$, a natural unbiased estimate of $\mathrm{MMD}^{2}\left[\mathcal{F},P,Q\right]$ is given by, 
    \begin{align}\label{eq:MMDXY}
        \emmd \left[\sfK, \sX_m, \sY_n \right] = \mathcal W_{\sX_m} + \mathcal W_{\sY_n} - 2 \mathcal B_{\sX_m, \sY_n} , 
    \end{align}   
    where 
    \begin{align}\label{eq:WX}
     \mathcal W_{\sX_m} :=   \frac{1}{m(m-1)}\sum_{1 \leq i \ne j \leq m} \sfK \left(X_{i},X_{j}\right) \text{ and }  \mathcal W_{\sY_n} := \frac{1}{n(n-1)}\sum_{1 \leq i \ne j \leq n} \sfK \left(Y_{i},Y_{j}\right) 
     \end{align} 
     is the average of the kernel dissimilarities within the samples in $\sX_m$ and $\sY_n$, respectively, and  
     \begin{align}\label{eq:BXY}
     \mathcal B_{\sX_m, \sY_n} :=  \frac{1}{mn}\sum_{i=1}^{m}\sum_{j=1}^{n} \sfK \left(X_{i},Y_{j}\right) 
    \end{align}
     is the average of the kernel dissimilarities between the samples in $\sX_m$ and $\sY_n$. 
    Throughout we will assume $N := m+n\rightarrow\infty$ such that 
    \begin{align}\label{eq:mn}
    \frac{m}{m+n}\rightarrow\rho \in (0,1) .
    \end{align} 
    Then $\emmd[\sfK, \sX_m, \sY_n]$ is a consistent estimate of $\mathrm{MMD}^2\left[\cF, P, Q \right]$ (see \citet[Theorem 7]{gretton2012kernel}), that is, 
    \begin{align}\label{eq:Kmn}
    \emmd[\sfK, \sX_m, \sY_n] \pto \mathrm{MMD}^2\left[\cF, P, Q \right] . 
    \end{align} 
    Hence, the test which rejects $H_0$ in \eqref{eq:H01PQ} for large values of $\mathrm{MMD}^2\left[\sfK, \sX_m, \sY_n \right]$ is universally consistent. In fact, for the consistency result it suffices to assume $\min\{m, n\}\rightarrow \infty$. The existence of the limit in \eqref{eq:mn} will be required for deriving the asymptotic distribution of the test statistic.
    
    \subsection{Aggregating Multiple Kernels} 
    \label{sec:multiplekernels}

    Fix $r \geq 1$ and suppose $\sfK_1, \sfK_2, \ldots, \sfK_r$ be $r$ distinct kernels each of which satisfy Assumption \ref{assumption:K}. Denote the vector of MMD estimates as 
    \begin{align}\label{eq:Kvector}
        \mathrm{MMD}^{2}\left[ \cK, \sX_m, \sY_n \right] = \left(\mathrm{MMD}^{2}[\sfK_{1}, \sX_m, \sY_n],\ldots, \mathrm{MMD}^{2}[\sfK_{r}, \sX_m, \sY_n]\right)^{\top} , 
    \end{align}
    where $\cK := \{\sfK_1, \sfK_2, \ldots, \sfK_r \}$. In this paper we propose a new test statistic that combines the contributions of the different kernels using the Mahalanobis distance of $\mathrm{MMD}^{2}\left[ \cK, \sX_m, \sY_n \right]$ as follows:  
    \begin{align}\label{eq:aggregateS}
    \left( \mathrm{MMD}^{2}\left[ \cK, \sX_m, \sY_n \right]  \right)^\top \bm S^{-1}  \left( \mathrm{MMD}^{2}\left[ \cK, \sX_m, \sY_n \right]  \right) , 
    \end{align}
    where is $\bm S$ is a consistent estimate of the limiting covariance matrix of $\mathrm{MMD}^{2}\left[ \cK, \sX_m, \sY_n \right]$ under $H_0$ (which we denote by ${\bm \Sigma}_{H_0}= ((\sigma_{ab}))_{1 \leq a, b \leq r}$). 
    Note that adjusting by the covariance matrix $\bm S$ brings the contributions of the individual MMD estimates in the same scale and by selecting a range of kernels/bandwidths in $\cK$ one can detect more fine-grained deviations from $H_0$, leading to significant power improvements as will be seen in Section \ref{sec:experiments}. (In Appendix \ref{sec:Sigma} we present general conditions under which $\bm \Sigma_{H_0}$ is invertible, which, in particular, hold for any collection of Gaussian or Laplace kernels.)

    In Corollary \ref{cor:TM} we compute 
    \begin{align}\label{eq:H0sigma} 
        \sigma_{ab}  & :=   \lim_{N \rightarrow \infty} (m+n)^2 \left(\Cov_{H_0}\left[ \mathrm{MMD}^{2}\left[ \cK, \sX_m, \sY_n \right] \right]\right)_{ab} \\ 
            & = \frac{2}{\rho^2(1-\rho)^2} \E_{X, X' \sim P}\left[\sfK_a^\circ(X, X')\sfK_b^\circ(X, X')\right] \nonumber , 
    \end{align}
    where 
    \begin{align}\label{eq:Kxycentered}
    \sfK_a^\circ (x, y)
    = \sfK_a(x,y) - \E_{X \sim P}\sfK_{a}(X, y) - \E_{X' \sim P}\sfK_{a}(x, X') + \E_{X, X' \sim P}\sfK_{a}(X, X') , 
    \end{align}
    is the centered version of the kernel $\sfK_a$, for $1 \leq a \leq r$. Therefore, a natural empirical estimate of $\bm \Sigma_{H_0}$ is given by $\hat {\bm \Sigma} = (( \hat \sigma_{ab} ))_{1 \leq a, b \leq r}$, where 
    \begin{align}\label{eq:H0sigmaestimate}
        \hat \sigma_{ab} & = \frac{2}{\hat{\rho}^2(1-\hat{\rho})^2} \cdot \frac{1}{m^2} \sum_{1 \leq i, j \leq m} \hat \sfK_a^\circ(X_i, X_j) \hat \sfK_b^\circ(X_i, X_j) , 
    \end{align} 
    with
    \begin{align}\label{eq:estimateKxy} 
    \hat \sfK_a^\circ (x, y)= \sfK_a(x,y) - \frac{1}{m} \sum_{u=1}^m \sfK_{a}(X_u, y) - \frac{1}{m} \sum_{v=1}^m \sfK_{a}(x, X_v) + \frac{1}{m^2} \sum_{u, v=1}^m\sfK_{a}(X_u, X_v) 
    \end{align} 
     being the empirical analogue of $\sfK_a^\circ$ and $\hat{\rho} = \frac{m}{m+n}$. Therefore, choosing $\bm S = \hat {\bm \Sigma}$ in \eqref{eq:aggregateS} we define the {\it Mahalanobis aggregated MMD} (MMMD) statistic as follows: 
    \begin{align}\label{eq:TM} 
    T_{m, n} : =  \left( \mathrm{MMD}^{2}\left[ \cK, \sX_m, \sY_n \right]  \right)^\top \hat{\bm \Sigma}^{-1}  \left( \mathrm{MMD}^{2}\left[ \cK, \sX_m, \sY_n \right]  \right) .
    \end{align} 
    In Corollary \ref{cor:TM} we show that $\hat {\bm \Sigma} \pto \bm \Sigma_{H_0}$, hence \eqref{eq:Kmn} implies that 
        \begin{align}\label{eq:TK}
            T_{m, n}  \pto \left( \mathrm{MMD}^{2}\left[ \bm\cF, P, Q \right]  \right)^\top \bm \Sigma^{-1}_{H_0}  \left( \mathrm{MMD}^{2}\left[ \bm\cF, P, Q \right]  \right) :=  T_{\cK} . 
        \end{align}
    where $\bm\cF = \left\{\cF_{1},\ldots,\cF_{r}\right\}$, with $\cF_{a}$ being the unit ball in the RKHS of $\sfK_{a}$, for all $1\leq a\leq r$, and 
    \begin{align}\label{eq:defbcFmmd}
        \mathrm{MMD}^{2}\left[ \bm\cF, P, Q \right] = \left(\emmd\left[\cF_{1},P,Q\right],\cdots,\emmd\left[\cF_{r},P,Q\right]\right)^{\top}.
    \end{align} 
    Note that $T_{\cK}=0$ under $H_0$ and $T_{\cK} > 0$ whenever $P \ne Q$. Hence, a test rejecting $H_0$ for `large' values of $T_{m, n}$ will be  universally consistent. 
    However, to construct a test based on $T_{m, n}$ we need to chose a cut-off (rejection region) based on the data. The first step towards this is to derive the limiting null distribution of $\mathrm{MMD}^{2}\left[ \cK, \sX_m, \sY_n \right] $. This is discussed in Section \ref{sec:H0asymptotic}. 
    
    \section{Asymptotic Null Distribution}
    \label{sec:H0asymptotic} 
    
    In this section we derive the limiting distribution of the vector of MMD estimates \eqref{eq:Kvector} under $H_0$ and, consequently, that of the proposed statistic $T_{m, n}$, using the framework of multiple Weiner-It\^{o} stochastic integrals. 
    (We recall the definition and basic properties of multiple Weiner-It\^{o} stochastic integrals in Appendix \ref{sec:stochasticintegral} of the supplementary material.) 
    
    \begin{theorem}\label{thm:K} 
    Suppose $\cK=\{\sfK_1, \sfK_2, \ldots, \sfK_r\}$ be a collection of $r$ distinct kernels such that $\sfK_a$ satisfies Assumption $\ref{assumption:K}$ and $\sfK_{a}\in L^{2}(\mathcal{X}^{2}, P^{2})$, for $1 \leq a \leq r$. Then under $H_0$, in the asymptotic regime \eqref{eq:mn}, 
    \begin{align}\label{eq:GK}
        (m+n) \mathrm{MMD}^{2}[ \cK, \sX_m, \sY_n ]  \dto G_\cK := \frac{1}{\rho(1-\rho)} \Big( I_2(\sfK^\circ_{1}), I_2(\sfK^\circ_{2}), \ldots, I_2(\sfK^\circ_{r}) \Big)^\top , 
    \end{align}
    where $I_2(\cdot)$ is the bivariate multiple Weiner-It\^{o} stochastic integral as defined in Appendix \ref{sec:stochasticintegral}, $\sfK_a^\circ$ is defined in \eqref{eq:Kxycentered}, for $1 \leq a \leq r$.  Moreover, the characteristic function of $G_\cK$ at $\bm \eta = (\eta_1, \eta_2, \ldots, \eta_r)^\top \in \R^r$ is given by:  
    \begin{align}\label{eq:expGK}
        \Phi\left(\bm \eta \right) := \E[e^{\iota \bm \eta^\top G_\cK} ]= \prod_{\lambda\in \Lambda(\bm \eta)}\dfrac{\exp\left(-\frac{\iota \lambda  }{\rho(1-\rho)}\right)}{\sqrt{1-\frac{2 \iota \lambda}{\rho (1-\rho)}}} , 
    \end{align}
    where $\Lambda(\bm \eta)$ is the set of eigenvalues (with repetitions) of the Hilbert-Schmidt operator $\cH_{\cK, \bm \eta}: L^2(\cX, P) \rightarrow L^2(\cX, P)$ defined as: 
    \begin{align}\label{eq:HKeta}
        \cH_{\cK, \bm \eta}[f(x)] = \int_{\cX} \left(\sum_{a=1}^r \eta_a \sfK_a^\circ(x,y)\right)f(y)\mathrm d P(y) . 
    \end{align}
    \end{theorem}

    The proof of Theorem \ref{thm:K} is given in Appendix \ref{sec:pfK} of the supplementary materials. 
    For an alternate representation of the limiting distribution in \eqref{eq:GK} see in Remark \ref{remark:limitingdistribution} in Appendix \ref{sec:pfK}.

    Theorem \ref{thm:K} allows us to obtain the limiting distribution of any smooth function of finitely many MMD estimates under $H_0$. In particular, for the MMMD statistic $T_{m, n}$ in \eqref{eq:TM} we have the following result. The proof is given in Appendix \ref{sec:corTMpf}. 
    
    \begin{corollary}\label{cor:TM} Suppose $\bm \Sigma_{H_0} := ((\sigma_{ab}))_{1 \leq a, b \leq r} $ and $\hat {\bm \Sigma}:= ((\hat \sigma_{ab}))_{1 \leq a, b \leq r}$ be as in \eqref{eq:H0sigma} and \eqref{eq:H0sigmaestimate}, respectively. Then 
    \begin{align}\label{eq:sigmaH0ab} 
    \sigma_{ab} = \frac{2}{\rho^2(1-\rho)^2} \E_{X, X' \sim P}\left[\sfK_a^\circ(X, X')\sfK_b^\circ(X, X')\right] , 
    \end{align}
    where $\sfK_a^\circ$, for $1 \leq a \leq r$, is as defined in \eqref{eq:Kxycentered}. Moreover, in the asymptotic regime \eqref{eq:mn}, 
    \begin{align}\label{eq:sigmaH0limit}
    \hat \sigma_{ab} \stackrel{a.s.} \rightarrow \sigma_{ab}, 
    \end{align}
    for $1 \leq a, b \leq r$. Furthermore, under $H_0$, for $G_{\cK}$ as in \eqref{eq:GK}, 
    \begin{align}\label{eq:GKH0}
    (m+n)^2 T_{m, n} \dto G_\cK^\top \bm \Sigma_{H_0}^{-1} G_{\cK} . 
    \end{align}  
    \end{corollary}
    
    \section{Calibration Using Gaussian Multiplier Bootstrap}
    \label{sec:H0implementation}
    In order to apply Corollary \ref{cor:TM} to obtain a valid level $\alpha$ test based on $T_{m, n}$ we need to estimate the quantiles of the limiting distribution in \eqref{eq:GKH0}, which depends on the (unknown) distribution $P$. Although the   distribution in \eqref{eq:GKH0} does not have a tractable closed form, we can efficiently estimate its quantiles based on the samples $\sX_m = \{X_1, X_2, \ldots, X_m\}$, using the kernel Gram matrix representation of the MMD estimate and the Gaussian multiplier bootstrap. 
    Towards this, for each kernel $\sfK_a$ define its Gram matrix based on $\sX_m$ as: 
    \begin{align*}
       \hat{\bm{\sfK}}_{a}= \left(\sfK_a(X_{i},X_{j})\right)_{1\leq i, j\leq m} ,
    \end{align*}
    and their centered versions as:  
    \begin{align}\label{eq:centered-kernel}
        \sfKmatrix_{a} = \bm C \hat{\bm{\sfK}}_{a} \bm C/m = \left(\dfrac{\hat \sfK_a^\circ(X_{i},X_{j})}{m}\right)_{1\leq i, j\leq m} , \quad \text{ where } \quad \bm C = \bm{I}-\frac{1}{m}\bm{1}\bm{1}^\top , 
    \end{align}
    and $\hat \sfK_a^\circ$ is as defined in \eqref{eq:estimateKxy}, for $1\leq a \leq r$.
    For $\hat{\rho}:= \frac{m}{m+n}$, denote 
    \begin{align}\label{eq:def-Em}
        \cE(\cK, \sX_m)  := 
        \begin{pmatrix}
            \bm{Z}_{m}^\top \sfKmatrix_1  \bm{Z}_{m} -  \frac{1}{\hat{\rho}(1-\hat{\rho})}\Tr[ \sfKmatrix_1 ] \\
            \bm{Z}_{m}^\top \sfKmatrix_2 \bm{Z}_{m} - \frac{1}{\hat{\rho}(1-\hat{\rho})} \Tr[ \sfKmatrix_2 ]\\
            \vdots\\
            \bm{Z}_{m}^\top \sfKmatrix_m \bm{Z}_{m} - \frac{1}{\hat{\rho}(1-\hat{\rho})} \Tr[ \sfKmatrix_m ]
        \end{pmatrix} , 
    \end{align}
    where $\bm{Z}_{m} \sim \cN_m(\bm 0, \frac{1}{\hat{\rho}(1-\hat{\rho})} \bm{I} )$ independent of $\sX_m$. In the following theorem we show that distribution of $\cE(\cK, \sX_m)$ conditional on $\sX_m$ converges to $G_\cK$ as in \eqref{eq:GK}. 
    
    \begin{theorem}\label{thm:estimateH0}  
    Suppose $\cK=\{\sfK_1, \sfK_2, \ldots, \sfK_r\}$ be a collection of $r\geq 1$ distinct kernels such that $\sfK_a$ satisfies Assumption \ref{assumption:K} 
    and $\sfK_a$ is bounded, for all $1\leq a\leq r$. Then under $H_0$, in the asymptotic regime \eqref{eq:mn}, $\cE(\cK, \sX_m)\big |\sX_m  \dto G_{\cK}$, 
    almost surely, where $G_\cK$ is as defined in \eqref{eq:GK}. 
    \end{theorem}

    The proof of Theorem \ref{thm:estimateH0} is given in Appendix \ref{sec:estimateH0pf}. It shows that the asymptotic distribution of $\cE(\cK, \sX_m)\big |\sX_m$ is the same as that of $(m+n) \mathrm{MMD}^{2}[ \cK, \sX_m, \sY_n ]$. Since $\cE(\cK, \sX_m)\big |\sX_m$ is completely determined by the data $\sX_m$, we can use it approximate the quantiles of any `nice' functions $G_{\cK}$. To this end, define 
    \begin{align}\label{eq:Testimate}
        \hat T_{m}:= \cE(\cK, \sX_m)^\top \hat{\bm{\Sigma}}^{-1} \cE(\cK, \sX_m) . 
    \end{align}  
    Now, by a direct computation $\Var \left[\cE(\cK, \sX_m)\middle|\sX_m\right] = \frac{2}{\hat\rho^2(1-\hat\rho)^2} ( ( \Tr [ \sfKmatrix_{a} \sfKmatrix_{b} ] ) )_{1 \leq a, b \leq m}$. Hence, from the proof of Corollary \ref{cor:TM} (specifically \eqref{eq:sigmaH0limit}) it follows that $\Var \left[\cE(\cK, \sX_m)\middle|\sX_m\right] = \hat{\bm{\Sigma}} \stackrel{a.s.} \rightarrow \bm{\Sigma}_{H_0}$. 
    This combined with Theorem \ref{thm:estimateH0} implies that under $H_0$, 
    \begin{align}
    \hat T_{m} \mid \sX_m \dto G_\cK^\top \bm \Sigma_{H_0}^{-1} G_{\cK} ,  
    \end{align} 
    almost surely. This shows that $\hat T_{m}$ has the same limiting distribution as $(m+n)^2 T_{m, n}$ under $H_0$ (recall \eqref{eq:GKH0}), hence, we can use the quantiles of $\hat T_{m}$ to calibrate the statistic $T_{m, n}$.  Specifically, for $\alpha \in (0, 1)$ denote by $\hat q_{\alpha, m}$ the $\alpha$-th quantile of distribution $\hat T_{m}|\sX_m$ and consider the test function 
    \begin{align}\label{eq:Tmnalpha}
    \phi_{m, n} = \bm 1\{ (m+n)^2 T_{m, n}>\hat q_{1-\alpha, m}\} . 
    \end{align}  
    Corollary \ref{cor:TM}, \eqref{eq:TK}, and \eqref{eq:Testimate} now implies the following result: 
    
    \begin{corollary}[Consistency]\label{cor:consistency} 
    Suppose the assumptions of Theorem \ref{thm:estimateH0} hold and $\phi_{m, n}$ be as defined above. Then $\lim_{m, n \rightarrow \infty} \E_{H_0}[\phi_{m, n}] = \alpha$. Moreover, for any $P \ne Q$, $\lim_{m, n \rightarrow \infty} \E_{H_1}\left[\phi_{m, n} \right] = 1$, that is, $\phi_{m, n}$ is  universally consistent. 
    \end{corollary}

    The result above shows that the MMMD statistic with cut-off chosen using the multiplier bootstrap method attains the exact asymptotic level and is universally consistent. In practice, to compute $\hat q_{1-\alpha, m}$ we generate $B$ replicates $\{\hat T_{m}^{(1)}, \hat T_{m}^{(2)}, \ldots, \hat T_{m}^{(B)}\}$ of $\hat T_m$, based on $B$ independent copies of $\bm{Z}_{m}$, and choose $\hat q_{1-\alpha, m}$ to be the sample $\alpha$-th quantile of $\{\hat T_{m}^{(1)}, \hat T_{m}^{(2)}, \ldots, \hat T_{m}^{(B)}\}$. 
    
    \begin{remark}\label{remark:sigmaH0}
    While implementing the test 
    we often replace $\hat{\bm{\Sigma}}^{-1}$ in \eqref{eq:TM} and \eqref{eq:Testimate}, by $( \hat{\bm{\Sigma}} + \lambda \bm{I}_m )^{-1}$, 
    for some suitably chosen regularization parameter $\lambda > 0$. Although the limiting covariance matrix $\bm \Sigma_{H_0}$ is invertible (see Corollary \ref{cor:Kvariance} in the supplementary material), hence, $\hat{\bm{\Sigma}}$ is also invertible for large sample sizes with probability 1, adding a small regularization provides numerical stability in finite samples. In fact, the conclusions in Corollary \ref{cor:consistency} remain valid, for any choice of 
    $\lambda = \lambda (\sX_m)$ converging almost surely to a deterministic constant $\lambda_0 > 0$ (see Section \ref{sec:experiments} for more details on the choice of  $\lambda$ in our experiments). 
    \end{remark}

    \section{Local Asymptotic Power}
    \label{sec:asymptoticpower} 
    
    Throughout this section we will assume that $\cX = \R^d$ and the distributions $P$ and $Q$ have densities $f_P$ and $f_Q$ with respect to the Lebesgue measure in $\R^d$. To quantify the notion of local alternatives, we will adopt the commonly used contamination model: 
    \begin{align}\label{eq:fPQ}
    f_Q(\cdot)=(1-\delta)f_P(\cdot)+\delta g(\cdot) , 
    \end{align} 
    where $\delta \in [0, 1)$ and $g \ne f_P$ is a probability density function with respect to the Lebesgue measure in $\R^d$ such that the following hold: 
    \begin{assumption}\label{assumption:fPQ}
    {\em The support of $g$ is contained in that of $f_P(\cdot)$ and $0<\Var_{X \sim P} [ \frac{g(X)}{f_P(X)} ]<\infty$. } 
    \end{assumption} 
    
    Under this assumption, contiguous local alternatives are obtained by considering local perturbations of the mixing proportion $\delta$ as follows (see \cite[Chapter 12]{lr}):  
    \begin{equation}\label{eq:H0N}
    H_0:\delta = 0 \qquad \mathrm{versus} \qquad H_1:\delta = h/\sqrt{N},
    \end{equation}  
    for some $h\neq 0$ and $N= m+n$. The following theorem derives the distribution of the MMMD statistic $T_{m, n}$ under $H_1$ as above.

    \begin{theorem}\label{thm:H0NK} 
    Suppose $\cK=\{\sfK_1, \sfK_2, \ldots, \sfK_r\}$ be a collection of $r$ distinct kernels such that $\sfK_a$ satisfies Assumption $\ref{assumption:K}$ and $\sfK_{a}\in L^{2}(\mathcal{X}^{2}, P^{2})$, for $1 \leq a \leq r$.
    Then under $H_1$ as in \eqref{eq:H0N}, in the asymptotic regime \eqref{eq:mn}, 
    \begin{align}\label{eq:H0NGK}
        (m+n) \mathrm{MMD}^{2}[ \cK, \sX_m, \sY_n ]  \dto G_{\cK, h} := \left( 
    \begin{array}{c}
     \gamma  I_2(\sfK^\circ_1) + 2h \sqrt{\gamma} I_1\left(\sfK^{\circ}_1\left[\frac{g}{f_P}\right]\right) + h^2 \mu_1 \\
     \gamma I_2(\sfK^\circ_2) + 2h \sqrt{\gamma} I_1 \left(\sfK^{\circ}_2\left[\frac{g}{f_P}\right]\right) + h^2 \mu_2  \\
      \vdots  \\
     \gamma I_2(\sfK^\circ_r) + 2h \sqrt{\gamma} I_1 \left(\sfK^{\circ}_r\left[\frac{g}{f_P}\right]\right) + h^2 \mu_r
    \end{array}
    \right) . 
    \end{align}
    where $\gamma = \frac{1}{\rho(1-\rho)}$, $\sfK_a^\circ[\frac{g}{f_P}](x) := \int_{\cX} \sfK_a^\circ (x, y) g(y) \mathrm d y$,  
    \begin{align}\label{eq:meanK}
    \mu_a := \E \left[ \sfK_a^\circ(X, X')  \frac{g(X) g(X') }{f_P(X) f_P(X')} \right] , 
    \end{align} 
    and $\sfK_a^\circ$ is defined in \eqref{eq:Kxycentered}, for $1 \leq a \leq r$.
    \end{theorem}
    
    The proof of the theorem is given in Appendix \ref{sec:localpowerpf} of the supplementary material. The following result is an immediately consequence of the above result together with the continuous mapping theorem and Corollary \ref{cor:TM}.

    \begin{corollary}\label{cor:H0NTM} 
    Under $H_1$ as in \eqref{eq:H0N}, $(m+n)^2 T_{m, n} \dto G_{\cK, h}^\top \bm \Sigma_{H_0}^{-1}G_{\cK, h}$. 
    \end{corollary}
    
    Using Corollary \ref{cor:H0NTM} we can derive the limiting local power of the test $\phi_{m, n}$ in \eqref{eq:Tmnalpha}. Specifically, suppose $F_{\cK, h}$ denotes the CDF of $G_{\cK, h}^\top \bm \Sigma_{H_0}^{-1}G_{\cK, h}$ and $q_{1-\alpha}$ be the $(1-\alpha)$-th quantile of the distribution $G_{\cK}^\top \bm \Sigma_{H_0}^{-1}G_{\cK}$. (Note that $G_{\cK, 0} = G_{\cK}$.) Since $\hat q_{1-\alpha, m} | \sX_m \stackrel{a.s.} \rightarrow q_{1-\alpha}$, Corollary \ref{cor:H0NTM} implies that 
        the asymptotic power of $\phi_{m, n}$ under $H_1$ as in \eqref{eq:H0N} is given by $\lim_{m, n \rightarrow \infty} \E_{H_1}[\phi_{m, n}]= 1- F_{\cK, h}(q_{1-\alpha})$. 
        This implies, $\phi_{m, n}$ has non-trivial asymptotic (Pitman)  efficiency and is rate-optimal, in the sense that, $\lim_{|h| \rightarrow \infty }\lim_{m, n \rightarrow \infty} \E_{H_1}[\phi_{m, n}] = 1$.

    \color{black} 
    
    \section{Numerical Experiments}
    \label{sec:experiments}
    
    In this section, we study the finite-sample performance of the proposed MMMD test across a range of simulation settings. Specifically, we will compare the MMMD test with the single kernel MMD test \citep{gretton2009fast} and the graph-based Friedman Rafsky (FR) test \citep{friedman1979multivariate}. Additional simulations are given in Appendix \ref{sec:experimentsadditional} of the supplementary materials.  Throughout we set the significance level $\alpha = 0.05$.

    For single kernel tests we use the Gaussian and Laplace kernels: 
    \begin{align}\label{eq:singlekernels}
        \sfK_{\mathrm{GAUSS}}(x,y) =  e^{-\frac{\|x-y\|^2}{\sigma^{2}}} \quad \text{and} \quad \sfK_{\mathrm{LAP}}(x,y) = e^{-\frac{\|x-y\|}{\sigma}},
    \end{align}
    with the bandwidth $\sigma$ is chosen using the median heuristic
    $$\sigma^2 :=  \lambda_{\text{med}}^2 = \text{median}\left\{\|Z_{i}-Z_{j}\|^{2}: 1 \leq i < j \leq n \right\},$$ 
    where $\sX_m \cup \sY_n = \{Z_1, Z_2, \ldots, Z_N\}$ is the pooled sample and 
    $\|\cdot\|$ denotes the Euclidean norm. We will refer to these tests as {\tt Gauss MMD} and {\tt LAP MMD}, respectively.
    
    For the MMMD statistic we will  use multiple Gaussian kernels, multiple Laplace kernels, or combination of Gaussian and Laplace kernels, with different bandwidths chosen follows: 
    \begin{itemize}
     
        \item {\tt Gauss MMMD}: This is the MMMD statistic with 5 Gaussian kernels with bandwidths 
        \begin{align}\label{eq:gaussmmmd}
        \bm \sigma = (\sigma_1, \sigma_2,  \sigma_3, \sigma_4, \sigma_5) =  (\tfrac{1}{2},\tfrac{1}{\sqrt{2}}, 1,\sqrt{2}, 2)\lambda_{\text{med}} . 
        \end{align}
      
        \item {\tt LAP MMMD}: This is the MMMD statistic with 5 Laplace kernels with bandwidths  
        \begin{align}\label{eq:laplacemmmd} 
        \bm \sigma = (\sigma_1, \sigma_2,  \sigma_3, \sigma_4, \sigma_5) =  (\tfrac{1}{2},\tfrac{1}{\sqrt{2}}, 1,\sqrt{2}, 2)\lambda_{\text{med}} . 
        \end{align}
    
        \item {\tt Mixed MMMD}: This is the MMMD statistic with 3 Gaussian kernels and 3 Laplace kernels with same set of bandwidths 
    \begin{align}\label{eq:mixedmmmd}
    \bm \sigma = (\sigma_1,\sigma_2,\sigma_3)  = (\tfrac{1}{\sqrt{2}}, 1, \sqrt{2})\lambda_{\text{med}} . 
    \end{align}
    \end{itemize}
    In our implementation we choose the regularity parameter $\lambda$ (recall Remark \ref{remark:sigmaH0}) as: 
    $\lambda = 10^{-5}\times \min_{1\leq a \leq r} \hat{\sigma}_{aa}$, for $\hat{\sigma}_{aa} > 0$ as in \eqref{eq:H0sigmaestimate}. Since $\lambda$ converges to $10^{-5}\times \min_{1\leq a \leq r} \sigma_{aa}$ almost surely (recall Corollary \ref{cor:TM}), the results in Corollary \ref{cor:consistency} remain valid. 
    The cutoffs of the tests are chosen based on the multiplier bootstrap as in \eqref{eq:Tmnalpha} using $B=500$ resamples.

    Finally, for the Friedman Rafsky (FR) test we use the implementation in the \texttt{R} package \texttt{gTests}, with the 5-MST (minimum spanning tree), which is the recommended practical choice in \citet{chen2017new}.

    \subsection{Dependence on Dimension} 
    \label{sec:dimension}
    
    In this section we study the performance of the different tests as dimension varies in the following settings. We fix sample sizes $m=n=100$, vary dimension over $d \in \{5, 10, 25, 50, 75, 100, 150\}$, and compute the empirical power by averaging over 500 iterations.  
    
    \begin{enumerate}[label=\textbf{(S\arabic*)}]
        \item\label{itm:S1} {\it Gaussian location-scale}: Here, we consider $P =  \cN_d(\bm{0} ,\bm{\Sigma}_{0})\text{ and } Q = \cN_d\left(0.1\bm{1}, 1.15\bm{\Sigma}_{0}\right)$, 
        where $\bm{\Sigma}_{0} = ((0.5^{|i-j|}))_{1\leq i,j\leq d}$ (see Figure \ref{fig:S12}(a)). 
        
        \item\label{itm:S2} {\it Gaussian and $t$-distribution mixture}: Here, we consider     
        $P =\tfrac{1}{2} \cN_d(\bm{0} ,\bm{\Sigma}_{0}) +\tfrac{1}{2} t_{10}(\bm{0} ,\bm{\Sigma}_{0})$ \text{ and }$Q =\tfrac{1}{2} \cN_d(\bm{0} ,1.22\bm{\Sigma}_{0}) +\tfrac{1}{2} t_{10}(\bm{0} ,1.22\bm{\Sigma}_{0})$, 
        where $\bm{\Sigma}_{0}$ is as above (see Figure \ref{fig:S12}(b)). 
    \end{enumerate} 
    
    \begin{figure}[h] 
    \centering
    \begin{subfigure}[c]{0.45\textwidth}
       \includegraphics[scale = 0.4]{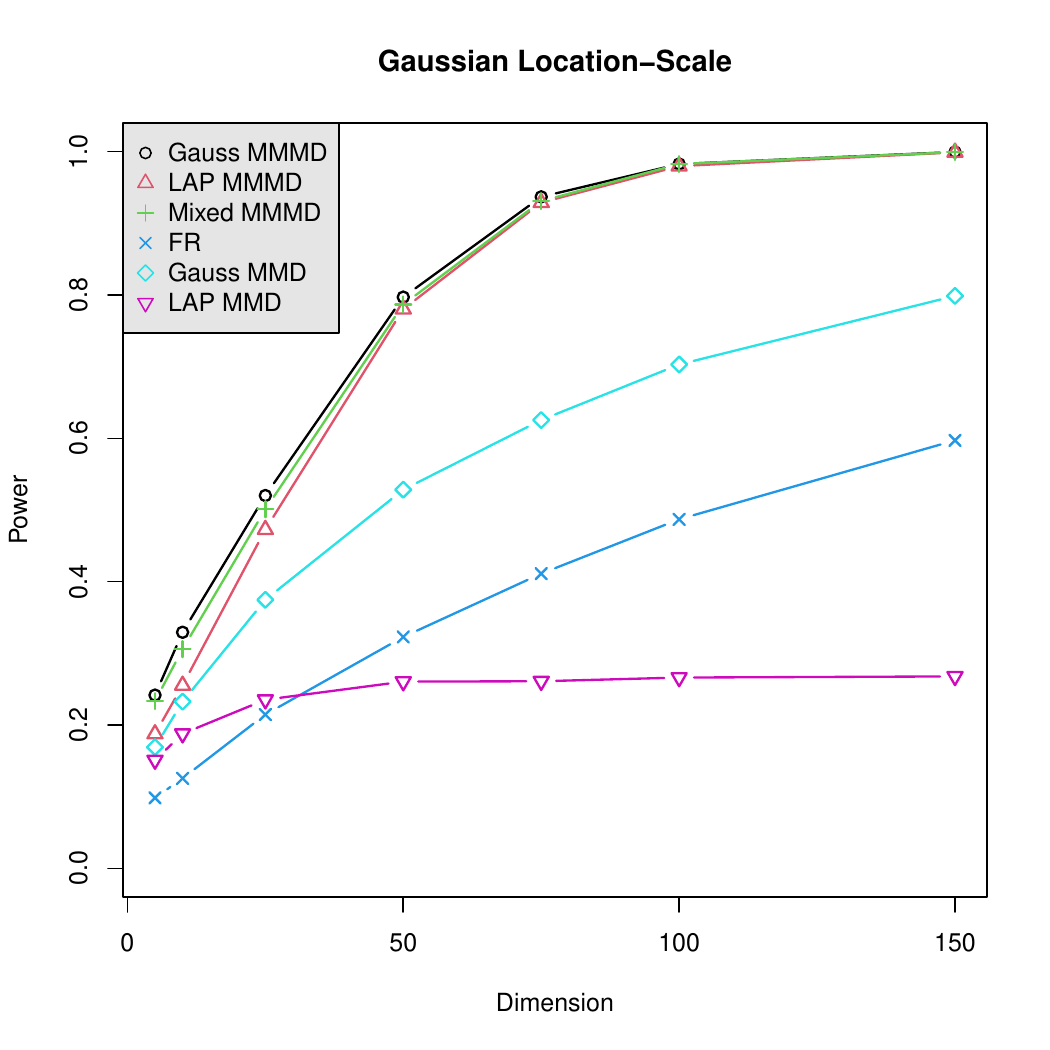} 
       \caption*{\small{(a)}}
    \end{subfigure}%
    \begin{subfigure}[c]{0.45\textwidth}
       \includegraphics[scale = 0.4]{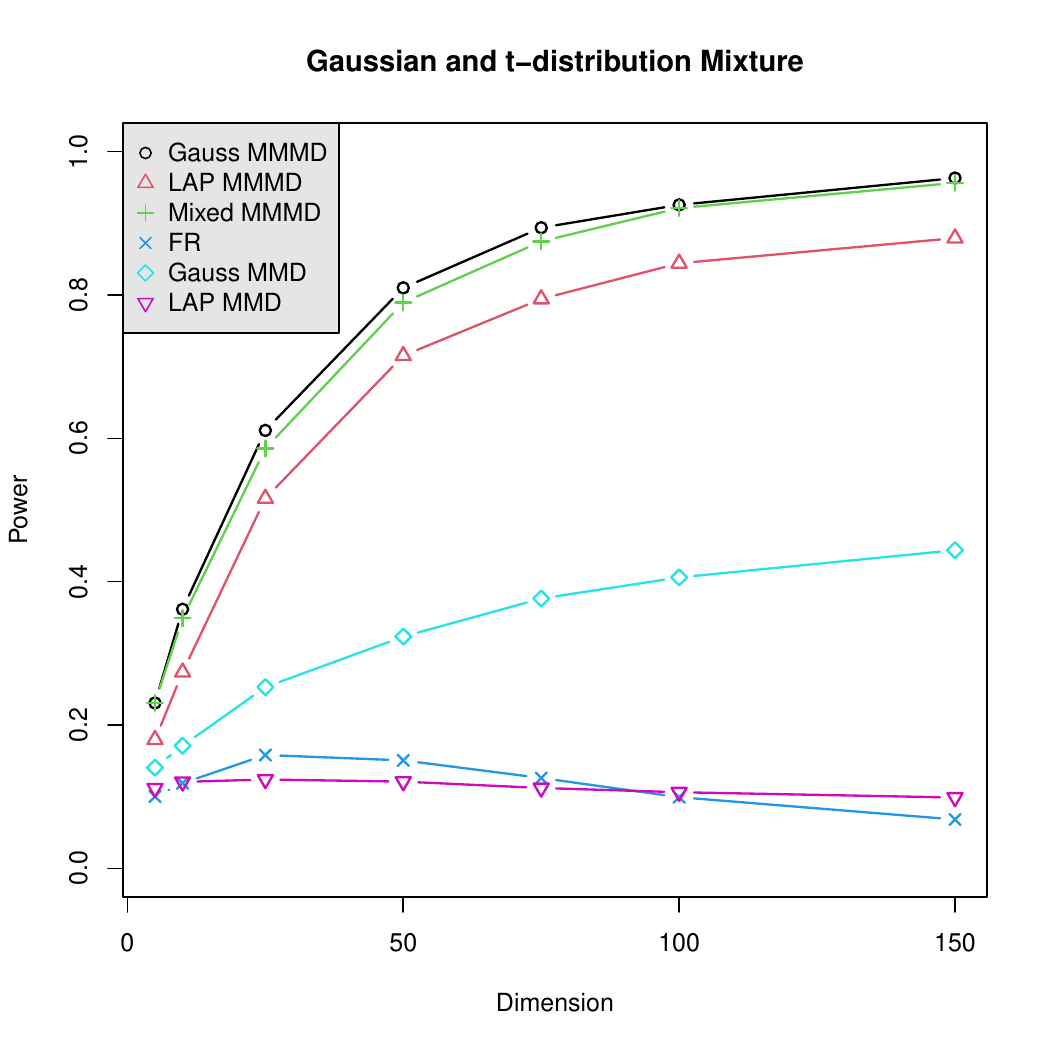} 
          \caption*{\small{(b)}} 
    \end{subfigure} 
    \caption{Empirical powers in (a) setting \ref{itm:S1} and (b) setting \ref{itm:S2}. }
    \label{fig:S12}
    \end{figure}

    The plots show that the multiple kernel MMMD tests have significantly more  power than the single kernel MMD tests and the FR test in both the 2 settings. Overall the {\tt Gauss MMMD} and the {\tt Mixed MMMD} tests perform the best, closely followed by the {\tt Lap MMMD}. This also shows the advantage of aggregating kernels across a range of dimensions, from low dimensions to dimensions that are comparable and even larger than the sample size. Additional simulations are provided in Appendix \ref{sec:dimensionadditional} of the supplementary material.

    \subsection{Mixture Alternatives} 
    \label{sec:mixtureexperiments}
    
    In this section we evaluate the performance of the tests for mixture alternatives by varying the mixing proportion. To this end, suppose $\bm{\Sigma}_{0} = ((0.5^{|i-j|}))_{1\leq i,j\leq d}$ and consider 
    \begin{align*}
    P = \varepsilon \cN_d(\bm{0} ,\bm{\Sigma}_{0}) + (1-\varepsilon)  t_{10}(\bm{0} ,\bm{\Sigma}_{0}) \text{ and } Q = \varepsilon \cN_d(\bm{0} ,1.25\bm{\Sigma}_{0}) + (1-\varepsilon) t_{10}(\bm{0} ,1.25\bm{\Sigma}_{0}). 
    \end{align*} 
    Figure \ref{fig:mixing} shows the empirical power (averaged over 500 iterations) of the different tests as $\varepsilon$ varies over $[0,1]$, with sample sizes $m=n=100$ and dimension $d=30$ (Figure \ref{fig:mixing}(a)) and $d=150$ (Figure \ref{fig:mixing} (b)). In both cases, the MMMD tests outperform the single kernel tests and the FR test, again illustrating the versatility of the aggregated tests.

    \begin{figure}[h]
    \centering
    \begin{subfigure}[c]{0.45\textwidth}
       \includegraphics[scale = 0.4]{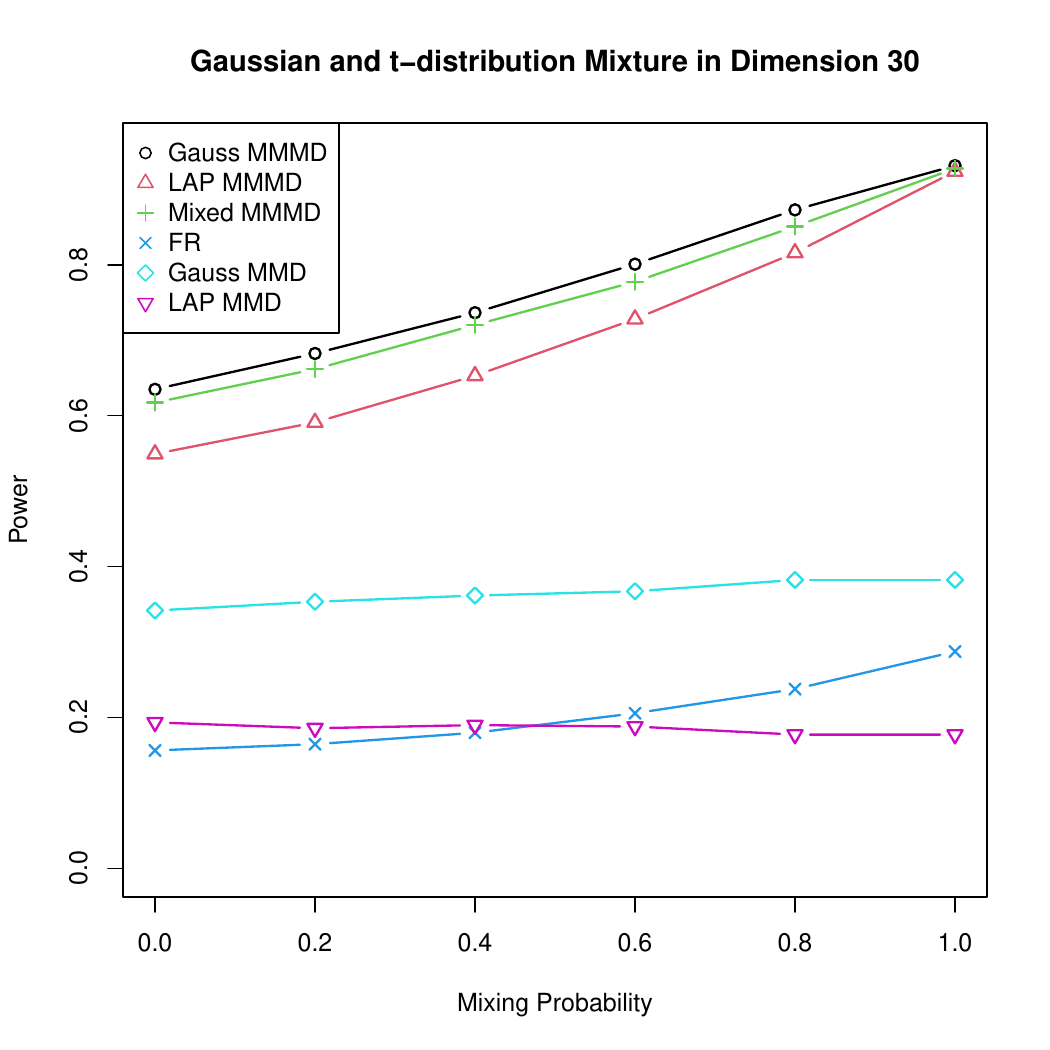} 
       \caption*{\small{(a)}}
    \end{subfigure}%
    \begin{subfigure}[c]{0.45\textwidth}
       \includegraphics[scale = 0.4]{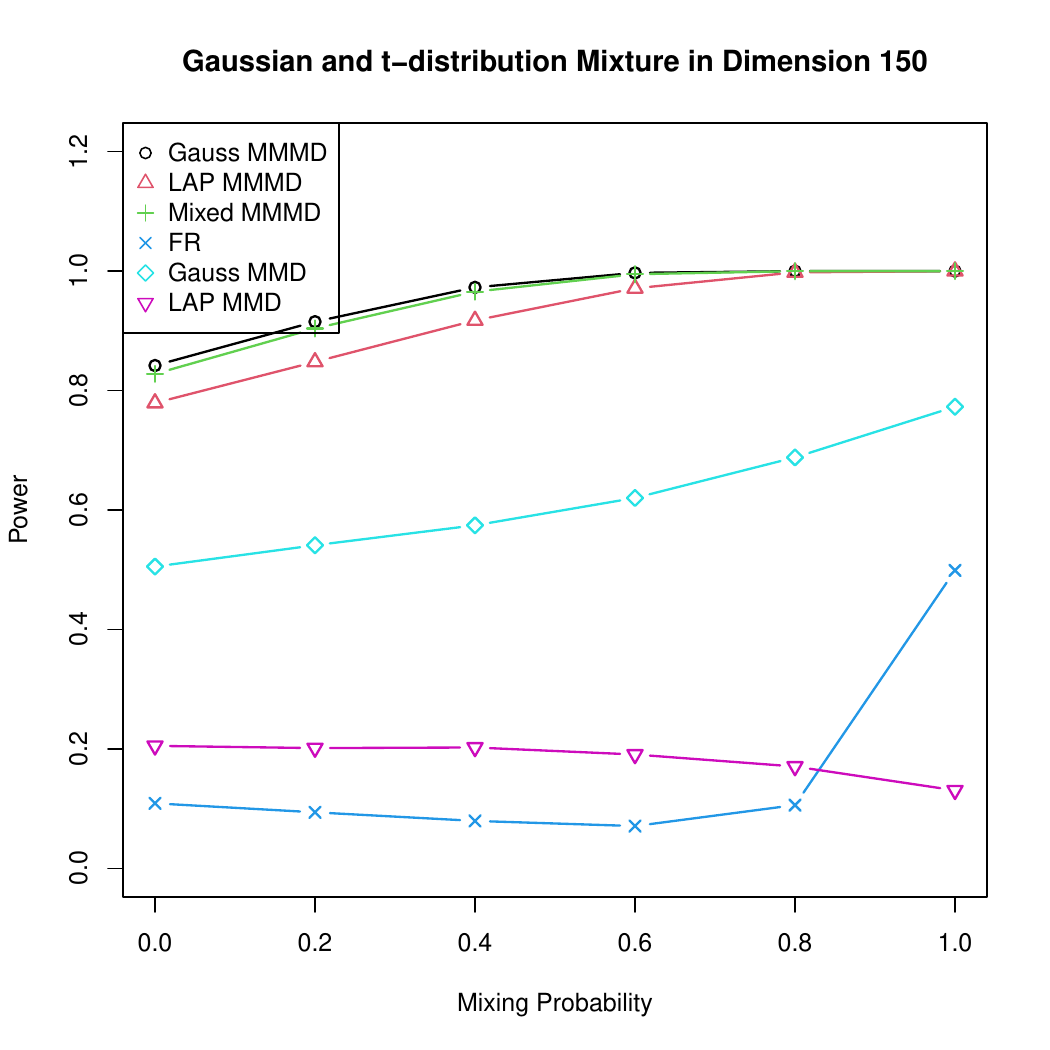} 
          \caption*{\small{(b)}}
    
    \end{subfigure}
    \caption{Empirical powers as a function of the mixing proportion for (a) $d=30$ and (b) $d=150$. }
    \label{fig:mixing}
    \end{figure}

    \subsection{Computational Complexity of the MMMD Test} 
    \label{sec:computation}
    
    \color{black}{ In Appendix \ref{sec:computationpf} we analyze the computational complexity of the MMMD test, when the rejection region is chosen based on $B$ replications of the statistic $\hat{T}_{m}$ from \eqref{eq:Testimate}. In particular, we show that the computational cost of MMMD test is $O(r^2N^2 + BrN^2 + B\log B)$ (assuming $r < N$). In practice, the number of resamples $B$ is usually chosen to be much larger than the number of kernels $r$, in which case the time complexity simplifies to $O(BrN^2 + B\log B)$. In fact, realistically one only aggregates over a bounded number of kernels, that is, $r = O(1)$, in which case the computational costs of the MMMD test and the MMD test differ only by a constant factor. In Appendix \ref{sec:computationpf}
    we also compare the running times of the MMD and the MMMD tests in simulations (see Table \ref{table:comptime}). Our experiments show that the MMMD tests provide significant power enhancement over the MMD test, with only a small increase in computation time. }

    \begin{remark}
    One way to reduce computation cost of the MMD test from $O(N^2)$ to $O(N)$ is the linear time MMD statistic \cite[Section 6]{gretton2012kernel}. 
    In Appendix \ref{sec:linear} we apply the Mahalanobis aggregation strategy to combine linear time statistics over multiple kernels and develop the associated theory. We also compare the power of the aggregated linear time MMD tests  with their single kernel counterparts and also with the quadratic time MMMD tests in simulations. 
    \end{remark}
    
    \color{black}
    
    \subsection{Comparison with the Permutation Test}
    \label{sec:permutation} 
     
     Another alternative to choosing the rejection threshold for $T_{m, n}$ is the permutation method. In fact, the permutation principle can be applied to calibrate any 2-sample test statistic based on the sample quantiles of the test statistic computed on $B$ permuted versions of the pooled data $\sX_m \cup \sY_n$. The resulting test is guaranteed to control Type I error in finite samples. In this paper we adopt the multiplier bootstrap over the permutation method for the following reasons: 
     
    \begin{itemize} 
    
    \item The independence of the Gaussian multipliers makes the asymptotic theory of the multiplier bootstrap method more tractable. Consequently, we are able to provide a holistic asymptotic theory for  the multiplier bootstrap based MMMD test, including limiting distributions (under both the null and the alternative), consistency, and local power analysis. 
     
    \item The multiplier bootstrap is computationally more efficient than the permutation method, both in terms of their asymptotic running times as well as power versus computation time trade-off in finite samples. Note that to obtain the permutation $p$-value we have to compute the MMMD statistic $T_{m,n}$ (recall \eqref{eq:TM}) for each of the $B$ random permutations of $N$ samples, where  $N = m+n$ is the total number of samples. Since with $r$ kernels it takes $O(r^2N^2)$ time to compute $T_{m,n}$ (see Appendix \ref{sec:computationpf}), the time complexity for the permutation test is $O(Br^2N^2)$, where $B$ is the number of permutations. On the other hand, we know from Section \ref{sec:computation} that the time complexity of the multiplier bootstrap based MMMD test \eqref{eq:Tmnalpha} is $O(r^2N^2 + BrN^2 + B\log B)$, which has a better dependence on $r$ than the permutation test. Even for fixed $r$ one can see significant gains in computation time in finite samples  (see Appendix \ref{sec:permutationsimulation}). In particular, our simulations show that the Type I error and power of the multiplier bootstrap and the permutation methods are comparable, but the computation time of the multiplier bootstrap method is much faster. 
    
    \end{itemize}

    \subsection{ Comparisons with Bandwidth Optimized MMD Tests and $p$-Value Combination Methods } 
    \label{sec:bandwithoptimizedcombination}

    Recall that in the previous sections while implementing the MMD test we chose the bandwidths for the Gaussian and Laplace kernels based on the median heuristic. Although this is the common choice in practice \citep{gretton2012kernel,ramdas2015decreasing}, it remains a heuristic because there is no theoretical understanding of its validity. 
    To address this issue, there has been studies that aim to find the ``best''  single kernel test by optimizing the bandwidth in such a way that the asymptotic power is maximized. This approach was first proposed by \cite{gretton2012optimal} for the linear time MMD test, which was subsequently extended to the quadratic time MMD test by \cite{sutherland2016generative}. The method involves splitting the data into 2 parts and using the first part to select the bandwidth by maximizing asymptotic power, or equivalently by maximizing the ratio (see \cite{liu2020learning}) $\frac{1}{\hat\sigma_{\lambda}^2} \mmd^2\left[\sfK_{\lambda},\sX_m,\sY_n\right] $ , 
    where $\hat\sigma_{\lambda}^2$ is a regularized estimator of the asymptotic variance of $\mmd^2\left[\sfK_{\lambda},\sX_m,\sY_n\right]$ under $H_1$. In Appendix \ref{sec:bandwithoptimized} we provide empirical comparisons of our test based on multiple kernels with the bandwidth optimized single kernel test in different simulations. 
    To mitigate the effect of data splitting we also implement the single kernel tests with twice the amount of data as in \citet[Section 5.3]{schrab2021mmd}. This emulates an oracle choice of bandwidth and represents the best single-kernel MMD test for the given data.
     In all the settings considered the MMMD tests have improved power than the bandwidth optimized single kernel tests. Also, MMMD tests with Gaussian/Laplace kernels have better power than the Gaussian/Laplace oracle MMD test (where bandwidth is optimized with double the sample size), respectively. It is worth noting that the bandwidths for the kernels in the MMMD tests are chosen as in \eqref{eq:gaussmmmd}, \eqref{eq:laplacemmmd}, and \eqref{eq:mixedmmmd}, respectively, which requires no optimization or data-splitting. Even so, the multiple kernel MMMD test is able to outperform the ``best'' single kernel, demonstrating the effectiveness of our aggregation scheme.

    Another possible aggregation strategy is to consider tests that combine $p$-values for multiple single kernel MMD tests. To illustrate how our aggregation strategy compares with $p$-value combination methods we consider the following experimental setup. We implement the {\tt Gauss MMMD} test based on 5 different Gaussian kernels with respective bandwidths $\bm \sigma = (\sigma_1, \sigma_2,  \sigma_3, \sigma_4, \sigma_5) =  ( \tfrac{1}{2}, \frac{1}{\sqrt 2}, 1,\sqrt{2}, 2)\lambda_{\text{med}}$, where $\lambda_{\text{med}}$ is defined after \eqref{eq:singlekernels}. The  {\tt Gauss MMMD} test  is calibrated using the multiplier bootstrap with $B = 500$ resamples. For comparison we consider the following $p$-value combinations. Here, $p_i$ denotes the $p$-value of the MMD test for a Gaussian kernel with bandwidth $\sigma_i$, for all $1\leq i\leq 5$, and the significance level $\alpha = 0.05$. 
    
    \begin{itemize}
        \item \textbf{Bonferonni:} Reject $H_0$ if $5\min_{1 \leq i \leq 5} p_i\leq \alpha$,
        \item \textbf{Harmonic Mean:} Reject $H_0$ if $2.214729\frac{5\log(5)}{\sum_{i=1}^5p_i^{-1}}\leq \alpha$, and
        \item \textbf{Bonferonni and Geometric Mean:} Reject $H_0$ if  $2\min \{ 5\min_{i=1}^5p_i, e\prod_{i=1}^5p_i^{1/5} \} \leq \alpha$.
    \end{itemize}
    The validity of the above $p$-value combinations follows from  \cite{vovk2020combining}. 
    
    The results of our experiments are given in Appendix \ref{sec:combination}. In all the simulation settings considered the {\tt Gauss MMMD} test emerges as the clear winner. This suggests that it is more advantageous to adopt our aggregation strategy over $p$-value combination methods for boosting the performance of kernel 2-sample tests.

    \color{black}

    \section{Real Data Applications}
    \label{sec:data}

    In this section we apply our method to compare images of digits in the noisy MNIST dataset.  Specifically, consider two noisy versions of the MNIST dataset: (1) MNIST with additive Gaussian noise (Section \ref{sec:mnistdatagaussian}), and (2) MNIST with reduced contrast and additive noise,  where, in addition to the Gaussian noise the contrast of the images is reduced (Section \ref{sec:mnistdatacontrast}).  As in the previous section, we implement the single kernel {\tt Gauss MMD} and {\tt LAP MMD} tests with the median bandwidth, the multiple kernel {\tt Gauss MMMD}, {\tt LAP MMMD}, and {\tt Mixed MMMD} tests with bandwidths as in \eqref{eq:gaussmmmd}, \eqref{eq:laplacemmmd}, \eqref{eq:mixedmmmd}, respectively, and the FR test using the R package {\tt gTests}.

    \subsection{MNIST with Additive Gaussian Noise} 
    \label{sec:mnistdatagaussian} 
    
    In this section we illustrate the performance of the proposed test in detecting different sets of digits when i.i.d. Gaussian noise with standard deviation $\sigma$ is added to each pixel. Figure \ref{fig:MNIST-addnoise} shows how such noisy data looks for (a) $\sigma = 0.6$ and (b) $\sigma=1$. 
    
    \begin{figure}[h]
      \begin{subfigure}[l]{0.28\textwidth}
        \includegraphics[width=\textwidth]{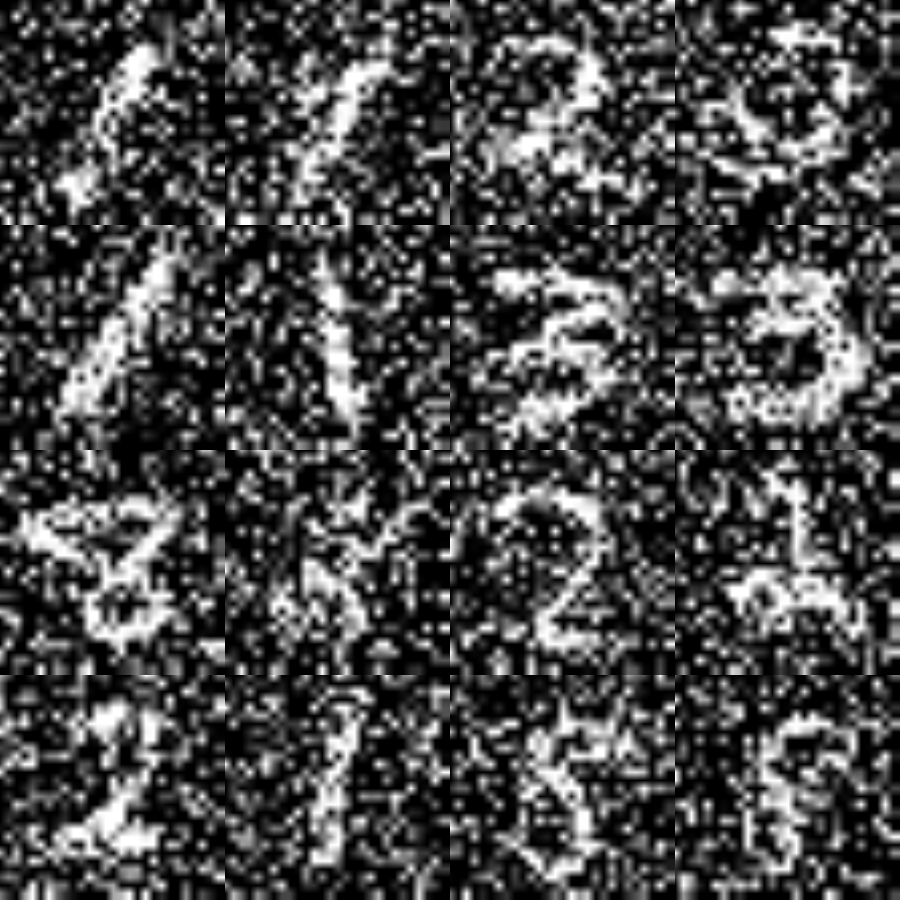}
        \caption*{\small{(a)}}
      \end{subfigure}
      \hfill
      \begin{subfigure}[l]{0.28\textwidth}
        \includegraphics[width=\textwidth]{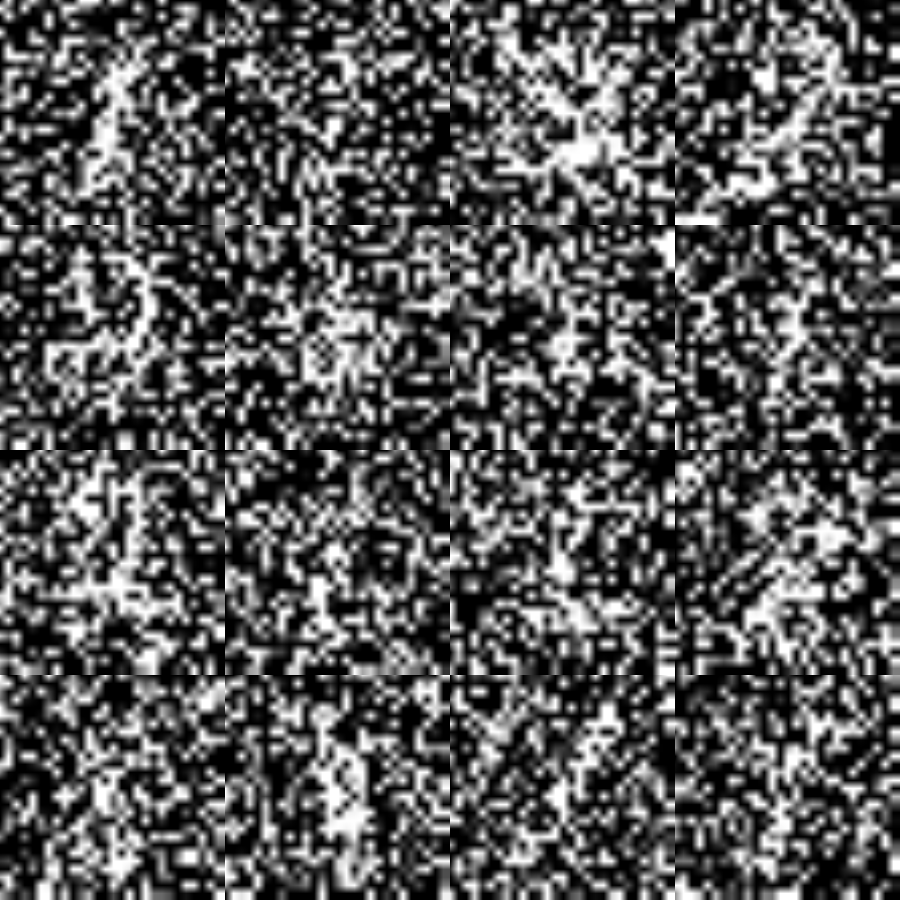}
        \caption*{\small{(b)}}
        
      \end{subfigure} 
      \begin{subfigure}[c]{0.36\textwidth}
       \includegraphics[width=\textwidth]{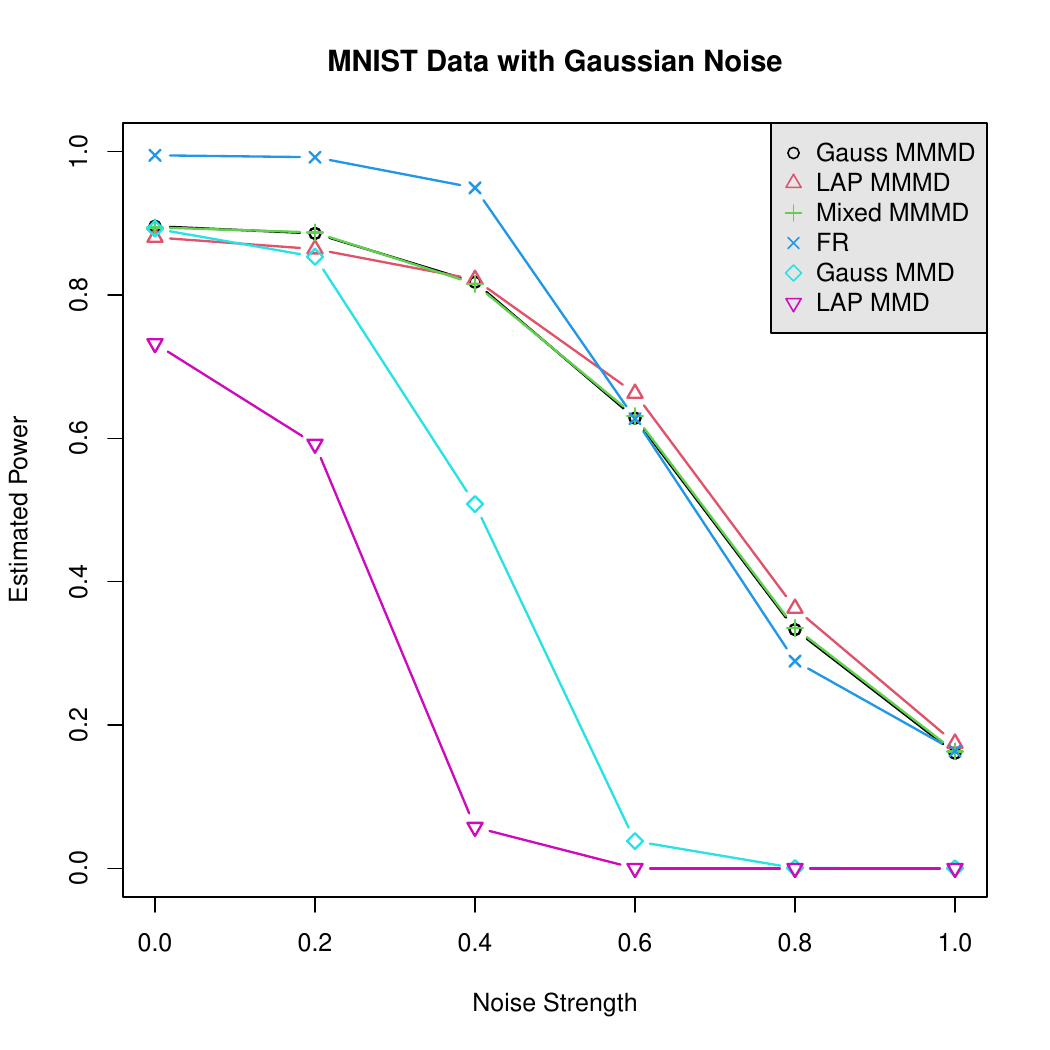}
       \caption*{\small{(c)}}
    \end{subfigure}
      \caption{MNIST data with additive Gaussian noise with (a) $\sigma=0.6$, (b) $\sigma=1$, and (c) estimated powers of the tests with increasing noise strength. } 
      \label{fig:MNIST-addnoise} 
    \end{figure} 
    
    To evaluate the proposed method we consider the following sets of digits: $P = \{ 1,2,3 \} \text{ and } Q = \{ 1,2,8 \}$,  
    and vary the standard error $\sigma\in (0,0.2,0.4,0.6,0.8,1)$. For each $\sigma$ we draw $100$ samples with replacements from the two sets and check if the tests successfully reject $H_0$ at level $\alpha=0.5$. We repeat this experiment $500$ times to estimate the power. Figure \ref{fig:MNIST-addnoise} (c) shows performance of the above mentioned tests, where we plot the power over the index of pair of sets of digits. This shows that for the clean data and for small noise levels, the singe kernel {\tt Gauss MMD} performs comparably with the MMMD tests. However, for larger noise levels the MMMD tests perform much than the single kernel tests. The FR test also perform well in this case across the range of the noise level.

    \subsection{MNIST with Reduced Contrast and Additive Gaussian Noise}
    \label{sec:mnistdatacontrast} 
    
    In this section we illustrate the performance of the different tests on the noisy version of the MNIST dataset considered in \citep{basu2017learning}. (The dataset is publicly available at \texttt{https://csc.lsu.edu/~saikat/n-mnist/}) Here, in  addition to additive Gaussian noise the contrast of the images is also reduced. Specifically, the contrast range is scaled down to half and an additive Gaussian noise is introduced with signal-to-noise ratio of 12. This emulates background clutter along with significant change in lighting conditions (see Figure \ref{fig:MNIST-AWGN} for an example of such a noisy image).
    
    \begin{figure}[!ht]
        \centering 
    \begin{subfigure}[c]{0.35\textwidth}
        \includegraphics[width = \textwidth]{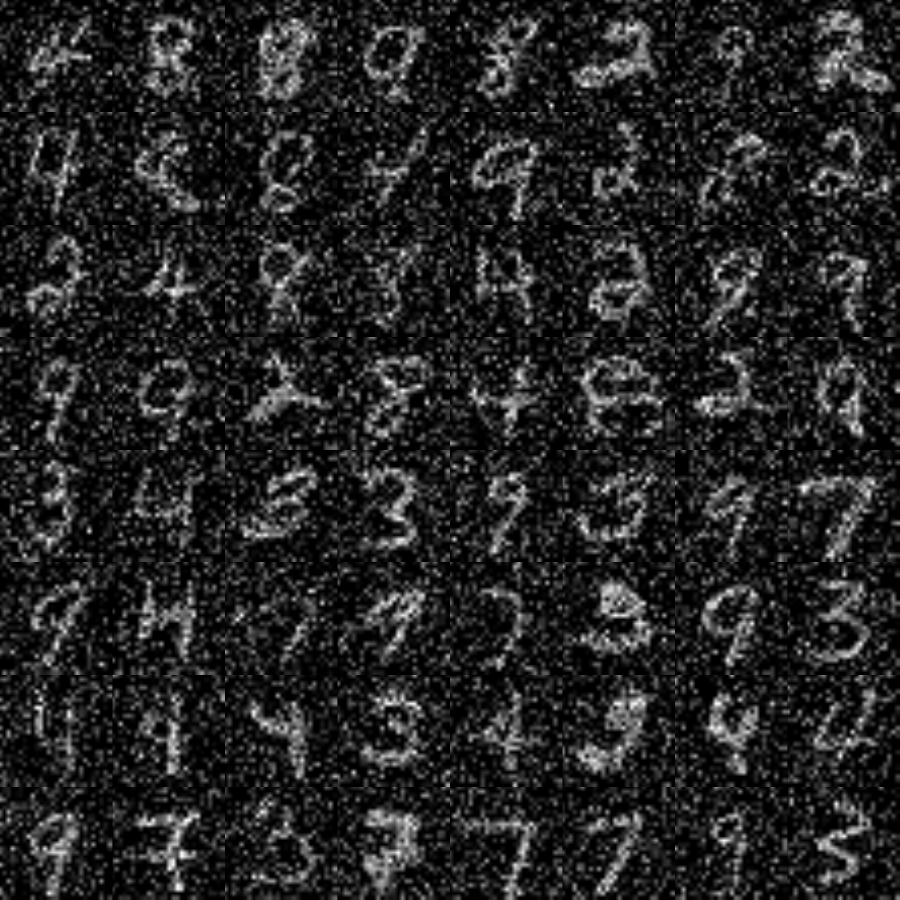}
         \caption*{\small{(a)}} 
    \end{subfigure}%
    \hspace{20pt}%
    \begin{subfigure}[c]{0.5\textwidth} 
       \includegraphics[width = 0.9\textwidth]{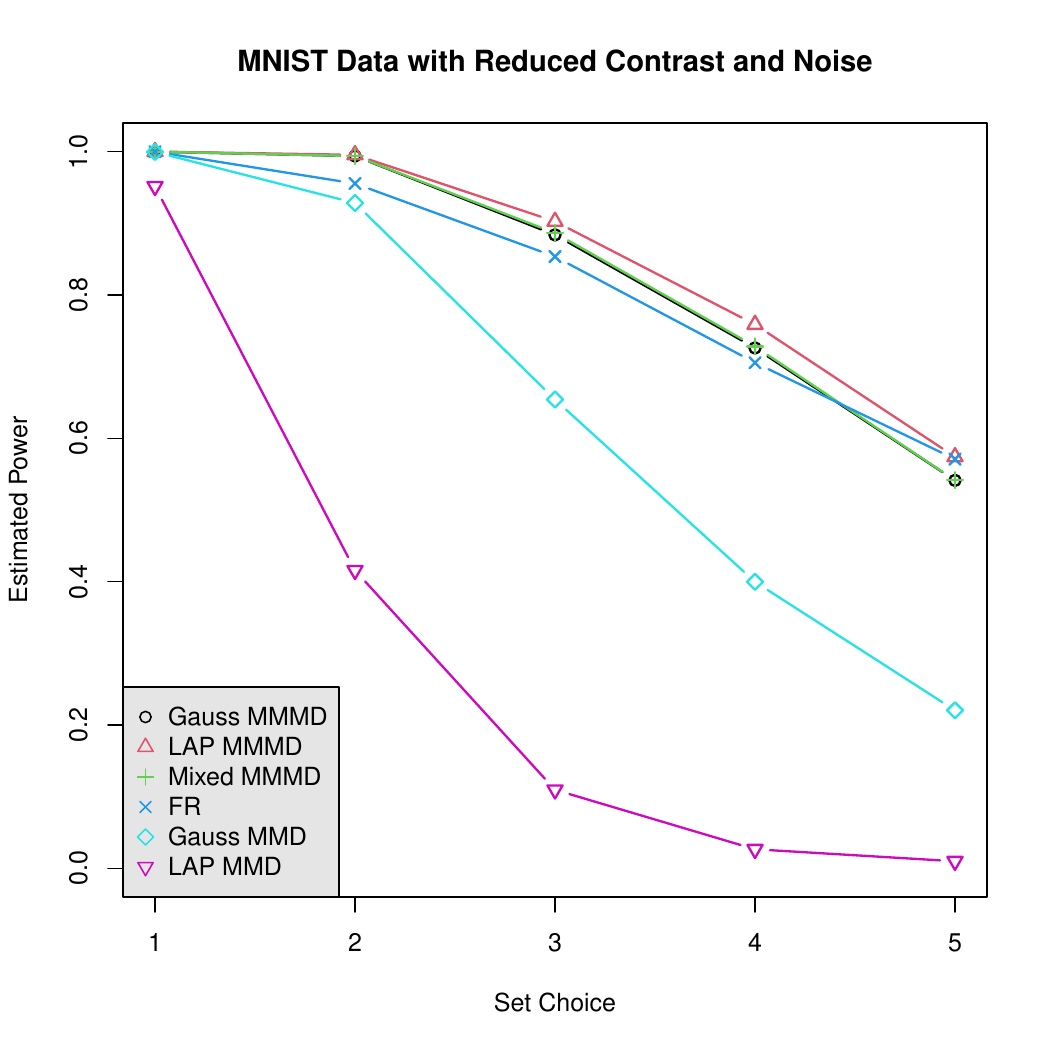} 
        \caption*{\small{(b)}} 
    \end{subfigure} 
        \caption{(a) MNIST dataset with reduced contrast and additive noise and (b) estimated power. } 
        \label{fig:MNIST-AWGN}
    \end{figure}
    
    We evaluate the performance of the different test for the following 5 pairs of sets of digits: (1) $P = \{ 2,4,8,9 \} $ and $Q = \{ 3,4,7,9 \}$; (2) $P = \{1,2,4,8,9\}$ and $Q = \{1,3,4,7,9\}$; (3) $P = \{0,1,2,4,8,9\}$ and $Q = \{ 0,1,3,4,7,9 \}$; (4) $P = \{0,1,2,4,5,8,9\}$ and $Q = \{0,1,3,4,5,7,9\}$; and (5) $P = \{0,1,2,4,5,6,8,9\}$ and $Q = \{0,1,3,4,5,6,7,9\}$. For each of the 5 cases above, we draw $150$ samples with replacements from the two sets and check if the tests successfully reject $H_0$ at level $\alpha=0.5$. We repeat this experiment $500$ times to estimate the power. 
    Figure \ref{fig:MNIST-AWGN} shows the power of the different for the above 5 sets. In this case, the multiple kernel tests and the FR test overall has the highest power across the 5 sets, followed by the {\tt Gauss MMD} and the {\tt Lap MMD}.

    \color{black}

    \section{Aggregation with Increasing Number of Kernels} 
    \label{sec:rkernels}
    
    In Section \ref{sec:multiplekernels} we showed that the Mahalanobis distance based aggregation of $r$ MMD statistics is consistent for the 2-sample problem, for any fixed $r \geq 1$. In this section we investigate whether consistency continues to hold when $r$ grows with the sample size. 
    To illustrate the strategy of combining multiple kernels more broadly, we also consider maximum and $L_2$ based aggregations with a growing number of kernels.

    \subsection{Maximum and $L_{2}$ Aggregation}\label{sec:MaxL2aggregate}

    There are many ways in which one can combine multiple kernels into a test statistic. For instance, we could consider maximum or $L_2$ based aggregations  as follows (assuming $m=n$ for simplicity): 
    \begin{align*}
        T_{m}^{\max}:= \max_{a=1}^{r_{m}}\emmd\left[\sfK_{a},\sX_{m},\sY_{m}\right]\text{ and }T_{m}^{L_{2}}:= \left\|\emmd\left[\cK_{r_{m}}, \sX_{m},\sY_{m}\right]\right\|_{2} , 
    \end{align*}
    where $r=r_m$ depends on $m$ and $\cK_{r_{m}} := \{\sfK_{a} : 1\leq a \leq r_{m}\}$. 
    The consistency and asymptotic distribution of these statistics when is $r$ fixed follows from results in Section \ref{sec:multiplekernels} and Theorem \ref{thm:K}, respectively. In the following proposition, using uniform convergence bounds for the MMD estimate, we construct tests based on $T_{m}^{\max}$ and  $T_{m}^{L_2}$ that are consistent in the growing $r$ regime. The proof is given in Appendix \ref{sec:concMaxL2pf}.

    \begin{proposition}\label{ppn:concMaxL2}
    Suppose $\cK_{r_{m}} = \left\{\sfK_{a}:1\leq a \leq r_{m}\right\}$ is a collection of $r_{m}$ distinct characteristic kernels such that $0\leq\sfK_{a}\leq K$, for all $1 \leq a \leq r_m$. Fix $\alpha \in (0, 1)$ and consider the following test functions: 
    \begin{align*}
        \phi_{m}^{\max} := \bm 1\left\{\left|T_{m}^{\max}\right|>C \sqrt{\frac{1}{m}\log\frac{6r_{m}}{\alpha}}\right\} \text{ and } 
        \phi_{m}^{L_{2}} := \bm{1}\left\{\left|T_{m}^{L_{2}}\right|>C \sqrt{\frac{r_{m}}{m}\log\frac{6r_{m}}{\alpha}}\right\} , 
    \end{align*} 
    where $C:= 8 K$. Then both $\phi_{m}^{\max}$ and $\phi_{m}^{L_{2}}$ have level $\alpha$ in finite samples. Moreover, $\phi_{m}^{\max}$ and $\phi_{m}^{L_{2}}$ are asymptotically consistent for \eqref{eq:H01PQ} if $\log r_{m} = o(m)$ and $r_{m}\log r_{m} = o(m)$, respectively.
    \end{proposition}

    \subsection{Mahalanobis Aggregation}

    In the growing $r$ regime the MMMD statistic takes the following form: 
    \begin{align}\label{eq:TmnMA}
        T_{m}^{MA}:= \left( \emmd\left[ \cK_{r_{m}}, \sX_m, \sY_m \right]  \right)^\top \bm \hat{\bm{\Sigma}}_{r_{m}}^{-1}  \left( \emmd\left[ \cK_{r_{m}}, \sX_m, \sY_m \right]  \right) , 
    \end{align}
    where $\hat{\bm{\Sigma}}_{r_{m}}$ has entries defined in \eqref{eq:H0sigmaestimate}.

    \begin{theorem}\label{thm:MahaAggrmntest} 
    Suppose the assumptions of Proposition \ref{ppn:concMaxL2} hold. 
    Then the test $$\phi_{m}^{MA} := \bm{1}\left\{\left|T_{m}^{MA}\right|> \frac{64K^2}{\sqrt m} \right\}$$ is asymptotically consistent whenever 
    $\lim_{m \rightarrow \infty}\inf_{1 \leq a \leq r_m} \left\{\emmd\left[\cF_{a}, P, Q\right] \right\}>0$ and $r_{m} \log r_{m} = o\left(\lambda_{m} \sqrt{m}\right)$, where $\lambda_{m}$ is the smallest eigenvalue of $\bm{\Sigma}_{r_m}$. 
    \end{theorem}

    The proof of Theorem \ref{thm:MahaAggrmntest} is given in Appendix \ref{sec:proofMahaAggrmn}. Essentially, the result shows that $T_{m}^{MA}$ leads to a consistent test for $r_m = o(\lambda_m \sqrt m)$, ignoring logarithmic factors. In comparison, for the maximum aggregation one can have $r_m$ grow up to sub-exponentially in $m$ and the $L_2$ aggregation allows any sub-linear growth for $r_m$ (recall Proposition \ref{ppn:concMaxL2}). One of the challenges in dealing with MMMD statistic in the growing $r_m$ regime is that in addition to the concentration of the vector of MMD statistic one has to ensure the concentration of $\bm{\Sigma}_{r_{m}}$, which necessitates a more stringent requirement on $r_m$ (in comparison to the $L_2$ aggregation) to guarantee consistency. 
    Towards this it is expected that the lowest eigenvalue of the population covariance matrix $\bm \Sigma_{r_m}$ plays a role in how large $r_m$ can be.

    Improving the dependence on $r$ in the above results and investigating the behavior of the various aggregation strategies when $r$ is comparable or even larger than $m$ are important future directions. However, from a practical standpoint one needs to exercise caution while selecting $r$. Although the aggregated tests remain consistent under appropriate growth conditions on $r$, in finite samples the power of the tests saturate and the tests also become conservative when $r$ is large. 
    This latter issue, which is already apparent for the single kernel test when the cutoff is chosen based on concentration inequalities (see \citet[Section 4]{gretton2012kernel}), become more significant when $r$ grows with $m$. Moreover, the computation of the $\hat{\bm{\Sigma}}^{-1}$ becomes less stable when $r$ becomes too large. In practice, as we see in the simulations, there is already significant improvement in power over single kernel tests just by aggregating over a few (up to 5) kernels. Further exploring the interplay between the choice of $r$, Type I error, and power is an interesting future direction. 
    
    \color{black} 
    
    \section{Broader Scope I: Local Power of Adaptive MMD Tests} 
    \label{sec:broaderscope} 
    
    The idea of using multiple kernels/bandwidths has recently emerged as a popular alternative to selecting a single bandwidth, for developing adaptive kernel two-sample tests that do not require data-splitting. 
    In this direction, \citet{kubler2020learning} proposed a method which does not require data splitting using the framework of post-selection inference.  
    However, this method requires asymptotic normality of the test statistic under $H_0$, hence, is restricted to the linear-time MMD estimate \citet[Section 6]{gretton2012kernel}, which leads to loss in power when compared to the more commonly used quadratic-time estimate \eqref{eq:WX}. \citet{fromont2013two,fromont2012kernels} and, more recently, \citet{schrab2021mmd} introduced another non-asymptotic aggregated test, hereafter referred to as MMDAgg, that is adaptive minimax (up to an iterated logarithmic term) over Sobolev balls.

    Our aggregation strategy, leads to a test that can be efficiently implemented, enjoys improved empirical power over single kernel tests for a range of alternatives, and scales well in high dimensions. Moreover, instead of minimax optimality, our focus is on establishing the asymptotic properties of the aggregated test. Towards this, we derive the joint distribution of the MMD estimates (under both local and fixed alternatives) and, consequently, establish the statistical (Pitman) efficiency of the proposed test. In fact, our  theoretical results apply to general aggregation schemes using which we can obtain the asymptotic local power of the aforementioned MMDAgg test. 
    To demonstrate this in this section we propose an asymptotic implementation of the MMDAgg test and sketch a heuristic argument that derives its limiting local power in the contamination model \eqref{eq:H0N}. The argument can be made rigorous by using tools from empirical process theory, however, since the purpose of this section is more illustrative than technical, we have not pursued this direction.

    To describe the asymptotic version of the MMDAgg test suppose $\cK=\{\sfK_1, \sfK_2, \ldots, \sfK_r\}$ is a finite collection of kernels and $\cW:=\{w_1, w_2, \ldots, w_r\}$ is an associated collection of positive weights such that $\sum_{s=1}^r w_s \leq 1$.  Moreover, for $\alpha \in (0, 1)$ and $1 \leq s \leq r$, let $\hat q_{1-\alpha, s, m}$ be the $\alpha$-th quantile of the distribution 
    \begin{align*}
         \sE(\sfK_s, \sX_m) :=  \bm{Z}_{m}^\top \sfKmatrix_s  \bm{Z}_{m} -  \frac{1}{\hat{\rho}(1-\hat{\rho})}\Tr[ \sfKmatrix_s ] . 
            \end{align*} 
    where $\sfKmatrix_s$ is as defined in \eqref{eq:centered-kernel}, for $1 \leq a \leq r$, and $\bm{Z}_{m} \sim \cN_m(\bm 0, \frac{1}{\hat{\rho}(1-\hat{\rho})} \bm{I} )$ is independent of $\sX_m$. The idea of the MMDAgg test is to reject $H_0$ if any one of the individual (single-kernel) test based on the kernels in $\cK$ rejects $H_0$ for a specially chosen cut-off (see \citet[Section 3.5]{schrab2021mmd} for details). Here, we consider an alternative implementation of MMDAgg test based on the Gaussian multiplier bootstrap discussed in Section \ref{sec:H0implementation}. To this end, define 
    $$u^*_{\alpha, m} := \mathrm{arg} \max\left\{ u \in (0, L): \P\left(\max_{1 \leq s \leq r} \left\{  \sE(\sfK_s, \sX_m)  - \hat q_{1-u w_s, s, m} \right\}  > 0 | \sX_m \right) \leq \alpha \right\} , $$  
    where $L := \min_{1 \leq s \leq r} w_s^{-1}$. (Note that the probability in the RHS above is over the randomness of $\bm Z_m$ (conditional on $\sX_m$), hence, $u^*_{\alpha, m}$ can be computed from the data by a grid search over $u \in (0, L)$.) The MMDAgg test would then reject $H_0$ if 
    \begin{align}\label{eq:phimnMMD}
    \phi_{m, n, \alpha}^{\mathrm{MMDAgg}} := \bm 1 \left\{  \max_{1 \leq s \leq r} \left\{ \mathrm{MMD}^{2}\left[ \sfK_s, \sX_m, \sY_n \right]  - \hat q_{1- w_s u^*_\alpha, s, m} \right\}  > 0 \right \} . 
    \end{align}
    To describe the asymptotic properties of this test, let $q_{\alpha, s}$ be the $\alpha$-th quantile of the distribution $\frac{1}{\rho(1-\rho)} I_2(\sfK^\circ_{s})$, for $1 \leq s \leq r$. Then  for each fixed $u \in (0, L)$, by Theorem \ref{thm:K}, Slutsky's theorem, and the continuous mapping theorem, as $m \rightarrow \infty$, 
    \begin{align*} 
    \max_{1 \leq s \leq r} \left\{  \sE(\sfK_s, \sX_m)  - \hat q_{1-u w_s, s, m} \right\} \dto \max_{1 \leq s \leq r} \left\{  \frac{1}{\rho(1-\rho)} I_2(\sfK^\circ_{s})  -  q_{1-u w_s, s} \right\} . 
    \end{align*} 
    since $\hat q_{1-u w_s, s, m}|\sX_m \stackrel{a.s.} \rightarrow q_{1-u w_s, s}$. Therefore, for each fixed $u \in (0, L)$, as $m \rightarrow \infty$, 
     $$ \P\left(\max_{1 \leq s \leq r} \left\{  \sE(\sfK_s, \sX_m)  - \hat q_{1-u w_s, s, m} \right\}  > 0 | \sX_m \right)  \rightarrow \P\left( \max_{1 \leq s \leq r} \left\{  \frac{1}{\rho(1-\rho)} I_2(\sfK^\circ_{s})  -  q_{1-u w_s, s} \right\} > 0 \right). $$ Now, since the convergence of the quantiles is uniform, we expect the following to hold as $m \rightarrow \infty$: $u^*_{\alpha, m} \stackrel{a.s.}\rightarrow u^*_{\alpha}$ and $\hat q_{1- w_s u^*_{\alpha, m} , s, m} | \sX_m \stackrel{a.s.}\rightarrow q_{1- w_s u^*_\alpha , s}$, 
    where 
    $$u^*_{\alpha} := \mathrm{arg} \max\left\{ u \in (0, L):  \P\left( \max_{1 \leq s \leq r} \left\{  \frac{1}{\rho(1-\rho)} I_2(\sfK^\circ_{s})  -  q_{1-u w_s, s} \right\} > 0 \right) \leq \alpha \right\} . $$ 
    Hence, under $H_1$ as in \eqref{eq:H0N}, by Theorem \ref{thm:H0NK}, Slutsky's theorem, and the continuous mapping theorem, 
    $$ \max_{1 \leq s \leq r} \left\{ \mathrm{MMD}^{2}\left[ \sfK_s, \sX_m, \sY_n \right]  - \hat q_{1-w_s u^*_\alpha , s, m} \right\}  \dto \max_{1 \leq s \leq r} \left\{  G_{\sfK_s, h} - q_{1-w_s u^*_\alpha , s} \right\} , $$ 
    where $G_{\sfK_s, h} := \gamma I_2(\sfK^\circ_s) + 2h \sqrt{\gamma} I_1 (\sfK^{\circ}_s [\frac{g}{f_P}]) + h^2 \mu_s$,  
    and $\mu_s$ is as defined in \eqref{eq:meanK}. 
    Therefore, the limiting power of the test \eqref{eq:phimnMMD} 
    is given by 
    \begin{align*} 
    \lim_{m, n \rightarrow \infty} \E_{H_1}[\phi_{m, n, \alpha}^{\mathrm{MMDAgg}} ] & = \P\left(\max_{1 \leq s \leq r} \left\{  G_{\sfK_s, h} - q_{1-w_s u^*_\alpha , s} \right\} > 0 \right) \nonumber \\ 
    & = 1 - \bm F_{\cK, h}(q_{1-w_1 u^*_\alpha , 1},  \ldots, q_{1-w_r u^*_\alpha , r}) , 
    \end{align*}
     where $\bm F_{\cK, h}$ is the cumulative distribution function of the vector $(G_{\sfK_1, h}, G_{\sfK_2, h}, \ldots, G_{\sfK_r, h})^\top$.

    Numerical results comparing the empirical power of the MMMD test with the MMDAgg test are reported in Appendix \ref{sec:adpativeMMDexperiments} in the supplementary materials.  The experiments show that the MMMD has better power than MMDAgg for a range of alternatives, which include perturbed uniform distributions in the Sobolev class, as well as mixture and local alternatives. This showcases both the practical relevance of the Mahalanobis aggregation strategy and the broader scope of our asymptotic results.

    \color{black}
    \section{Broader Scope II: Aggregation with Diverging Bandwidths}
    \label{sec:kernelbandwidth}
    
    In the previous sections we established the universal consistency and derived  the asymptotic null distribution of the MMMD test, for kernels with fixed bandwidths. However, in practice bandwidths are often chosen in data-driven manner which depend on the sample size $N$. For instance, to obtain tests which are optimal (in detecting smooth departures from the null hypothesis), the scaling parameter $\lambda := \frac{1}{\sigma^2}$ has to diverge with the sample size (see \cite{kernelnonparametricsmoothalternatives,schrab2021mmd}). For such choices of the scaling parameter, the test statistic has an asymptotically normal distribution under $H_0$, hence the rejection threshold can be readily obtained without any permutation/bootstrap resampling. Combining this idea with our aggregation strategy, in this section we construct a new test that combines multiple Gaussian kernels, with appropriately chosen diverging scaling parameters, that has a multivariate normal distribution under $H_0$. 
    
    For $\lambda > 0$, let $\sfK_{\lambda}(x,y):= e^{-\lambda\|x-y\|^2}$ 
    be the Gaussian kernel with scaling parameter $\lambda$. For $r\geq 1$, consider the collection of kernels, $\cK_{\bm \nu}:=\left\{\sfK_{\nu_{a}}:1\leq a\leq r\right\}$, where $\bm \nu = (\nu_1, \nu_2, \ldots, \nu_r)$ is a set of scaling paramters which can possibly depend on $N$. Throughout this section we make the following assumptions:
    
    \begin{assumption}\label{assmp:relate}
    There exists $\{\eta_{s}> 0:1\leq s\leq r\}$ such that $\nu_{s} = \eta_{s}\lambda_{N}$, for all $1\leq s\leq r$, where $\lambda_{N} = o\left(N^{4/d}\right)$ such that $\lambda_{N}\ra\infty$, in the asymptotic regime \eqref{eq:mn}.   
     \end{assumption}

    \begin{assumption}\label{assmp:PQdensitypq}
    Suppose $\mathcal{X} = \mathbb{R}^{d}$, for $d\geq 1$, and the distribution $P$ has a density $f_{P}\in L_{2}(\R^d)$ with respect to the Lebesgue measure on $\mathbb{R}^{d}$.
    \end{assumption}
    
    Under these assumptions we have the following theorem (see Appendix \ref{sec:H0KNpf} for the proof).

    \begin{theorem}\label{thm:asymkernelonN}
     Suppose the collection of kernel $\cK_{\bm \nu}$ satisfy Assumption \ref{assmp:relate} and Assumption \ref{assmp:PQdensitypq} holds. Then under $H_{0}$ in the asymptotic regime \eqref{eq:mn},
    \begin{align*}
        \frac{mn}{\sqrt{2}(m+n)}\lambda_{N}^{d/4}\emmd\left[\cK_{\bm \nu},\sX_{m},\sY_{n}\right]\dto \mathcal N_{r}\left(\bm{0}, \Gamma \right) , 
    \end{align*}
    where $\Gamma =  (\gamma_{ab})_{1 \leq a, b \leq r}$ is a $r\times r$ matrix with entries $\gamma_{ab} = \frac{\pi^{d/2}\|f_{P}\|_{2}^2}{\left(\eta_{a} + \eta_{b}\right)^{d/2}}$, for $1\leq a,b\leq r$. 
    \end{theorem} 
    
    In Lemma \ref{lemma:estnormfP} (in Appendix \ref{sec:H0KNpf}) we provide a consistent estimate $\|\hat f_{P}\|_{2}^2$ of $\|f_{P}\|_{2}^2$ as in Theorem 4 in \cite{kernelnonparametricsmoothalternatives}. Combining Theorem \ref{thm:asymkernelonN} and Lemma \ref{lemma:estnormfP} gives the following result.

    \begin{corollary}\label{corollary:KNH0} 
     Suppose the conditions of Theorem \ref{thm:asymkernelonN} hold.  
        Then under $H_{0}$, 
        in the asymptotic regime \eqref{eq:mn},
        \begin{align*}
            {\bm V}_{m, n} := \frac{mn}{\sqrt{2}(m+n)} \lambda_{N}^{d/4}\|\hat f_{P}\|_{2}^{-1} \emmd\left[\cK_{\bm \nu},\sX_{m},\sY_{n}\right]\dto \mathcal N_{r}(\bm{0}, \tilde{\bm \Gamma} ) , 
        \end{align*} 
        where $\|\hat f_{P}\|_{2}$ is defined in Lemma \ref{lemma:estnormfP} and $\tilde{\bm \Gamma} = (\tilde{\gamma}_{ab})_{1 \leq a, b \leq r}$ is a $r\times r$ matrix with entries $\tilde{\gamma}_{ab} = (\frac{\pi}{\eta_{a} + \eta_{b}})^{d/2}$, for $1\leq a,b\leq r$. Consequently, $\{\bm V_{m, n}^\top \tilde{\bm \Gamma}^{-1} \bm V_{m, n} > \chi^2_{r, 1-\alpha}\}$ is an asymptotically level $\alpha$ test. 
    \end{corollary}
    
    Note that the test above has a tractable chi-square distribution under $H_0$, hence, its rejection region can be readily obtained without any bootstrap resampling (unlike the general test with fixed bandwidths discussed in Section \ref{sec:H0implementation}). Furthermore, the test in Corollary \ref{corollary:KNH0} will be optimal in detecting certain smooth alternatives for an appropriately chosen bandwidth (depending on the smoothness parameter), similar to the single kernel test  (see Section 3 in \cite{kernelnonparametricsmoothalternatives}). Moreover, we expect the test in 
    Corollary \ref{corollary:KNH0} to have better power than its single sample counterpart in finite samples, for specific types of smooth alternatives. 
    \\

    \color{black}

    \small{
    \noindent\textbf{Acknowledgements.} The authors are grateful to the Editor, the Associate Editor, and the anonymous referees for their insightful comments which led to several new results and greatly improved the quality and the presentation of the paper. B. B. Bhattacharya was supported by NSF CAREER grant DMS 2046393 and a Sloan Research Fellowship. }

    \bibliographystyle{abbrvnat}
    \bibliography{kernelbib}
    
    \appendix 

\normalsize 

\section*{Supplement to ``Boosting the Power of Kernel Two Sample Tests''}

This supplementary material is organized as follows: We recall the definition and properties of multiple Weiner-It\^{o} stochastic integrals in Appendix \ref{sec:stochasticintegral}. The proofs of Theorem \ref{thm:K} and Corollary \ref{cor:TM} are given in Appendix \ref{sec:H0pf}. Theorem \ref{thm:estimateH0} and Theorem \ref{thm:H0NK} are proved in Appendix \ref{sec:estimateH0pf} and Appendix \ref{sec:localpowerpf}, respectively. In Appendix \ref{sec:H1asymptotic} we derive the joint distribution of the MMD estimates under the alternative. In Appendix \ref{sec:computationpf} we analyze the computational complexity, both theoretically and in simulations. Additional simulations comparing the MMMD tests with the single kernel tests are given in Appendix \ref{sec:experimentsadditional}. The Mahalanobis aggregation strategy is implemented for the linear time statistic in Appendix \ref{sec:linear}.  The proofs of Proposition \ref{ppn:concMaxL2} and Theorem \ref{thm:MahaAggrmntest} are given in Appendix \ref{sec:maxL2MMMMDpf}. Simulations comparing the MMMD test and the MMDAgg test are presented in Appendix \ref{sec:adpativeMMDexperiments}. The proof of Theorem \ref{thm:asymkernelonN} is given in Appendix \ref{sec:H0KNpf}. The invertibility of  $\bm{\Sigma}_{H_{0}}$ is discussed in Appendix \ref{sec:Sigma}. We collect the proofs of various technical lemmas in Appendix \ref{sec:lmpf}.

\section{Multiple Weiner-It\^{o} Stochastic Integrals} 
\label{sec:stochasticintegral}

In this section, we recall the definition and basic properties of multiple Weiner-It\^{o} stochastic integrals as presented in  \citet{stochasticintegral}. Suppose $\mathcal X$ is a separable metric space, $\mathscr B(\mathcal X)$ is the sigma-algebra generated by the open sets of $\mathcal X$, and $P$ is a non-atomic probability measure on $\mathcal X$. We denote this probability space by $(\cX, \sB(\cX), P)$. 

\begin{definition}
	A \textit{Gaussian stochastic measure} on $(\mathcal X, \sB(\mathcal X), P)$ is a collection of random variables $\{\mathcal{Z}_P(A): A\in \sB(\mathcal X)\}$ defined on a common probability space $(\Omega, \cF, \mu)$ such that  the following hold: 
	\begin{itemize}
	\item $\mathcal{Z}_P(A)\sim \cN(0, P(A))$, for all $A\in \sB(\mathcal X)$. 
	\item For any finite collection of disjoint sets $A_{1},\ldots, A_{t} \in \mathscr B(\mathcal X)$, the random variables $$\{\mathcal{Z}_P(A_1), \mathcal{Z}_P(A_2), \ldots, \mathcal Z_P(A_t)\}$$ are independent and 
\begin{align*}
\mathcal{Z}_P\left(\bigcup_{s=1}^{t}A_{s}\right) = \sum_{s=1}^{t}\mathcal{Z}_P(A_{s}).
\end{align*}
\end{itemize} 
\end{definition}

For $d \geq 1$, denote by $L^2(\cX^d, \sB(\cX^d), P^d)$ the space of measurable functions $f: \mathcal X^d \rightarrow \R$ such that 
$$\|f\|^2:= \int_{\cX^d} |f(x_1, x_2, \ldots, x_d)|^2 \mathrm d P(x_1) \mathrm d P(x_2) \ldots, \mathrm d P(x_d) < \infty.$$ 
Define $\cE_d \subseteq L^2(\cX^d, \sB(\cX^d), P^d)$ as the set of all elementary functions having the form
\begin{equation}\label{eq:itelty}
f(t_1, t_2, \ldots, t_d) =\sum_{1 \leq i_1, i_2, \ldots, i_d \leq m} a_{i_1, i_2, \ldots, i_d} \bm 1\{ (t_1, t_2, \ldots,  t_d) \in A_{i_1}\times \cdots \times A_{i_d} \}, 
\end{equation}
where $A_1, A_2, \ldots, A_m \in \sB(\cX)$ are pairwise disjoint and $ a_{i_1, i_2, \ldots, i_d}$ is zero if two indices are equal. 
The multiple Weiner-It\^{o} integral for functions in $\cE_d$ is defined as follows: 

\begin{definition}\label{defn:integral_elementary} (Multiple Weiner-It\^{o} integral for elementary functions) The $d$-dimensional Weiner-It\^{o} stochastic integral, with respect to the Gaussian stochastic measure $\{Z_P(A), A \in \sB(\cX)\}$, for the function $f \in \cE_d $ in  \eqref{eq:itelty} is defined as 
$$I_d(f):=\int_{\cX^d} f(x_1, x_2, \ldots, x_d) \prod_{a=1}^d \mathrm d \cZ_P(x_a):=\sum_{1 \leq i_1, i_2, \ldots, i_d \leq m} a_{i_1, i_2, \ldots, i_d} \cZ_P(A_{i_1})\times \cdots \times \cZ(A_{i_d}).$$ 
\end{definition} 

The multiple Weiner-It\^{o} integral for elementary functions satisfies the following two properties \citep{stochasticintegral}: 

\begin{itemize}

\item (Boundedness) For $f \in  \cE_d$, $\E[I_d(f)^2] \leq d! \|f\|^2 < \infty$.  

\item (Linearity) For $f,g\in \cE_d$, 
$I_d(f+g) \stackrel{a.s.}{=} I_d(f)+I_d(g)$. 

\end{itemize}
This shows that $I_d$ is a bounded linear operator from $\cE_d $ to $L^2(\Omega, \cF, \mu)$, the collection of square-integrable random variables defined on $(\Omega, \cF, \mu)$. Since $\cE_d$ is dense in $L^2(\cX^d, \sB(\cX^d), P^d)$ (by \citet[Theorem 2.1]{stochasticintegral}), using the BLT theorem (see \citet[Theorem I.7]{reedsimon}) $I_d$ can be uniquely extended to $L^2(\cX^d, \sB(\cX^d), P^d)$ by taking limits. This leads to the following definition: 

\begin{definition}\label{defn:integral_II} (Multiple Weiner-It\^{o} integral for general $L_2$-functions) The $d$-dimensional Weiner-It\^{o} stochastic integral for a function $f \in L^2(\cX^d, \sB(\cX^d), P^d)$ is defined as the $L_2$ limit of the sequence $\{I_d(f_n)\}_{n \geq 1}$,  where $\{f_n\}_{n \geq 1}$ is a sequence such that $f_n \in \cE_d$ with $\lim_{n \rightarrow \infty}\|f_n-f\|=0$. 
This is denoted by:  
\begin{align}\label{eq:integral}
I_d(f):=\int_{\cX^d} f(x_1, x_2, \ldots, x_d) \prod_{a=1}^d \mathrm d \cZ_P(x_a). 
\end{align}
\end{definition}

As in the case of elementary functions, it can be easily checked that $I_d(f)$ satisfies the following properties: 
\begin{itemize}

\item (Boundedness) For $f \in L^2(\cX^d, \sB(\cX^d), P^d)$, $\E[I_d(f)^2] \leq d! \|f\|^2 < \infty$.  

\item (Linearity) For $f, g\in L^2(\cX^d, \sB(\cX^d), P^d)$, 
$I_d(f+g) \stackrel{a.s.}{=} I_d(f)+I_d(g)$. 

\end{itemize}
It is also important to note that multiple Weiner-It\^{o} integrals do not behave like classical (non-stochastic) integrals with respect to product measures, since by definition diagonal sets do not contribute to It\^{o} integrals. Nevertheless, one can express the multiple Weiner-It\^{o} integral for a product function in terms of univariate stochastic integrals using the Wick product (cf.~\citet[Theorem 7.26]{janson1997gaussian}). In the bivariate case, with 2 functions $f, g \in L^2(\cX^2, \sB(\cX^2), P^2)$, this simplifies to  
\begin{align}\label{eq:fgstochasticintegral}
\int_{\cX^2} f(x) g(y) \mathrm d \cZ_P(x) \mathrm d \cZ_P(y) = \int_{\cX} f(x) \mathrm d \cZ_P(x)  \int_{\cX} g(y) \mathrm d \cZ_P(y) - \int_{\cX} f(x) g(x) \mathrm d x . 
\end{align}

\section{Proof of Theorem \ref{thm:K} and Corollary \ref{cor:TM}}
\label{sec:H0pf}
To begin with note that the definition of $\emmd$ in \eqref{eq:MMDXY} can be extended to any measurable and symmetric function $\mathsf H \in L^{2}(\mathcal{X}^2, P^2)$ (not necessarily positive definite) in a natural way as follows: 
\begin{align}\label{eq:H}
\emmd \left[\mathsf H, \sX_m, \sY_n \right] = \mathcal W_{\sX_m} + \mathcal W_{\sY_n} - 2 \mathcal B_{\sX_m, \sY_n} , 
\end{align}   
where $\mathcal W_{\sX_m}$, $\mathcal W_{\sY_n}$, and $\mathcal B_{\sX_m, \sY_n} $ are as defined in \eqref{eq:WX} and \eqref{eq:BXY}, respectively, with $\sf K$ replaced by $\mathsf{H}$. The main ingredient in the proof of Theorem \ref{thm:K} is the following result:

\begin{proposition}[{\citet{gretton2012kernel}, Theorem 5}] 
\label{prop:H0joint} 
For any measurable and symmetric function $\mathsf{H} \in L^{2}(\mathcal{X}^2, P^2)$, in the asymptotic regime \eqref{eq:mn}, 
\begin{align}\label{eq:H0distribution}
    (m+n)\emmd \left[\mathsf H, \sX_m, \sY_n \right] := \frac{1}{\rho(1-\rho)} \sum_{s=1}^{\infty}\lambda_{s}\left(Z_{s}^2- 1 \right) , 
\end{align} 
where $\{Z_{s}: s \geq 1\}$ are i.i.d. $\cN(0, 1)$ and 
$\{\lambda_{s}\}_{s \geq 1}$ are the eigenvalues (with repetitions) of the Hilbert-Schimdt operator $\cH_{\mathsf{H}^\circ}$ defined as: 
\begin{align}\label{eq:kernelH}
    \cH_{\mathsf{H}^\circ}[f(x)] = \int_{\cX} \mathsf{H}^{\circ} (x,y) f(y) \mathrm d P(y) , 
\end{align} 
with $\mathsf{H}^\circ (x,y) := \mathsf{H}(x,y) - \E_{X \sim P}\mathsf{H}(X,y) - \E_{X' \sim P}\mathsf{H}(x, X') + \E_{X, X' \sim P} \mathsf{H}(X, X')$.  Moreover, the characteristic function of $Z(\mathsf{H})$ at $t \in \R$ is  given by: 
\begin{align}\label{eq:expZ}
\Phi_{Z(\mathsf{H})} (t) := \E\left[ e^{\iota t Z(\mathsf{H})} \right]  = \prod_{s=1}^\infty \frac{e^{-\frac{\iota \lambda_s t}{\rho(1-\rho)}}}{\sqrt{1-\frac{2  \iota \lambda_s t }{ \rho(1-\rho) }}} .  
\end{align}
\end{proposition}

\begin{remark} 
The convergence in \eqref{eq:H0distribution} is a consequence of \citet[Theorem 5]{gretton2012kernel}, while the expression of the characteristic function in \eqref{eq:expZ} follows from \citet[Proposition 6.1]{janson1997gaussian}.  
(Note that Proposition \ref{prop:H0joint} also follows from the proof of Proposition \ref{ppn:H0N} in Section \ref{sec:localpowerpf} by setting $h=0$). 
\end{remark}

We now present the proof of Theorem \ref{thm:K} in Appendix \ref{sec:pfK}. The proof of Corollary \ref{cor:TM} is given in Appendix \ref{sec:corTMpf}.

\subsection{Proof of Theorem \ref{thm:K}}  
\label{sec:pfK} 
First recall the definition of $\mathrm{MMD}^{2}[ \cK, \sX_m, \sY_n ]$ from \eqref{eq:Kvector}.  Note that for $\bm \eta = (\eta_1, \eta_2, \ldots, \eta_r)^\top \in \mathbb{R}^r$,
\begin{align}\label{eq:etaMMD}
     \bm \eta^\top   \mathrm{MMD}^{2}[ \cK, \sX_m, \sY_n ] = \sum_{a=1}^r  \eta_a \mathrm{MMD}^{2}[ \sfK_a, \sX_m, \sY_n ] & = \mathrm{MMD}^{2}[ \mathsf{H}_{\bm \eta}, \sX_m, \sY_n ] , 
\end{align} 
where $\mathsf{H}_{\bm \eta} := \sum_{a=1}^r \eta_a \sfK_a$. Clearly, $\mathsf{H}_{\bm \eta}$ is a measurable and symmetric function and $\mathsf{H}_{\bm \eta} \in L^{2}\left(\mathcal{X}, P^2\right)$ (by Assumption \ref{assumption:K}). Then by Proposition \ref{prop:H0joint}, 
\begin{align}\label{eq:ZH} 
Z_{m, n}(\mathsf{H}_{\bm \eta})  := (m+n)   \mathrm{MMD}^{2} \left[\mathsf{H}_{\bm \eta}, \sX_m, \sY_n \right] \dto Z(\mathsf{H}_{\bm \eta}) = \frac{1}{\rho(1-\rho)} \sum_{\lambda\in \Lambda(\bm \eta)}\lambda\left(Z_{\lambda}^2- 1 \right) . 
\end{align}
where $\{Z_\lambda\}_{\lambda \in \Lambda(\bm \eta)}$ are i.i.d. $\cN(0, 1)$ 
and $\Lambda(\bm \eta)$ are the eigenvalues (with repetitions) of the Hilbert-Schimdt operator: 
\begin{align}\label{eq:Hxyeta}
    \cH_{\mathsf{H}_{\bm \eta}^\circ}[f(x)] = \int \mathsf{H}_{\bm \eta}^\circ (x,y) f(y)\mathrm d P(y) , 
\end{align} 
with 
\begin{align}\label{eq:Hxyetaxy}
    \mathsf{H}_{\bm \eta}^\circ (x,y) & : = \mathsf{H}_{\bm \eta}(x,y) - \E_{X \sim P}\mathsf{H}_{\bm \eta}(X,y) - \E_{X' \sim P}\mathsf{H}_{\bm \eta}(x, X') + \E_{X, X' \sim P} \mathsf{H}_{\bm \eta}(X, X') \nonumber \\ 
   & = \sum_{a=1}^r  \eta_a \sfK^\circ_{a}(x, y)  , 
\end{align} 
Hence, the operator  $\cH_{\mathsf{H}_{\bm \eta}^\circ}$ in \eqref{eq:Hxyeta} is same as the operator $\cH_{\cK, \bm \eta}$ defined in \eqref{eq:HKeta}. 
Then Proposition \ref{prop:H0joint}, \eqref{eq:ZH}, and \eqref{eq:Hxyetaxy} implies, 
\begin{align}\label{eq:expZH}
    \E\left[ e^{\iota Z_{m, n}(\mathsf{H}_{\bm \eta})} \right]  \rightarrow  \E\left[ e^{\iota  Z(\mathsf{H}_{\bm \eta})} \right] = \prod_{\lambda\in\Lambda(\bm \eta)} \frac{e^{-\frac{\iota \lambda }{\rho(1-\rho)}}}{\sqrt{1-\frac{2  \iota \lambda }{ \rho(1-\rho) }}} = \Phi(\bm \eta) , 
\end{align}
where $\Phi(\bm \eta)$ is as defined in \eqref{eq:expGK}. Since $\bm \eta \in \R^r$ was chosen arbitrarily and $\Phi(\bm \eta)$ is continuous at $\bm \eta = \bm 0 \in \mathbb R^r$ (by Lemma \ref{lm:phicontinuity}), Levy's continuity theorem \citet[Theorem 3.3.17]{durrett} implies that there exists a random variable $Z_{\cK}$ with characteristic function $\Phi(\bm \eta)$ such that,
\begin{align}\label{eq:ZK}
   \mathrm{MMD}^{2}[ \cK, \sX_m, \sY_n ] \dto Z_{\mathcal K}. 
\end{align}

We now show that the limit $Z_\cK$ in \eqref{eq:ZK} can be expressed as $G_{\cK}$ in \eqref{eq:GK}. Towards this, note that by the linearity of multiple stochastic integrals, 
\begin{align}\label{eq:phiGKeta}
\E\left[ e^{\iota \bm \eta^\top G_\cK} \right] = \E\left[ e^{\frac{\iota}{\rho(1-\rho)} \sum_{a=1}^r  \eta_a I_2(\sf K_a^\circ) } \right]  & = \E\left[ e^{\frac{\iota}{\rho(1-\rho)} I_2(\sum_{a=1}^r  \eta_a \sf K_a^\circ) } \right] \nonumber \\ 
& = \E\left[ e^{\frac{\iota}{\rho(1-\rho)} I_2( \mathsf{H}_{\bm \eta}^\circ) } \right] , 
\end{align} 
where the last step uses \eqref{eq:Hxyetaxy}. Then by \citet[Theorem 6.1]{janson1997gaussian}, 
$I_{2}(\mathsf{H}_{\bm \eta}^\circ)$ has characteristic function: 
\begin{align}\label{eq:phiHseta}
    \Phi_{\mathsf{H}_{\bm \eta}^\circ}(s):= \E\left[ e^{ \iota s I_{2}(\mathsf{H}_{\bm \eta}^\circ) } \right]= \prod_{\lambda\in \Lambda(\bm \eta)} \dfrac{ e^{- \iota  \lambda s } }{\sqrt{1-2 \iota \lambda  s}}
\end{align}
where $\Lambda(\bm \eta)$ is the set of eigenvalues (with repetition) of the bilinear form $\cB_{\mathsf{H}_{\bm{\eta}}^\circ}: L^2(\cX) \times L^2(\cX) \rightarrow \R$:  
\begin{align*}
\cB_{\mathsf{H}_{\bm \eta}^\circ}(f_P, f_2) :=  \frac{1}{2}\E\left[I_2(\mathsf{H}_{\bm \eta}^\circ) I_1(f_{1})I_1(f_{2})\right], 
\end{align*} 
for any $f_{1},f_{2}\in L^{2}(\mathcal{X}, P)$. Now, using the multiplication formula for stochastic integrals (cf.~\citet[Theorem 7.33]{janson1997gaussian}) gives, 
\begin{align}
    \frac{1}{2}\E\left[I_2(\mathsf{H}_{\bm \eta}^\circ) I_1(f_{1})I_1(f_{2})\right] 
    &= \frac{1}{2}\int_{\mathcal{X}^{2}}\left[\mathsf{H}_{\bm \eta}^\circ(x,y)f_{1}(x)f_{2}(y) + \mathsf{H}_{\bm \eta}^\circ (x,y)f_{1}(y)f_{2}(x)\right]\mathrm d P(x)\mathrm d P(y)\nonumber \\
    &= \int_{\mathcal{X}^{2}} \mathsf{H}_{\bm \eta}^\circ (x,y)f_{1}(x)f_{2}(y)\mathrm d P(x)\mathrm d P(y) , \nonumber 
\end{align}
where the last equality follows by symmetry of the function $\mathsf{H}_{\bm \eta}^\circ$. 
This shows that the bilinear form $\cB_{\mathsf{H}_{\bm \eta}^\circ}$ has the same set of eigenvalues (with repetitions) as that of the operator $\cH_{\mathsf{H}_{\bm \eta}^\circ}$ defined in \eqref{eq:Hxyetaxy}. Hence, combining \eqref{eq:phiGKeta} and \eqref{eq:phiHseta} it follows that, 
\begin{align}\label{eq:etaGK}
\E\left[ e^{\iota \bm \eta^\top G_\cK} \right] =  \Phi_{\mathsf{H}_{\bm \eta}^\circ}\left(\frac{1}{\rho(1-\rho)}\right) =  \prod_{\lambda\in \Lambda(\bm \eta)} \frac{e^{-\frac{\iota \lambda }{\rho(1-\rho)}}}{\sqrt{1-\frac{2  \iota \lambda }{ \rho(1-\rho) }}} , 
\end{align}
where $\Lambda(\bm \eta)$ is the set of eigenvalues (with repetitions) of the operator $\cH_{\mathsf{H}_{\eta}^\circ}$, which the same as the operator $\cH_{\cK, \bm \eta}$ (by \eqref{eq:HKeta}). Note that the RHS of \eqref{eq:etaGK} equals the function $\Phi(\bm \eta)$ defined in \eqref{eq:expGK}, which implies that $Z_\cK$ in \eqref{eq:ZK} has the same distribution as $G_\cK$ in \eqref{eq:GK}. This completes the proof of Theorem \ref{thm:K}. \hfill $\Box$

\begin{remark}[Alternative description of  the limiting distribution] \label{remark:limitingdistribution}
Note that, for $1 \leq a \leq r$, by the spectral theorem (see \citet[Theorem 6.35]{renardy2006introduction}): 
$$\sfK_a^\circ(x, y) = \sum_{s=1}^\infty \lambda_{s, a}  \phi_{s, a}(x) \phi_{s, a}(y) , $$  
where $\{\lambda_{s, a}\}_{s \geq 1}$ and $\{ \phi_{s, a} \}_{s \geq 1}$ are, respectively, the eigenvalues and the eigenvectors of the operator: $\cH_{\sfK_a}[f(x)] = \int_{\cX} \sfK_a^\circ(x,y) f(y)\mathrm d P(y)$.  
Hence, by the linearity of the stochastic integral,  \eqref{eq:fgstochasticintegral}, and orthonormality of the eigenvectors,  
\begin{align}\label{eq:alternatelimit}
 I_2(\sfK_a^\circ) 
& = \sum_{s=1}^\infty \lambda_{s, a}   \int_{\cX} \int_{\cX} \phi_{s, a} (x) \phi_{s, a} (y) \mathrm d \cZ_P(x) \mathrm d \cZ_P(y) \nonumber \\ 
& = \sum_{s=1}^\infty \lambda_{s, a}  \left ( \left( \int_{\cX} \phi_{s, a}(x) \mathrm d \cZ_P(x) \right)^2 - \int_{\cX} \phi_{s, a}(x)^2 \mathrm d x \right) \nonumber \\ 
 & \stackrel{D} = \sum_{s=1}^{\infty}\lambda_{s, a} \left( Z_{s, a}^2 - 1 \right) , 
\end{align}
where $Z_{s, a} \stackrel{D}{:=} \int_{\cX} \phi_{s, a}(x) \mathrm d \cZ_P(x)$. 
Note that $\{Z_{s, a}\}_{s \geq 1, 1 \leq a \leq r}$ is a collection of Gaussian variables with 
\begin{align}\label{eq:covarianceZ}
\Cov(Z_{s, a}, Z_{s', b}) = \int_{\cX} \phi_{s, a}(x) \phi_{s', b}(x) \mathrm d \cZ_P(x), 
\end{align}
 for $1 \leq a, b \leq r$ and $s, s' \geq 1$.  Hence, the limiting distribution of 
$(m+n) \mathrm{MMD}^{2}[ \cK, \sX_m, \sY_n ]$ can be alternately expressed as (recall \eqref{eq:GK}): 
$$G_\cK \stackrel{D} = \left( \sum_{s=1}^{\infty}\lambda_{s, a} \left( Z_{s, a}^2 - 1 \right) \right)_{1 \leq a \leq r}.$$   
Note that the orthonormality conditions imply that for each fixed $a \in \{1, 2, \ldots, r\}$, the collection $\{Z_{s, a}\}_{s \geq 1}$ is distributed as i.i.d. $\mathcal N(0, 1)$. This gives the well-known representation of the marginal distribution of the MMD estimate as an infinite weighted sum of independent centered $\chi^2_1$ variables (see \citet[Theorem 12]{gretton2012kernel}). Jointly these infinite sums are dependent, due to the dependence among the collection $\{Z_{s, a}\}_{s \geq 1, 1 \leq a \leq r}$ for $1 \leq a \ne b \leq r$ with covariance structure as in \eqref{eq:covarianceZ}. 
\end{remark}

\subsection{Proof of Corollary \ref{cor:TM}}
\label{sec:corTMpf}

Recall the definition of the centered kernel $\sfK_a^\circ$ from \eqref{eq:Kxycentered}. Then it is easy to see that 
\begin{align*}
    \mathrm{MMD}^{2}\left[\sfK_a, \sX_m, \sY_n \right] = \mathrm{MMD}^{2}\left[\sfK_a^\circ, \sX_m, \sY_n \right] . 
\end{align*} 
Observe that $\E\left[\sfK_a^\circ(X_{1},X_{2})\middle|X_{1}\right]= 0 $ and hence $\Cov\left(\sfK_a^\circ(X_1, X_2) \sfK_b^\circ(X_1, X_3)\right) = 0$, for $X_1, X_2, X_3$ i.i.d. from the distribution $P$. Then by a direct computation it follows that, for $1 \leq a, b \leq r$, 
\begin{align}\label{eq:sigmamatrix}
\sigma_{ab} = \frac{2}{\rho^2(1-\rho)^2} \E_{X, X' \sim P}\left[\sfK_a^\circ(X, X')\sfK_b^\circ(X, X')\right] . 
\end{align} 
This proves \eqref{eq:sigmaH0ab}. 

Next, we prove \eqref{eq:sigmaH0limit}. For all $1\leq a \leq r$, define
\begin{align*} 
\bm{\sfK}^\circ_a = \left(\left(\sfK_a^\circ(X_{i},X_{j})/m\right)\right)_{1\leq i,j\leq m}
\end{align*}
where $\sfK_a^\circ$ is defined in Theorem \ref{thm:K}. Also, recall the definition of $\sfKmatrix_{a}$, for $1 \leq a \leq r$, from \eqref{eq:centered-kernel}. Now, observe that for any $1 \leq a, b \leq r$,
\begin{align}\label{eq:trace-bd}
    \left|\Tr\left[ \sfKmatrix_{a} \sfKmatrix_{b} \right] - \Tr\left[ \bm{\sfK}^\circ_a   \bm{\sfK}^\circ_b  \right]\right| & \leq   \left|\Tr\left[\sfKmatrix_{a} \sfKmatrix_{b} \right] - \Tr\left[ \bm{\sfK}^\circ_a  \sfKmatrix_{b}  \right]\right| +  \left|\Tr\left[  \bm{\sfK}^\circ_a  \sfKmatrix_{b}  \right] - \Tr\left[ \bm{\sfK}^\circ_a   \bm{\sfK}^\circ_b  \right]\right|  \nonumber \\   
   & \leq \| \sfKmatrix_{a} \| \| \sfKmatrix_{a} - \bm{\sfK}^\circ_a  \| + \| \bm{\sfK}^\circ_a \| \| \sfKmatrix_{b} - \bm{\sfK}^\circ_b  \| 
   \end{align} 
Since $\sfK_a^\circ\in L^{2}(\mathcal{X}^{2}, P^2)$, by the strong law of large number for $U$-statistics (see \citet[Theorem 5.4.A]{serfling}) 
$$\| \bm{\sfK}^\circ_a \|^2 = \frac{1}{m^2} \sum_{1 \leq i, j \leq m} {\sfK}^\circ_a  (X_i, X_j)^2 \stackrel{a.s.} \rightarrow \E_{X, X' \sim P} [{\sfK}^\circ_a(X, X')] . $$ 
Also, following the proof of Lemma \ref{lm:HHhatL2} we have $\| \sfKmatrix_{a} - \bm{\sfK}^\circ_a  \| \stackrel{a.s.} \rightarrow 0$.  This implies, $\| \sfKmatrix_{a} \|^2 \stackrel{a.s.} \rightarrow \E_{X, X' \sim P} [{\sfK}^\circ_a (X, X')] $. 
Thus, combining the above conclusions with \eqref{eq:trace-bd} gives,
\begin{align}\label{eq:K-bar-tilde-equiv}
   \left|\Tr\left[ \sfKmatrix_{a} \sfKmatrix_{b} \right] - \Tr\left[ \bm{\sfK}_a^\circ   \bm{\sfK}_b^\circ  \right]\right| |\overset{a.s.}{\rightarrow}0
\end{align}
Since, for all $1\leq a \leq r$, $\sfK_a^\circ\in L^{2}(\mathcal{X}^{2}, P^{2})$, then by \citet[Theorem 5.4.A]{serfling},
\begin{align}\label{eq:lim-K-bar}
     \Tr\left[ \bm{\sfK}^\circ_a   \bm{\sfK}^\circ_b  \right]  = \frac{1}{m^{2}}\sum_{1 \leq i, j \leq m}\sfK_a^\circ(X_{i}, X_{j})\sfK_b^\circ(X_{i},X_{j}) & \stackrel{a.s.}{\rightarrow}\E_{ X \sim P}\left[\sfK_a^\circ(X_{1},X_{2})\sfK_b^\circ(X_{1},X_{2})\right] \nonumber \\ 
     & =  \frac{\rho^2(1-\rho)^2}{2}\sigma_{ab}  , 
\end{align}
where the last equality follows from \eqref{eq:sigmamatrix}. Now, recalling \eqref{eq:H0sigmaestimate} and applying \eqref{eq:K-bar-tilde-equiv} and \eqref{eq:lim-K-bar} note that 
\begin{align*}
 \hat \sigma_{ab} = \frac{2}{\hat{\rho}^2(1-\hat{\rho})^2} \cdot \frac{1}{m^{2}}\sum_{1 \leq i, j \leq m}\hat \sfK_a^\circ(X_{i},X_{j}) \hat \sfK_b^\circ(X_{i},X_{j}) & = \frac{2}{\hat{\rho}^2(1-\hat{\rho})^2}\Tr\left[ \sfKmatrix_{a} \sfKmatrix_{b} \right] 
 \nonumber \\ 
 & \stackrel{a.s.}{\rightarrow} \sigma_{ab} .  
\end{align*} 
This completes the proof of \eqref{eq:sigmaH0limit}.
The result in \eqref{eq:GKH0} follows from \eqref{eq:sigmaH0limit} combined with Theorem \ref{thm:K},  Slutsky's theorem, and the continuous mapping theorem.  

\section{Proof of Theorem \ref{thm:estimateH0}} 
\label{sec:estimateH0pf}
Suppose $\cK = \{\sfK_1, \sfK_2, \ldots, \sfK_r \}$ be a collection of $r$ characteristic kernels satisfying the conditions of Theorem \ref{thm:estimateH0}. We will prove Theorem \ref{thm:estimateH0} by showing that every linear combination of $\cE(\cK, \sX_m)$ converges to the corresponding linear combination of $G_\cK$. To this end, suppose $\bm \eta = (\eta_1, \eta_2, \ldots, \eta_r)^\top \in \mathbb{R}^r$ and define $\mathsf{H}_{\bm \eta} = \sum_{a=1}^r \eta_a \sfK_a$. Let $\mathsf{H}_{\bm \eta}^\circ = \sum_{a=1}^r \eta_a \sfK_a^\circ$ be as defined in \eqref{eq:Hxyetaxy} 
and $\{\lambda_s(\mathsf{H}_{\bm \eta}^\circ)\}_{s \geq 1}$ be the eigenvalues of the operator $\cH_{\mathsf{H}_{\bm \eta}^\circ}$ as in \eqref{eq:Hxyeta}. Also, define 
\begin{align*}
    \hat{\bm{\mathsf{H}}}_{\bm \eta}^\circ = \bm C \hat{\bm{\mathsf{H}}}_{\bm \eta} \bm C/m , 
\end{align*}
where $ \hat{\bm{\mathsf{H}}}_{\bm \eta} = (\mathsf{H}_{\bm \eta}(X_{i},X_{j}))_{1\leq i, j\leq m}$ and $\bm C = \bm{I}-\frac{1}{m}\bm{1}\bm{1}^{\top}$ is the centering matrix as defined in \eqref{eq:centered-kernel}. Note that 
\begin{align*}
    \hat{\bm{\mathsf{H}}}_{\bm \eta}^\circ = \left( \left( \dfrac{\hat{\mathsf{H}}_{\bm \eta}^\circ (X_i, X_j)}{m}  \right) \right)_{1 \leq i, j \leq m}, 
\end{align*}
where, similar to \eqref{eq:estimateKxy}, 
\begin{align}\label{eq:estimateHcentered} 
\hat{\mathsf{H}}_{\bm \eta}^\circ (x, y)=\mathsf{H}_{\bm \eta}(x,y) - \frac{1}{m} \sum_{u=1}^m \mathsf{H}_{\bm \eta}(X_u, y) - \frac{1}{m} \sum_{v=1}^m \mathsf{H}_{\bm \eta}(x, X_v) + \frac{1}{m^2} \sum_{1 \leq u, v \leq m} \mathsf{H}_{\bm \eta}(X_u, X_v) . 
\end{align} 
Recalling the definition of $\hat{\bm{\sfK}}_a^\circ$ from \eqref{eq:centered-kernel}, observe that 
\begin{align}\label{eq:Hetahatcentered}
\hat{\bm{\mathsf{H}}}_{\bm \eta}^\circ = \sum_{a=1}^r \eta_a \hat{\bm{\sfK}}_a^\circ . 
\end{align} 
Let $\{\lambda_{s}(\hat{\bm{\mathsf{H}}}_{\bm \eta}^\circ)\}_{1 \leq s \leq m}$ be the eigenvalues of the matrix $\hat{\bm{\mathsf{H}}}_{\bm \eta}^\circ$. Recall 
that $\gamma= \frac{1}{\rho(1-\rho)}$. Then we have the following proposition: 

\begin{proposition}\label{ppn:HZW} Suppose $\{\lambda_s(\mathsf{H}_{\bm \eta}^\circ)\}_{s \geq 1}$ and $\{\lambda_{s}(\hat{\bm{\mathsf{H}}}_{\bm \eta}^\circ)\}_{1 \leq s \leq m}$ be as defined above. Then there exists a set $\cQ_0 \in \sB(\cX)$ (not depending on $\bm \eta$) with $\P(\cQ_0) = 1$ such that on the set $\cQ_0$, 
\begin{align}\label{eq:HZW}
    \sum_{s=1}^{m} \lambda_s(\hat{\bm{\mathsf{H}}}_{\bm \eta}^\circ) \left(W_{s}^2-\gamma \right)| \sX_m \overset{D}{\rightarrow} \sum_{s=1}^{\infty}\lambda_s(\mathsf{H}_{\bm \eta}^\circ)\left(Z_{s}^2 - \gamma \right) , 
    \end{align} 
    as $m\rightarrow\infty$, where $\{W_{s}, Z_{s} : s \geq 1\}$ 
    are i.i.d. from $\cN(0, \gamma)$ independent of $\sX_m=\{X_1, X_2, \ldots, X_m\}$.
\end{proposition}

The proof of Proposition \ref{ppn:HZW} is given in Appendix \ref{sec:HZWpf}. We first show how this can be used to complete the proof of Theorem \ref{thm:estimateH0}. To this end, note that by definition, for all $s\geq 1$, $W_{s} = \sqrt{\gamma}W_{s}^{\circ}$,
where $\{W_{s}^{\circ}: s\geq 1\}$ are i.i.d. from $\cN(0,1)$. Define, for $s \geq 1$, 
\begin{align*}
    \hat{W}_{s} = \sqrt{\hat{\gamma}}W_{s}^{\circ}, 
\end{align*} 
where $\hat{\gamma} := \frac{1}{\hat{\rho}(1-\hat{\rho})} = \frac{mn}{(m+n)^2}$. Now observe that by Lemma \ref{lemma:HmL2} on $\cQ:= \cQ_1\bigcap\cQ_0$,
\begin{align*}
    \E\left[\left(\sum_{s=1}^{m}\lambda_s(\hat{\bm{\mathsf{H}}}_{\bm \eta}^\circ)(W_{s}^2-\gamma) - \sum_{s=1}^{m}\lambda_s(\hat{\bm{\mathsf{H}}}_{\bm \eta}^\circ)(\hat{W}_{s}^2-\hat{\gamma})^2\right)^2\middle| \sX_m\right] 
    & = \sum_{s=1}^{m}\lambda_s(\hat{\bm{\mathsf{H}}}_{\bm \eta}^\circ)^2(\gamma - \hat{\gamma})^2 \\
    & = (\gamma-\hat{\gamma})^2\|\hat{\bm{\mathsf{H}}}_{\bm \eta}^\circ\|^2 \rightarrow 0 , 
\end{align*}
and hence on the set $\cQ$,
\begin{align}\label{eq:ppnHZW-hat}
    \sum_{s=1}^{m} \lambda_s(\hat{\bm{\mathsf{H}}}_{\bm \eta}^\circ) \left(\hat{W}_{s}^2-\hat{\gamma} \right)| \sX_m \overset{D}{\rightarrow} \sum_{s=1}^{\infty}\lambda_s(\mathsf{H}_{\bm \eta}^\circ)\left(Z_{s}^2 - \gamma \right) . 
\end{align}
By the spectral decomposition,
\begin{align*}
    \hat{\bm{\mathsf{H}}}_{\bm \eta}^\circ = \bm Q_m \bm \Lambda_m \bm Q_m^\top , 
    \end{align*} 
  where $\bm \Lambda_m = \text{diag}(\lambda_s(\hat{\bm{\mathsf{H}}}_{\bm \eta}^\circ))_{1 \leq s \leq m}$ and $\bm Q_m^\top \bm Q_m = \bm Q_m \bm Q_m^\top = \bm{I}$ is an orthogonal matrix. Observe that for $\bm{W}_m=(\hat{W}_{1}, \ldots, \hat{W}_{m})^{\top}$, 
\begin{align}\label{eq:eigendecomp-apply-1}
    \sum_{s=1}^{m} \lambda_s(\hat{\bm{\mathsf{H}}}_{\bm \eta}^\circ) \left(\hat{W}_{s}^2-\hat{\gamma} \right) = \bm{W}_m^{\top} \bm \Lambda_m \bm{W}_m - \hat{\gamma}\Tr[\hat{\bm{\mathsf{H}}}_{\bm \eta}^\circ] =  \bm Z_m^{\top} \hat{\bm{\mathsf{H}}}_{\bm \eta}^\circ \bm Z_m - \hat{\gamma}\Tr[\hat{\bm{\mathsf{H}}}_{\bm \eta}^\circ] , 
\end{align}
where $\bm Z_m = \bm Q_m \bm W_m \sim \cN_m(\bm 0, \hat{\gamma} \bm{I})$ is independent of $\sX_m$. This is because $\bm Q_m$ is orthogonal and, hence, $\bm Z_m | \sX_m \sim \cN_m(\bm 0, \hat{\gamma} \bm{I})$, which implies $\bm Z_m \sim \cN_m(\bm 0, \hat{\gamma} \bm{I})$. 
By \eqref{eq:Hetahatcentered}, \eqref{eq:eigendecomp-apply-1}, and recalling the definition of $\cE(\cK, \sX_m)$ from \eqref{eq:def-Em} it follows that
\begin{align}\label{eq:est-alpha-beta-decomp}
    \sum_{s=1}^{m}\lambda_s(\hat{\bm{\mathsf{H}}}_{\bm \eta}^\circ)\left(\hat{W}_s^2-\hat{\gamma}\right) 
    & = \bm \eta^\top \cE(\cK, \sX_m) .
\end{align} 
Hence, by \eqref{eq:ppnHZW-hat} on a set $\cQ$ with $\P(\cQ) =1$, as $m\rightarrow\infty$, 
\begin{align}\label{eq:alpha-beta-dist-convg}
 \bm \eta^\top \cE(\cK, \sX_m) | \sX_{m}
    \overset{D}{\rightarrow} \sum_{s=1}^{\infty}\lambda_s(\mathsf{H}_{\bm \eta}^\circ)\left(Z_{s}^2-\gamma\right) \stackrel{D}= Z(\mathsf{H}_{\bm \eta}) , 
\end{align} 
where the last step uses \eqref{eq:ZH}. From \eqref{eq:etaMMD} and \eqref{eq:ZH} we know that $Z(\mathsf{H}_{\bm \eta})$ is the limiting distribution of $\bm \eta^\top \mathrm{MMD}^{2}[ \cK, \sX_m, \sY_n ]$. Hence, by Theorem \ref{thm:K}, $\sum_{s=1}^{\infty}\lambda_s(\mathsf{H}_{\bm \eta}^\circ)\left(Z_{s}^2-\gamma\right) \overset{D}{=} \bm \eta^\top G_{\cK}$. This implies, by \eqref{eq:est-alpha-beta-decomp} and \eqref{eq:alpha-beta-dist-convg}, on the set $\cQ$, 
$$ \bm \eta^\top \cE(\cK, \sX_m) | \sX_{m} \dto \bm \eta^\top G_{\cK} , $$
as $m \rightarrow \infty$. Hence, by the Cramer-Wold device, on the set $\cQ$, 
$\cE(\cK, \sX_m) | \sX_{m} \dto G_{\cK}$, as $m \rightarrow \infty$. This completes the proof of Theorem \ref{thm:estimateH0}. \hfill $\Box$

\subsection{Proof of Proposition \ref{ppn:HZW}} 
\label{sec:HZWpf}

The proof of Proposition \ref{ppn:HZW} is organized as follows.
First we show that the $L_2$ norm of the matrix $\hat{\bm{\mathsf{H}}}_{\bm \eta}^\circ$ converges to the $L_2$ norm of the operator $\mathcal{H}_{\mathsf{H}_{\bm \eta}^\circ}$ almost surely (proof is given in Appendix \ref{sec:HmL2pf}). 

\begin{lemma}\label{lemma:HmL2} 
There exists a set $\cQ_1 \in \sB(\cX)$ (not depending on $\bm \eta$) with $\P(\cQ_1) = 1$ such that on $\cQ_1$, 
\begin{align}\label{eq:HmL2}
\|\hat{\bm{\mathsf{H}}}_{\bm \eta}^\circ\| \rightarrow \| \mathsf{H}_{\bm \eta}^\circ \| , 
\end{align} 
as $m \rightarrow \infty$.  
\end{lemma}

Next, we show that the $\ell$-th moment (power sum) of the eigenvalues of $\hat{\bm{\mathsf{H}}}_{\bm \eta}^\circ$ converges to the $\ell$-th moment (power sum) of the eigenvalues of $\mathcal{H}_{\mathsf{H}_{\bm \eta}^\circ}$, for $\ell \geq 3$, almost surely (proof is given in Appendix \ref{sec:momentHmpf}).

\begin{lemma}\label{lemma:momentHm} 
There exists a set $\cQ_2 \in \sB(\cX)$ (not depending on $\bm \eta$) with $\P(\cQ_2) = 1$ such that on $\cQ_2$, 
\begin{align}\label{eq:momentHm}
    \sum_{s=1}^{m} \lambda_s(\hat{\bm{\mathsf{H}}}_{\bm \eta}^\circ)^{\ell}  \rightarrow \sum_{s=1}^{\infty}\lambda_{s}(\mathsf{H}_{\bm \eta}^\circ)^{\ell} , 
\end{align} 
for all $\ell\geq 3$, as $m \rightarrow \infty$.  
\end{lemma}

Finally, we show that if \eqref{eq:HmL2} and \eqref{eq:momentHm} holds, then the convergence in \eqref{eq:HZW} holds (proof is given in Appendix \ref{sec:convergencepf}).  

\begin{lemma}\label{lemma:convergence}
Suppose \eqref{eq:HmL2} and \eqref{eq:momentHm} holds. Then on the set $\cQ_0 = \cQ_1 \bigcap \cQ_2$, the convergence in \eqref{eq:HZW} holds. 
\end{lemma} 

Since $\P(\cQ_0) =1$ and $\cQ_0$ does not depend on $\bm \eta$, the above 3 lemmas combined completes the proof of Proposition \ref{ppn:HZW}.

\subsubsection{Proof of Lemma \ref{lemma:HmL2}}
\label{sec:HmL2pf}

Define 
\begin{align*}
    \bm{\mathsf{H}}^\circ_{\bm \eta} = \left(\left(\dfrac{\mathsf{H}_{\bm \eta}^\circ(X_{i},X_{j})}{m}\right)\right)_{1\leq i,j\leq m} , 
    \end{align*}
where 
\begin{align}\label{eq:HXetaxy}
    \mathsf{H}_{\bm \eta}^\circ (X_i, X_j) & : = \mathsf{H}_{\bm \eta}(X_i, X_j) - \E_{X}\mathsf{H}_{\bm \eta}(X, X_j) - \E_{X'}\mathsf{H}_{\bm \eta}(X_i, X') + \E_{X, X'}\mathsf{H}_{\bm \eta}(X, X') ,  
\end{align}  
where $X, X'$ are i.i.d. samples from $P$. 
First we will 
show that  $\bm{\mathsf{H}}^\circ_{\bm \eta}$ and $\hat{\bm{\mathsf{H}}}_{\bm \eta}^\circ$ are asymptotically close. 

\begin{lemma}\label{lm:HHhatL2} There exists a set $\cR_1 \in \sB(\cX)$ (not depending on $\bm \eta$) with $\P(\cR_1) = 1$ such that 
\begin{align}\label{eq:HHhatL2}
\lim_{m \rightarrow \infty} \| \hat{\bm{\mathsf{H}}}^\circ_{\bm \eta} - \bm{\mathsf{H}}^\circ_{\bm \eta} \|^{2} = 0. 
\end{align} 
\end{lemma}

\begin{proof} 
Note that 
$$\| \hat{\bm{\mathsf{H}}}^\circ_{\bm \eta} - \bm{\mathsf{H}}^\circ_{\bm \eta} \|^{2} = \frac{1}{m^2}\sum_{1 \leq i, j \leq m} \left( \hat{\mathsf{H}}_{\bm \eta}^\circ(X_{i},X_{j}) - \mathsf{H}_{\bm \eta}^\circ(X_{i},X_{j}) \right)^2 .
$$ 
Hence, by \eqref{eq:estimateHcentered} and \eqref{eq:HXetaxy}, 
\begin{align}\label{eq:HHL2}
    &\|\bm{\mathsf{H}}^\circ_{\bm \eta} - \hat{\bm{\mathsf{H}}}^\circ_{\bm \eta} \|^{2} \leq 3 (T_1 + T_2 + T_3) , 
    \end{align}
where 
\begin{align}
T_1 & := \frac{1}{m}\sum_{i=1}^{m}\left(\frac{1}{m}\sum_{v=1}^{m}\mathsf{H}_{\bm \eta}(X_{i},X_{v})- \E_{X' \sim P} [\mathsf{H}_{\bm \eta}(X_{i},X') ] \right)^2 , \nonumber \\ 
T_2 & := \frac{1}{m}\sum_{j=1}^{m}\left(\frac{1}{m}\sum_{u=1}^{m}\mathsf{H}_{\bm \eta}(X_{u},X_{j})- \E_{X' \sim P} [\mathsf{H}_{\bm \eta}(X', X_{j}) ] \right)^2 , \nonumber \\ 
 T_3 & := \left(\frac{1}{m^{2}}\sum_{1\leq u, v \leq m}\mathsf{H}_{\bm \eta}(X_{u}, X_{v})-\E_{X, X' \sim P} [ \mathsf{H}_{\bm \eta}(X, X') ] \right)^2 . \nonumber 
\end{align} 
Since $\mathsf{H}_{\bm \eta} = \sum_{a=1}^r \eta_a \sfK_a$,  
\begin{align}\label{eq:HHT1}
 T_1     & \leq C_{1}^{(r)} \sum_{a=1}^ r \eta_a^2 \left\{ \frac{1}{m}\sum_{i=1}^{m}\left(\frac{1}{m}\sum_{v=1}^{m}\sfK_{a}(X_{i},X_{v})- \E_{X' \sim P} [ \sfK_{a}(X_{i}, X') ] \right)^2 \right \}. 
\end{align}
where $C_{1}^{(r)}>0$ is a constant depending on $r$ only. By Lemma \ref{lemma:kernel-avg-convg}, for every $1 \leq a \leq r$ there exists a set  $\cB_{\sfK_a} \in \sB(\cX)$ with $\P(\cB_{\sfK_a}) = 1$ such that 
$$ \lim_{m \rightarrow \infty} \frac{1}{m}\sum_{i=1}^{m}\left(\frac{1}{m}\sum_{v=1}^{m}\sfK_{a}(X_{i},X_{v})- \E_{X' \sim P} [ \sfK_{a}(X_{i}, X') ] \right)^2 = 0.$$
Define $\cB_{\cK} = \bigcap_{s=1}^r \cB_{\sfK_a} $. Clearly, $\P(\cB_{\cK}) = 1$. 
By Lemma \ref{lemma:kernel-avg-convg} and \eqref{eq:HHT1}, on the set $\cB_{\cK}$, 
$\lim_{m \rightarrow \infty} T_1 = 0$. Similarly, on the set $\cB_{\cK}$, $\lim_{m \rightarrow \infty} T_2 = 0$. 

Also, by \citet[Theorem 5.4.A]{serfling} there is a set $\cE_{\cK}$ (depending $\sfK_{1}, \sfK_{2}, \ldots, \sfK_{r}$, but not on $\bm \eta$) with $\P(\cE_{\cK}) = 1$ such that on the set $\cE_{\cK}$, for all $1 \leq a \leq r$, 
$$\lim_{m \rightarrow \infty}\frac{1}{m^{2}}\sum_{1\leq u, v \leq m}\sfK_{a}(X_{u}, X_{v})-\E_{X, X' \sim P} [ \sfK_{a}(X, X') ] = 0 . $$

Then on the set $\cE_{\cK}$, 
\begin{align}\label{eq:HHT3}
   T_3 
    &\leq C_{2}^{(r)} \sum_{a=1}^r \eta_a^2 \left( \left(\frac{1}{m^{2}}\sum_{1\leq u, v\leq m}\sfK_{a}(X_{u}, X_{v}) - \E_{X, X' \sim P} [\sfK_{a}(X, X')] \right)^2 \right)  \rightarrow 0, 
\end{align}
where $C_{2}^{(r)}>0$ is a constant depending on $r$ only.
Combining \eqref{eq:HHT1} and \eqref{eq:HHT3} with \eqref{eq:HHL2} shows that \eqref{eq:HHhatL2} holds on $\cR_1 = \cB_{\cK} \bigcap \cE_{\cK}$. Since $\P(\cR_1) = 1$, this completes the proof of Lemma \ref{lm:HHhatL2}. 
\end{proof}

 Now, we compute $\|\bm{\mathsf{H}}^\circ_{\bm \eta} \|$. By \citet[Theorem 5.4.A]{serfling}, there exists a set $\cR_2 \in \sB(\cX)$ with $\P(\cR_2) = 1$ such that on the set $\cR_2$,  
\begin{align}\label{eq:normheta0}
    \| \bm{\mathsf{H}}^\circ_{\bm \eta} \|^{2} & = \frac{1}{m^{2}}\sum_{1 \leq i , j \leq n} \mathsf{H}_{\bm \eta}^\circ(X_{i},X_{j})^2 \nonumber \\ 
    & = \sum_{a=1}^r  \eta_a^{2} \left(  \frac{1}{m^{2}}\sum_{1 \leq i , j \leq n} \sfK_a^\circ(X_{i},X_{j})^2 \right) + \sum_{1 \leq s \ne t \leq r} \eta_a \eta_b \left( \frac{1}{m^{2}}\sum_{1 \leq i, j \leq n } \sfK_a^\circ(X_{i},X_{j} ) \sfK_b^\circ (X_{i},X_{j}) \right) \nonumber \\
     & \rightarrow \sum_{a=1}^r  \eta_a^{2} \E_{X, X' \sim P}[\sfK_a^\circ(X, X')^2] + \sum_{1 \leq a \ne b \leq r} \eta_a \eta_b \E_{X, X' \sim P} [ \sfK_a^\circ(X, X') \sfK_b^\circ (X, X') ]   \nonumber \\ 
    & = \E \left[ \left( \sum_{a=1}^r  \eta_a  \sfK_a^\circ(X, X') \right)^2 \right] \nonumber \\
    & = \E_{X, X' \sim P} \left[ \mathsf{H}_{\bm \eta}^\circ(X, X')^2 \right]
     = \left\|\mathsf{H}_{\bm \eta}^\circ\right\|^{2} .  
\end{align}
Combining the above together with Lemma \ref{lm:HHhatL2} and choosing $\cQ_1  = \cR_1 \bigcap \cR_2$, the result in Lemma \ref{lemma:HmL2} follows.

\subsubsection{Proof of Lemma \ref{lemma:momentHm}}
\label{sec:momentHmpf}

Define 
\begin{align*}
  \bm{\mathsf{H}}_{\bm \eta}^{\circ, -} = \left(\left((1-\delta_{ij})\dfrac{\mathsf{H}_{\bm \eta}^\circ(X_{i},X_{j})}{m}\right)\right)_{1\leq i,j\leq m}
\end{align*}
where $\delta_{ij} = 1$ if $i=j$ and $0$ otherwise. First, we will show that the $\ell$-th moment (power sums) of the eigenvalues of $\bm{\mathsf{H}}_{\bm \eta}^{\circ,-}$ converges to the $\ell$-th moment of the eigenvalues of $\mathcal{H}_{\mathsf{H}_{\bm \eta}^\circ}$, for $\ell \geq 3$, almost surely.

\begin{lemma}\label{lemma:power-sum-convg}
Suppose $\{\lambda_s(\bm{\mathsf{H}}_{\bm \eta}^{\circ, -}): 1\leq s \leq m\}$ are the eigenvalues of $\bm{\mathsf{H}}_{\bm \eta}^{\circ, -}$. Then there exists a set $\cE_1 \in \sB(\cX)$ with $\P(\cE_1) = 1$ such  that on $\cE$
\begin{align*}
  \lim_{m \rightarrow \infty}  \sum_{s=1}^{m}\lambda_s(\bm{\mathsf{H}}_{\bm \eta}^{\circ, -})^{\ell}\ = \sum_{s=1}^{\infty}\lambda_{s} (\mathsf{H}_{\bm \eta}^\circ)^{\ell} , 
\end{align*} 
for all $\ell \geq 3$.  
\end{lemma} 

\begin{proof}
For fixed $\ell\geq 3$, observe that,
\begin{align}\label{eq:eigenvaluehatH}
    \sum_{s=1}^{m}\lambda_s(\bm{\mathsf{H}}_{\bm \eta}^{\circ, -})^{\ell} & = \tr\left[\left(\bm{\mathsf{H}}_{\bm \eta}^{\circ, -}\right)^{\ell}\right] \nonumber \\ 
    & = \frac{1}{m^{\ell}}\sum_{1 \leq i_{1}\neq i_{2}\neq \cdots\neq i_{\ell} \leq m} \mathsf{H}_{\bm \eta}^\circ(X_{i_{1}},X_{i_{2}})\mathsf{H}_{\bm \eta}^\circ(X_{i_{2}},X_{i_{3}})\cdots \mathsf{H}_{\bm \eta}^\circ(X_{i_{\ell}},X_{i_{1}}) .
\end{align} 
Recalling $\mathsf{H}_{\bm \eta}^\circ = \sum_{a=1}^r \eta_a \sfK_a^\circ$, \eqref{eq:eigenvaluehatH} can be written as: 
\begin{align}\label{eq:eigenvalueestimateKH}
    \sum_{s=1}^{m}\lambda_s(\bm{\mathsf{H}}_{\bm \eta}^{\circ, -})^{\ell}  
  & = \sum_{j_{1},\ldots,j_{\ell}\in \{1,2, \ldots, r\}}\prod_{t=1}^{\ell} \eta_{j_{t}} \left( \frac{1}{m^{\ell}}\sum_{1 \leq i_{1} \neq \cdots\neq i_{\ell} \leq m}\prod_{t=1}^{\ell} \sfK^{\circ}_{j_{t}}(X_{i_{t}},X_{i_{t+1}}) \right). 
\end{align} 

Now, since $\mathsf{H}_{\bm \eta}^\circ\in L^{2}(\mathcal{X}^{2}, P^{2})$, by the spectral theorem 
\begin{align}\label{eq:eigenvaluesHeta}
\mathsf{H}_{\bm \eta}^\circ(x, y) = \sum_{s=1}^\infty \lambda_s \phi_s(x) \phi_s(y) , 
\end{align}
where $\{ \lambda_1, \lambda_2, \ldots \}$ are the eigenvalues of  $\mathsf{H}_{\bm \eta}^\circ$ and $\{\phi_1, \phi_2, \ldots \}$  are the corresponding eigenvectors which form an orthonormal basis of $L^{2}(\mathcal{X}, P)$. Using the spectral representation in \eqref{eq:eigenvaluesHeta} and the orthonormality of the eigenvectors it follows that 
\begin{align}
    \sum_{s=1}^{\infty}\lambda_{s}(\mathsf{H}_{\bm \eta}^\circ)^{\ell} & = \int \mathsf{H}_{\bm \eta}^\circ(x_{1},x_{2})\mathsf{H}_{\bm \eta}^\circ(x_{2},x_{3})\cdots\mathsf{H}_{\bm \eta}^\circ(x_{\ell},x_{1})\mathrm{d}P(x_{1})\cdots\mathrm{d}P(x_{\ell}) \label{eq:FinnerFinite}\\ 
& = 
\sum_{j_{1},\ldots,j_{\ell}\in \{1,2, \ldots, r \}}\prod_{t=1}^{\ell} \eta_{j_{t}}\E\left[\sfK^{\circ}_{j_{1}}(X_{1},X_{2})\cdots \sfK^{\circ}_{j_{\ell}}(X_{\ell},X_{1})\right] \label{eq:eigenvalueKH},
     \end{align}  
     Note that the R.H.S. of \eqref{eq:FinnerFinite} is finite, since $\sum_{s=1}^{\infty}\lambda_{s}(\mathsf{H}_{\bm{\eta}}^\circ)^\ell\leq \|\mathsf{H}_{\bm{\eta}}^\circ\|^\ell$, by Finner's inequality \citep{finner1992generalization} (see also \citet[Theorem 3.1]{LZdense}) since $\mathsf{H}_{\bm \eta}^\circ = \sum_{a=1}^r \eta_a \sfK_a^\circ$. 
Hence, using \citet[Theorem 5.4A]{serfling}, \eqref{eq:eigenvalueestimateKH}, and \eqref{eq:eigenvalueKH}, we can find a set $\cE_1 \in \sB(\cX)$ with $\P(\cE_1) = 1$ such that on $\cE_1$, as $m \rightarrow \infty$,  
\begin{align*}
    \sum_{s=1}^{m}\lambda_s(\bm{\mathsf{H}}_{\bm \eta}^{\circ, -})^{\ell} 
    & \rightarrow \sum_{s=1}^{\infty}\lambda_{s}(\mathsf{H}_{\bm \eta}^\circ)^{\ell} , 
    \end{align*}
 for all $\ell\geq 3$. This completes the proof of Lemma \ref{lemma:power-sum-convg}.
\end{proof}

Now, we will show the \eqref{eq:momentHm}, that is, 
\begin{align*}
  \lim_{m \rightarrow \infty}  \sum_{s=1}^{m}\lambda_s(\hat{\bm{\mathsf{H}}}_{\bm \eta}^\circ)^{\ell} = \sum_{s=1}^{\infty}\lambda_{s}(\mathsf{H}_{\bm \eta}^\circ)^{\ell} , 
\end{align*} 
for all $\ell\geq 3$, on set $\mathcal{Q}_2$ with $\P(\cQ_2)=1$. 
To this end, note that for all $\ell\geq 3$, 
\begin{align}\label{eq:lambdaHHhat}
   & \left|\sum_{s=1}^{m}\lambda_s(\hat{\bm{\mathsf{H}}}_{\bm \eta}^\circ)^{\ell} - \sum_{s=1}^{\infty}\lambda_{s}(\mathsf{H}_{\bm \eta}^\circ)^{\ell}\right| \nonumber \\ 
    & \leq \left|\sum_{s=1}^{m}\lambda_s(\hat{\bm{\mathsf{H}}}_{\bm \eta}^\circ)^{\ell} - \sum_{s=1}^{m}\lambda_s(\bm{\mathsf{H}}_{\bm \eta}^{\circ, -})^{\ell}\right| + \left|\sum_{s=1}^{m}\lambda_s(\bm{\mathsf{H}}_{\bm \eta}^{\circ, -})^{\ell} - \sum_{s=1}^{\infty}\lambda_{s}(\mathsf{H}_{\bm \eta}^\circ)^{\ell}\right| . 
\end{align} 
On the set $\cE_1$ as in Lemma \ref{lemma:power-sum-convg}, the second term above converges to zero as $m \rightarrow \infty$. To bound the first term, observe that 
\begin{align*}
    \left|\sum_{s=1}^{m}\lambda_s( \hat{\bm{\mathsf{H}}}_{\bm \eta}^\circ )^{\ell} - \sum_{s=1}^{m}\lambda_s( \bm{\mathsf{H}}_{\bm \eta}^{\circ, -})^{\ell}\right|\leq \sum_{s=1}^{\infty}\left|\lambda_s^{+}( \hat{\bm{\mathsf{H}}}_{\bm \eta}^\circ)^{\ell} -\lambda_s^{+}( \bm{\mathsf{H}}_{\bm \eta}^{\circ, -} )^{\ell}\right| + \sum_{s=1}^{\infty}\left|\lambda_s^{-}( \hat{\bm{\mathsf{H}}}_{\bm \eta}^\circ )^{\ell} -\lambda_s^{-}(\bm{\mathsf{H}}_{\bm \eta}^{\circ, -})^{\ell}\right| ,  
\end{align*} 
where 
\begin{itemize} 

\item $\lambda_1^{+}( \hat{\bm{\mathsf{H}}}_{\bm \eta}^\circ) \geq \lambda_2^{+}( \hat{\bm{\mathsf{H}}}_{\bm \eta}^\circ) \geq \ldots \geq 0$ are the non-negative eigenvalues of $\hat{\bm{\mathsf{H}}}_{\bm \eta}^\circ$ (and similarly for $\bm{\mathsf{H}}_{\bm \eta}^{\circ, -}$) arranged in non-increasing order.  

\item $\lambda_1^{-}( \hat{\bm{\mathsf{H}}}_{\bm \eta}^\circ) \leq \lambda_2^{+}( \hat{\bm{\mathsf{H}}}_{\bm \eta}^\circ) \leq \ldots \leq 0$ are the non-positive eigenvalues of $\hat{\bm{\mathsf{H}}}_{\bm \eta}^\circ$ (and similarly for $\bm{\mathsf{H}}_{\bm \eta}^{\circ, -}$) arranged in non-decreasing order.  

\end{itemize}
(Note that we have set $\lambda_s^{\pm}$ to zero whenever appropiate to extend the sequence to infinity.) Now, recalling the definitions of $ \bm{\mathsf{H}}_{\bm \eta}^\circ$ and $ \bm{\mathsf{H}}_{\bm \eta}^{\circ, -}$ gives, 
\begin{align}\label{eq:Hm-norm-convg}
 \|  \bm{\mathsf{H}}_{\bm \eta}^\circ - \bm{\mathsf{H}}_{\bm \eta}^{\circ, - } \| = \frac{1}{m^2} \sum_{i=1}^m  \mathsf{H}_{\bm \eta}^{\circ}(X_i, X_i)^2 \rightarrow 0 , 
\end{align}  
on a set $\bar{\cE}_2 \in \cB(\cX)$. Therefore, recalling Lemma \ref{lm:HHhatL2}, by \eqref{eq:Hm-norm-convg} and the Hoffman-Wielandt inequality \citep{eigenvalueHW} (see also \citet[Theorem 2.2]{Koltch2000}), 
\begin{align}\label{eq:lambdaHeta}
    \lim_{m \rightarrow \infty} \left( \sum_{s=1}^\infty \left( \lambda_s^{+}( \hat{\bm{\mathsf{H}}}_{\bm \eta}^\circ ) - \lambda_s^{+}(\bm{\mathsf{H}}_{\bm \eta}^{\circ, -}) \right)^2  \right)^{\frac{1}{2}} \leq \lim_{m \rightarrow \infty } \left(\|\bm{\mathsf{H}}_{\bm \eta}^{\circ, -} - \bm{\mathsf{H}}_{\bm \eta}^\circ \| + \|\hat{\bm{\mathsf{H}}}_{\bm \eta}^{\circ} - \bm{\mathsf{H}}_{\bm \eta}^\circ \| \right) = 0 . 
\end{align} 
 on the set $\cE_2: = \bar{\cE}_2\bigcap\mathcal{R}_{1}$. Moreover, by Lemma \ref{lemma:HmL2}, \eqref{eq:normheta0} and \eqref{eq:Hm-norm-convg}, on the set $\cQ_1$,  $\lim_{m \rightarrow \infty } \| \hat{\bm{\mathsf{H}}}_{\bm \eta}^\circ \| = \|\mathsf{H}_{\bm \eta}^\circ\|$. Hence, by Lemma \ref{lm:HHhatL2} on the set $\cQ_1 \bigcap \cE_2$ there is a constant $C > 0$ such that  $$\max\{ \| \hat{\bm{\mathsf{H}}}_{\bm \eta}^\circ \|, \| \bm{\mathsf{H}}_{\bm \eta}^{\circ, -} \|, \|\mathsf{H}_{\bm \eta}^\circ\| \} \leq C.$$ This implies, on the set $\cQ_1 \bigcap \cE_2$, 
\begin{align}\label{eq:eigen-bdd-omega}
 \max \{ |\lambda_s^{+}( \hat{\bm{\mathsf{H}}}_{\bm \eta}^\circ )|, |\lambda_s^{+}(\bm{\mathsf{H}}_{\bm \eta}^{\circ, -})| \} \leq \frac{C}{\sqrt{s}}, 
    \end{align} 
for all $s \geq 1$. Then using \eqref{eq:lambdaHeta}, \eqref{eq:eigen-bdd-omega}, and the dominated convergence theorem, 
 \begin{align*}
    \lim_{m \rightarrow \infty}  \sum_{s=1}^{\infty}\left|\lambda_s^{+}( \hat{\bm{\mathsf{H}}}_{\bm \eta}^\circ )^{\ell} -\lambda_s^{+}(\bm{\mathsf{H}}_{\bm \eta}^{\circ, -})^{\ell}\right| = 0 , 
    \end{align*} 
  for all $\ell \geq 3$, on the set  $\cQ_1 \bigcap \cE_2$. 
Similarly,  $\lim_{m \rightarrow \infty} \sum_{s=1}^{\infty} | \lambda_s^{-}(\hat{\bm{\mathsf{H}}}_{\bm \eta}^{\circ})^{\ell} -\lambda_s^{-}(\bm{\mathsf{H}}_{\bm \eta}^{\circ, -})^{\ell} | = 0 $, for all $\ell \geq 3$, on the set  $\cQ_1 \bigcap \cE_2$. Therefore, on the set $\cQ_2:=\cQ_1 \bigcap \cE_1 \bigcap \cE_2$, from \eqref{eq:lambdaHHhat}, 
\begin{align*}
  \lim_{m \rightarrow \infty}  \sum_{s=1}^{m}\lambda_s(\hat{\bm{\mathsf{H}}}_{\bm \eta}^\circ)^{\ell}  = \sum_{s=1}^{\infty}\lambda_{s}(\mathsf{H}_{\bm \eta}^\circ)^{\ell} , 
    \end{align*}
for all $\ell\geq 3$. Since $\P(\cQ_2) = 1$, this completes the proof of Lemma \ref{lemma:momentHm}.

\subsubsection{Proof of Lemma \ref{lemma:convergence}} 
\label{sec:convergencepf}

Recall that $\{\lambda_s(\hat{\bm{\mathsf{H}}}_{\bm \eta}^\circ)\}_{1 \leq s \leq m}$ are the eigenvalues of $\hat{\bm{\mathsf{H}}}_{\bm \eta}^\circ$. For $s > m$ define $\lambda_s(\hat{\bm{\mathsf{H}}}_{\bm \eta}^\circ)=0$. Consider, for all $m \geq 1$, 
\begin{align}\label{eq:def-Ym}
    Y_{m}:= \sum_{s=1}^{\infty}\lambda_s(\hat{\bm{\mathsf{H}}}_{\bm \eta}^\circ)(W_s^2-\gamma) . 
\end{align}
Then by \citet[Proposition 7.1]{bbbpdsm} observe that
\begin{align*}
    M_{Y_{m}|\sX_{m}}(t) := \E\left[ e^{ tY_{m} } \middle|\sX_{m}\right] = \prod_{s=1}^{m}\dfrac{\exp\left(-\gamma\lambda_s(\hat{\bm{\mathsf{H}}}_{\bm \eta}^\circ) t\right)}{\sqrt{1-2\gamma\lambda_s( \hat{\bm{\mathsf{H}}}_{\bm \eta}^\circ )t}}, \text{ for all } |t|<\frac{1}{8\gamma}\left(\sum_{s=1}^{m}\lambda_s( \hat{\bm{\mathsf{H}}}_{\bm \eta}^\circ )^2\right)^{-\frac{1}{2}} . 
\end{align*}
By definition, $\sum_{s=1}^{m}\lambda_s(\hat{\bm{\mathsf{H}}}_{\bm \eta}^\circ)^2 = \sum_{s=1}^{\infty}\lambda_s(\hat{\bm{\mathsf{H}}}_{\bm \eta}^\circ)^2 = \| \hat{\bm{\mathsf{H}}}_{\bm \eta}^\circ \|^2$. Taking logarithm and expanding gives, 
\begin{align}\label{eq:MGF-Hm-log}
    \log M_{Y_{m}|\sX_{m}}(t) 
    & = \gamma^{2}t^{2}\| \hat{\bm{\mathsf{H}}}_{\bm \eta}^\circ \|^{2} + \frac{1}{2}\sum_{k=3}^{\infty}\sum_{s=1}^{m}\dfrac{(2\gamma\lambda_s(\hat{\bm{\mathsf{H}}}_{\bm \eta}^\circ)t)^k}{k} . 
\end{align}
Also, Denote $Z(\mathsf{H}_{\bm \eta}^\circ) = \sum_{s=1}^{\infty}\lambda_s(\mathsf{H}_{\bm \eta}^\circ)(Z_{s}^2-\gamma)$. Then by Lemma \ref{lm:ZHM}, 
\begin{align}\label{eq:MGF-h-tilde-log}
    & \log M_{Z(\mathsf{H}_{\bm \eta}^\circ)}(t) \nonumber \\ 
    & = \log \E\exp\left(tZ(\mathsf{H}_{\bm \eta}^\circ)\right) \nonumber \\ 
    & = \gamma^{2}t^{2}\|\mathsf{H}_{\bm \eta}^\circ\|^{2} + \frac{1}{2}\sum_{k=3}^{\infty}\sum_{s=1}^{\infty}\dfrac{(2\gamma\lambda_s(\mathsf{H}_{\bm \eta}^\circ)t)^{k}}{k} \text{ for all } |t|<\frac{1}{8\gamma}(\sum_{s=1}^{\infty}\lambda_s( \mathsf{H}_{\bm \eta}^\circ )^2)^{-\frac{1}{2}} . 
\end{align}

Let $\cQ_1$ and $\cQ_2$ be as in Lemma \ref{lemma:HmL2} and Lemma \ref{lemma:momentHm}, respectively, and define $\cQ_0= \cQ_1 \bigcap Q_2$. On $\cQ_0$, there exists a constant $C > 0$ such that $\|\hat{\bm{\mathsf{H}}}_{\bm \eta}^\circ \| < C$ (by \eqref{eq:HmL2}) and $\|\mathsf{H}_{\bm \eta}^\circ\|<C$ (since $\mathsf{H}_{\bm \eta}^\circ \in L^2(\cX^2, P^2)$). Then by \eqref{eq:MGF-Hm-log} and \eqref{eq:MGF-h-tilde-log}, both MGF's exists for $|t|\leq \frac{1}{8\gamma C}$ on $\cQ_0$. Hereafter, we will assume that the we are on the set $\cQ_0$. Using 
$\sum_{s=1}^{\infty}\lambda_s(\mathsf{H}_{\bm \eta}^\circ)^2 = \|\mathsf{H}_{\bm \eta}^\circ\|^2<\infty$ gives,  for all $s \geq 1$, 
\begin{align}\label{eq:bdd-lambda}
    \left|\lambda_s(\mathsf{H}_{\bm \eta}^\circ)\right|\leq \dfrac{\|\mathsf{H}_{\bm \eta}^\circ\|}{\sqrt{s}}\leq\dfrac{C}{\sqrt{s}} . 
\end{align} 
Then notice that for all $|t| \leq \frac{1}{8\gamma C} <\frac{1}{8\gamma}\|\mathsf{H}_{\bm \eta}^\circ\|^{-1}$,
\begin{align}\label{eq:Sh-mgf-abs-summable}
    \sum_{k=3}^{\infty}\sum_{s=1}^{\infty}\left|\dfrac{(2\gamma\lambda_s(\mathsf{H}_{\bm \eta}^\circ)t)^k}{k}\right|\leq \sum_{k=3}^{\infty}\sum_{s=1}^{\infty}\dfrac{(2\gamma)^{k}\|\mathsf{H}_{\bm \eta}^\circ\|^k}{(8\gamma)^{k}\|\mathsf{H}_{\bm \eta}^\circ\|^{k} s^{k/2}k}\leq \sum_{k=3}^{\infty}\sum_{s=1}^{\infty}\dfrac{1}{4^{k} s^{3/2}k}<\infty , 
\end{align} 
and hence the second term in \eqref{eq:MGF-h-tilde-log} is absolutely summable.
Similarly, since $\sum_{s=1}^{\infty}\lambda_s(\hat{\bm{\mathsf{H}}}_{\bm \eta}^\circ)^2 = \|\hat{\bm{\mathsf{H}}}_{\bm \eta}^\circ \|^2\leq C$, 
\begin{align}\label{eq:bdd-lambda-m}
    \left|\lambda_s(\hat{\bm{\mathsf{H}}}_{\bm \eta}^\circ)\right|\leq \dfrac{C}{\sqrt{s}}, 
 \end{align}
 for all $s \geq 1$, and the second term in \eqref{eq:MGF-Hm-log} is also absolutely summable for all $|t|\leq\frac{1}{8\gamma C} $. 
Now, for any $N \geq 1$ and for all $|t|\leq \frac{1}{8\gamma C}$, 
\begin{align}\label{eq:bdd-MGF-expansion} 
    \bigg|\sum_{k=3}^{\infty}\sum_{s=1}^{\infty}\dfrac{(2\gamma\lambda_s(\hat{\bm{\mathsf{H}}}_{\bm \eta}^\circ)t)^k}{k}
    & - \sum_{k=3}^{\infty}\sum_{s=1}^{\infty}\dfrac{(2\gamma\lambda_s(\mathsf{H}_{\bm \eta}^\circ)t)^k}{k}\bigg|\nonumber\\
    \leq &  \left|\sum_{k=3}^{N}\sum_{s=1}^{\infty}\dfrac{(2\gamma\lambda_s( \hat{\bm{\mathsf{H}}}_{\bm \eta}^\circ )t)^k}{k} - \sum_{k=3}^{N}\sum_{s=1}^{\infty}\dfrac{(2\gamma\lambda_s(\mathsf{H}_{\bm \eta}^\circ)t)^k}{k}\right|\nonumber\\
    &+ \sum_{k=N+1}^{\infty}\left[\left|\sum_{s=1}^{\infty}\dfrac{(2\gamma\lambda_s( \hat{\bm{\mathsf{H}}}_{\bm \eta}^\circ )t)^k}{k}-\sum_{s=1}^{\infty}\dfrac{(2\gamma\lambda_s(\mathsf{H}_{\bm \eta}^\circ)t)^k}{k}\right|\right] . 
\end{align} 
Note that the first term in \eqref{eq:bdd-MGF-expansion} converges to zero as $m \rightarrow 0$ on $\cQ_0$ (by \eqref{eq:momentHm}). Therefore, it suffices to show that the second term converges to zero in \eqref{eq:bdd-MGF-expansion}. Towards this, note that for $k\geq 3$, by \eqref{eq:bdd-lambda} and \eqref{eq:bdd-lambda-m}, 
\begin{align*}
    \left|\sum_{s=1}^{\infty}\dfrac{(2\gamma\lambda_s( \hat{\bm{\mathsf{H}}}_{\bm \eta}^\circ )t)^k}{k}-\sum_{s=1}^{\infty}\dfrac{(2\gamma\lambda_s(\mathsf{H}_{\bm \eta}^\circ)t)^k}{k}\right|
    & \leq \dfrac{(2\gamma)^{k}|t|^{k}}{k}\sum_{s=1}^{\infty} \left\{ \left|\lambda_s( \hat{\bm{\mathsf{H}}}_{\bm \eta}^\circ )\right|^{k} + \left|\lambda_s(\mathsf{H}_{\bm \eta}^\circ)\right|^{k} \right\} \\
    & \leq \dfrac{2C^{k}(2\gamma)^{k}|t|^{k}}{k}\sum_{s=1}^{\infty}\frac{1}{s^{k/2}}\\
    &\leq \dfrac{2C^{k}(2\gamma)^{k}|t|^{k}}{k} \sum_{s=1}^{\infty}\frac{1}{s^{3/2}} .  
\end{align*}
Then for $|t|\leq \frac{1}{8\gamma C}$, 
\begin{align*} 
\sum_{k=N+1}^{\infty}\left[\left|\sum_{s=1}^{\infty}\dfrac{(2\gamma\lambda_s( \hat{\bm{\mathsf{H}}}_{\bm \eta}^\circ )t)^k}{k}-\sum_{s=1}^{\infty}\dfrac{(2\gamma\lambda_s(\mathsf{H}_{\bm \eta}^\circ)t)^k}{k}\right|\right] \leq 4 \sum_{s=1}^{\infty}\frac{1}{s^{3/2}} \sum_{k=N+1}^{\infty}\dfrac{1}{4^{k}k} , 
\end{align*}
which converges to zero as $m\rightarrow\infty$ and then $N\rightarrow\infty$. This implies the RHS of \eqref{eq:bdd-MGF-expansion} converges to zero as $m\rightarrow\infty$ and then $N\rightarrow\infty$. 
Thus, by \eqref{eq:MGF-Hm-log}, \eqref{eq:MGF-h-tilde-log}, and Lemma \ref{lemma:HmL2}, on the set $\cQ_0$, 
\begin{align*}
    \lim_{m \rightarrow \infty} M_{Y_{m}|\sX_{m}}(t) = M_{Z(\mathsf{H}_{\bm \eta}^\circ)}(t), \text{ for all } |t|<\frac{1}{8\gamma C} . 
\end{align*} 
Hence, recalling \eqref{eq:def-Ym}, $Y_{m}|\sX_{m} \overset{D}{\rightarrow}Z(\mathsf{H}_{\bm \eta}^\circ)$ on the set $\cQ_0$. Since $\P(\cQ_0) = 1$, this completes the proof of Lemma \ref{lemma:convergence}. \hfill $\Box$

\section{Proof of Theorem \ref{thm:H0NK}} 
\label{sec:localpowerpf}

Suppose $\mathsf H \in L^{2}(\mathcal{X}^2, P^2)$ is a measurable and symmetric function (not necessarily positive definite) and recall the definition of $\emmd \left[\mathsf H, \sX_m, \sY_n \right] $ from \eqref{eq:H}. Note that 
\begin{align}\label{eq:MMDK}
\emmd \left[\mathsf H, \sX_m, \sY_n \right] = \mathcal W_{\sX_m} + \mathcal W_{\sY_n} - 2 \mathcal B_{\sX_m, \sY_n} = \mathcal W_{\sX_m}^\circ + \mathcal W_{\sY_n}^\circ - 2 \mathcal B_{\sX_m, \sY_n}^\circ , 
\end{align}
where $\mathsf{H}^\circ$ is as in \eqref{eq:Kxycentered} and 
\begin{align*}
 \mathcal W_{\sX_m}^\circ :=   \frac{1}{m(m-1)}\sum_{1 \leq i \ne j \leq m} \mathsf H^\circ \left(X_{i},X_{j}\right) \text{ and }  \mathcal W_{\sY_n}^\circ := \frac{1}{n(n-1)}\sum_{1 \leq i \ne j \leq n} \mathsf H^\circ \left(Y_{i},Y_{j}\right) 
 \end{align*} 
and  $$\mathcal B_{\sX_m, \sY_n}^\circ :=  \frac{1}{mn}\sum_{i=1}^{m}\sum_{j=1}^{n} \mathsf H^\circ \left(X_{i},Y_{j}\right) . $$ 
Therefore, to obtain the limiting distribution of $\emmd \left[\mathsf H, \sX_m, \sY_n \right] $ we need to derive the joint distribution of $(\mathcal W_{\sX_m}^\circ$, $\mathcal W_{\sY_n}^\circ, \mathcal B_{\sX_m, \sY_n}^\circ)$ under $H_1$. To this end, recall the definition of the Hilbert-Schmidt operator $\cH_{\mathsf{H}^\circ}$ from \eqref{eq:kernelH}. This operator has countably many eigenvalues $\{\lambda_s\}_{s \geq 1}$ with eigenvectors $\{\phi_s\}_{s\geq 1}$ satisfying: 
\begin{align}\label{eq:Heigenvectors}
    \int_{\mathcal{X}} \mathsf{H}^\circ (x,y)\phi_{s}(y) \mathrm d P(y) = \lambda_{s}\phi_{s}(x) \text{ and }\int_{\mathcal{X}}\phi_{s}(x)\phi_{s'}(x) \mathrm d P(x) = \delta_{s, s'} , 
\end{align}
for $s, s' \geq 1$ and $\delta_{s, s'}=1$ if $s=s'$ and zero otherwise. 
Note that, since $\E_{X \sim P} [\mathsf{H}^\circ (X, y)] = 0$, for all $y\in \mathcal{X}$, whenever $\lambda_{s}\neq 0$, an application of Fubini's theorem and \eqref{eq:Heigenvectors} implies that 
$\E_{X \sim P}[\phi_{s}(X)] = 0$. Moreover, (see, for example, \citet[Theorem 4, Chapter X and Section XI.6]{dunford1965linear} or \citet[Theorem 8.94 and Theorem 8.83]{renardy2006introduction}), 
\begin{align}\label{eq:Hintegral}
\sum_{s=1}^\infty \lambda_s^2 = \int_{\cX^2} \mathsf{H}^\circ(x, y)^2 \mathrm d x \mathrm dy = \| \mathsf{H}^\circ \|^2 < \infty , 
\end{align}
and the spectral theorem, 
\begin{align}\label{eq:Hexpansion}
    \mathsf{H}^\circ (x,y) = \sum_{s=1}^{\infty} \lambda_{s}\phi_{s}(x)\phi_{s}(y),
\end{align}
where the convergence is in $L^2$.

The following result gives the joint distribution of $(\mathcal W_{\sX_m}^\circ$, $\mathcal W_{\sY_n}^\circ, \mathcal B_{\sX_m, \sY_n}^\circ)$ and, hence, that of $\emmd \left[\mathsf H, \sX_m, \sY_n \right]$ under contiguous local alternatives \eqref{eq:H0N} in the contamination model \eqref{eq:fPQ}.

\begin{proposition}  
\label{ppn:H0N} 
Suppose $\mathsf{H} \in L^{2}(\mathcal{X}^2, P^2)$ be a measurable and symmetric function. Then under $H_1$ as in \eqref{eq:fPQ} in the asymptotic regime \eqref{eq:mn}, 
\begin{align}\label{eq:QmnLH2}
  \begin{pmatrix}
    m \cW_{\sX_m}^{\circ}  \\ 
    n \cW_{\sY_n}^{\circ}  \\ 
    \sqrt{mn} \cB_{\sX_m, \sY_n}^{\circ} 
        \end{pmatrix} 
    \dto 
    \begin{pmatrix}
    \sum_{s=1}^\infty \lambda_{s}\left(W_{s}^{2}-1\right)\\
    \sum_{s=1}^\infty \lambda_{s}\left( \left(W_s' + h \sqrt{1-\rho} L_s \right)^{2}-1\right)\\
    \sum_{s=1}^\infty \lambda_{s} W_s \left( W_s'+ h \sqrt{1-\rho} L_s \right)
    \end{pmatrix} , 
\end{align}
where 
\begin{itemize} 

\item $\{W_{s}, W'_{s}: s \geq 1\}$ are independent standard Gaussian random variables, 

\item $\{\lambda_{s}\}_{s \geq 1}$  are the eigenvalues (with repetitions) and the eigenvectors $\{\phi_s\}_{s\geq 1}$ of the Hilbert-Schimdt operator $\mathcal{H}_{\mathsf{H}_{\bm \eta}^\circ}$ as in \eqref{eq:Heigenvectors},  

\item $L_s :=\E_{X \sim P} [ \frac{\phi_s(X) g(X)}{f_P(X)} ]$, for $s \geq 1$.

\end{itemize}
Consequently, under $H_1$, 
\begin{align}\label{eq:H0Ndistribution}
    (m+n)\mathrm{MMD}^2\left[\mathsf{H}, \sX_m,\sY_n\right]   \dto \tilde Z(\mathsf{H})  :=  \gamma \sum_{s=1}^{\infty}\lambda_{s}\left(\left(Z_s + \frac{h}{\sqrt{\gamma}} L_s \right)^{2} - 1 \right)  , 
    \end{align}
where $\gamma = \frac{1}{\rho(1-\rho)}$, $\{Z_s : s \geq 1\}$ are i.i.d. $\cN(0, 1)$. Moreover, 
\begin{align}\label{eq:ZHexpectation}
\E[\tilde Z(\mathsf{H})] = h^2 \sum_{s=1}^\infty \lambda_s L_s^2 = h^2 \E_{X, X' \sim P} \left[\mathsf{H}^\circ(X, X') \frac{g(X) g(X') }{f_P(X) f_P(X')} \right]  < \infty
\end{align} 
and the characteristic function of $\tilde Z(\mathsf{H})$ at $t \in \R$ is  given by: 
\begin{align}\label{eq:H0NexpZ}
\Phi_{\tilde Z(\mathsf{H})} (t)  := \E\left[ e^{\iota t \tilde Z(\mathsf{H})} \right]   = \frac{e^{ \iota t h^2 \sum_{s=1}^\infty \lambda_s L_s^2 - \sum_{s=1}^\infty \left\{  \iota \gamma \lambda_s  t   + \frac{ \gamma h^2 \lambda_s^2 L_s^2 t^2 }{( 1 -2 \iota \lambda_s \gamma t ) }   \right\} } }{ \prod_{s=1}^\infty \sqrt{1- 2  \iota \gamma \lambda_s  t }} .  
\end{align}  
\end{proposition}

The proof of Proposition \ref{ppn:H0N} is given in Section \ref{sec:H0Npf}. The first step in the proof is to truncate the asymptotic expansions of $(\mathcal W_{\sX_m}^\circ$, $\mathcal W_{\sY_n}^\circ, \mathcal B_{\sX_m, \sY_n}^\circ)$ based on the spectral theorem. Next, we show that the summands that appear in the truncation together with the log-likelihood ratio has a multivariate normal distribution under $H_0$. Then invoking Le Cam's third lemma \citep[Example 6.7]{van2000asymptotic} we obtain the asymptotic distribution of truncated between and within kernel discrepancies under $H_1$ (Lemma \ref{lm:H0NUVmn}). To complete the proof we need to show that tail in the truncation is asymptotically negligible, which leverages the fact the eigenvalues of the kernels are summable in the $L_2$ sense (see Lemma \ref{lm:L} and Lemma \eqref{lm:Wxymn}).

\begin{remark}
The arguments in the proof of Proposition \ref{ppn:H0N} is similar those in \citet{chikkagoudar2014limit} on the limiting distribution of  degenerate two-sample $U$-statistics for parametric contiguous alternatives. However, the computations are different for the contamination model and also, because our ultimate goal is to derive the joint distribution of the vector of MMD estimates under local alternatives, it is convenient for us in Theorem \ref{thm:H0NK} to express the limiting distribution in terms of stochastic integrals.
\end{remark}

We now apply Proposition \ref{ppn:H0N} show it can be used to complete the proof of Theorem \ref{thm:H0NK}. As in \eqref{eq:etaMMD}, for $\bm \eta = (\eta_1, \eta_2, \ldots, \eta_r)^\top \in \mathbb{R}^r$,
\begin{align}\label{eq:etaH0N}
     \bm \eta^\top   \mathrm{MMD}^{2}[ \cK, \sX_m, \sY_n ] = \sum_{a=1}^r  \eta_a \mathrm{MMD}^{2}[ \sfK_a, \sX_m, \sY_n ] & = \mathrm{MMD}^{2}[ \mathsf{H}_{\bm \eta}, \sX_m, \sY_n ] , 
\end{align} 
where $\mathsf{H}_{\bm \eta} := \sum_{a=1}^r \eta_a \sfK_a$. 
Then by Proposition \ref{ppn:H0N}, under $H_1$, 
\begin{align}\label{eq:ZHlinear}
Z_{m, n}(\mathsf{H}_{\bm \eta})  & := (m+n)   \mathrm{MMD}^{2} \left[\mathsf{H}_{\bm \eta}, \sX_m, \sY_n \right] \nonumber \\ 
& \dto \tilde Z(\mathsf{H}_{\bm \eta}) = \gamma \sum_{s=1}^{\infty}\lambda_{s}\left(\left(Z_s + \frac{h}{\sqrt{\gamma}} L_s \right)^{2} - 1 \right) , 
\end{align}
where, by \eqref{eq:H0Ndistribution}, $\{\lambda_{s}\}_{s \geq 1}$ are the eigenvalues (with repetitions) and the eigenvectors $\{\phi_s\}_{s\geq 1}$ of the operator $\cH_{\mathsf{H}_{\bm \eta}^\circ}$. 

Note that by the linearity of the stochastic integral and arguments as in \eqref{eq:alternatelimit}, 
\begin{align}\label{eq:ZHT1}
\sum_{a=1}^r  \eta_a  I_2(\sfK_a^\circ)  = I_2(\mathsf{H}_{\bm \eta}^\circ) 
 = \int_{\cX} \int_{\cX} \mathsf{H}_{\bm \eta}^\circ (x, y) \mathrm d \cZ_P(x) \mathrm d \cZ_P(y)   \stackrel{D} = \sum_{s=1}^{\infty}\lambda_{s} \left( Z_s^2 - 1 \right) . 
\end{align} 
where $\{Z_s\}_{s \geq 1} \stackrel{D} = \left\{\int_{\cX} \phi_s(x) \mathrm d \cZ_P(x)\right \}_{s \geq 1}$. This also implies, 
\begin{align}\label{eq:ZHT2}
\sum_{s=1}^\infty \lambda_s \E_{X \sim P} \left[ \frac{\phi_s(X) g(X)}{f_P(X)} \right] Z_s & = \sum_{s=1}^\infty \lambda_s  Z_s\int_{\cX} \phi_s(x)g(x)\mathrm d x\nonumber \\ 
& \stackrel{D} = \sum_{s=1}^\infty \lambda_s \int \phi_s(x) g(x) \mathrm d x  \left( \int_{\cX} \phi_s(y) \mathrm d \cZ_P(y) \right)\nonumber \\ 
& = \left(\int_{\cX} \mathsf{H}_{\bm \eta}^\circ (x, y) g(x) \mathrm d x \right) \mathrm d \cZ_P(y) \nonumber \\ 
& = I_1\left(\mathsf{H}_{\bm \eta}^{\circ}\left[\frac{g}{f_P}\right]\right) = \sum_{a=1}^r  \eta_a  I_1\left(\sfK_a^{\circ}\left[\frac{g}{f_P}\right]\right) , 
\end{align}  
where the notations are as defined in Theorem \ref{thm:H0NK}. 
Also, by \eqref{eq:ZHexpectationpf}, 
\begin{align}\label{eq:ZHT3} 
\sum_{s=1}^{\infty} \lambda_s L_s^2 & = \E_{X, X' \sim P} \left[ \mathsf{H}_{\bm \eta}^\circ(X, X') \frac{g(X) g(X') }{f_P(X) f_P(X')}  \right] \nonumber \\ 
 & = \sum_{a=1}^r  \eta_a \E_{X, X' \sim P} \left[ \mathsf{K}_a ^\circ(X, X') \frac{g(X) g(X') }{f_P(X) f_P(X')}  \right] .
\end{align} 
Using \eqref{eq:ZHT1}, \eqref{eq:ZHT2} and \eqref{eq:ZHT3} in \eqref{eq:ZHlinear} shows that $\tilde Z(\mathsf{H}_{\bm \eta}) \stackrel{D} = \bm \eta^\top G_{\cK, h}$ 
where $G_{\cK, h} $ is as defined in \eqref{eq:H0NGK}. This implies, from \eqref{eq:etaH0N},  $$\bm \eta^\top   \mathrm{MMD}^{2}[ \cK, \sX_m, \sY_n ] \dto \bm \eta^\top G_{\cK, h}.$$ Since $\bm \eta \in \R^r$ is arbitrary, this completes the proof of Theorem \ref{thm:H0NK}.

\subsection{Proof of Proposition \ref{ppn:H0N}}
\label{sec:H0Npf}

Fix $L \geq 1$ and define the $L$-truncated versions of $\mathcal W_{\sX_m}$, $\mathcal W_{\sY_n}$, and $\mathcal B_{\sX_m, \sY_n} $ as follows: 
\begin{align}\label{eq:BWxy}
 \mathcal W_{\sX_m}^{\circ (L)}  & := \frac{1}{m(m-1)}  \sum_{s=1}^{L} \sum_{1 \leq i\neq j \leq m}  \lambda_{s}\phi_{s}(X_i)\phi_{s}(X_j) \nonumber\\
 & = \frac{1}{m(m-1)}\sum_{s=1}^{L} \lambda_{s}\left(\left(\sum_{i=1}^m \phi_{s}(X_{i})\right)^{2}-\sum_{i=1}^m \phi_{s}^{2}(X_{i})\right), \nonumber \\
\mathcal W_{\sY_n}^{\circ (L)}   & := \frac{1}{n(n-1)}  \sum_{s=1}^{L} \sum_{1 \leq i\neq j \leq n}  \lambda_{s}\phi_{s}(Y_i)\phi_{s}(Y_j) \nonumber \\
& = \frac{1}{n(n-1)}\sum_{s=1}^{L} \lambda_{s}\left(\left(\sum_{i=1}^n \phi_{s}(Y_{i})\right)^{2}-\sum_{i=1}^n \phi_{s}^{2}(Y_{i})\right) ,  \nonumber \\ 
\mathcal B_{\sX_m, \sY_n}^{\circ (L)} & :=  \frac{1}{m n} \sum_{s=1}^{L}  \lambda_{s} \sum_{i=1}^m \sum_{j=1}^n \phi_{s}(X_i)\phi_{s}(Y_j) . 
\end{align} 
Define $U_{s, m} := \frac{1}{\sqrt m} \sum_{i=1}^m \phi_s(X_i)$ and $V_{s, n} := \frac{1}{\sqrt n} \sum_{i=1}^n \phi_s(Y_i)$, for $1 \leq s \leq L$ and the vectors 
\begin{align}\label{eq:UVm}
\bm U_{m}^{(L)} = (U_{s, m})_{1 \leq s \leq L} \quad \text{ and } \quad \bm V_{n}^{(L)} = (V_{s, n})_{1 \leq s \leq L}. 
\end{align}
Note that, under $H_0$, by the law of large numbers and \eqref{eq:Heigenvectors} 
$$\frac{1}{m} \sum_{i=1}^m \phi^2_s(X_i) \pto \E_{X \sim P}[\phi_s(X)^2] = 1 \quad  \text{ and } \quad \frac{1}{n} \sum_{i=1}^n \phi^2_s(Y_i) \pto \E_{Y \sim P}[\phi_s(Y)^2] = 1 . $$ Hence, recalling \eqref{eq:BWxy}, under $H_0$, 
\begin{align}\label{UxVy}
(m-1)\mathcal W_{\sX_m}^{\circ (L)}  = \sum_{s=1}^{L} \lambda_{s}\left(U_{s, m}^2  - 1 \right) + o_P(1), \quad (n-1) \mathcal W_{\sY_n}^{\circ (L)}  = \sum_{s=1}^{L} \lambda_{s}\left(V_{s, n}^2  - 1 \right) + o_P(1) , 
\end{align} 
and 
\begin{align}\label{eq:UVxy}
\sqrt{mn} \mathcal B_{\sX_m, \sY_n}^{\circ (L)} & :=  \sum_{s=1}^{L}  \lambda_{s} U_{s, m} V_{s, n} . 
\end{align}
Therefore, obtain the joint distribution of  $(\mathcal W_{\sX_m}^{\circ (L)}, \mathcal W_{\sY_n}^{\circ (L)}, \mathcal B_{\sX_m, \sY_n}^{\circ (L)})^\top$ it suffices to find the joint distribution of 
\begin{align}\label{eq:QmnL}
\bm Q_{m, n}^{(L)} := \left( \sum_{s=1}^{L} \lambda_{s}\left(U_{s, m}^2  - 1 \right), \sum_{s=1}^{L} \lambda_{s}\left(V_{s, n}^2  - 1 \right),  \sum_{s=1}^{L}  \lambda_{s} U_{s, m} V_{s, n} \right)^\top .  
\end{align}
This is derived in the following lemma. Here, $\bm 0$ denotes the zero vector in $\R^L$ and $\bm{I}_{2L}$ denotes the $2L \times 2L$ identity matrix.

\begin{lemma}\label{lm:H0NUVmn} Fix $L \geq 1$ and suppose  $\bm U_{m}^{(L)}$ and $\bm V_{n}^{(L)}$ be as defined in \eqref{eq:UVm}. Then under $H_1$ as in \eqref{eq:H0N} the following hold in the asymptotic regime \eqref{eq:mn}, 
\begin{align}\label{eq:UVmnH1}
\begin{pmatrix}
    \bm U_{m}^{(L)} \\ 
    \bm V_n^{(L)} 
        \end{pmatrix}
  \dto \cN_{2L} \left ( 
        \begin{pmatrix}
   \bm 0 \\
   \bm \theta  
     \end{pmatrix} , \bm{I}_{2L}
        \right) . 
\end{align} 
where $\bm \theta := h \sqrt{1 - \rho} \cdot ( L_1, L_2 \ldots, L_r)^\top$ and $L_s$ as in Proposition \ref{ppn:H0N}. Consequently, 
\begin{align}\label{eq:QmnLH1}
  \begin{pmatrix}
    (m-1) \cW_{\sX_m}^{\circ (L)}  \\ 
    (n-1) \cW_{\sY_n}^{\circ (L)}  \\ 
    \sqrt{mn} \cB_{\sX_m, \sY_n}^{\circ (L)} 
        \end{pmatrix} 
    \dto 
      \bm Q^{(L)} : = 
    \begin{pmatrix}
    \sum_{s=1}^L\lambda_{s}\left(W_{s}^{2}-1\right)\\
    \sum_{s=1}^L \lambda_{s}\left( \left(W_s' + h \sqrt{1-\rho} L_s \right)^{2}-1\right)\\
    \sum_{s=1}^L \lambda_{s} W_s \left( W_s'+ h \sqrt{1-\rho} L_s    \right)
    \end{pmatrix} , 
\end{align} 
where $\{W_{s}, W'_{s}: s \geq 1\}$ are independent standard Gaussian random variables. 
\end{lemma}

\begin{proof} To prove \eqref{eq:UVmnH1} we will first derive the joint distribution of $ \bm U_{m}^{(L)}$, $\bm V_n^{(L)}$, and the log-likelihood ratio, and then invoke LeCam's third lemma \citet[Example 6.7]{van2000asymptotic}. For the hypothesis in \eqref{eq:fPQ} the likelihood ratio $L_N$ is given by: 
$$L_N:=\sum_{i=1}^n \log\left[ \frac{\left(1-\frac{h}{\sqrt{N}}\right)f_P(Y_i)+\frac{h}{\sqrt{N}}g(Y_i)}{f_P(Y_i)}\right].$$ By local asymptotic normality (see,
for example, \citet[Chapter 7]{van2000asymptotic}), $L_N$ can be written as:
$$L_N=\dot{L}_N-\frac{(1-\rho) h^2}{2} \cdot \delta_{f_P, g} +o_{P}(1), $$ where $$\dot{L}_N=\frac{h}{\sqrt{N}}\sum_{i=1}^n \left(\frac{g(Y_i)}{f_P(Y_i)}-1\right) \text{ and } \delta_{f_P, g} := \int_{\cX} \left(\frac{g(x)}{f_P(x)}-1\right)^2 f_P(x) \mathrm d x. $$
Note that $\Cov[U_{s, m} \dot L_N] = 0$ and 
\begin{align*}
\Cov[V_{s, n} \dot L_N]  = \E[V_{s, n} \dot L_N] 
& = \sqrt{\frac{n}{N}} \cdot \frac{h}{n} \sum_{i=1}^n \E \left[ \phi_s(Y_i) \left(\frac{g(Y_i)}{f_P(Y_i)}-1\right) \right ]  \nonumber \\ 
& \pto h \sqrt{1-\rho} \cdot \E_{X \sim P} \left[ \frac{\phi_s(X) g(X)}{f_P(X)}  \right ]  ,
\end{align*}
by the law of large numbers and the fact $\E_{X \sim P} [\phi_s(X)]= 0$ (since we can assume $\lambda_s \neq 0$). Hence, by the multivariate central limit theorem, under $H_0$, 
\begin{align}\label{eq:UVmnH0}
\begin{pmatrix}
    \bm U_{m}^{(L)} \\ 
    \bm V_n^{(L)} \\ 
    \dot L_N 
        \end{pmatrix}
        \dto \cN_{2L+1} \left ( 
        \begin{pmatrix}
   \bm 0_{L \times 1}\\
   \bm 0_{L \times 1}\\
   - \frac{(1-\rho) h^2}{2} \delta_{f_P, g}
    \end{pmatrix} , 
     \begin{pmatrix}
   \bm{I}_L & \bm 0_{L \times L} & \bm 0_{L \times 1} \\
\bm 0_{L \times L} & \bm{I}_L & \bm \theta \\
   \bm 0_{1 \times L} & \bm \theta^\top & (1-\rho) h^2 \delta_{f_P, g} \\
       \end{pmatrix} 
        \right) 
    \end{align}
  where $\bm \theta$ is as defined in Lemma \ref{lm:H0NUVmn} and $\bm 0_{K \times L}$ is the $K \times L$ zero-matrix, for $K, L \geq 1$. Then by LeCam's third lemma \citep[Example 6.7]{van2000asymptotic} the result in \eqref{eq:UVmnH1} follows. 

Now, since $\bm Q_{m, n}^{(L)}$ (recall \eqref{eq:QmnL}) is a continuous function of $\bm U_{m}^{(L)}$ and $\bm V_{n}^{(L)}$, the result in \eqref{eq:QmnLH1} follows from \eqref{UxVy}, \eqref{eq:UVxy}, \eqref{eq:UVmnH0}, and the continuous mapping theorem. 
\end{proof} 

Next, we show that $\bm Q^{(L)}$ (as defined in Lemma \ref{lm:H0NUVmn}) converges as $L \rightarrow \infty$. 

\begin{lemma}\label{lm:L} Let $\bm Q^{(L)}$ be as defined in \eqref{eq:QmnLH1}. Then as $L \rightarrow \infty$, 
\begin{align}\label{eq:L}
    \bm Q^{(L)} \stackrel{L^2} \rightarrow   
    \bm Q :=   \begin{pmatrix}
    Q_1 \\
    Q_2 \\
    Q_3 
    \end{pmatrix}
    :=
       \begin{pmatrix}
    \sum_{s=1}^\infty \lambda_{s}\left(W_{s}^{2}-1\right)\\
    \sum_{s=1}^\infty \lambda_{s}\left( \left(W_s' + h \sqrt{1-\rho} L_s \right)^{2}-1\right)\\
    \sum_{s=1}^\infty \lambda_{s} W_s \left( W_s'+ h \sqrt{1-\rho} L_s \right) , 
    \end{pmatrix} , 
\end{align}
where $\{W_{s}, W'_{s}: s \geq 1\}$ are independent standard Gaussian random variables. 
\end{lemma}

\begin{proof} Note that 
$$\sum_{s=1}^{\infty} \text{Var}\left[ \lambda_{s}(W_s^2 -1)\right]  = 2 \sum_{s=1}^{\infty} \lambda_s^2 < \infty , $$
by \eqref{eq:Hintegral}. Hence, as $L \rightarrow \infty$,   
$$\sum_{s=1}^L \lambda_{s}\left(W_{s}^{2}-1\right) \stackrel{L^2} \rightarrow \sum_{s=1}^\infty \lambda_{s}\left(W_{s}^{2}-1\right) = Q_1.$$

Next, we denote $\theta_s := h \sqrt{1-\rho} L_s = h \sqrt{1-\rho}  \E_{X \sim P} [ \frac{\phi_s(X) g(X)}{f_P(X)} ]$. Then by the Cauchy-Schwarz inequality, 
\begin{align}\label{eq:Qexpectation}
\theta_s^2 & \leq h^2 (1-\rho) \E_{X \sim P}[\phi_s(X)^2]  \E_{X \sim P} \left[ \frac{g(X)^2}{f_P(X)^2} \right] \nonumber \\ 
& = h^2 (1-\rho) \E_{X \sim P} \left[ \frac{g(X)^2}{f_P(X)^2} \right] < \infty , 
\end{align}
since $\E_{X \sim P} \left[ \phi_s(X)^2 \right]  = 1$, for all $s \geq 1$, 
and $\E_{X \sim P} [ \frac{g(X)^2}{f_P(X)^2} ] < \infty$ by Assumption \ref{assumption:fPQ}. Then 
\begin{align}\label{eq:Qvariance}
\sum_{s=1}^\infty \Var\left[ \lambda_{s} W_s \left(W_s' + \theta_s  \right)  \right] & = \sum_{s=1}^\infty \lambda_{s}^2 \E[W_s^2] \E \left[\left(W_s' + \theta_s \right)^{2} \right]   = \sum_{s=1}^\infty \lambda_{s}^2 (1+ \theta_s^2)   < \infty , 
\end{align}
by \eqref{eq:Hintegral} and \eqref{eq:Qexpectation}. Hence, as $L \rightarrow \infty$,   
$$\sum_{s=1}^L \lambda_{s} W_s ( W_{s}' + \theta_s)^{2}  \stackrel{L^2} \rightarrow \sum_{s=1}^\infty \lambda_{s} W_s ( W_{s} + \theta_s)^{2} = Q_3. $$

It remains to establish the convergence to $Q_2$ in \eqref{eq:L}. Tp this end, fix $L \geq 1$ and denote 
$$Q_2^{(L)} := \sum_{s=1}^L \lambda_{s}\left( \left(W_s' + \theta_s \right)^{2}-1\right) .$$ 
Then for $L' > L \geq 1$,  
\begin{align}\label{eq:QLconvergence}
\E\left[ \left( Q_2^{(L')} - Q_2^{(L)} \right)^2 \right]  \leq \Var \left[ \left( Q_2^{(L')} - Q_2^{(L)} \right) \right] + \left( \E\left[ Q_2^{(L')} - Q_2^{(L)} \right ] \right)^2 . 
\end{align} 
This implies, 
\begin{align}\label{eq:QLvariance}
\Var\left[ \left( Q_2^{(L')} - Q_2^{(L)} \right) \right]  & = \sum_{s=L+1}^{L'} \Var\left[ \lambda_{s}\left( \left(W_s' + \theta_s \right)^{2}-1\right)  \right]  \nonumber \\ 
& = \sum_{s=L+1}^{L'} \lambda_{s}^2 \Var\left[  \left(W_s' + \theta_s \right)^{2} \right] \nonumber \\ 
& = 2 \sum_{s=L+1}^{L'} \lambda_{s}^2 (1+ 2 \theta_s^2) \rightarrow 0 ,
\end{align} 
as $L, L' \rightarrow \infty$, using \eqref{eq:Hintegral} and \eqref{eq:Qexpectation}. Next consider, 
\begin{align}
\left|\E\left[ Q_2^{(L')} - Q_2^{(L)} \right ] \right|  = \sum_{s=L+1}^{L'} \lambda_s \theta_s^2 & = h \sqrt{1-\rho}   \sum_{s=L+1}^{L'}   \lambda_s \left(\int_{\cX} \phi_s(x) g(x) \mathrm d x \right)^2 \nonumber \\ 
 & =  h \sqrt{1-\rho}    \int_{\cX} \int_{\cX}  \sum_{s=L+1}^{L'} \lambda_s  \phi_s(x) \phi_s(y) g(y)  g(y) \mathrm d x \mathrm d y.  \nonumber 
\end{align}
Denote $C:= h \sqrt{1-\rho}$ and $M:= \E_{X \sim P}[\frac{g(X)^2}{f_P(X)^2}] < \infty$ (by Assumption \ref{assumption:fPQ}). Then by the  Cauchy-Schwarz inequality, 
\begin{align}\label{eq:QLexpectation}
\left|\E\left[ Q_2^{(L')} - Q_2^{(L)} \right ] \right|^2 & \leq C^2 M^2 \int_{\cX} \int_{\cX}   \left( \sum_{s=L+1}^{L'}  \lambda_s  \phi_s(x) \phi_s(y)  \right)^2 f_P(x) f_P(y) \mathrm d x \mathrm d y \nonumber \\ 
& \rightarrow 0 , 
\end{align}
as $L, L' \rightarrow \infty$, since the convergence in \eqref{eq:Hexpansion} is in $L^2$. 
Combining \eqref{eq:QLexpectation} and \eqref{eq:QLvariance} with \eqref{eq:QLconvergence} it follows that  $Q_2^{(L)}$ converges in $L^2$ to $Q_2$. 
This completes the proof of Lemma \ref{lm:L}.  
\end{proof}

Next, we show that $(\cW_{\sX_m}^{\circ (L)} , \cW_{\sY_n}^{\circ (L)} , \cB_{\sX_m, \sY_n}^{\circ (L)} )^\top$ (recall \eqref{eq:BWxy}) and $(\cW_{\sX_m}^\circ, \cW_{\sY_n}^\circ, \cB_{\sX_m, \sY_n}^\circ)^\top$ are asymptotically close. 

\begin{lemma}\label{lm:Wxymn} As $L \rightarrow \infty$, 
\begin{align*} 
\sup_{m, n \geq 1}\E\left\|\begin{pmatrix}
  (m-1)  \cW_{\sX_m}^{\circ (L)}  \\ 
  (n-1)  \cW_{\sY_n}^{\circ (L)}  \\ 
   \sqrt{mn} \cB_{\sX_m, \sY_n}^{\circ (L)} 
        \end{pmatrix} - \begin{pmatrix}
   (m-1) \cW_{\sX_m}^\circ \\ 
   (n-1) \cW_{\sY_n}^\circ \\ 
    \sqrt{mn} \cB_{\sX_m, \sY_n}^\circ 
        \end{pmatrix}\right\|^2 \rightarrow 0 . 
\end{align*}  
\end{lemma}

\begin{proof}   
Note that by \eqref{eq:Hexpansion} and Fubini's theorem, 
\begin{align}
   \begin{pmatrix}
   (m-1) \cW_{\sX_m}^\circ \\ 
   (n-1) \cW_{\sY_n}^\circ \\ 
   \sqrt{mn} \cB_{\sX_m, \sY_n}^\circ 
        \end{pmatrix} & =
           \begin{pmatrix}
        \frac{1}{m}\sum_{1 \leq i\neq j \leq m} \mathsf H^\circ(X_{i},X_{j})\\
        \frac{1}{n} \sum_{1 \leq i\neq j \leq n} \mathsf H^\circ(Y_{i},Y_{j})\\
        \frac{1}{\sqrt{mn}} \sum_{i=1}^m \sum_{i=1}^n \mathsf H^\circ(X_{i},Y_{j})
    \end{pmatrix} \nonumber \\ 
    & =   
     \begin{pmatrix}
   \frac{1}{m}   \sum_{s=1}^{\infty} \lambda_{s} \sum_{1 \leq i\neq j \leq m} \phi_{s}(X_i)\phi_{s}(X_j) \\
    \frac{1}{n} \sum_{s=1}^{\infty} \lambda_{s} \sum_{1 \leq i\neq j \leq m}  \phi_{s}(Y_i)\phi_{s}(Y_j) \\
    \frac{1}{\sqrt{mn}}\sum_{s=1}^{\infty}\lambda_{s} \sum_{i=1}^{m} \sum_{j=1}^{n} \phi_{s}(X_{i}) \phi_{s}(Y_{j}) 
    \end{pmatrix} \label{eq:exist-infinite-sum}
    \end{align} 
where the existence of such infinite sums will be proved in the following. 
Using $\E_{X \sim P}[\phi_{s}(X)]=0$ for $s \geq 1$, it is easy to show that for $s \neq s'$,
\begin{align*}
    \E\left[\left(\sum_{1 \leq i \ne j \leq m} \phi_{s}(X_{i})\phi_{s}(X_{j})\right)\left(\sum_{1 \leq i \ne j \leq m} \phi_{s'}(X_{i})\phi_{s'}(X_{j})\right)\right] = 0
\end{align*}
and using \eqref{eq:Heigenvectors}, 
\begin{align*}
    \E\left[\left(\sum_{1 \leq i \ne j \leq m} \phi_{s}(X_{i})\phi_{s}(X_{j})\right)^2\right] = 4\E\left[\sum_{1 \leq i<j \leq m}\left(\phi_{s}(X_{i})\phi_{s}(X_{j})\right)^2\right] =  2 m(m-1) . 
\end{align*}
Define $c_m = \sqrt{2 m(m-1)}$. Then $\{ \frac{1}{c_m}\sum_{1 \leq i \ne j \leq m}\phi_{s}(X_{i})\phi_{s}(X_{j}) \}_{s \geq 1}$ is a collection of orthonormal random variables. Hence, by \citet[Lemma 6.8]{MR1892228} the  following infinite sum 
\begin{align*}
    \sum_{s=1}^{\infty}\lambda_{s}\left(\frac{1}{c_m}\sum_{1 \leq i \ne j \leq m} \phi_{s}(X_{i})\phi_{s}(X_{j})\right) 
\end{align*}
exists, which also proves the existence of the infinite sum in \eqref{eq:exist-infinite-sum}. Moreover, 
\begin{align*}
 \E\left[ (m-1)\left(\cW_{\sX_m}^\circ - \cW_{\sX_m}^{\circ (L)}  \right)\right]^2  
 & \leq \frac{m^2}{c_m^2} \E\left[\sum_{s=L+1}^{\infty}\lambda_{s}\left(\frac{1}{c_m}\sum_{1 \leq i \ne j \leq n}\phi_{s}(X_{i})\phi_{s}(X_{j})\right)\right]^{2}\\
 & \leq  \sum_{s=L+1}^{\infty}\lambda_{s}^2 \rightarrow 0, 
\end{align*} 
as $L \rightarrow \infty$ (recall \eqref{eq:Hintegral}), uniformly in $m, n$. Similarly, it can be shown that 
$$(n-1)^2\E\left[ \cW_{\sY_n} - \cW_{\sY_n}^{\circ (L)}  \right]^2  \rightarrow 0 \text{ and } mn\E\left[ \cB_{\sX_m, \sY_n} - \cB_{\sX_m, \sY_n}^{(L)} \right]^2  \rightarrow 0$$ as $L \rightarrow \infty$, uniformly in $m, n$. 
\end{proof}

Combining Lemmas \ref{lm:H0NUVmn}, \ref{lm:L}, \ref{lm:Wxymn} and using \citet[Lemma 6]{jansonustat} we get,  
\begin{align*}
  \begin{pmatrix}
    (m-1) \cW_{\sX_m}^{\circ}  \\ 
    (n-1) \cW_{\sY_n}^{\circ}  \\ 
    \sqrt{mn} \cB_{\sX_m, \sY_n}^{\circ} 
        \end{pmatrix} 
    \dto 
     \begin{pmatrix}
    \sum_{s=1}^\infty \lambda_{s}\left(W_{s}^{2}-1\right)\\
    \sum_{s=1}^\infty \lambda_{s}\left( \left(W_s' + \theta_s \right)^{2}-1\right)\\
    \sum_{s=1}^\infty \lambda_{s} W_s \left( W_s'+ \theta_s    \right)
    \end{pmatrix} , 
\end{align*} 
where $\theta_s := h \sqrt{1-\rho} \cdot \E_{X \sim P}[ \frac{\phi_s(X) g(X)}{f_P(X)}  ]$, for $s \geq 1$. This establishes \eqref{eq:QmnLH2}. Then \eqref{eq:MMDK} and the continuous mapping theorem gives, 
\begin{align}\label{eq:MMDKH1}
& (m+n)  \emmd \left[\mathsf H, \sX_m, \sY_n \right] \nonumber \\ 
& = (m+n)  \mathcal W_{\sX_m}^\circ + (m+n) \mathcal W_{\sY_n}^\circ - 2 (m+n) \mathcal B_{\sX_m, \sY_n}^\circ \nonumber \\  
& \dto \frac{1}{\rho} \sum_{s=1}^\infty \lambda_{s}\left(W_{s}^{2}-1\right) + \frac{1}{1-\rho} \sum_{s=1}^\infty \lambda_{s}\left( \left(W_s' + \theta_s \right)^{2}-1\right) + \frac{2}{\sqrt{ \rho(1-\rho) }} \sum_{s=1}^\infty \lambda_{s} W_s \left( W_s'+ \theta_s    \right)  \nonumber \\ 
& = \sum_{s=1}^{\infty}\lambda_{s}\left(\left(\frac{1}{\sqrt{\rho}}  W_s -\frac{1}{\sqrt{1-\rho}} W_s' + h \E_{X \sim P}\left[ \frac{\phi_s(X) g(X)}{f_P(X)}  \right ]\right)^{2}-\frac{1}{\rho(1-\rho)} \right)  \\ 
& \stackrel{D} =  \frac{1}{\rho(1-\rho)} \sum_{s=1}^{\infty}\lambda_{s}\left(\left( Z_s + h \sqrt{\rho(1-\rho)} \E_{X \sim P}\left[ \frac{\phi_s(X) g(X)}{f_P(X)}  \right ]\right)^{2} - 1 \right) , \nonumber 
\end{align}
where $\{Z_s\}_{s \geq 1}$ are i.i.d. $\cN(0, 1)$ and the rearrangement of the terms in \eqref{eq:MMDKH1} can be justified by truncation and taking limits. This completes the proof of \eqref{eq:H0Ndistribution}. 

To show \eqref{eq:ZHexpectation} note that 
\begin{align} 
\E[\tilde Z(\mathsf{H})] = h^2 \sum_{s=1}^{\infty} \lambda_s L_s^2 & = h^2  \sum_{s=1}^{\infty}   \lambda_s \left(\int_{\cX} \phi_s(x) g(x) \mathrm d x \right)^2 \nonumber \\ 
 & =  h^2    \int_{\cX} \int_{\cX}  \sum_{s=1}^{\infty} \lambda_s  \phi_s(x) \phi_s(y) g(y)  g(y) \mathrm d x \mathrm d y \label{eq:exintexp}\\ 
 & =  h^2    \int_{\cX} \int_{\cX} \mathsf{H}^\circ(x, y) g(y)  g(y) \mathrm d x \mathrm d y  \tag*{(by \eqref{eq:Hintegral})} \nonumber \\ 
&=h^2 \E_{X, X' \sim P} \left[\mathsf{H}^\circ(X, X') \frac{g(X) g(X') }{f_P(X) f_P(X')}  \right]  < \infty \label{eq:ZHexpectationpf},
\end{align}
where the exchange of expectation and integral in \eqref{eq:exintexp} is valid by arguments similar to \eqref{eq:QLexpectation} and finiteness of the expectation follows by the Cauchy-Schwarz inequality, Assumption \ref{assumption:fPQ}, and the fact $\mathsf{H}^\circ \in L^2(\cX^2, P^2)$. Finally, the expression for the characteristic function in \eqref{eq:H0NexpZ} follows from \citet[Theorem 6.2]{janson1997gaussian}, since 
$$\sum_{s=1}^\infty \lambda_s^2 L_s^2 = \sum_{s=1}^\infty  \lambda_s^2 \left( \E_{X \sim P}  \left[ \frac{\phi_s(X) g(X)}{f_P(X)} \right] \right)^2 < \infty , $$
by arguments as in \eqref{eq:Qvariance}. This completes the proof of Proposition \ref{ppn:H0N}.  \hfill $\Box$

\begin{remark}\label{remark:localpower}
Note that Theorem \ref{thm:H0NK}, the continuous mapping theorem, and Corollary \ref{cor:TM} implies that under $H_1$ as in \eqref{eq:H0N},  
\begin{align}\label{eq:GKH1}
(m+n)^2 T_{m, n} \dto G_{\cK, h}^\top \bm \Sigma_{H_0}^{-1}G_{\cK, h} . 
\end{align}
This allows us to derive the limiting local power of the test $\phi_{m, n}$ in \eqref{eq:Tmnalpha}. Specifically, suppose $F_{\cK, h}$ denotes the CDF of $G_{\cK, h}^\top \bm \Sigma_{H_0}^{-1}G_{\cK, h}$ and $q_{1-\alpha}$ be the $1-\alpha$-th quantile of the distribution $G_{\cK}^\top \bm \Sigma_{H_0}^{-1}G_{\cK}$. (Note that $G_{\cK, 0} = G_{\cK}$.) Since $\hat q_{1-\alpha, m} | \sX_m \stackrel{a.s.} \rightarrow q_{1-\alpha}$, \eqref{eq:GKH1} implies that 
	the asymptotic power of $\phi_{m, n}$ under $H_1$ as in \eqref{eq:H0N} is given by  
	$$\lim_{m, n \rightarrow \infty} \E_{H_1}[\phi_{m, n}]= 1- F_{\cK, h}(q_{1-\alpha}).$$
	This implies, $\phi_{m, n}$ has non-trivial asymptotic (Pitman)  efficiency and is rate-optimal, in the sense that, $\lim_{|h| \rightarrow \infty }\lim_{m, n \rightarrow \infty} \E_{H_1}[\phi_{m, n}] = 1$.
\end{remark}

\section{Distribution Under Alternative} 
\label{sec:H1asymptotic}

In this section we derive the asymptotic distribution of  $\mathrm{MMD}^{2}\left[ \cK, \sX_m, \sY_n \right]$ under the alternative, that is, when $P \neq Q$. For this we write $\mathrm{MMD}^{2}\left[ \cK, \sX_m, \sY_n \right]$  as a two-sample $U$-statistics as noted in \citet{kim2022minimax},
\begin{align}\label{eq:hxxyy}
\mathrm{MMD}^{2}\left[ \sfK, \sX_m, \sY_n \right] = \frac{1}{m(m-1)}\frac{1}{n(n-1)}\sum_{1 \leq i_1 \neq i_2 \leq m} \sum_{1 \leq j_1 \neq j_2 \leq n} \bm{h}\left(X_{i_1}, X_{i_2},Y_{j_1}, Y_{j_2}\right) , 
\end{align}
where $\bm{h}(x,x',y,y') = \left(h_{a}(x,x',y,y')\right)_{1 \leq a \leq r}$ and 
\begin{align}\label{eq:haxxyy}
    h_{a}(x,x',y,y') = \sfK_{a}(x,x') + \sfK_{a}(y,y') - \sfK_{a}(x,y') - \sfK_{a}(x',y) ,   
\end{align}
for $1 \leq a \leq r$. Using the Hoeffding's decomposition we can easily derive the joint distribution of $\mathrm{MMD}^{2}\left[ \cK, \sX_m, \sY_n \right]$ under the alternative. In this case the asymptotic distribution will be a $r$-dimensional multivariate normal as in Theorem \ref{thm:MMDdistributionPQ} below. To express the limiting covariance matrix we need the following definitions: For $1 \leq a \leq r$, let 
\begin{align}\label{eq:Delta1s}
    \Delta_a^{(1)}(x) := \int_{\cX} \sfK_{a}(x,x') \mathrm dP(x')  - \int_{\cX} \sfK_{a}(x,y') \mathrm d Q(y')
\end{align}
and 
\begin{align}\label{eq:Delta2s}
    \Delta_a^{(2)}(y) := \int_{\cX} \sfK_{a}(y, y') \mathrm dQ(y')  - \int_{\cX} \sfK_{a}(x',y) \mathrm d P(x') . 
\end{align}
Also, denote $\bm\Delta^{(1)}(x) :=(\Delta_{a}^{(1)}(x))_{1 \leq a \leq r}$ and $\bm\Delta^{(2)}(x) :=(\Delta_{a}^{(2)}(x))_{1 \leq a \leq r}$. Then we have the following theorem:

\begin{theorem}\label{thm:MMDdistributionPQ}
Suppose $\cK=\{\sfK_1, \sfK_2, \ldots, \sfK_r\}$ be a collection of $r$ distinct characteristic kernels and $\bm\cF=\{\cF_1, \cF_2, \ldots, \cF_r\}$ be the unit balls of their respective RHKS.  Suppose $\sfK_{a}\in L^{2}(\mathcal{X}^{2}, P^{2})\cap L^{2}(\mathcal{X}^{2}, Q^{2})\cap L^{2}(\mathcal{X}^{2}, P\times Q)$, for all $1\leq a \leq r$. Then for $P \ne Q$, 
\begin{align*}
    \sqrt{m+n}\left(\mathrm{MMD}^{2}\left[ \cK, \sX_m, \sY_n \right] -  \mathrm{MMD}^{2}\left[ \bm\cF, P, Q\right]\right)\overset{D}{\longrightarrow}\cN_r \left(\bm{0}, \bm{\Sigma}_{H_1} \right) , 
\end{align*}
where 
\begin{align}\label{eq:sigmaH1}
    \bm{\Sigma}_{H_1} := 4\left(\rho \Var_{X \sim P}{\left[\bm\Delta^{(1)}(X)\right]} + (1 - \rho) \Var_{Y \sim Q}{\left[\bm \Delta^{(2)} (Y)\right]}\right) . 
\end{align} 
and $\mathrm{MMD}^{2}\left[ \bm\cF, P, Q \right] = \left(\mathrm{MMD}^{2}\left[\mathcal{F}_{1}, P, Q\right],\cdots, \mathrm{MMD}^{2}\left[\mathcal{F}_{r}, P, Q \right]\right)^\top$.
\end{theorem}

\begin{proof} 
For $1 \leq a \leq r$, let 
\begin{align}\label{eq:Deltat1s}
    \tilde \Delta_{a}^{(1)}(x) := \mathbb{E}_{X' \sim P, Y, Y' \sim Q}[h_{a}(x, X', Y, Y')] - \mathrm{MMD}^{2}[\mathcal{F}_{a}, P, Q ] 
\end{align}
and 
\begin{align}\label{eq:Deltat2s}
     \tilde \Delta_{a}^{(1)}(y) = \mathbb{E}_{X, X' \sim P, Y' \sim Q}[h_{a}(X, X',y,Y')] - \mathrm{MMD}^{2}[\mathcal{F}_{a}, P, Q ] . 
\end{align} 
Recalling \eqref{eq:hxxyy}, the first-order Hoeffding's projection for $\mathrm{MMD}^{2}\left[\cK, \sX_m, sY_n \right] - \mathrm{MMD}^{2}[\bm\cF, P, Q]$ is given by, 
\begin{align*}
    \hat{\bm{U}}_{m,n} = \frac{2}{m}\sum_{i=1}^{m} \tilde{\bm{\Delta}}^{(1)} (X_{i}) + \frac{2}{n}\sum_{j=1}^{n} \tilde{\bm{\Delta}}^{(2)}(Y_{j}) , 
\end{align*}
where $\tilde{\bm{\Delta}}^{(1)}(x) = (\tilde \Delta_{a}^{(1)}(x))_{1 \leq a \leq r}$ and  $\tilde{\bm{\Delta}}^{(2)}(x) = (\tilde \Delta_{a}^{(2)}(x))_{1 \leq a \leq r}$. 
Then by \cite[Theorem 12.6]{van2000asymptotic}, 
\begin{align}\label{eq:MMDUmndifference}
   \left\| \sqrt{m+n}\left(\mathrm{MMD}^{2}\left[ \sfK, P, Q \right] - \mathrm{MMD}^{2}[\bm\cF, P, Q] - \hat{\bm U}_{m, n} \right) \right\| = o_{P}(1) . 
\end{align}
By the multivariate central limit theorem, 
$$\sqrt{m+n}~\hat{\bm U}_{m, n} \dto \cN_r(\bm 0, \bm \Gamma) , $$
where 
$$\bm \Gamma = 4\left(\rho \Var_{X \sim P}{\left[\tilde{\bm \Delta}^{(1)}(X)\right]} + (1 - \rho) \Var_{Y \sim Q}{\left[ \tilde{\bm \Delta}^{(2)} (Y)\right]}\right) = \bm \Sigma , $$
since, recalling \eqref{eq:Delta1s} and \eqref{eq:Deltat1s}, $\Var_{X \sim P}{[ \tilde{\bm \Delta}^{(1)}(X)]} = \Var_{X \sim P}[\bm{\Delta}^{(1)}(X)]$ and, from \eqref{eq:Delta2s} and \eqref{eq:Deltat2s}, $\Var_{Y \sim Q}{[ \tilde{\bm \Delta}^{(2)}(Y)]} = \Var_{Y \sim Q}[\bm{\Delta}^{(2)}(Y)]$. This together with \eqref{eq:MMDUmndifference} completes the proof of Theorem \ref{thm:MMDdistributionPQ}.
\end{proof}

\color{black}

\color{black}
\section{Computational Complexity of the MMMD test} 
\label{sec:computationpf}
For a single kernel $\sfK$ computing $\emmd[\sfK,\cX_{m},\cY_{n}]$ from \eqref{eq:MMDXY} has a cost of $O(N^2)$. Thus, for the collection of kernels $\cK = \
\{ \sfK_1, \sfK_2, \ldots, \sfK_r \}$ computing the vector of MMD estimates $\emmd[\cK, \cX_{m}, \cY_{n}]$ from \eqref{eq:Kvector} requires $O(rN^2)$ time. Next, we derive the cost of computing the matrix $\hat{\bm{\Sigma}}$ (recall \eqref{eq:H0sigmaestimate}). 

\begin{lemma}\label{lemma:compsigmaab}
    For any fixed $1\leq a<b\leq r$, computing $\hat{\sigma}_{ab}$ requires $O(m^2)$ time. Hence, the matrix $\hat{\bm{\Sigma}}$ can be computed in $O(r^2m^2)$ time. 
\end{lemma}

\begin{proof} 
For any fixed $1\leq a,b\leq r$, by \eqref{eq:estimateKxy} observe that,
\begin{align}\label{eq:expandKhatdot}
    \sum_{1 \leq i,j \leq m} \hat{\sfK}^{\circ}_{a}(X_{i}, X_{j})\hat{\sfK}_{b}^{\circ}(X_{i}, X_{j}) = \eta_{1} - 2\eta_{2} + \eta_{3}
\end{align}
where $ \eta_{1}  := \sum_{1 \leq i,j \leq m}\sfK_{a}(X_{i}, X_{j})\sfK_{b}(X_{i}, X_{j})$, 
\begin{align}
   \eta_{2} & := \frac{1}{m}\sum_{i=1}^{m}\left(\sum_{v=1}^{m}\sfK_{a}(X_{i}, X_{v})\right)\left(\sum_{v=1}^{m}\sfK_{b}(X_{i},X_{v})\right) , \nonumber \\ 
   \eta_{3} & := \frac{1}{m^2}\left(\sum_{1 \leq i,j \leq m} \sfK_{a}(X_{i}, X_{j})\right)\left(\sum_{1 \leq i,j \leq m} \sfK_{b}(X_{i}, X_{j})\right) . \nonumber 
\end{align}
By definition, $\eta_{1} \eta_{2}$ and $\eta_{3}$ can be computed in $O(m^2)$ time. Hence, $\hat{\sigma}_{ab}$ can be computed in $O(m^2)$ time and the matrix $\hat{\bm{\Sigma}}$ can be computed in $O(r^2m^2)$ time. 
\end{proof}

Given the matrix $\hat{\bm{\Sigma}}$, the matrix $\hat{\bm{\Sigma}}^{-1}$ can be computed in $O(r^3)$ time. Finally, given $\hat{\bm\Sigma}^{-1}$ and the vector $\emmd\left[\cK,\cX_{m},\cY_{n}\right]$ computing the MMMD statistic $T_{m,n}$ has a cost of $O(r^2)$. Combining the above, the total computational cost of $T_{m, n}$ is then $O(rN^2 + r^2m^2+ r^3)= O(r^2N^2 + r^3)$. Note that in all relevant cases $r < N$, in which case the time complexity for computing the MMMD statistic $T_{m,n}$ simplifies to $O(r^2N^2 + r^3) = O(r^2N^2)$.

To obtain the rejection threshold of the MMMD statistic, we need to compute the  sample quantile of $B$ replicates of the statistic $\hat{T}_{m}$ from \eqref{eq:Testimate}. The following lemma derives the complexity of computing the vector $\cE(\cK, \sX_m)$ in \eqref{eq:def-Em}.

\begin{lemma}\label{lemma:compcomplexEm}
    For $1\leq a \leq r$ and any vector $\bm{v} \in\mathbb{R}^{m}$ computing $\bm{v}^{\top}\sfKmatrix_{a}\bm{v} - \frac{1}{\hat{\rho}(1-\hat{\rho})}\Tr[\sfKmatrix_{a}]$ can be computed in $O(m^2)$ time. Hence, the vector $\cE(\cK, \sX_m)$ in \eqref{eq:def-Em} can be computed in $O(r m^2)$ time. 
\end{lemma}

\begin{proof} 
From \eqref{eq:centered-kernel} note that,
\begin{align}\label{eq:vmquadtoCvmquad}
    \bm{v}^{\top}\sfKmatrix_{a}\bm{v} = \frac{1}{m}\left(\bm{C}\bm{v}\right)^{\top}\hat{\bm{\sfK}}_{a}\left(\bm{C}\bm{v}\right).
\end{align}
Note that $\bm{C}\bm{v} = \bm{v} - \frac{1}{m}\bm{1}\left(\bm{1}^{\top}\bm{v}\right)$ can be computed in $O(m)$ time. Hence, by \eqref{eq:vmquadtoCvmquad} computing $\bm{v}^{\top}\sfKmatrix_{a}\bm{v}$ incurs a cost of $O(m^2)$. Further, by \eqref{eq:estimateKxy} we get,
\begin{align*}
    \Tr\left[\sfKmatrix_{a}\right] = \sum_{i=1}^{m}\hat{\sfK}^{\circ}_{a}(X_{i},X_{i}) = \sum_{i=1}^{m}\sfK_{a}(X_{i},X_{i}) - \frac{1}{m}\sum_{i=1}^{m}\sum_{j=1}^{m}\sfK_{a}(X_{i},X_{j}). 
\end{align*}
Hence, $\Tr[\sfKmatrix_{a}]$ can be computed in $O(m^2)$ time. Thus, the overall computational cost of $\bm{v}^{\top}\sfKmatrix_{a}\bm{v} - \frac{1}{\hat{\rho}(1-\hat{\rho})}\Tr[\sfKmatrix_{a}]$ is $O(m^2)$. 

Also, note that generating $\bm{Z}_{m}\sim\cN_{m}(\bm 0, \frac{1}{\hat{\rho}(1-\hat{\rho})}\bm{I})$ has a cost of $O(m)$. Thus, computing the vector $\cE(\cK, \cX_{m})$ requires $O(rm^2)$ time. 
\end{proof}

Note that given $\hat{\bm{\Sigma}}^{-1}$ and the vector $\cE(\cK, \cX_{m})$, $\hat{T}_{m}$ can  be computed in $O(r^2)$ time (recall \eqref{eq:Testimate}).  Thus, from Lemma \ref{lemma:compcomplexEm},  the overall cost of computing $\hat{T}_{m}$ is $O(rm^2 + r^2)$. Hence, we can compute $\hat{T}_{m}$ over $B$ replications in  $O(B(rm^2 + r^2))$ time, and get the sample quantile in $O(B\log B)$ time. 
Combining this with the cost of computing the MMMD statistic, the overall computational cost of MMMD test is (assuming $r < N$), 
\begin{align*}
    O(r^2N^2 + BrN^2 + B\log B).
\end{align*}

To understand the trade-off between computational time and power in practice, we consider the following simulation setting: 
\begin{align}\label{eq:testpowertime}
    P = \cN_{d}\left(\bm{0}, \bm{I}_{d}\right)\text{ and }Q = \cN_{d}\left(\bm{0}, \sigma^2\bm{I}_{d}\right).
\end{align}
with $d = 20$ and $\sigma^{2} = 1.1$, with varying sample size. Table \ref{table:comptime} shows the empirical power (averaged over 500 iterations) and time (in seconds) of the single kernel {\tt Gauss MMD} and {\tt LAP MMD} tests with the median bandwidth and  the multiple kernel {\tt Gauss MMMD}, {\tt LAP MMMD}, and {\tt Mixed MMMD} tests with bandwidths as in \eqref{eq:gaussmmmd}, \eqref{eq:laplacemmmd}, \eqref{eq:mixedmmmd}, respectively. 
The results show that MMMD tests provide a significant improvement in power, at the cost of a minor increase in computation time.

\begin{table}[!ht]
\small 
    \centering
    \begin{tabular}{|c|cc|cc|cc|cc|cc|}
    \hline
    \multirow{2}{*}{Sample Size} & \multicolumn{2}{c|}{LAP MMD}      & \multicolumn{2}{c|}{Gauss MMD}    & \multicolumn{2}{c|}{LAP MMMD}      & \multicolumn{2}{c|}{Gauss MMMD}    & \multicolumn{2}{c|}{Mixed MMMD}    \\ \cline{2-11} 
                                 & \multicolumn{1}{c|}{Power} & Time & \multicolumn{1}{c|}{Power} & Time & \multicolumn{1}{c|}{Power} & Time  & \multicolumn{1}{c|}{Power} & Time  & \multicolumn{1}{c|}{Power} & Time  \\ \hline
    50                           & \multicolumn{1}{c|}{0.002} & 22.5 & \multicolumn{1}{c|}{0.026} & 23.6 & \multicolumn{1}{c|}{0.088} & 23.8  & \multicolumn{1}{c|}{0.14}  & 23.8  & \multicolumn{1}{c|}{0.148} & 23.8  \\ \hline
    100                          & \multicolumn{1}{c|}{0.006} & 28.7 & \multicolumn{1}{c|}{0.072} & 28.8 & \multicolumn{1}{c|}{0.196} & 30.4  & \multicolumn{1}{c|}{0.262} & 29.4  & \multicolumn{1}{c|}{0.266} & 28.9  \\ \hline
    200                          & \multicolumn{1}{c|}{0.07}  & 40.6 & \multicolumn{1}{c|}{0.126} & 40.4 & \multicolumn{1}{c|}{0.482} & 40.4  & \multicolumn{1}{c|}{0.5}   & 39.0  & \multicolumn{1}{c|}{0.48}  & 40.4  \\ \hline
    300                          & \multicolumn{1}{c|}{0.132} & 52.3 & \multicolumn{1}{c|}{0.238} & 50.6 & \multicolumn{1}{c|}{0.652} & 53.9  & \multicolumn{1}{c|}{0.72}  & 52.8  & \multicolumn{1}{c|}{0.706} & 53.9  \\ \hline
    400                          & \multicolumn{1}{c|}{0.24}  & 63.0 & \multicolumn{1}{c|}{0.328} & 62.5 & \multicolumn{1}{c|}{0.788} & 64.8  & \multicolumn{1}{c|}{0.88}  & 64.4  & \multicolumn{1}{c|}{0.856} & 65.0  \\ \hline
    500                          & \multicolumn{1}{c|}{0.35}  & 73.4 & \multicolumn{1}{c|}{0.414} & 73.5 & \multicolumn{1}{c|}{0.904} & 85.4  & \multicolumn{1}{c|}{0.932} & 76.5  & \multicolumn{1}{c|}{0.944} & 83.7  \\ \hline
    600                          & \multicolumn{1}{c|}{0.428} & 86.1 & \multicolumn{1}{c|}{0.564} & 85.1 & \multicolumn{1}{c|}{0.964} & 110.0 & \multicolumn{1}{c|}{0.966} & 92.7  & \multicolumn{1}{c|}{0.968} & 108.4 \\ \hline
    700                          & \multicolumn{1}{c|}{0.578} & 93.8 & \multicolumn{1}{c|}{0.658} & 95.5 & \multicolumn{1}{c|}{0.984} & 135.6 & \multicolumn{1}{c|}{0.992} & 116.8 & \multicolumn{1}{c|}{0.996} & 139.9 \\ \hline
    \end{tabular}
    \caption{Power and computation time (in seconds) for the hypothesis in \eqref{eq:testpowertime}. } 
    \label{table:comptime}
\end{table}
    
\normalsize

\color{black}

\section{Additional Simulations} 
\label{sec:experimentsadditional}

In this section, we will compare the MMMD test with the single kernel MMD test and the graph-based Friedman Rafsky (FR) test with parameters as in Section \ref{sec:experiments}.

\subsection{Dependence on Sample Size} 
\label{sec:dependencesimulation}

In this section we illustrate how the different tests performing as the sample  size varies, with dimension held fixed. Toward this, we fix $d=2$ and consider 
 \begin{align*}
        P = \cN_{2}(\bm{0}, \bm{I}_{2})  \text{ and }  Q= \cN_{2}(\bm{0}, 1.25 \cdot \bm{I}_2) , 
    \end{align*} 
 and vary the sample sizes over $m=n\in \{50, 100, 200, 300, 400, 500\}$. Figure \ref{fig:S9-2} shows the empirical Type I-error and power  of the aforementioned tests. Figure \ref{fig:S9-2}(a) shows that all the tests have good Type I error control. Figure \ref{fig:S9-2}(b) shows that the multiple kernel MMMD tests have better power than the single kernel MMD tests, with the {\tt Gauss MMMD} and the {\tt Mixed MMMD} tests performing the best.

\begin{figure}[ht] 
  \begin{subfigure}[c]{0.45\linewidth}
    \centering
    \includegraphics[width=0.85\linewidth]{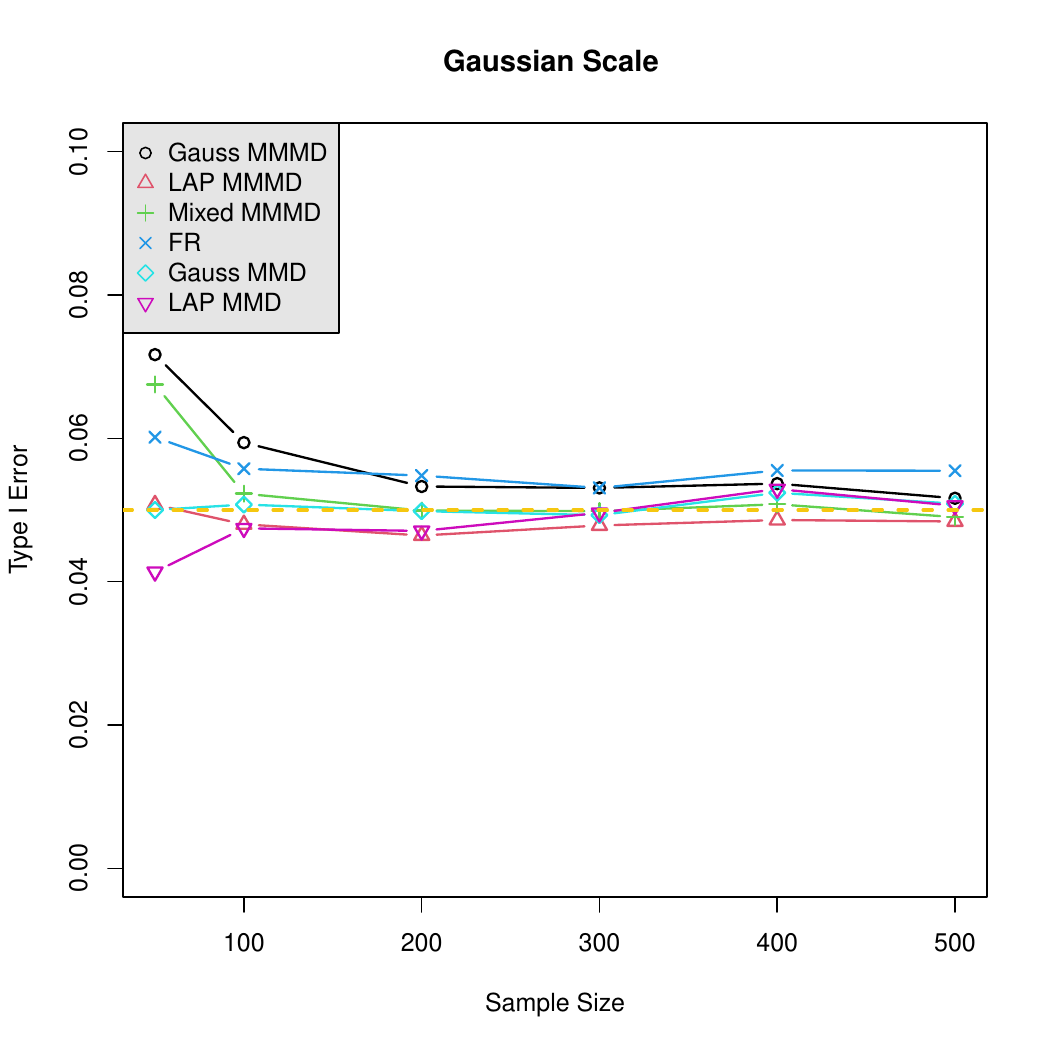} 
    \caption*{\small{(a)}} 
  \end{subfigure}
   \begin{subfigure}[c]{0.45\linewidth}
    \centering
    \includegraphics[width=0.85\linewidth]{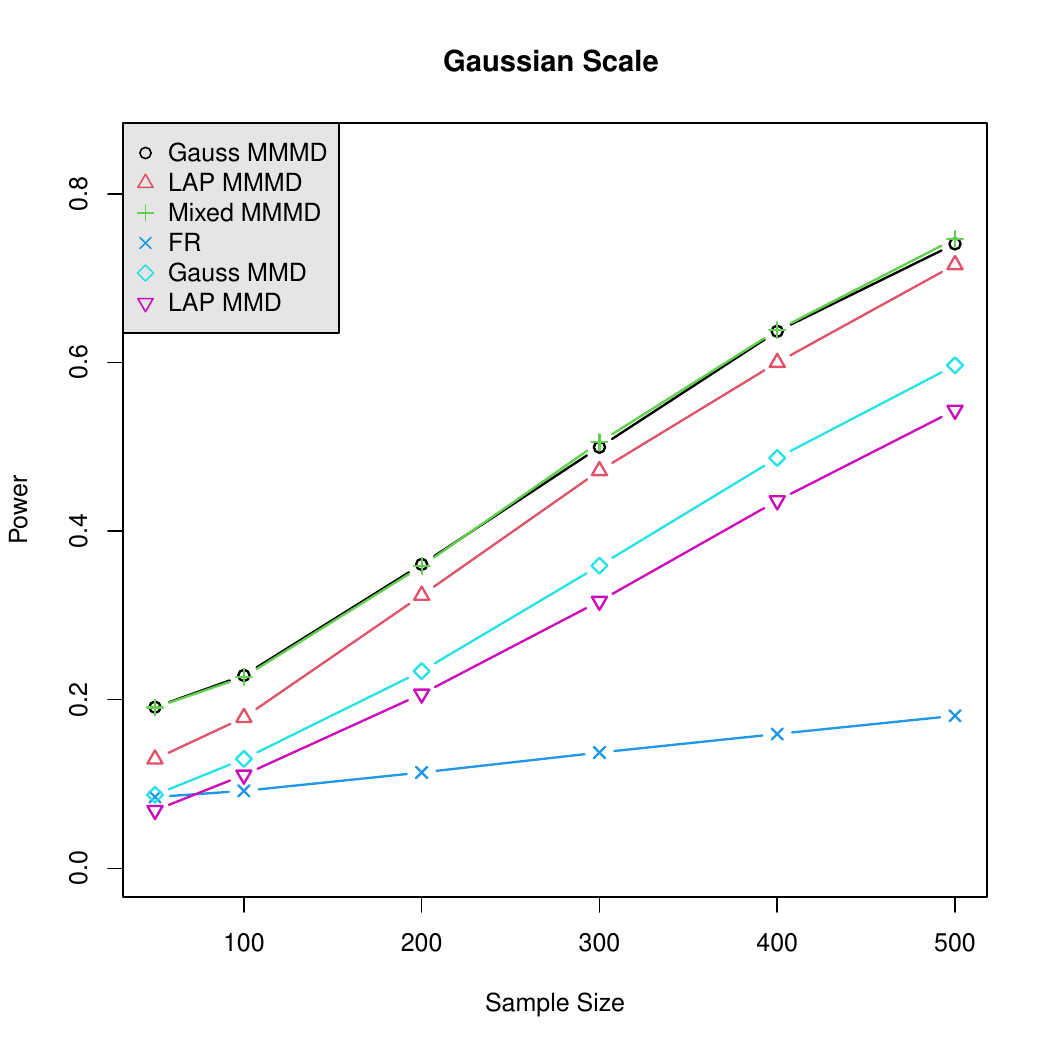}  
       \caption*{\small{(b)}} 
  \end{subfigure} 
  \caption{(a) Empirical type-I errors and (b) powers of the different tests in dimension  $d=2$ with varying sample size. }
  \label{fig:S9-2} 
\end{figure}

\subsection{Dependence on Dimension} 
\label{sec:dimensionadditional}

Throughout this section, as in Section \ref{sec:dimension}, we fix sample sizes $m=n=100$, vary dimension over $d \in \{5, 10, 25, 50, 75, 100, 150\}$. 

We begin by recalling the setting \ref{itm:S1}. Figure \ref{fig:Type1S12} shows the probability of Type I error (averaged over 500 iterations) in this setting.  We observe from the table that kernel tests become more conservative as dimension increases, although, quite interestingly, this issue is significantly mitigated for the multiple kernel MMMD tests.

\begin{figure}[h] 
\centering
 \includegraphics[width=0.95\linewidth]{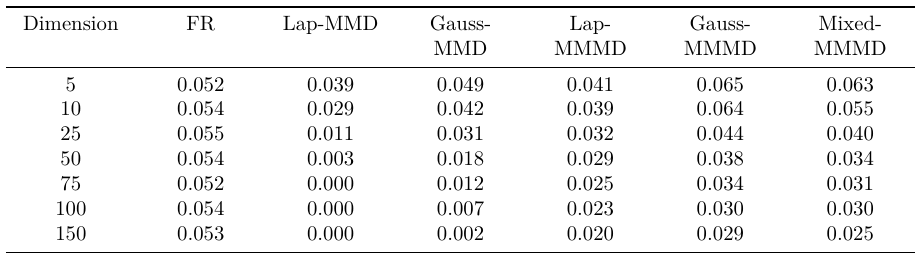} 
\vspace{-0.1in} 
\caption{Probability of Type I error for the setting in \ref{itm:S1}. }
\label{fig:Type1S12}
\end{figure} 

Next, we consider the following 2 settings. The empirical power by averaging over 500 iterations.  

\begin{enumerate}[label=\textbf{(S\arabic*)}]\label{enum: Over-Dim Settings}
\setcounter{enumi}{2}
    
     \item\label{itm:S3} {\it $t$-distribution scale}: Here, we consider $P =  t_{10}(\bm{0} ,\bm{\Sigma}_{0})\text{ and } Q = t_{10}\left(\bm{0}, 1.22\bm{\Sigma}_{0}\right)$, 
    where $t_{10}$ is the $t$-distribution with 10 degrees of freedom and $\bm{\Sigma}_{0}$ is as above (see Figure \ref{fig:S12}(b)). 
    
      \item\label{itm:S4} {\it Gaussian and Laplace mixture}: 
     Here, we consider     
    $$P =\tfrac{1}{2} \cN_d(\bm{0} ,\bm{\Sigma}_{1}) +\tfrac{1}{2} t_{10}(\bm{0} ,\bm{\Sigma}_{1}) \text{ and } Q =\tfrac{1}{2} \cN_d(\bm{0} ,1.3\bm{\Sigma}_{1}) +\tfrac{1}{2} t_{10}(\bm{0} , 1.3\bm{\Sigma}_{1}) ,  $$
    where $\bm{\Sigma}_{1} = ((0.7^{|i-j|}))_{1\leq i,j\leq d}$ (see Figure \ref{fig:S34}(b)). 
\end{enumerate} 

\begin{figure}[h] 
\centering
\begin{subfigure}[c]{0.45\textwidth}
   \includegraphics[width=0.85\linewidth]{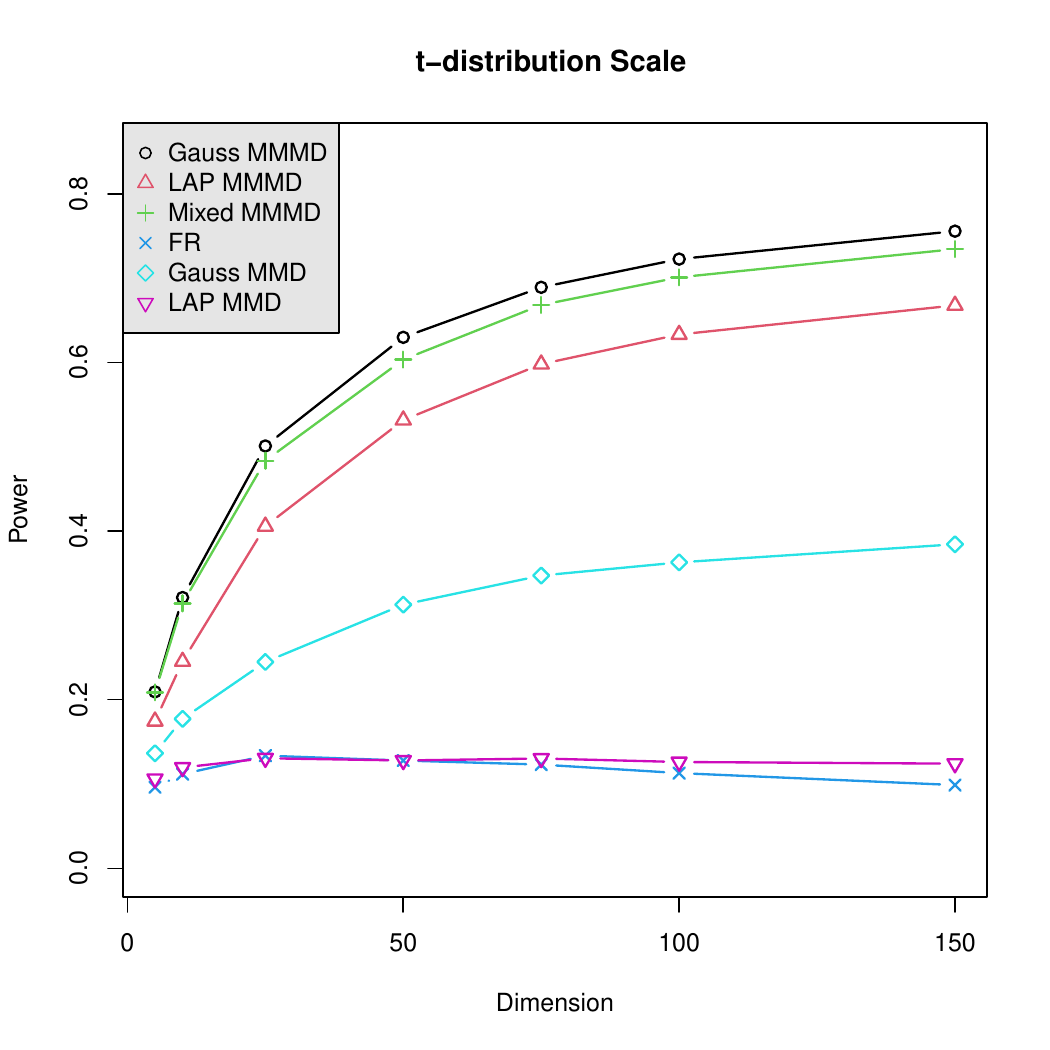} 
      \caption*{\small{(a)}}
\end{subfigure} 
\begin{subfigure}[c]{0.45\textwidth}
   \includegraphics[width=0.85\linewidth]{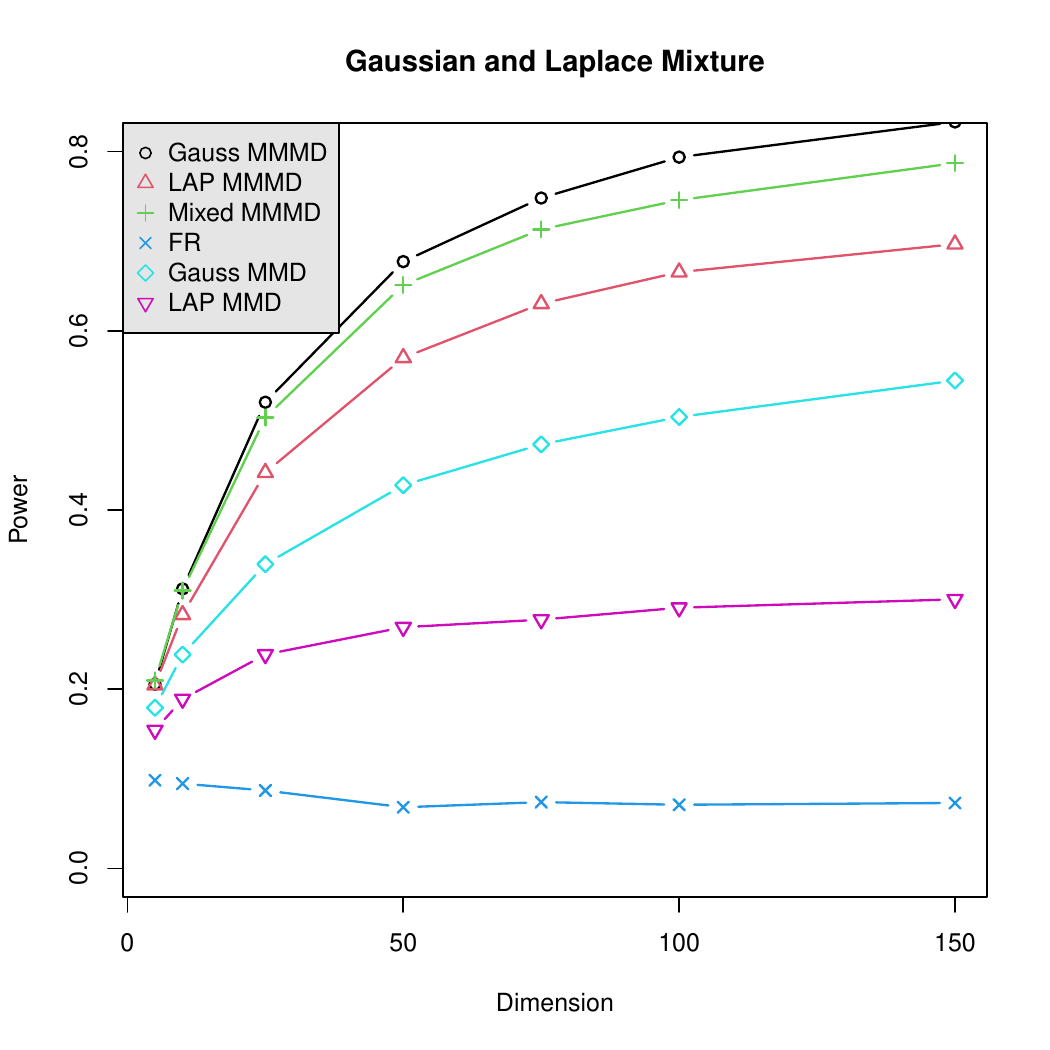} 
      \caption*{\small{(b)}} 
\end{subfigure} 
\caption{Empirical powers in (a) setting \ref{itm:S3} and (b) setting \ref{itm:S4}. }
\label{fig:S34}
\end{figure}
The plots show that the multiple kernel MMMD tests have significantly more  power than the single kernel MMD tests and the FR test in both the settings \ref{itm:S3} and \ref{itm:S4}. 

\subsection{Local Alternatives} 
\label{sec:localexperiments} 

Recall that in Section \ref{sec:asymptoticpower} we derived the asymptotic local power of the MMMD statistic. Here we validate this in the following simulation setting: 
    \begin{align}\label{eq:scalelocalalternative}
        P = \cN_d(\bm{0} ,\bm{I}_{d})\text{ and } Q = \cN_d\left(\bm{0} , (1+\tfrac{h}{\sqrt N})\bm{I}_{d} \right) , 
    \end{align} 
    where $N = m+n$. Figure \ref{fig:dimension}(a) shows the empirical power (averaged over 500 iterations)  of the different tests, for dimension $d = 20$, sample sizes $m=n=100$, as the signal strength varies over [0, 2].  The plots show that the MMMD methods have significantly better local power than the other tests, illustrating the attractive efficiency property of our aggregation method.

\begin{figure}[h]
\centering 
\begin{subfigure}[c]{0.45\textwidth}
   \includegraphics[width=0.85\linewidth]{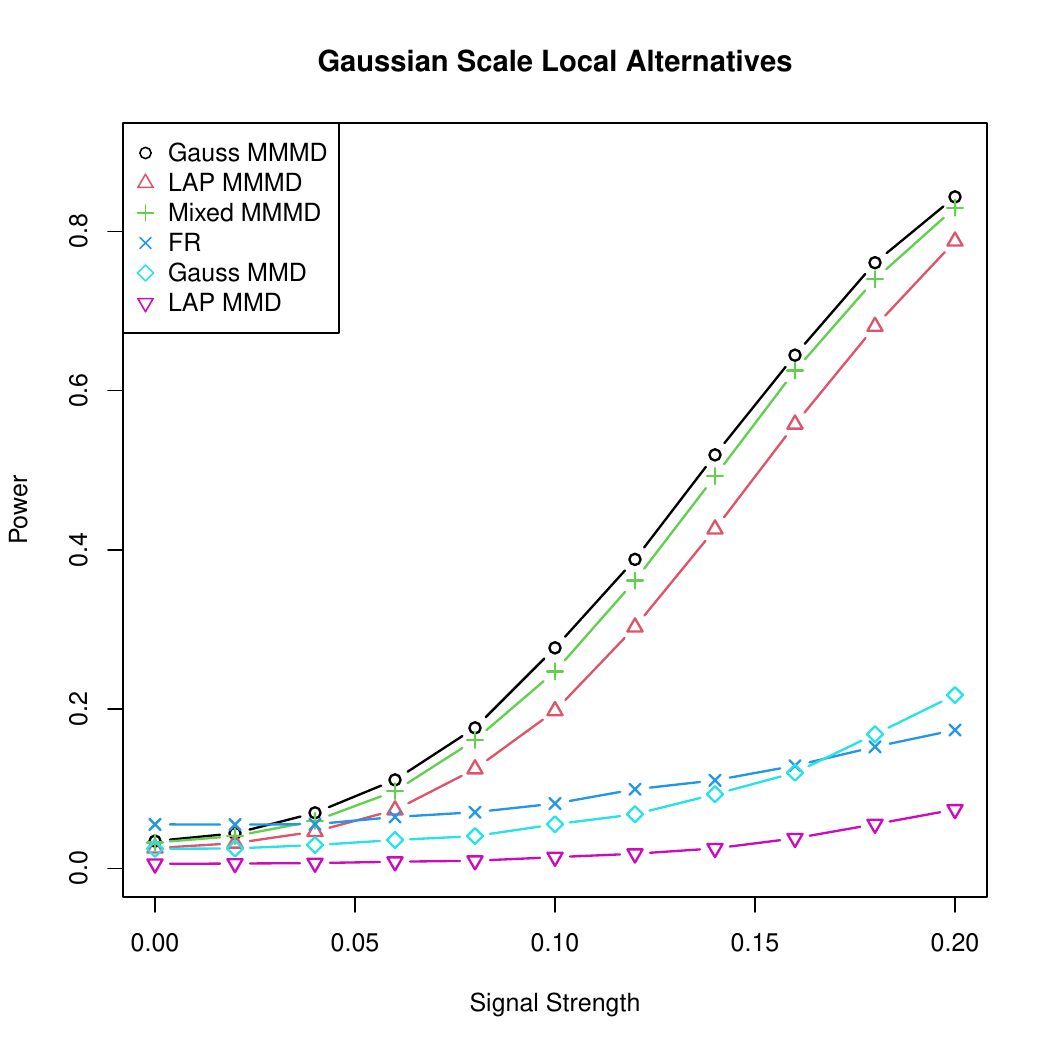} 
   \caption*{\small{(a)}}
\end{subfigure} 
\begin{subfigure}[c]{0.45\textwidth}
   \includegraphics[width=0.85\linewidth]{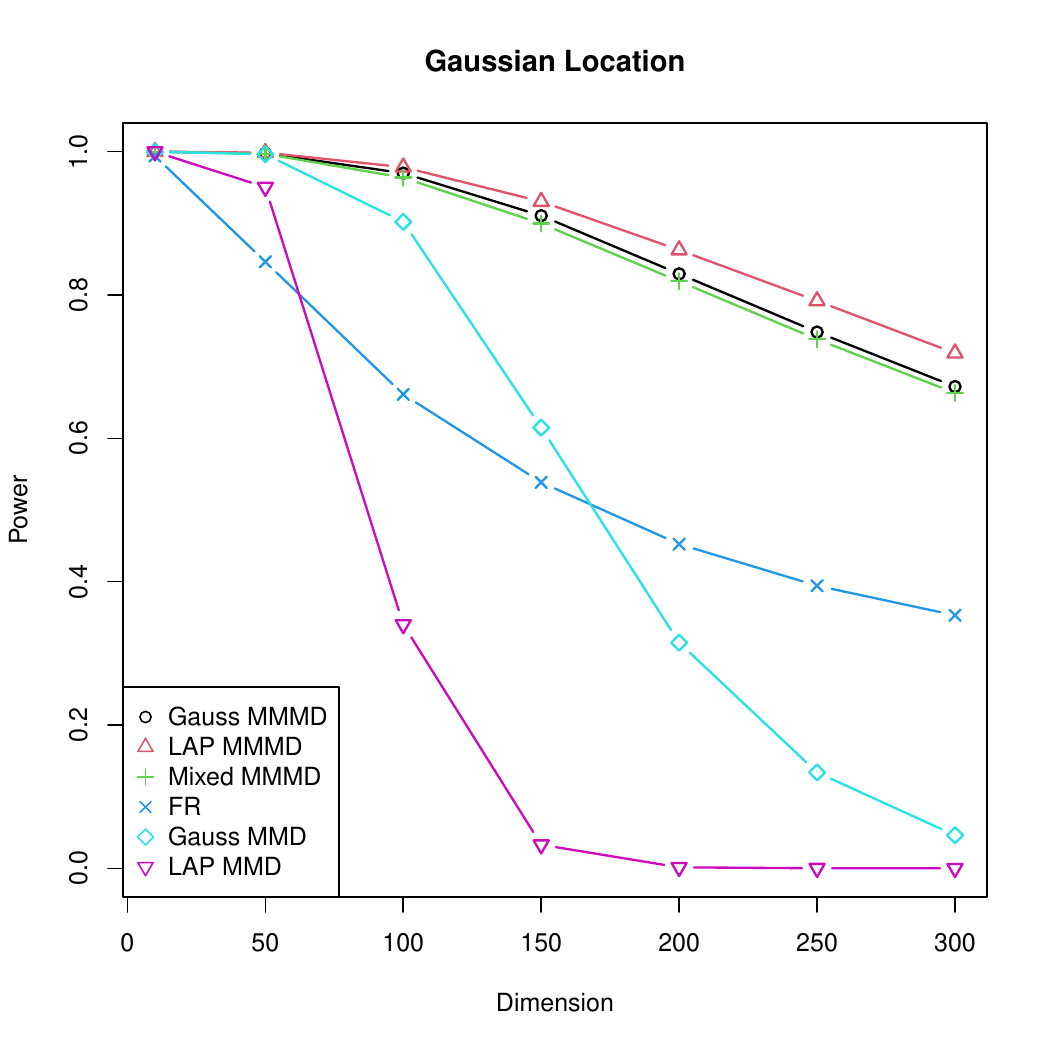}
      \caption*{\small{(b)}}
 
\end{subfigure} 
\caption{Empirical powers for (a) local alternatives as in \eqref{eq:scalelocalalternative} and (b) high-dimensional alternatives as in \eqref{eq:normaldimension}. }
\label{fig:dimension}
\end{figure}

\subsection{High-Dimensional Alternatives}

 To fairly evaluate the performance of kernel  tests in high-dimensions we consider, as suggested by \citet{ramdas2015}, pairs of distributions for which the Kullback-Leibler (KL) divergence remain constant, as the dimension increases. Specifically, we consider   
    \begin{align}\label{eq:normaldimension}
        P = \cN_d\left(\bm{0} , \bm{I}_{d}\right)\text{ and } Q = \cN_d\left( (1.25/\sqrt{d}) \bm{1}, \bm{I}_{d}\right) . 
    \end{align} 
It is easy to check that the KL divergence between $P$ and $Q$ is $\frac{1.25^{2}}{2}$, which does not change with $d$. In this case for the {\tt Gauss MMMD}/{\tt LAP MMMD} tests we use 5 different Gaussian/Laplace kernels with respective bandwidths $\bm \sigma = (\sigma_1, \sigma_2,  \sigma_3, \sigma_4, \sigma_5) =  (\tfrac{1}{2 \sqrt 2}, \tfrac{1}{2}, \frac{1}{\sqrt 2}, 1,\sqrt{2})\lambda_{\text{med}}$. Also, for the {\tt Mixed MMMD} test is implemented with 4 Gaussian kernels and 4 Laplace kernels with same set of bandwidths: $\bm \sigma = (\sigma_1, \sigma_2,  \sigma_3, \sigma_4) =( \frac{1}{2},\frac{1}{\sqrt{2}}, 1, \sqrt{2})\lambda_{\text{med}}$. 
Figure \ref{fig:dimension}(b) shows the empirical powers (averaged over 500 iterations) of the different tests  as a function of dimension with sample sizes $m=n=100$. As expected, the power of all the tests decreases as dimension increases. However, for the multiple kernel tests the power decrement is much slower and overall they perform significantly better than the single kernel tests.

\color{black}

\subsection{Finite Sample Comparison with the Permutation Test}
\label{sec:permutationsimulation}

To compare the multiplier bootstrap and permutation methods we consider the following simulation settings.

\begin{itemize}
	\item Consider
    \begin{align}\label{eq:permexp1}
        P = \cN_{d}(\bm 0, \bm I) \text{ and }Q = \cN_{d}(\bm 0, 1.05^2\bm I). 
    \end{align}
    We fix sample sizes $m = n\in \{200, 300\}$ and vary the dimension over $$d\in\{2, 5, 10, 20, 50, 100, 150, 200, 300\}.$$  Table \ref{tab:permexp1} shows the type I error rate, empirical power (averaged over 100 iterations),  and the computation time (in seconds) for the MMMD test using the multiplier bootstrap and the permutation method with  $B=1000$ bootstrap samples and $B=1000$ random permutations, respectively. 
    	
	\item Consider
    \begin{align}\label{eq:permexp2}
        P = \cN_2(\bm 0, \bm I)\text{ and }Q = \cN_2(\bm 0, 1.25^2\bm I).
    \end{align}
    In this case, we fix the number of bootstrap resamples/permutations $B \in \{100, 1000\}$ and vary sample sizes $m = n\in\{200, 300, 400, 500, 600, 700\}$. The empirical Type-I error rate, power (averaged  over $100$ iterations), and time (in seconds) are shown in Table \ref{tab:permexp2}.
\end{itemize}

\begin{table}[h]
    \footnotesize 
    \begin{minipage}{.5\linewidth}
    \centering
    \begin{tblr}{
        cells = {c},
        cell{1}{1} = {c=7}{},
        cell{2}{1} = {r=2}{},
        cell{2}{2} = {c=3}{},
        cell{2}{5} = {c=3}{},
        vlines,
        hline{1-2,4-13} = {-}{},
        hline{3} = {2-7}{},
        colsep = 1.5pt,
        rowsep = 0.5pt
      }
      $m = n =200$ &             &        &          &                      &        &         \\
      Dimension    & Permutation &        &          & Multiplier Bootstrap &        &         \\
                   & Type I      & Power  & Time     & Type I               & Power  & Time    \\
      $2$          & $0.05$      & $0.05$ & $476.49$ & $0.06$               & $0.1$  & $5.52$  \\
      $5$          & $0.06$      & $0.04$ & $496.96$ & $0.02$               & $0.03$ & $4.89$  \\
      $10$         & $0.07$      & $0.13$ & $467.86$ & $0.03$               & $0.06$ & $4.20$  \\
      $20$         & $0.04$      & $0.19$ & $427.00$ & $0.01$               & $0.15$ & $5.27$  \\
      $50$         & $0.05$      & $0.33$ & $466.40$ & $0.03$               & $0.4$  & $6.43$  \\
      $100$        & $0.02$      & $0.7$  & $526.36$ & $0.05$               & $0.63$ & $7.69$  \\
      $150$        & $0.04$      & $0.87$ & $491.51$ & $0.02$               & $0.8$  & $8.92$  \\
      $200$        & $0.05$      & $0.94$ & $513.26$ & $0.03$               & $0.9$  & $10.59$ \\
      $300$        & $0.06$      & $1$    & $568.53$ & $0.04$               & $1$    & $16.36$ 
      \end{tblr}
    \end{minipage}%
    \begin{minipage}{.5\linewidth}
        \begin{tblr}{
            cells = {c},
            cell{1}{1} = {c=7}{},
            cell{2}{1} = {r=2}{},
            cell{2}{2} = {c=3}{},
            cell{2}{5} = {c=3}{},
            vlines,
            hline{1-2,4-13} = {-}{},
            hline{3} = {2-7}{},
            colsep = 1.5pt,
            rowsep = 0.5pt
          }
          $m = n = 300$ &             &        &           &                      &        &         \\
          Dimension     & Permutation &        &           & Multiplier Bootstrap &        &         \\
                        & Type I      & Power  & Time      & Type I               & Power  & Time    \\
          $2$           & $0.06$      & $0.06$ & $1095.07$ & $0.06$               & $0.08$ & $3.59$  \\
          $5$           & $0.05$      & $0.11$ & $1080.09$ & $0.05$               & $0.09$ & $3.82$  \\
          $10$          & $0.06$      & $0.12$ & $1084.52$ & $0.04$               & $0.13$ & $3.59$  \\
          $20$          & $0.08$      & $0.29$ & $1074.10$ & $0.03$               & $0.22$ & $4.73$  \\
          $50$          & $0.04$      & $0.65$ & $1145.30$ & $0.02$               & $0.49$ & $7.58$  \\
          $100$         & $0.03$      & $0.85$ & $1219.39$ & $0.06$               & $0.83$ & $9.52$  \\
          $150$         & $0.06$      & $0.98$ & $1243.97$ & $0.04$               & $0.94$ & $12.78$ \\
          $200$         & $0.04$      & $1$    & $1242.91$ & $0.05$               & $0.97$ & $14.04$ \\
          $300$         & $0.03$      & $1$    & $1355.28$ & $0.04$               & $1$    & $20.58$ 
          \end{tblr}
    \end{minipage}
    \caption{ Type I error rate, power, and time (in seconds) for the hypothesis in \eqref{eq:permexp1}. }
    \label{tab:permexp1}
\end{table}

\begin{table}[h]
    \footnotesize
    \begin{minipage}{.5\linewidth}
    \centering
    \begin{tblr}{
        cells = {c},
        cell{1}{1} = {c=7}{},
        cell{2}{1} = {r=2}{},
        cell{2}{2} = {c=3}{},
        cell{2}{5} = {c=3}{},
        vlines,
        hline{1-2,4-10} = {-}{},
        hline{3} = {2-7}{},
        colsep = 1.5pt,
      rowsep = 0.5pt
      }
      $B = 100$   &             &        &           &                      &        &         \\
      {Sample\\Size} & Permutation &        &           & Multiplier Bootstrap &        &         \\
                  & Type I      & Power  & Time      & Type I               & Power  & Time    \\
      $200$       & $0.06$      & $0.43$ & $48.18$   & $0.08$               & $0.44$ & $5.39$  \\
      $300$       & $0.08$      & $0.45$ & $112.54$  & $0.05$               & $0.5$  & $5.72$  \\
      $400$       & $0.08$      & $0.56$ & $258.88$  & $0.08$               & $0.68$ & $5.51$  \\
      $500$       & $0.08$      & $0.68$ & $423.72$  & $0.02$               & $0.7$  & $7.83$  \\
      $600$       & $0.07$      & $0.87$ & $676.67$  & $0.07$               & $0.83$ & $11.69$ \\
      $700$       & $0.04$      & $0.86$ & $1045.71$ & $0.04$               & $0.87$ & $17.82$ 
      \end{tblr}
    \end{minipage}%
    \begin{minipage}{.5\linewidth} 
        \begin{tblr}{
            cells = {c},
            cell{1}{1} = {c=7}{},
            cell{2}{1} = {r=2}{},
            cell{2}{2} = {c=3}{},
            cell{2}{5} = {c=3}{},
            vlines,
            hline{1-2,4-10} = {-}{},
            hline{3} = {2-7}{},
            colsep = 1.5pt,
      rowsep = 0.5pt
          }
          $B = 1000$  &             &        &            &                      &        &         \\
          {Sample\\Size} & Permutation &        &            & Multiplier Bootstrap &        &         \\
                      & Type - I    & Power  & Time       & Type - I             & Power  & Time    \\
          $200$       & $0.06$      & $0.42$ & $483.70$   & $0.04$               & $0.42$ & $9.41$  \\
          $300$       & $0.06$      & $0.49$ & $1103.24$  & $0.03$               & $0.51$ & $9.65$  \\
          $400$       & $0.07$      & $0.56$ & $2248.31$  & $0.06$               & $0.57$ & $10.45$ \\
          $500$       & $0.06$      & $0.67$ & $3995.38$  & $0.06$               & $0.66$ & $11.52$ \\
          $600$       & $0.04$      & $0.85$ & $6530.46$  & $0.05$               & $0.79$ & $16.77$ \\
          $700$       & $0.05$      & $0.84$ & $10073.89$ & $0.06$               & $0.91$ & $23.54$ 
          \end{tblr}
       \end{minipage}
    \caption{ Type I error rate, power, and time (in seconds) for the hypothesis in \eqref{eq:permexp2}. } 
    \label{tab:permexp2}
\end{table}

In both the above cases we consider the \texttt{Gauss MMD} test with bandwidths from \eqref{eq:gaussmmmd}. From Table \ref{tab:permexp1} and Table \ref{tab:permexp2} we observe that while the Type I error rate and power of the multiplier bootstrap and the permutation methods are comparable, the time required by the multiplier bootstrap is significantly faster than the permutation test. This advocates the use of the multiplier bootstrap method for calibrating the MMMD statistic.

\color{black}

\subsection{ Comparisons with Bandwidth Optimized Single-Kernel Tests } 
\label{sec:bandwithoptimized}

In this section we report our experimental results comparing the MMMD tests with the bandwidth optimized single kernel test, as described in Section \ref{sec:bandwithoptimizedcombination}. Specifically, we compare the performance of \texttt{Gauss MMMD}, \texttt{Lap MMMD} and \texttt{Mixed MMMD} test (with bandwidths chosen as in \eqref{eq:gaussmmmd}, \eqref{eq:laplacemmmd}, and \eqref{eq:mixedmmmd}, respectively) with the single Gaussian/Laplace kernel MMD tests, where the bandwidths are optimized as described in Section \ref{sec:bandwithoptimizedcombination}. We refer to the bandwidth optimized single kernel tests as \texttt{Opt Gauss MMD} and \texttt{Opt Lap MMD}, respectively. 
We also implement the oracle single kernel tests with twice the amount of data as described in Section \ref{sec:bandwithoptimizedcombination}. We refer to these tests as \texttt{Oracle Gauss MMD} and \texttt{Oracle Lap MMD} for the Gaussian and Laplace kernel choices, respectively.
We borrow the implementation of \texttt{Opt Gauss MMD}, \texttt{Opt Lap MMD}, \texttt{Oracle Gauss MMD}, and \texttt{Oracle Lap MMD} from Section 5.3 in \cite{schrab2021mmd}. Specifically, the bandwidth is optimized by maximizing the ratio $ \frac{1}{\hat\sigma_{\lambda}^2} \mmd^2\left[\sfK_{\lambda},\sX_m,\sY_n\right] $ over $\lambda\in \{c\lambda_{\texttt{med}}: c = 0.1, 0.2, \ldots, 2\}$. 

We compare the performance of the above tests in the following simulation settings: 

\begin{itemize} 

\item {\it Mixture alternatives}: Here, as in Section \ref{sec:mixtureexperiments} we consider 
$$P = \varepsilon \cN_d(\bm{0} ,\bm{\Sigma}_{0}) + (1-\varepsilon)  t_{10}(\bm{0} ,\bm{\Sigma}_{0}) \text{ and } Q = \varepsilon \cN_d(\bm{0} ,1.25\bm{\Sigma}_{0}) + (1-\varepsilon) t_{10}(\bm{0} ,1.25\bm{\Sigma}_{0}), $$ 
where $\bm{\Sigma}_{0} = ((0.5^{|i-j|}))_{1\leq i,j\leq d}$. 
Figure \ref{fig:bandwidthoptimizedcombination}(a) shows the empirical power (averaged over 500 iterations) of the different tests as $\varepsilon$ varies over $[0,1]$, with sample sizes $m=n=100$, and dimension $d=30$.  

\item {\it Local alternatives}: As in \eqref{eq:scalelocalalternative}, suppose 
    \begin{align*} 
        P = \cN_d(\bm{0} ,\bm{I}_{d})\text{ and } Q = \cN_d\left(\bm{0} , (1+\tfrac{h}{\sqrt N})\bm{I}_{d} \right) , 
    \end{align*} 
    where $N = m+n$. Figure \ref{fig:bandwidthoptimizedcombination}(b) shows the empirical power (averaged over 500 iterations)  of the different tests, for dimension $d = 20$, sample sizes $m=n=100$, as the signal strength varies over [0, 2].

\item {\it Perturbed $1$-dimensional uniform distribution}: This is the example from \cite{schrab2021mmd} described in Section \ref{sec:uniformR}. In particular, $P$ is the uniform distribution on $[0, 1]$ and $Q$ is a perturbed version of uniform distribution on $[0, 1]$ as in \eqref{eq:ftheta}. Figure \ref{fig:bandwidthoptimizedcombination}(c) shows the empirical power (averaged over 500 iterations)  of the different tests, for sample sizes $m=n=100$ and $R=1,2,3,4,5,6$. 

\end{itemize}

\begin{figure}[!ht]
    \centering 
    \begin{subfigure}[c]{0.32\textwidth}
    \includegraphics[scale = 0.28]{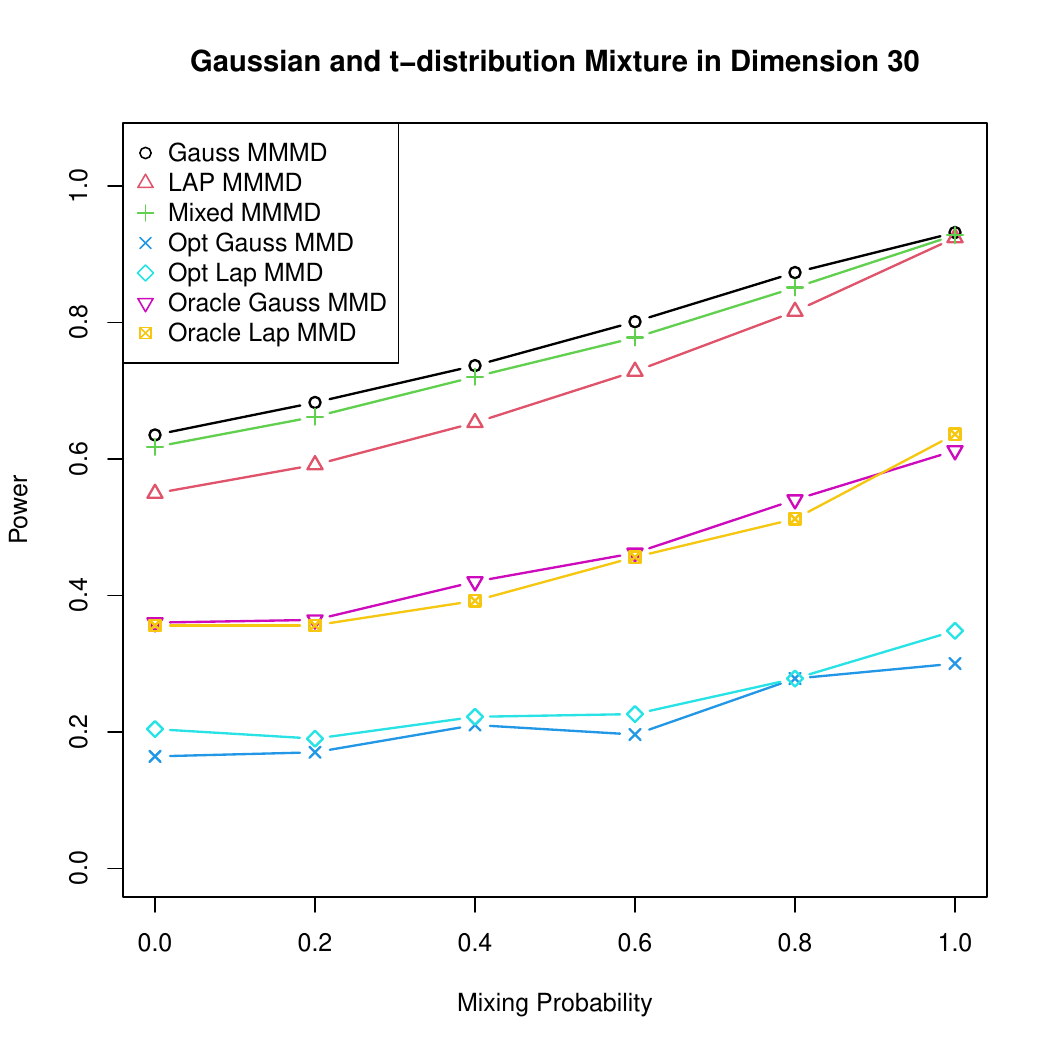} 
     \caption*{\small{(a)}} 
\end{subfigure}    
\begin{subfigure}[c]{0.32\textwidth}
    \includegraphics[scale = 0.28]{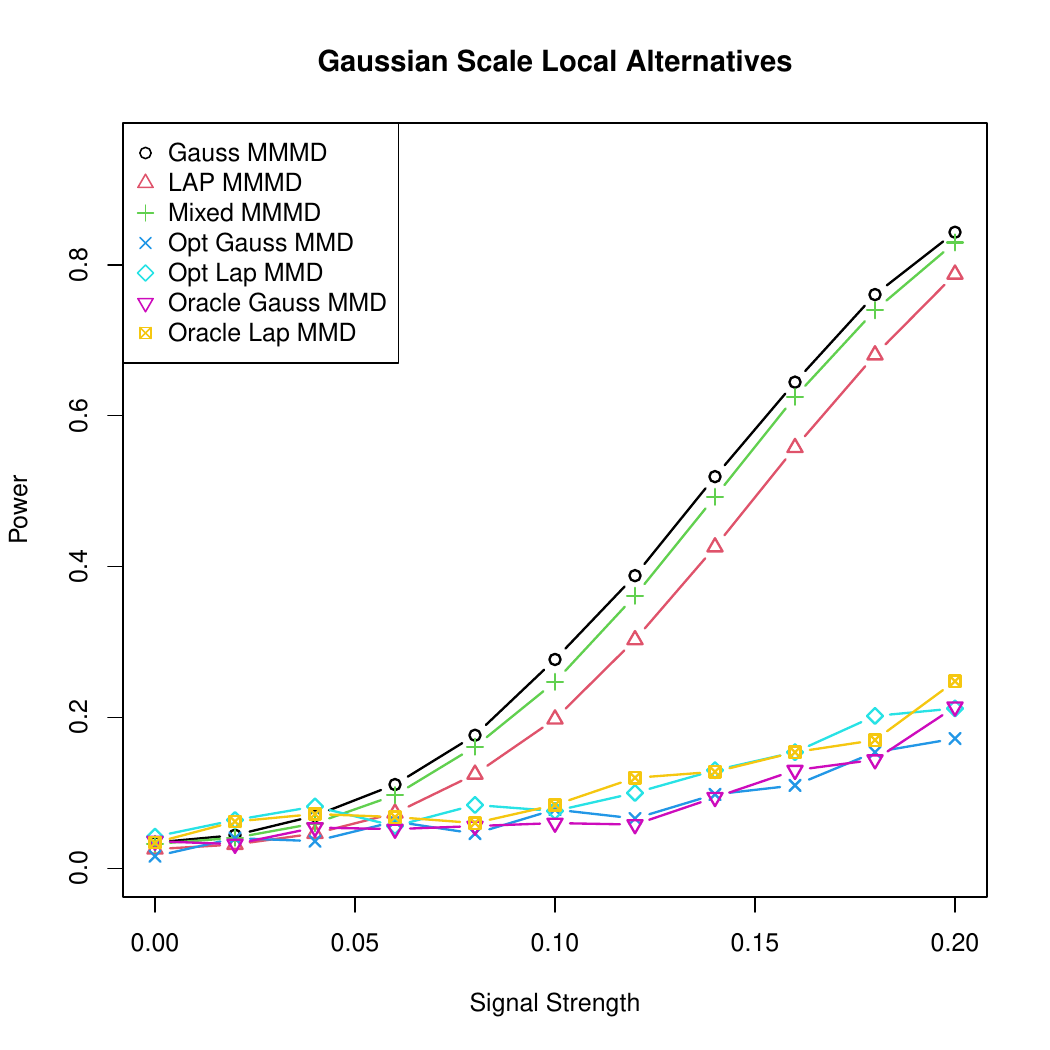} 
     \caption*{\small{(b)}} 
\end{subfigure}    
\begin{subfigure}[c]{0.32\textwidth}
   \includegraphics[scale = 0.28]{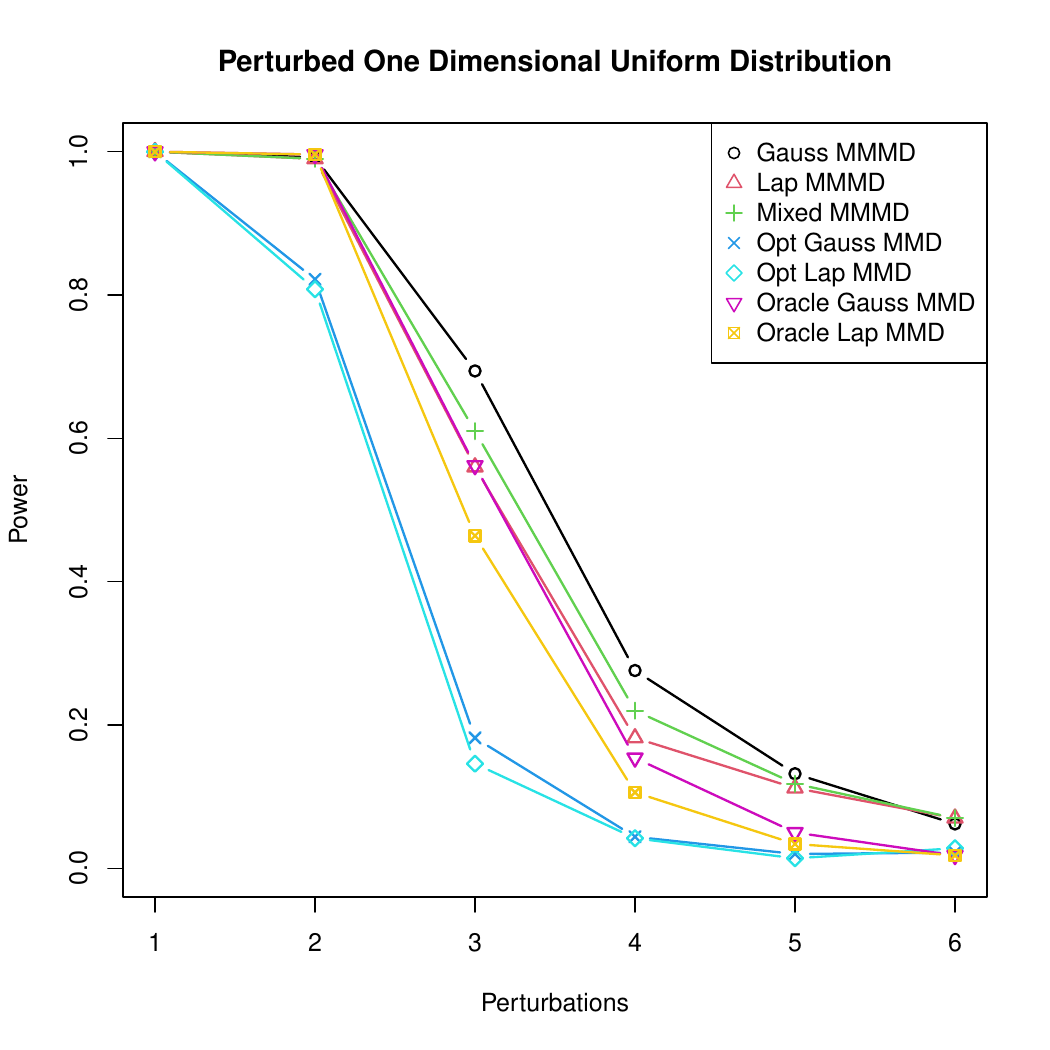} 
    \caption*{\small{(c)}} 
\end{subfigure} 
    \caption{ Empirical power comparison of MMMD, \texttt{Opt MMD}, and \texttt{Oracle MMD} tests. }
    \label{fig:bandwidthoptimizedcombination}
\end{figure}

As mentioned in Section \ref{sec:bandwithoptimizedcombination}, in all the above cases the MMMD tests have improved power than the bandwidth optimized single kernel tests and also their oracle counterparts.

\subsection{ Comparisons with $p$-Value Combination Methods } 
\label{sec:combination}

In this section we present our results comparing the {\tt Gauss MMMD} test (with bandwidth from \eqref{eq:gaussmmmd}) with the $p$-value combination strategies described in Section \ref{sec:bandwithoptimizedcombination}. We compare the performance of the different tests in the following simulation settings: 
\begin{enumerate}
    \item {\it High Dimensional Alternatives:} As in \eqref{eq:normaldimension}, consider 
    \begin{align}\label{eq:highdimpvalues} 
        P = \cN_d\left(\bm{0} , \bm{I}_{d}\right)\text{ and } Q = \cN_d\left( (1.25/\sqrt{d}) \bm{1}, \bm{I}_{d}\right).
    \end{align} 
    Figure \ref{fig:pcomb}(a) shows the empirical powers (averaged over 100 iterations) of the different tests as a function of the dimension with sample sizes $m = n= 100$.  
    
    \item {\it Scale Alternatives: } As in \eqref{eq:testpowertime}, suppose  
    \begin{align}\label{eq:scalepvalues}
        P = \cN_{d}\left(\bm{0}, \bm{I}_{d}\right)\text{ and }Q = \cN_{d}\left(\bm{0}, \sigma^2\bm{I}_{d}\right) , 
    \end{align}
    with $d = 20$ and $\sigma^{2} = 1.1$. Figure \ref{fig:pcomb}(b) shows the empirical power (averaged over 100 iterations) of the different tests with varying sample sizes.
\end{enumerate}

In both the above cases, the {\tt Gauss MMMD} outperforms the $p$-value combination methods, illustrating the effectiveness of our aggregation strategy.  

\begin{figure}
    \centering
    \begin{subfigure}{.5\textwidth}
      \centering
      \includegraphics[width=\linewidth]{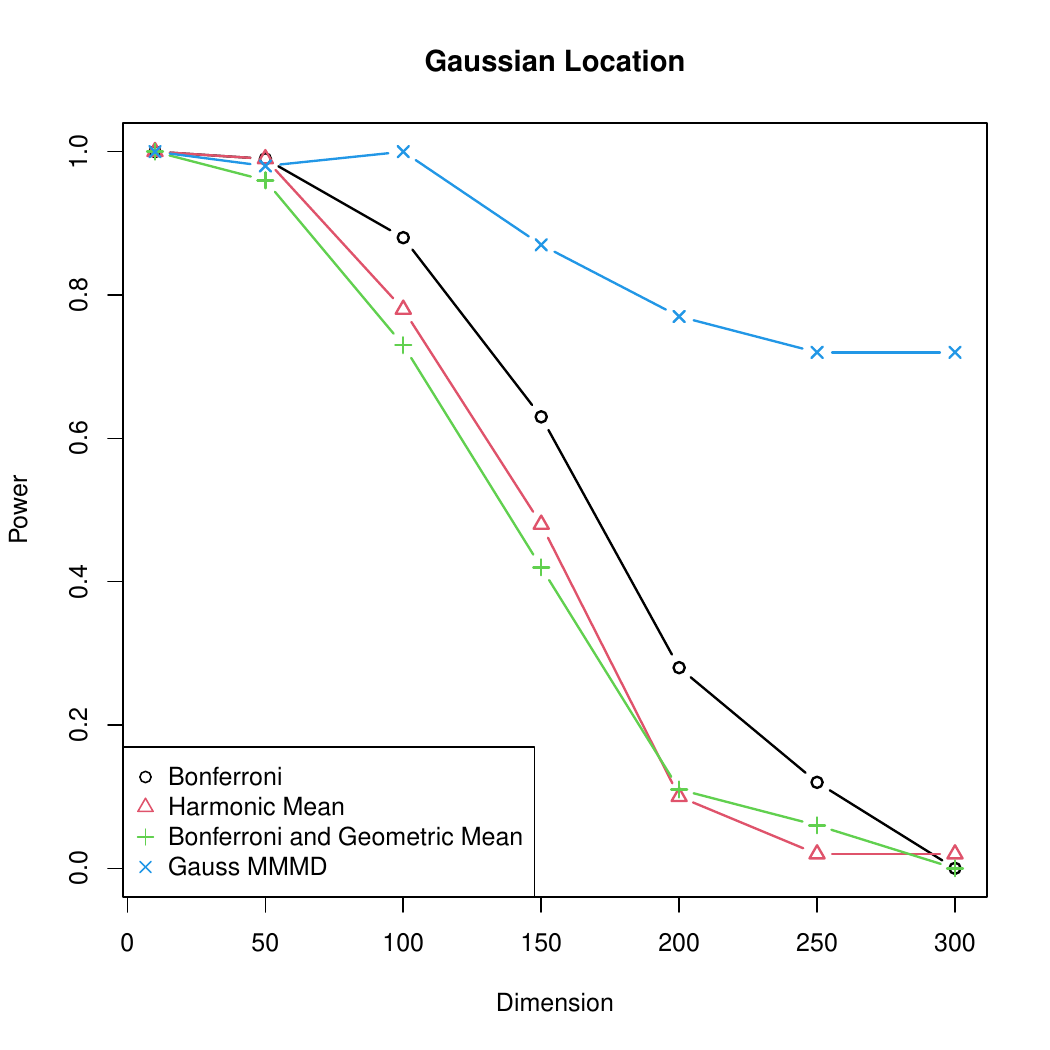}
      \caption*{\small (a)}
    \end{subfigure}%
    \begin{subfigure}{.5\textwidth}
      \centering
      \includegraphics[width=\linewidth]{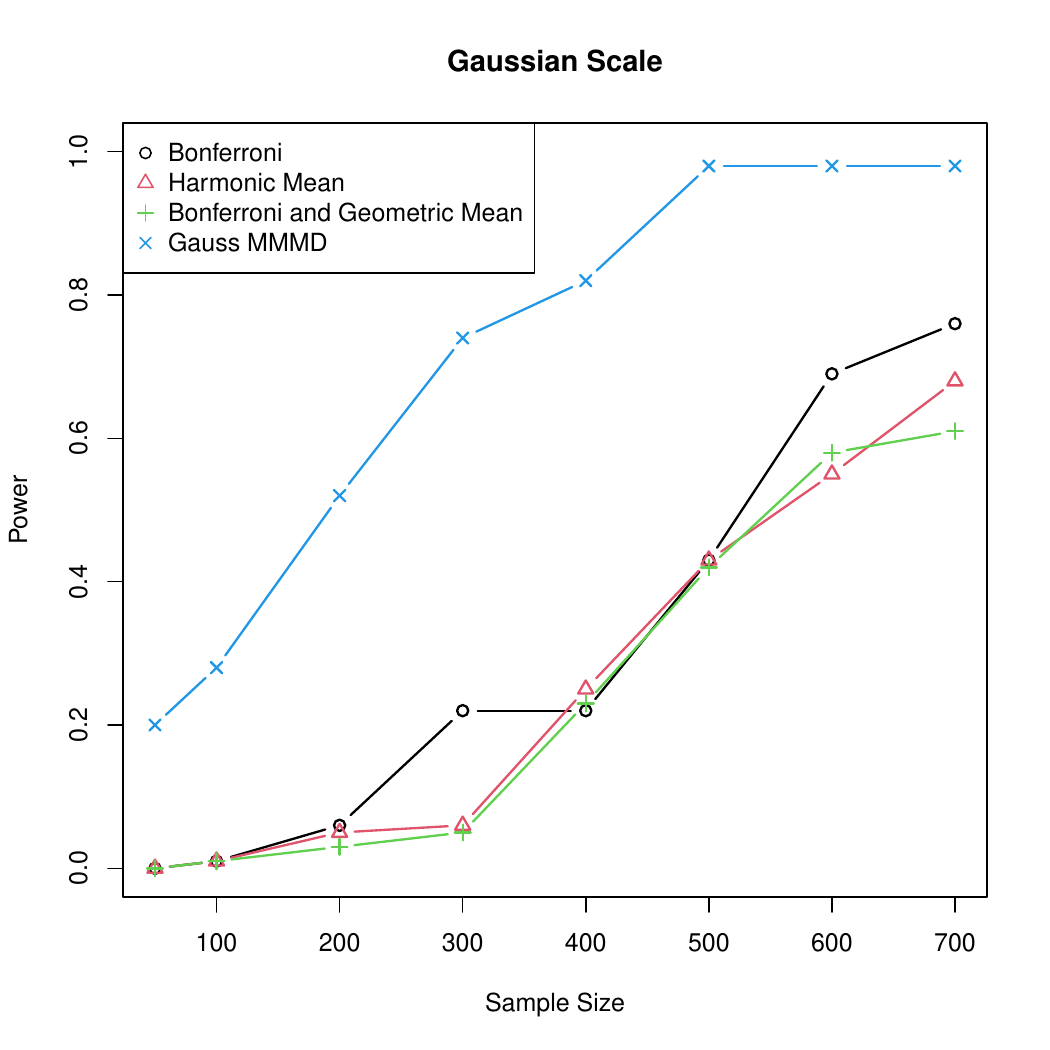}
      \caption*{\small (b)}
    \end{subfigure}
    \caption{Empirical powers for (a) high-dimensional alternatives and (b) scale alternatives as in \eqref{eq:highdimpvalues} and \eqref{eq:scalepvalues}, respectively.}
    \label{fig:pcomb}
    \end{figure}

\color{black}

 \section{Linear Time Statistic and Aggregation}
 \label{sec:linear}

For a kernel $\sfK:\cX\times\cX\ra \R$ define,
\begin{align*}
    \sfL_{\sfK}(z_{1},z_{2}):= \sfK(x_{1},x_{2}) + \sfK(y_{1},y_{2}) - \sfK(x_{1},y_{2}) - \sfK(x_{2},y_{1}) , 
\end{align*}
where $z_{1} = (x_{1},y_{1})$ and $z_{2} = (x_{2},y_{2})$. For simplicity we will assume $m=n$.
Then the linear-time statistic is defined as 
\begin{align}\label{eq:defMMDl}
    \emmd_{\sfL}\left[\sfK, \sX_{m}, \sY_{m}\right] = \frac{1}{\lf m/2\rf}\sum_{i=1}^{\lf m/2\rf}\sfL_{\sfK}(Z_{2i-1}, Z_{2i}) , 
\end{align}
where $\sX_{m} = \{X_{1},\ldots,X_{m}\}$, $\sY_{m} = \{Y_{1},\ldots,Y_{m}\}$, and $Z_{i} = (X_{i}, Y_{i})$, for $1\leq i\leq m$. The statistic \eqref{eq:defMMDl} is an unbiased estimate of the population MMD \eqref{eq:KPQ} and can be computed in linear time (see \citet[Lemma 14]{gretton2012kernel}).
Moreover, the computation of $\emmd_{\sfL}$ requires only $O(1)$ storage, and thus is appealing to use in the presence of data streams. 

Given a collection of kernels $\cK = \{\sfK_{1},\ldots,\sfK_{r}\}$ define the vector of linear time MMD estimates as follows: 
\begin{align*}
    \emmd_{\sfL}\left[\cK,\sX_{m},\sY_{m}\right] = \left(\emmd_{\sfL}\left[\sfK_{1}, \sX_{m}, \sY_{m}\right],\ldots, \emmd_{\sfL}\left[\sfK_{r}, \sX_{m}, \sY_{m}\right]\right)^\top
\end{align*}
Since \eqref{eq:defMMDl} is an average of a set of i.i.d. distributed random variables, by the multivariate Central Limit Theorem we have as $m\ra\infty$, under the same conditions as in Theorem \ref{thm:MMDdistributionPQ} (in Appendix \ref{sec:H1asymptotic}), 
    \begin{align}\label{eq:linearH0H1}
        \sqrt{m}\left(\emmd_{\sfL}\left[\cK,\sX_{m},\sY_{m}\right] - \emmd\left[\bm\cF,P,Q\right]\right)\dto \mathrm{N}_{r}(\bm{0},\bm{\Sigma}^{(\sfL)}) , 
    \end{align}
    where $\bm{\Sigma}^{(\sfL)} = (\sigma_{ab}^{\sfL})_{1 \leq a, b \leq r}$ is a $r\times r$ matrix having entries,
    \begin{align*}
        \sigma_{ab}^{(\sfL)}:= 2\Cov\left[\sfL_{\sfK_{a}}(Z,Z'), \sfL_{\sfK_{b}}(Z,Z')\right], \quad \text{ for } 1\leq a,b\leq r , 
    \end{align*}
    where $(Z,Z') \sim P\times Q$. 
To construct a test statistic from \eqref{eq:linearH0H1} we need to estimate the matrix  
$\bm{\Sigma}^{(\sfL)}$. This can be easily obtained by considering 
\begin{align*}
    \hat{\sigma}_{ab}^{(\sfL)}:= 2\left[\frac{1}{\lf m/2\rf}\sum_{i=1}^{\lf m/2\rf}\sfL_{\sfK_{a}}(Z_{2i-1}, Z_{2i})\sfL_{\sfK_{b}}(Z_{2i-1},Z_{2i}) - \bar{\sfL}_{\sfK_{a}}\bar{\sfL}_{\sfK_{b}}\right] , 
\end{align*} 
where $\bar{\sfL}_{\sfK_{a}}:= \frac{1}{\lf m/2\rf}\sum_{i=1}^{\lf m/2\rf}\sfL_{\sfK_{a}}(Z_{2i-1}, Z_{2i})$, for $1\leq a,b\leq r$. Note that by the law of large numbers, $\hat{\sigma}_{ab}^{(\sfL)} \stackrel{a.s.}\rightarrow \sigma_{ab}^{(\sfL)}$, for $1\leq a,b\leq r$, hence, $\hat{\bm{\Sigma}}^{(\sfL)}\asto\bm{\Sigma}^{(\sfL)}$. This together with \eqref{eq:linearH0H1} leads to the following result: 

\begin{theorem}\label{thm:LinTimeConvgH0}
  Under the assumptions of Theorem \ref{thm:MMDdistributionPQ}, 
    \begin{align*}
        \sqrt{m}\hat{\bm{\Sigma}}_{\sfL}^{-\frac{1}{2}}\left(\emmd_{\sfL}\left[\cK,\sX_{m},\sY_{m}\right] - \emmd\left[\bm\cF,P,Q\right]\right)\dto \mathcal{N}_{r}\left(\bm{0}, \bm I_{r}\right) . 
    \end{align*}
\end{theorem} 

    Note that under $H_{0}$, $\emmd\left[\bm\cF,P,Q\right]= 0$. Hence, we define the  \textit{Mahalanobis Aggregated Linear MMD} (MLMMD) statistic as: 
    \begin{align*}
        T_{m}^{(\sfL)} := \left(\emmd_{\sfL}\left[\cK,\sX_{m},\sY_{m}\right]\right)^{\top}\hat{\bm{\Sigma}}_{\sfL}^{-1}\left(\emmd_{\sfL}\left[\cK,\sX_{m},\sY_{m}\right]\right).
    \end{align*} 
    Clearly, $T_{m}^{(\sfL)}$ can be computed in linear time, when $r=O(1)$. Also, by Theorem \ref{thm:LinTimeConvgH0}, under $H_{0}$, $m T_{m}^{(\sfL)}\dto \chi_{r}^2$. 
Now, for any $\alpha>0$, consider the test 
\begin{align}\label{eq:LMMMDtest}
    \phi_{m}^{(\sfL)} = \bm{1}\left\{mT_{m}^{(\sfL)}> \chi^2_{r, 1-\alpha} \right\} , 
\end{align} 
where $\chi^2_{r, 1-\alpha}$ is the $(1-\alpha)$-th quantile of the $\chi^2_r$ distribution. Theorem \ref{thm:LinTimeConvgH0} implies that \eqref{eq:LMMMDtest} is asymptotically level $\alpha$ and universally consistent for the hypothesis in \eqref{eq:H01PQ}. 

\begin{figure}[!ht]
    \centering 
    \begin{subfigure}[c]{0.32\textwidth}
    \includegraphics[scale = 0.285]{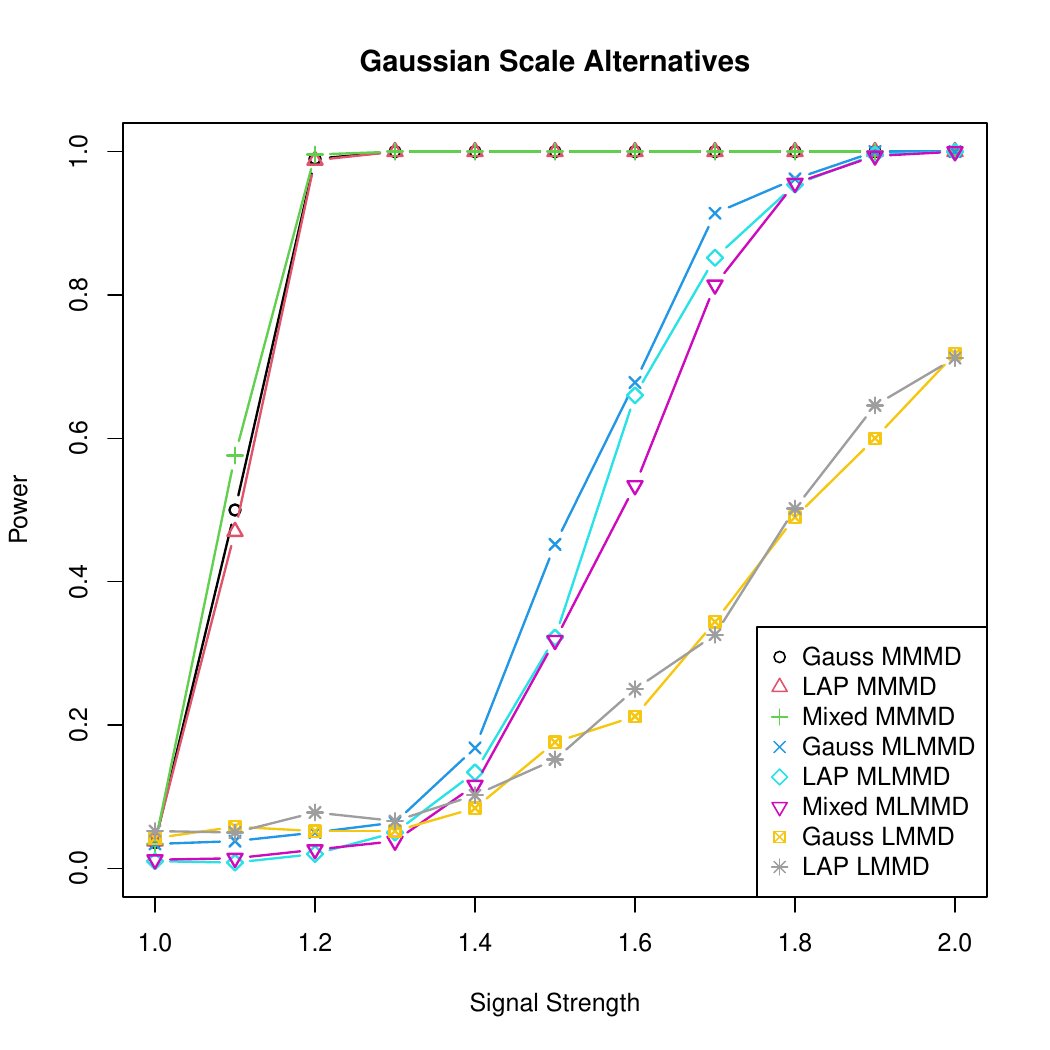} 
     \caption*{\small{(a)}} 
\end{subfigure}    
\begin{subfigure}[c]{0.32\textwidth}
    \includegraphics[scale = 0.285]{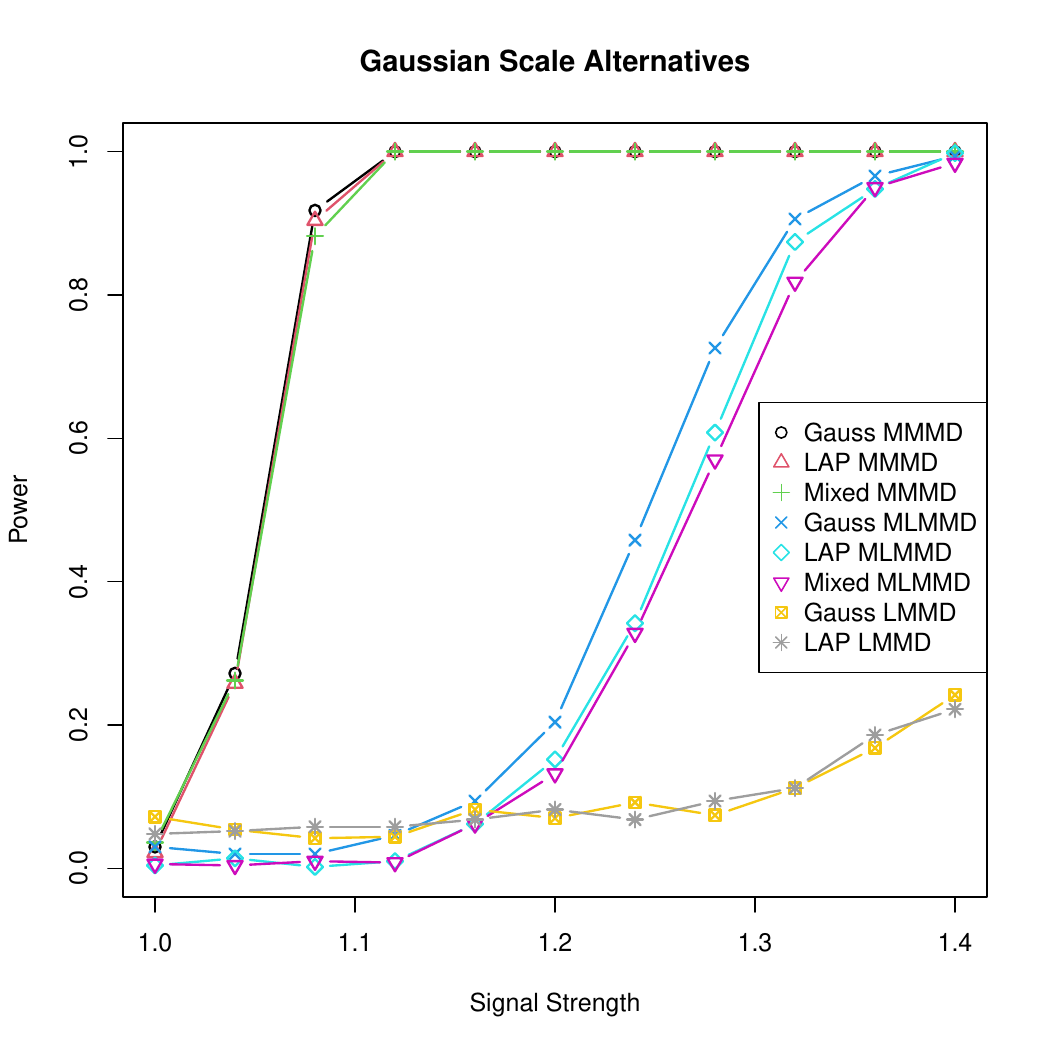} 
     \caption*{\small{(b)}} 
\end{subfigure}    
\begin{subfigure}[c]{0.32\textwidth}
   \includegraphics[scale = 0.285]{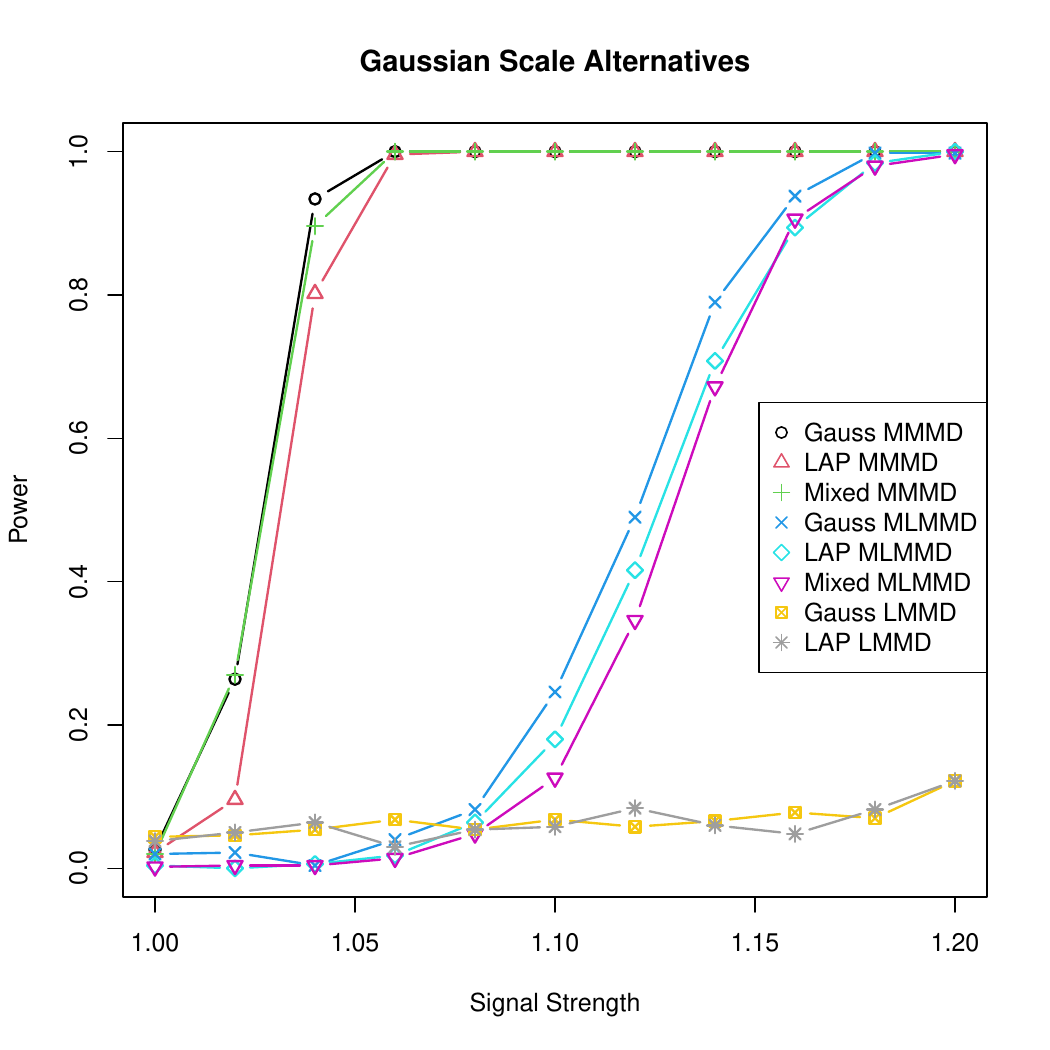} 
    \caption*{\small{(c)}} 
\end{subfigure} 
    \caption{ \textcolor{black}{ Empirical powers of the MMMD, MLMMD and LMMD tests for (a) $d=20$ (b) $d = 75$ and (c) $d= 300$. }} 
    \label{fig:MMMDvsLMMMD}
\end{figure}

\subsection{Empirical Performance of Linear MMD Aggregation}
\label{sec:empiricalsimulations}

\color{black}

In this section we empirically compare the performance of the the linear time MLMMD test in \eqref{eq:LMMMDtest} with their single kernel counterparts and the quadratic time MMMD test in   \eqref{eq:Tmnalpha}.

For the MLMMD statistic we consider a collection of Gaussian, Laplace, and mixed kernels as in Section \ref{sec:experiments}. We refer to these as \texttt{Gauss MLMMD}, \texttt{LAP MLMMD} and \texttt{Mixed MLMMD}, with the kernels chosen from the collections in  \eqref{eq:gaussmmmd}, \eqref{eq:laplacemmmd} and \eqref{eq:mixedmmmd}, respectively. For the single kernel counterparts we consider the Gaussian and Laplace kernels from \eqref{eq:singlekernels} with median bandwidth, and refer to them as \texttt{Gauss LMMD} and \texttt{LAP LMMD}, respectively. For comparison, we consider the following simulation setting,
\begin{align*}
    P = \cN_{d}\left(\bm 0, \bm I_{d}\right)\text{ and }Q = \cN_{d}\left(\bm 0, \sigma^2\bm I_{d}\right) , 
\end{align*}
where $\sigma^2$ is chosen based on the dimension $d$. We show the empirical powers (averaged over 500 iterations) of the different tests, with sample sizes $m=n=200$, dimension $d = 20$ (Figure \ref{fig:MMMDvsLMMMD} (a)), $d = 75$ (Figure \ref{fig:MMMDvsLMMMD} (b)),  and $d=300$ (Figure \ref{fig:MMMDvsLMMMD} (c)), as the signal strength $\sigma^2$ varies over $[0, 2]$, $[0,1.4]$ and $[0,1.2]$, respectively. 
The MLMMD tests outperform single kernel LMMD tests, highlighting, once again, the advantage of aggregation over multiple kernels. 
Also, as expected, the quadratic time MMMD tests outperform the linear time MLMMD tests in terms of power. Nevertheless, the MLMMD tests are not far behind and, as the signal strength becomes larger, and the powers of the all the tests approach 1.

\color{black}

\section{Proofs from Section \ref{sec:rkernels} }
\label{sec:maxL2MMMMDpf}

The proofs of the consistency in the growing $r$ regime rely on finite sample concentration bounds of the MMD estimate \eqref{eq:MMDXY}  around the population $\mmd$ (from \eqref{eq:KPQ}). To this end, we recall the following result from \cite{gretton2012kernel}, which can be proved by an application of the bounded difference inequality. 
\begin{proposition}\label{prop:concMMDmn} \citep[Theorem 17]{gretton2012kernel}
    Let $\sfK$ be a characteristic kernel with unit ball $\cF$ in the corresponding RKHS such that $0\leq \sfK\leq K$. Then for $t>0$,
    \begin{align*}
        \P\left(\left|\emmd\left[\sfK, \sX_{m}, \sY_{m}\right] - \emmd\left[\cF, P, Q\right]\right|>t\right)\leq 6\exp\left(-\frac{t^2 m }{64 K^2}  \right) . 
    \end{align*}
\end{proposition}

We will use this result to establish the consistency of the maximum and $L_2$ aggregations (Proposition \ref{ppn:concMaxL2}) in Appendix \ref{sec:concMaxL2pf} and the consistency of the Mahalanobis aggregation (Theorem \ref{thm:MahaAggrmntest}) in Appendix \ref{sec:proofMahaAggrmn}.  

\subsection{Proof of Proposition \ref{ppn:concMaxL2} }
\label{sec:concMaxL2pf}
Consider $\bm{\cF}_{r_m} := \left\{\cF_a:1\leq a\leq r_{m}\right\}$ where $\cF_a$ is the unit ball in the RKHS generated by $\sfK_a$, for all $1\leq a\leq r_m$. First we proceed with proving validity and consistency of $\phi_{m}^{\max}$. By union bound and Proposition \ref{prop:concMMDmn} notice that,
\begin{align}\label{eq:maxconc}
    \P\bigg(\left|T_m^{\max} - \max_{a=1}^{r_m}\mmd^2\left[\cF_a, P,Q\right]\right| > t\bigg)\leq 6r_m\exp\left(-\frac{t^2m}{64K^2}\right)
\end{align}
Recall that under $H_0$, $\mmd^2\left[\cF_a, P, Q\right] = 0$ for all $1\leq a\leq r_m$. Hence \eqref{eq:maxconc} gives,
\begin{align*}
    \P_{H_0}\left(\left|T_m^{\max}\right|>t_{\max}\right)\leq 6r_m\exp\left(\log\frac{\alpha}{6r_m}\right)\leq \alpha,
\end{align*}
where $t_{\max} = 8K\sqrt{\frac{1}{m}\log\frac{6r_m}{\alpha}}$.
Now note that,
\begin{align*}
    \P\left(\left|T_m^{\max}\right|\leq t_{\max}\right)
    &\leq\P\left(\left|T_{m}^{\max} - \max_{a=1}^{r_m}\mmd^2\left[\cF_a, P,Q\right]\right|> \max_{a=1}^{r_m}\mmd^2\left[\cF_a, P,Q\right] - t_{\max}\right)\\
    &\leq \P\left(\left|T_{m}^{\max} - \max_{a=1}^{r_m}\mmd^2\left[\cF_a, P,Q\right]\right|> \mmd^2\left[\cF_1, P,Q\right] - t_{\max}\right)
\end{align*}
By assumption from Proposition \ref{prop:concMMDmn} recall that $\log r_{m}\ll m$. Hence under $H_1$, for large enough $m, \mmd^2\left[\cF_1, P,Q\right] - t_{\max}>\mmd^2\left[\cF_1, P,Q\right]/2>0$ and by \eqref{eq:maxconc},
\begin{align*}
    \P_{H_1}\left(\left|T_m^{\max}\right|\leq t_{\max}\right)
    &\leq 6r_m\exp\left(-\frac{\left(\mmd^2\left[\cF_1, P,Q\right]/2\right)^2m}{64K^2}\right) = o(1),
\end{align*}
which completes the proof of consistency of $\phi_m^{\max}$. Now we proceed with proving validity and consistency of $\phi_{m}^{L_2}$. Once again by union bound an Proposition \ref{prop:concMMDmn} notice that,
\begin{align}\label{eq:concL2}
    \P\left[\left|T_m^{L_2} - T_{m,0}^{L_2}\right|>t\right]
    &\leq \sum_{a=1}^{r_m}\P\left(\left|\mmd^2\left[\sfK_a,\sX_m,\sY_m\right] - \mmd^2\left[\cF_a, P,Q\right]\right|>\frac{t}{\sqrt{r_m}}\right)\nonumber\\
    &\leq 6r_m\exp\left(-\frac{t^2m}{64r_{m}K^2}\right),
\end{align}
where $T_{m,0}^{L_2}:= \left\|\mmd^2\left[\bm\cF_{r_m}, P,Q\right]\right\|_2$ (recall \eqref{eq:defbcFmmd}). Once again recall that under $H_0$, $\mmd^2\left[\cF_a, P, Q\right] = 0$ for all $1\leq a\leq r_m$. Then \eqref{eq:concL2} along with $t_{L_2} = 8K\sqrt{\frac{r_m}{m}\log \frac{6r_m}{\alpha}}$ gives,
\begin{align*}
    \P_{H_0}\left(\left|T_m^{L_2}\right|>t_{L_2}\right)\leq 6r_m\exp\left(\log\frac{\alpha}{6r_m}\right)\leq \alpha.
\end{align*}
Finally to show consistency of $\phi_m^{L_2}$, notice that,
\begin{align*}
    \P\left(\left|T_{m}^{L_2}\right|\leq t_{L_2}\right)
    &\leq \P\left(\left|T_m^{L_2} - T_{m,0}^{L_2}\right|>T_{m,0}^{L_2} - t_{L_2}\right)\\
    &\leq \P\left(\left|T_m^{L_2} - T_{m,0}^{L_2}\right|> \mmd^2\left[\cF_1, P,Q\right] - t_{L_2}\right).
\end{align*}
By assumption recall that $r_m\log r_m\ll m$. Thus, under $H_1$, for large enough $m$, $\mmd^2\left[\cF_1, P,Q\right] - t_{L_2}>\mmd^2\left[\cF_1, P,Q\right]/2>0$. Then by \eqref{eq:concL2} we get,
\begin{align*}
    \P_{H_1}\left(\left|T_m^{L_2}\right|\leq t_{L_2}\right)\leq 6r_m\exp\left(-\frac{\left(\mmd^2\left[\cF_1, P, Q\right]/2\right)^2m}{64r_mK^2}\right) = o(1).
\end{align*}

\subsection{Proof of Theorem \ref{thm:MahaAggrmntest} }\label{sec:proofMahaAggrmn} 

For any matrix $\bm A = ((a_{ij}))_{1 \leq i, j \leq m}$, we denote by $\|\bm A\|_2 = \max_{\|\bm x\|_2=1} \| \bm A \bm x\|_2$ the operator norm and by $\|\bm A\|_F= (\sum_{1 \leq i, j \leq m} a_{ij}^2)^{\frac{1}{2}}$ the Frobenius norm of $\bm A$. From \eqref{eq:H0sigma} and \eqref{eq:H0sigmaestimate} recall that $\bm\Sigma_{r_{m}} = ((\sigma_{ab}))_{1 \leq a,b \leq r_{m}}$ and $\hat{\bm\Sigma}_{r_{m}} = ((\hat\sigma_{ab}))_{1 \leq a,b \leq r_{m}}$,  
with  
\begin{align}\label{eq:sigmahatsigma}
\sigma_{ab} = 16\mathbb{E}\left[\sfK_a^\circ(X, X')\sfK_b^\circ(X, X')\right] \text{ and } \hat\sigma_{ab} = \frac{16}{m^2} \sum_{1 \leq i, j \leq m} \hat \sfK_a^\circ(X_i, X_j) \hat \sfK_b^\circ(X_i, X_j) , 
\end{align} 
where $X, X' \sim P$, for $1\leq a,b\leq r_{m}$. The main technical ingredient in the proof of Theorem \ref{thm:MahaAggrmntest} is the following lemma which shows that $\hat\sigma_{ab}$ concentrates around $\sigma_{ab}$ for any $1\leq a,b\leq r_{m}$. Consequently, the smallest and largest eigenvalues of $\hat{\bm\Sigma}_{r_{m}}$ also concentrates around the smallest and largest eigenvalues of $\bm\Sigma_{r_{m}}$, respectively.

\begin{lemma}\label{lemma:concsigmahatsigma}
    Suppose $\sup_{1 \leq a \leq r_m}\sfK_{a}\leq K$.  Then for all $1 \leq a, b \leq r_m$ and $t > 1600K^2/m$, 
\begin{align}\label{eq:hatsigmaexponential}
\P\left[\left|\hat\sigma_{ab} - \sigma_{ab}\right|> t \right] \leq 2\exp\left(-\frac{m t^2}{ 4 LK^4 } \right) , 
\end{align}   
    where $L:=165888$. Furthermore, for $t > 1600K^2r_m/m$, 
   \begin{align}\label{eq:lambda1m}
       \max\left\{ \P\left[\left|\hat\lambda_{1} - \lambda_{1}\right|>t\right] ,  \P\left[\left|\hat\lambda_{m} - \lambda_{m}\right|>t\right] \right\} \leq 2r_{m}^2\exp\left(-\frac{m t^2}{4 K^4L r_{m}^2 } \right) , 
   \end{align}
where $\lambda_1, \lambda_m$ are the largest and smallest eigenvalues of $\bm{\Sigma}_{r_{m}}$, respectively, and $\hat{\lambda}_1, \hat{\lambda}_m$ are the largest and smallest eigenvalues of 
$\hat{\bm\Sigma}_{r_{m}}$, respectively. 
\end{lemma}

The proof of Lemma \ref{lemma:concsigmahatsigma} is given in Appendix \ref{sec:concSigmainv}. We now use this result to prove Theorem \ref{thm:MahaAggrmntest}. To begin with, recalling \eqref{eq:TmnMA} note that 
\begin{align*}
|T_m^{MA}| & = \bigg|\emmd\left[\cK_{r_{m}},\cX_{m},\cY_{m}\right]^{\top}
    \hat{\bm\Sigma}_{r_{m}}^{-1}\emmd\left[\cK_{r_{m}},\cX_{m},\cY_{m}\right] \bigg| \nonumber \\ 
    &\leq \frac{1}{\hat \lambda_{m} }\left\|\emmd\left[\cK_{r_{m}},\cX_{m},\cY_{m}\right] \right\|_{2}^2 ,  
\end{align*} 
where $\hat \lambda_{m}\geq 0$ is the smallest eigenvalue of the matrix $\hat{\bm\Sigma}_{r_{m}}$. Now, consider the event $\mathcal A_1 = \{ |\hat \lambda_{m} - \lambda_{m} | \leq \omega_m \}$, for $\omega_m = \frac{r_m}{\sqrt m} \log r_m$. Note that $r_m/m = o(\omega_m)$. Hence, by \eqref{eq:lambda1m} in Lemma \ref{lemma:concsigmahatsigma}, $\P[\mathcal A_1^c] \rightarrow 0$. 
Also, note that 
$$\left\{ 
\left\|\emmd\left[\cK_{r_{m}},\cX_{m},\cY_{m}\right] \right\|_{2}^2 > t \hat \lambda_{m} \right\} \subseteq \left\{\left\|\emmd\left[\cK_{r_{m}},\cX_{m},\cY_{m}\right] \right\|_{2}^2 > t (\lambda_{m} - \omega_m) \right\}$$
on the event $\mathcal A_1$. Since $\emmd\left[\bm{\cF}_{r_{m}}, P, P \right] = 0$, by a union bound argument along with Proposition \ref{prop:concMMDmn} we get,
\begin{align}
    \P_{H_0}\left[ \{ |T_m^{MA}| >t \} \cap \mathcal A_1 \right]
    &\leq \P_{H_0}\left[\left\|\emmd\left[\cK_{r_{m}},\cX_{m},\cY_{m}\right] \right\|_{2}^2 > t  ( \lambda_{m} - \omega_m) \right]\nonumber\\ 
    &\leq \P_{H_0}\left[\left\|\emmd\left[\cK_{r_{m}},\cX_{m},\cY_{m}\right] \right\|_{\infty}> \sqrt{ \frac{t (\lambda_{m} - \omega_m) }{r_m} } \right]\nonumber\\ 
    &\leq 6r_{m}\exp\left(-\frac{t (\lambda_{m} - \omega_m) m}{64K^2 r_m} \right) \nonumber \\ 
     &\leq 6r_{m}\exp\left(-\frac{t \omega_{m} m}{32K^2 r_m} \right) =  6r_{m}\exp\left(-\frac{ t \sqrt m \log r_m}{32K^2} \right) , \label{eq:Tmsecondconc}
\end{align}  
where the last inequality holds for large $m$, since $\omega_m = o(\lambda_m)$ by assumption in Theorem \ref{thm:MahaAggrmntest}. 
Now, let $t = t_0 : = \frac{64K^2}{\sqrt m}$. By   
applying  \eqref{eq:Tmsecondconc}, under $H_0$ we get, 
$$\P_{H_0}\left[ \{ |T_m^{MA}| >t_0\} \cap \mathcal A_1 \right] \leq \frac{6}{r_m} \rightarrow 0.$$ Since $\P[\mathcal A_1^c] \rightarrow 0$, this implies, under $H_0$, $\P_{H_0}\left[ |T_m^{MA}| >t_0 \right] \rightarrow 0$. 

Next, we proceed to control the probability of Type II error. First note that
\begin{align}\label{eq:KPQprobability}
\P\left[ \| \emmd\left[\cK_{r_{m}},\cX_{m},\cY_{m}\right] - \emmd\left[ \bm{\cF}_{r_m}, P, Q \right] \|_\infty > t \right] 
& \leq 2 r_m \exp\left(-\frac{t^2 m}{64K^2} \right) . 
\end{align} 
Denote the event 
$$\mathcal B = \left\{ \left|\emmd\left[ \sfK_{a},\cX_{m},\cY_{m}\right] - \emmd\left[ \cF_{a}, P, Q \right]\right| \leq \frac{128K \log r_m}{\sqrt m} \text{ for all } 1 \leq a \leq r \right\}.$$ 
The inequality in \eqref{eq:KPQprobability} implies that $\P\left[ \mathcal B \right] \rightarrow 1$. Also, consider the event $\mathcal A_2 = \{ |\hat \lambda_{1} - \lambda_{1} | \leq \omega_m \}$. By \eqref{eq:lambda1m} in Lemma \ref{lemma:concsigmahatsigma}, $\P[\mathcal A_2^c] \rightarrow 0$. Now, denote 
$L_{r_m} := \inf_{1 \leq a \leq r_m} \emmd\left[ \cF_{a}, P, Q \right]$. On the set $\mathcal A_2 \cap \mathcal B$,
\begin{align*}
& |T_m^{MA}| \nonumber \\ 
& = \bigg|\emmd\left[\cK_{r_{m}},\cX_{m},\cY_{m}\right]^{\top}
    \hat{\bm\Sigma}_{r_{m}}^{-1}\emmd\left[\cK_{r_{m}},\cX_{m},\cY_{m}\right] \bigg| \nonumber \\ 
    & \geq \frac{1}{\hat \lambda_{1}}\left\|\emmd\left[\cK_{r_{m}},\cX_{m},\cY_{m}\right] \right\|_{2}^2 \nonumber \\  
       & \geq \frac{1}{\hat \lambda_{1}} \left( \left\|\emmd\left[ \bm{\cF}_{r_m}, P, Q \right] \right\|_{2}^2 - \frac{256 K\sqrt{r_m}\log r_m}{\sqrt m} \left\|\emmd\left[ \cK_{r_m}, P, Q\right] \right\|_2 - \frac{128^2 K^2r_m \log^2 r_m}{m}\right)\nonumber \\  
        & \geq \frac{1}{\hat \lambda_{1}} \left( r_m L_{r_m}^2  - \frac{256K^2 r_m \log r_m}{\sqrt m} - \frac{128^2K^2 r_m \log^2 r_m}{m}  \right)  \nonumber \\  
        & \geq \frac{r_m}{\lambda_{1} + \omega_m } \left( L_{r_m}^2  - \frac{256K^2\log r_m}{\sqrt m} - \frac{128^2K^2 \log^2 r_m}{m}  \right)  ,  
\end{align*} 
where the last step uses the condition in set $\mathcal A_2$. Note that $0 < \lambda_1 \leq \| \Sigma_{r_m} \|_F = O(r_m)$, since the kernels are bounded, and under $H_1$, $\lim_{m \rightarrow \infty} L_{r_m} > 0$ by assumption of Theorem \ref{thm:MahaAggrmntest}.
This shows, under $H_1$,    
\begin{align*} 
\P_{H_1}\left[ \{ |T_m^{MA}| < t_0 \} \cap \mathcal A_2 \cap \mathcal B \right] \rightarrow 0 , 
\end{align*}
since $t_0 = \frac{64K^2}{\sqrt m} = o(1)$. Since $\P[\mathcal A_2^c \cup \mathcal B^c] \rightarrow 0$, this implies, $\P_{H_1}\left[ |T_m^{MA}| < t_0 \right] \rightarrow 0$.

\subsubsection{Proof of Lemma \ref{lemma:concsigmahatsigma}}\label{sec:concSigmainv}

    Recalling the definition from \eqref{eq:Kxycentered} and doing an expansion we get,
    \begin{align}\label{eq:expKadotKbdot}
        \E\left[\sfK_{a}^{\circ}(X_{1},X_{2})\sfK_{b}^{\circ}(X_{1},X_{2})\right] 
        & = \E\left[\sfK_{a}(X_{1},X_{2})\sfK_{b}(X_{1},X_{2})\right] - 2\E\left[\sfK_{a}(X_{1},X_{2})\sfK_{b}(X_{1},X_{3})\right]\nonumber\\
        &\ + \E\left[\sfK_{a}(X_{1},X_{2})\right]\E\left[\sfK_{b}(X_{1},X_{2})\right]
    \end{align}
    Further by \eqref{eq:expandKhatdot} we have,
    \begin{align}\label{eq:sigmahatabexpand}
        \frac{1}{m^2}\sum_{1\leq i,j \leq m} \hat\sfK_{a}^{\circ}(X_{i},X_{j})\hat\sfK_{b}^{\circ}(X_{i},X_{j}) = T_{1} -2T_{2} + T_{3} , 
    \end{align}
    where $T_i = \eta_i/m^2$, for $1 \leq i \leq 3$, with $\eta_1, \eta_2, \eta_3$ as defined in \eqref{eq:expandKhatdot}.  Observe that,
    \begin{align*}
        \E T_{1} = \frac{m(m-1)}{m^2}\E\left[\sfK_{a}(X_{1},X_{2})\sfK_{b}(X_{1},X_{2})\right] + \frac{1}{m}\E\left[\sfK_{a}(X_{1},X_{1})^2\right] . 
    \end{align*}
    Hence,
    \begin{align}\label{eq:zeta1conc}
        \left|\E T_{1} - \E\left[\sfK_{a}(X_{1},X_{2})\sfK_{b}(X_{1},X_{2})\right]\right|\leq \frac{2K^2}{m} . 
    \end{align}
    By similar computations we have,
    \begin{align}\label{eq:zeta2conc}
        \left|\E T_{2} - \E\left[\sfK_{a}(X_{1},X_{2})\sfK_{b}(X_{1},X_{3})\right]\right|\leq \frac{8K^2}{m} 
    \end{align}
    and 
    \begin{align}\label{eq:zeta3conc}
        \left|\E T_{3} - \E\left[\sfK_{a}(X_{1},X_{2})\right]\E\left[\sfK_{b}(X_{1},X_{2})\right]\right|\leq \frac{24K^2}{m} . 
    \end{align}
Now recalling \eqref{eq:expandKhatdot}, \eqref{eq:sigmahatabexpand} along with \eqref{eq:zeta1conc}, \eqref{eq:zeta2conc} and \eqref{eq:zeta3conc} we get,
    \begin{align}\label{eq:expbddkakabdot}
        \left|\E\left[\frac{1}{m^2}\sum_{1 \leq i, j \leq m}\hat\sfK_{a}^{\circ}(X_{i},X_{j})\hat\sfK_{b}^{\circ}(X_{i},X_{j})\right] - \E\left[\sfK_{a}^{\circ}(X_{1},X_{2})\sfK_{b}^{\circ}(X_{1},X_{2})\right]\right|\leq \frac{50K^2}{m} . 
    \end{align}
    For $1 \leq i \leq 3$, denote by $T_i^{(a)}$ the analogue of $T_i$ when the variable $X_a$ is replaced by an independent copy $X_a'$, keeping the other variables $(X_b)_{b \ne a}$ fixed.   
    Then recalling the bounds on the kernels, it follows that 
    $$|T_1^{(a)}- T_1| \leq \frac{4K^2}{m},|T_2^{(a)}- T_2| \leq \frac{8K^2}{m}\text{ and }|T_3^{(a)}- T_3| \leq \frac{16K^2}{m}$$
    for all $1 \leq a \leq m$. Now applying the bounded difference inequality \citep{mcdiarmid1989method} and recalling \eqref{eq:sigmahatabexpand} shows for all $t\geq 0$,
    \begin{align}\label{eq:Kconcentration}
        & \P\left[\left|\frac{1}{m^2}\sum_{1 \leq i, j \leq m}\hat\sfK_{a}^{\circ}(X_{i},X_{j})\hat\sfK_{b}^{\circ}(X_{i},X_{j}) - \E\left[\frac{1}{m^2}\sum_{1 \leq i, j \leq m}\hat\sfK_{a}^{\circ}(X_{i},X_{j})\hat\sfK_{b}^{\circ}(X_{i},X_{j})\right]\right|>t\right] \nonumber \\ 
        & \leq 2\exp\left(-\frac{mt^2}{648K^4}\right) . 
    \end{align}
    Finally recalling \eqref{eq:sigmahatsigma} and \eqref{eq:expbddkakabdot}, for all $t\geq 0$ we get,
    \begin{align*}
    	& \P\left[\left|\hat\sigma_{ab} - \sigma_{ab}\right|> t + \frac{800K^2}{m} \right] \nonumber \\ 
	& = \P\left[\left|\frac{1}{m^2}\sum_{1 \leq i, j \leq m}\hat\sfK_{a}^{\circ}(X_{i},X_{j})\hat\sfK_{b}^{\circ}(X_{i},X_{j}) - \E\left[\sfK_{a}^{\circ}(X_{1},X_{2})\sfK_{b}^{\circ}(X_{1},X_{2})\right]\right|> \frac{t}{16} + \frac{50K^2}{m}\right] \nonumber \\ 
	& = \P\left[\left|\frac{1}{m^2}\sum_{1 \leq i, j \leq m}\hat\sfK_{a}^{\circ}(X_{i},X_{j})\hat\sfK_{b}^{\circ}(X_{i},X_{j}) - \E\left[\frac{1}{m^2}\sum_{1 \leq i, j \leq m}\hat\sfK_{a}^{\circ}(X_{i},X_{j})\hat\sfK_{b}^{\circ}(X_{i},X_{j})\right]\right|> \frac{t}{16} \right] \nonumber \\  
        & \leq 2\exp\left(-\frac{mt^2}{LK^4}\right) , 
    \end{align*} 
    where the last step uses \eqref{eq:Kconcentration}. Note that $t > 1600K^2/m$ implies, $t-800K^2/m \geq t/2$. Hence, for $t > 1600K^2/m$, 
    $$ \P\left[\left|\hat\sigma_{ab} - \sigma_{ab}\right|> t \right] \leq 2\exp\left(-\frac{m}{LK^4} \left( t-\frac{800}{m} \right)^2 \right) \leq 2\exp\left(-\frac{m t^2}{ 4 LK^4 } \right) , 
 $$
which completes the proof of \eqref{eq:hatsigmaexponential} in Lemma \ref{lemma:concsigmahatsigma}. 

Now, using \eqref{eq:hatsigmaexponential} and an union bound argument gives,
\begin{align}\label{eq:Sigmaconc}
    \P\left[\left\|\hat{\bm{\Sigma}}_{r_{m}} - \bm{\Sigma}_{r_{m}}\right\|_{F}>t\right] & \leq  \P\left[ \max_{1 \leq a, b \leq r_{m}} | \hat{\sigma}_{ab} - \sigma_{ab} | > \frac{t}{r_m} \right] \nonumber \\ 
    & \leq 2r_{m}^2\exp\left(-\frac{m t^2}{4 L K^4r_{m}^2} \right) , 
\end{align} 
whenever $t \geq 1600K^2 r_m/ m$. Recall that $\hat\lambda_{m}$ and $\lambda_{m}$ are the smallest eigenvalues of $\hat{\bm{\Sigma}}_{r_{m}}$ and $\bm\Sigma_{r_{m}}$ respectively. Then by Weyl's inequality,
\begin{align*}
    \left|\hat\lambda_{m} - \lambda_{m}\right| \leq \left\|\hat{\bm{\Sigma}}_{r_{m}} - \bm{\Sigma}_{r_{m}}\right\|_{2} \leq \left\|\hat{\bm{\Sigma}}_{r_{m}} - \bm{\Sigma}_{r_{m}}\right\|_{F} . 
    \end{align*}
Hence, for $t \geq 1600K^2 r_m/ m$, 
\begin{align}\label{eq:lambdam}
    \P\left[\left|\hat\lambda_{m} - \lambda_{m}\right|>t\right]\leq 2r_{m}^2\exp\left(-\frac{m t^2}{4 L K^4r_{m}^2 } \right) . 
\end{align}
Similarly, 
\begin{align*}
    \left|\hat\lambda_{1} - \lambda_{1}\right| \leq \left\|\hat{\bm{\Sigma}}_{r_{m}} - \bm{\Sigma}_{r_{m}}\right\|_{2} \leq \left\|\hat{\bm{\Sigma}}_{r_{m}} - \bm{\Sigma}_{r_{m}}\right\|_{F} . 
    \end{align*}
Hence, for $t \geq 1600 K^2r_m/ m$, 
\begin{align}\label{eq:lambda1}
    \P\left[\left|\hat\lambda_{1} - \lambda_{1}\right|>t\right]\leq 2r_{m}^2\exp\left(-\frac{m t^2}{4 LK^4 r_{m}^2 } \right) . 
\end{align}
The bounds in \eqref{eq:lambdam} and \eqref{eq:lambda1} together completes the proof of \eqref{eq:lambda1m} in Lemma \ref{lemma:concsigmahatsigma}.

\color{black}

\section{Comparison with MMDAgg Test} 
\label{sec:adpativeMMDexperiments}

In this section we compare the MMMD test with the MMDAgg test as implemented in \citet{schrab2021mmd}. 

\subsection{Perturbed $1$-dimensional Uniform Distribution} 
\label{sec:uniformR}
\textcolor{black}{ Here, $P$ is the uniform distribution on $[0, 1]$ and $Q$ is a perturbed version of uniform distribution on $[0, 1]$ as considered \citet{schrab2021mmd}. } 
The perturbed density at $x\in \mathbb{R}$ is given by: 
\begin{align}\label{eq:ftheta}
    f_{\bm{\theta}}(x) = \bm{1}\left\{x\in [0,1]\right\} + \frac{c_1}{R}\sum_{v\in\{1,2,\ldots,R\}}\theta_{v}G\left(Rx - v\right)
\end{align}
where $c_1=2.7$, $\bm{\theta} = (\theta_{1},\theta_{2},\ldots, \theta_{R})\in \{-1,1\}^{R}$, $R$ is the number of perturbations, and 
\begin{align*}
    G(t):= \exp\left(-\dfrac{1}{1-(4t + 3)^2}\right)\bm{1}\left\{t\in (-1, -\tfrac{1}{2})\right\} - \exp\left(-\dfrac{1}{1-(4t + 1)^2}\right)\bm{1}\left\{t\in (-\tfrac{1}{2}, 0)\right\} . 
\end{align*} 
It is known that for $R$ large enough the difference between the uniform density and the perturbed uniform density lies in the Sobolev ball \citep{kernelnonparametricsmoothalternatives}. Figure \ref{fig:MMDAggMixLocal}(a) shows the empirical powers of the \texttt{MMDAgg} test with Gaussian and Laplace kernels, with bandwidths chosen according to  the increasing weight strategy as in \citet{schrab2021mmd}; the empirical powers of the \texttt{Gauss MMMD}, \texttt{LAP MMMD}, and \texttt{Mixed MMMD} with bandwidth chosen as in  \eqref{eq:gaussmmmd}, \eqref{eq:laplacemmmd} and \eqref{eq:mixedmmmd},  respectively; and the empirical power of the FR test. The sample sizes are set to $m=n=500$, the perturbations range over $R=1,2,3,4,5,6$, and the power is computed over 500 repetitions, with a new value of $\bm{\theta}\in \{-1,1\}^R$ sampled uniformly in each iteration. The plot shows that the MMMD tests have better finite-sample power than the MMAgg tests, particularly for larger perturbations.

\begin{figure}[!ht]
    \centering 
    \begin{subfigure}[c]{0.32\textwidth}
    \includegraphics[scale = 0.28]{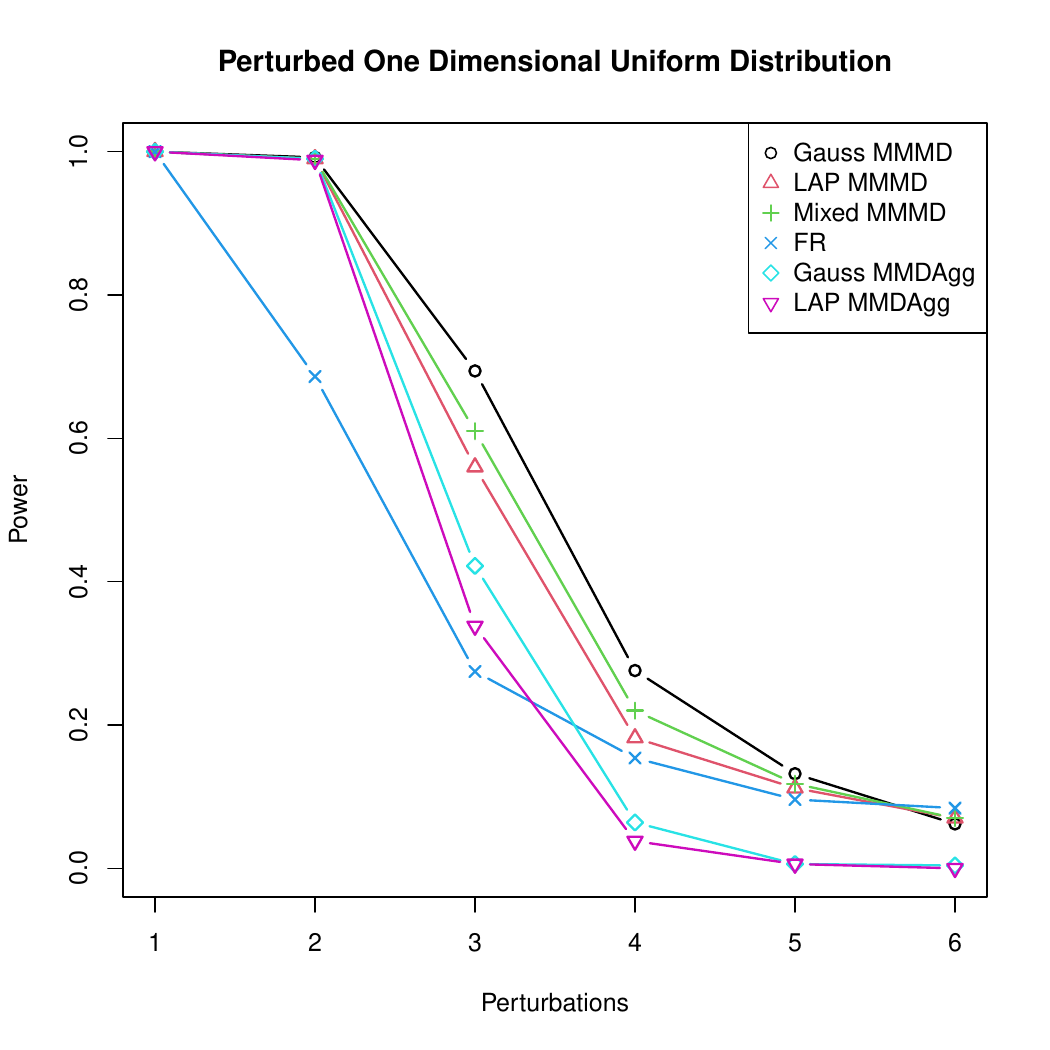} 
     \caption*{\small{(a)}} 
\end{subfigure}    
\begin{subfigure}[c]{0.32\textwidth}
    \includegraphics[scale = 0.28]{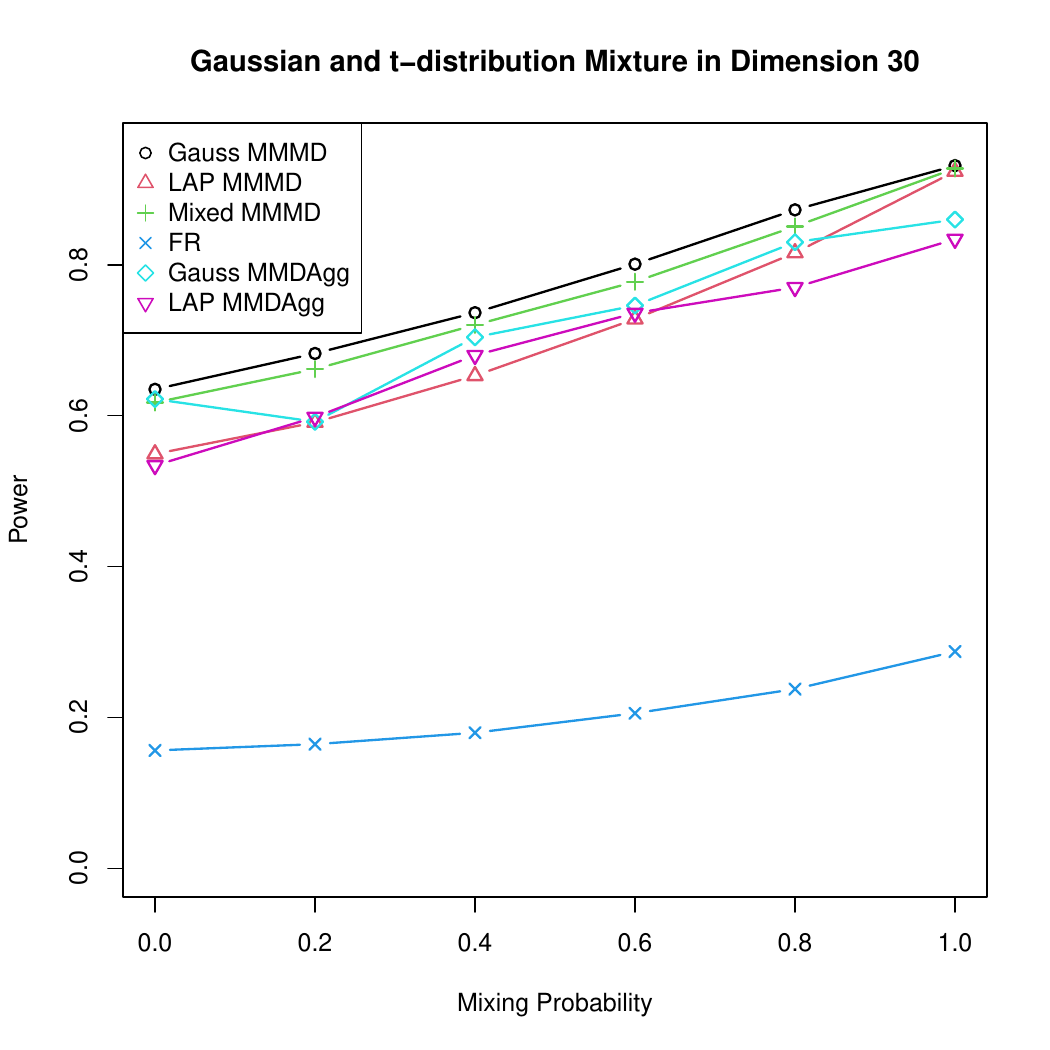} 
     \caption*{\small{(b)}} 
\end{subfigure}    
\begin{subfigure}[c]{0.32\textwidth}
   \includegraphics[scale = 0.28]{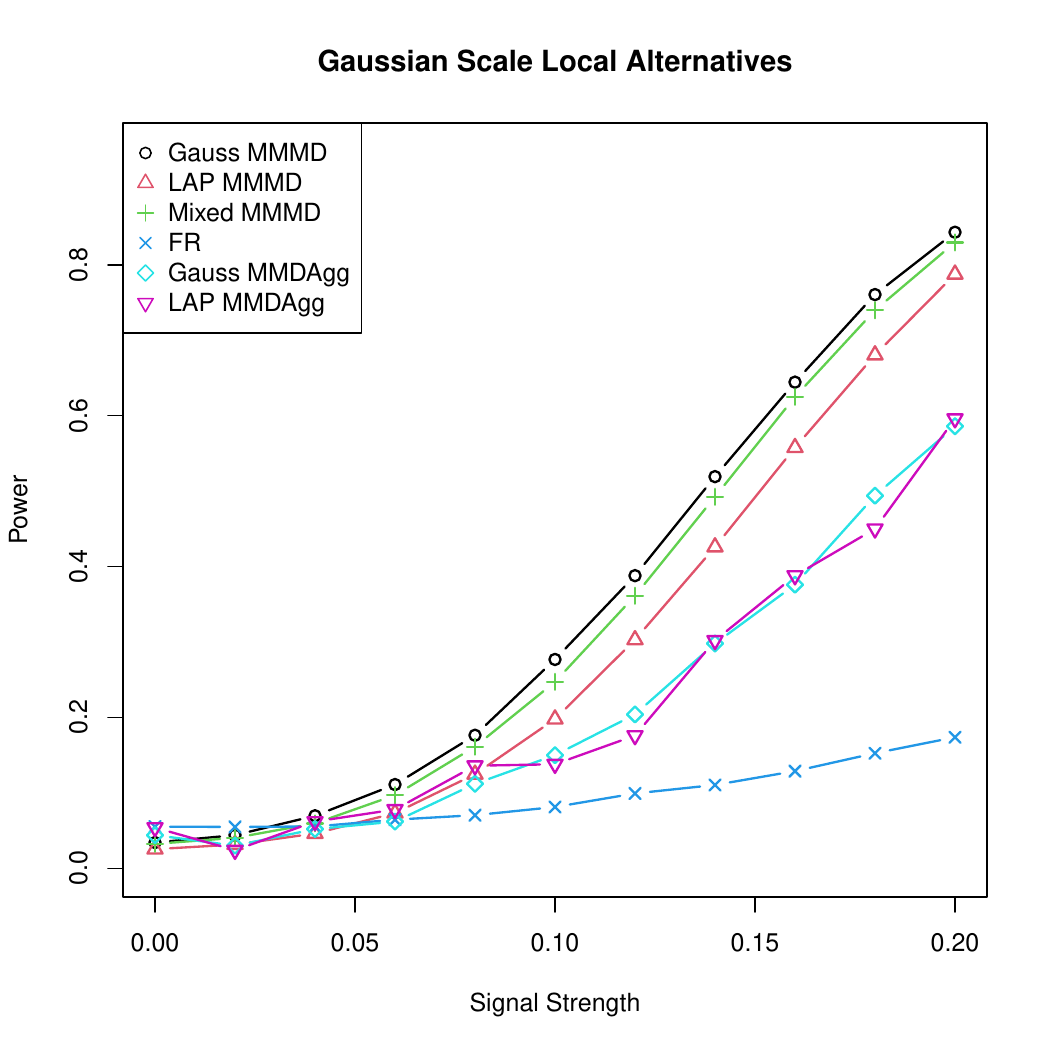} 
    \caption*{\small{(c)}} 
\end{subfigure} 
    \caption{Empirical powers of the MMMD and MMDAgg tests for (a) the perturbed uniform distribution, (b) mixture alternatives, and (c) local alternatives. }
    \label{fig:MMDAggMixLocal}
\end{figure}

\subsection{Mixture and Local Alternatives}

Next, we compare the empirical power (by repeating the experiment 500 times)  of the MMMD tests with the MMDAgg tests (based on Gaussian and Laplace kernels) for the mixtures alternative in $d=30$ as in Section \ref{sec:mixtureexperiments} and the local alternative in $d=20$ as in Section \ref{sec:localexperiments}. For the MMMD tests we use bandwidths as in \eqref{eq:gaussmmmd}, \eqref{eq:laplacemmmd}, and \eqref{eq:mixedmmmd}, while for the MMDAgg tests we consider the increasing weight strategy with collection of bandwidths $\Lambda(-2,2)$ as defined in \citet[Section 5.3]{schrab2021mmd}. 
For the mixture alternative the MMMD tests perform slightly better than the MMDAgg tests (see Figure \ref{fig:MMDAggMixLocal}(b)), while for the local alternative the MMMD tests show significant improvement over the MMDAgg tests (see Figure \ref{fig:MMDAggMixLocal}(c)).

\color{black}

\section{Proof of Theorem \ref{thm:asymkernelonN}} 
\label{sec:H0KNpf}

Given a collection of bandwidths $\bm \nu = (\nu_1, \nu_2, \ldots, \nu_r)$ recall that 
\begin{align}\label{eq:MMDvecKn}
    \emmd\left[\cK_{\bm \nu},\sX_{m},\sY_{n}\right] := \left(\emmd[\sfK_{\nu_{1}},\sX_{m},\sY_{n}],\cdots,\emmd[\sfK_{\nu_{r}},\sX_{m},\sY_{n}]\right)^{\top} . 
\end{align}
Then considering $\bm{\alpha}:=(\alpha_{1},\ldots,\alpha_{r})^{\top}\in\R^{r}$ note that, 
\begin{align}\label{eq:MMDalphaKH}
    \bm{\alpha}^{\top}\emmd\left[\cK_{\bm \nu},\sX_{m},\sY_{n}\right] & = \sum_{a=1}^{r}\alpha_{a}\emmd\left[\sfK_{\nu_{a}}, \sX_{m},\sY_{n}\right] \nonumber \\ 
  &  = \emmd\left[\sfH_{\bm \nu},\sX_{m},\sY_{n}\right] , 
\end{align}
where $\sfH_{\bm \nu}:=\sum_{a=1}^{r}\alpha_{a}\sfK_{\nu_{a}}$.
\begin{proposition}\label{prop:asympHNr}
Let $\sfH_{\bm \nu}$ be as defined above and suppose the assumptions of Theorem \ref{thm:asymkernelonN} hold. Then under $H_{0}$ in the asymptotic regime \eqref{eq:mn} the following hold: 
\begin{align*}
    \frac{mn}{\sqrt{2}(m+n)}\lambda_{N}^{d/4}\emmd\left[\sfH_{\bm \nu},\sX_{m},\sY_{n}\right]\dto\pi^{d/4}\|f_{P}\|_{2} \cdot Z , 
\end{align*}
where $Z\sim \mathcal{N} \left(0,\bm{\alpha}^{\top} \Gamma \bm{\alpha}\right)$, for $\Gamma$ as defined in Theorem \ref{thm:asymkernelonN}. 
\end{proposition} 

The proof of Proposition \ref{prop:asympHNr} is given in Appendix \ref{sec:proofofpropasumpHNr}. Since $\bm{\alpha}\in \R^{r}$ is chosen arbitrarily, the proof of Theorem \ref{thm:asymkernelonN} follows from  Proposition \ref{prop:asympHNr} and the Cram\'er-Wold device. The proof of Corollary \ref{corollary:KNH0} is given in Section \ref{sec:corKNpf}.

\subsection{Proof of Proposition \ref{prop:asympHNr}}\label{sec:proofofpropasumpHNr}

First, note that 
\begin{align}\label{eq:HHcircequiv}
    \emmd\left[\sfH_{\bm \nu},\cX_{m},\cY_{n}\right] = \emmd\left[\sfH_{\bm \nu}^{\circ},\cX_{m},\cY_{n}\right]
\end{align}
where $\sfH_{\bm \nu}^{\circ}$ is the centerd version of $\sfH_{\bm \nu}$ defined by,
\begin{align*}
    \sfH_{\bm \nu}^{\circ}(x,y) = \sfH_{\bm \nu}(x,y) - \E_{X\sim P}\sfH_{\bm \nu}(X,y) - \E_{X'\sim P}\sfH_{\bm \nu}(x,X') + \E_{X,X'\sim P}\sfH_{\bm \nu}(X,X').
\end{align*} Clearly, 
\begin{align}\label{eq:HcircalphaK}
    \sfH_{\bm \nu}^{\circ} = \sum_{a=1}^{r}\alpha_{a}\sfK_{\nu_{a}}^{\circ}, 
\end{align}
where $\sfK_{\nu_{a}}^{\circ}$ is defined in \eqref{eq:Kxycentered}. Now, recall that $N= m+n$ and $\frac{m}{m+n} \rightarrow \rho$. Define for any $\sfK: \cX \times \cX \rightarrow \mathbb R$, 
\begin{align}\label{eq:MMDXYN}
    \overline{\mmd}^2 \left[\sfK, \sX_m, \sY_n \right] = \hat{\mathcal W}_{\sX_m} + \hat{\mathcal W}_{\sY_n} - 2 \hat{\mathcal B}_{\sX_m, \sY_n} , 
\end{align}   
where 
\begin{align*}
 \overline{\mathcal W}_{\sX_m} :=   \frac{1}{N^2 \rho^2}\sum_{1 \leq i \ne j \leq m} \sfK \left(X_{i},X_{j}\right) \text{ and }  \overline{\mathcal W}_{\sY_n} := \frac{1}{N^2 (1-\rho)^2}\sum_{1 \leq i \ne j \leq n} \sfK \left(Y_{i},Y_{j}\right) , 
 \end{align*} 
and  
 \begin{align*}
 \overline{\mathcal B}_{\sX_m, \sY_n} :=  \frac{1}{N^2 \rho (1-\rho) }\sum_{i=1}^{m}\sum_{j=1}^{n} \sfK \left(X_{i},Y_{j}\right) . 
\end{align*}

By Theorem 1 and Lemma 4 of \cite{kernelnonparametricsmoothalternatives} we get, 
\begin{align*}
    \frac{mn}{\sqrt{2}(m+n)}\lambda_{N}^{d/4}\Bigg(\emmd\left[\sfK_{\nu_{a}}^{\circ},\sX_{m},\sY_{n}\right] - \overline{\mmd}^2\left[\sfK_{\nu_{a}}^{\circ},\cX_{m},\cY_{n}\right]\Bigg) = o_{P}(1).
\end{align*}
Hence, recalling \eqref{eq:HcircalphaK}, 
\begin{align}\label{eq:MMDequivhatMMD}
    \frac{mn}{\sqrt{2}(m+n)}\lambda_{N}^{d/4}\Bigg(\emmd\left[\sfH_{\bm \nu}^{\circ},\sX_{m},\sY_{n}\right] - \overline{\mmd}^2\left[\sfH_{\bm \nu}^{\circ},\cX_{m},\cY_{n}\right]\Bigg) = o_{P}(1).
\end{align}
Recalling \eqref{eq:HHcircequiv}, it suffices to derive the asymptotic distribution of $\overline{\mmd}^2\left[\sfH_{\bm \nu}^{\circ},\cX_{m},\cY_{n}\right]$ to complete the proof of Proposition \ref{prop:asympHNr}. For this, we invoke the framework of \cite{hall1984central}, which derives central limit theorems for degenerate $U$-statistics with kernels depending on the sample size. This requires the verification of a technical condition that arises from the martingale central limit theorem. We formalize this in the following lemma (see Appendix \ref{sec:proofoftouseHall} for the proof).

\begin{lemma}\label{lemma:touseHall}
Under the assumptions of Theorem \ref{thm:asymkernelonN}, in the asymptotic regime \eqref{eq:mn},
\begin{align*}
    \frac{\E\left[\sfH_{\bm \nu}^{\circ}(X_{1},X_{2})^4\right]}{N^2\E\left[\sfH_{\bm \nu}^{\circ}(X_{1},X_{2})^2\right]^2} + \frac{\E\left[\sfH_{\bm \nu}^{\circ}(X_{1},X_{2})^2\sfH_{\bm \nu}^{\circ}(X_{1},X_{3})^2\right]}{N\E\left[\sfH_{\bm \nu}^{\circ}(X_{1},X_{2})^2\right]^2} + \frac{\E\left[\sfG_{\bm \nu}^2(X_{1},X_{2})\right]}{\E\left[\sfH_{\bm \nu}^{\circ}(X_{1},X_{2})^2\right]^2}\ra 0 , 
\end{align*}
where 
\begin{align}\label{eq:Gkernel}
    \sfG_{\bm \nu}(x,y):=\E\left[\sfH_{\bm \nu}^{\circ}(x,X_{3})\sfH_{\bm \nu}^{\circ}(y,X_{3})\right] , 
\end{align}
for all $x,y\in \R^d$ and $X_{1},X_{2},X_{3}$ i.i.d. from $P$.
\end{lemma}

Adapting the arguments from \cite{hall1984central} and using Lemma \ref{lemma:touseHall} then gives, 
\begin{align}\label{eq:disthatMMD}
    \frac{mn}{\sqrt{2}(m+n)}\left(\E\left[\sfH_{\bm \nu}^{\circ}(X_{1},X_{2})^2\right]\right)^{-1/2}\overline{\mmd}^2\left[\sfH_{\bm \nu}^{\circ}, \cX_{m},\cY_{n}\right]\dto \mathcal N (0,1) . 
\end{align}

Next, we show that $\E\left[\sfH_{\bm \nu}^{\circ}(X_{1},X_{2})^2\right]$ scales as $\lambda_{N}^{-d/2}$ and compute its limit. The proof is given in Appendix \ref{sec:proofofEHcircL2lim}.

\begin{lemma}\label{lemma:asympEHNcirc} 
Under the assumptions of Theorem \ref{thm:asymkernelonN}, in the asymptotic regime \eqref{eq:mn},
    \begin{align*}
        \lambda_{N}^{d/2}\E\left[\sfH_{\bm \nu}^{\circ}(X_{1},X_{2})^2\right]\ra \pi^{d/2} \|f_{P}\|_{2}^2\sum_{u,v=1}^{r}\frac{\alpha_{u}\alpha_{v}}{(\eta_{u} + \eta_{v})^{d/2}}
    \end{align*}
    where $X_{1},X_{2}$ are i.i.d. $P$.
\end{lemma}

The proof of Proposition \ref{prop:asympHNr} is now complete by collecting \eqref{eq:HHcircequiv}, \eqref{eq:MMDequivhatMMD}, \eqref{eq:disthatMMD} and Lemma \ref{lemma:asympEHNcirc}.

\subsubsection{Proof of Lemma \ref{lemma:asympEHNcirc}}\label{sec:proofofEHcircL2lim}
By \eqref{eq:HcircalphaK},
\begin{align}\label{eq:HNcirc12Kcirc12}
    \E\left[\sfH_{\bm \nu}^{\circ}(X_{1},X_{2})^2\right] = \sum_{a,b=1}^{r}\alpha_{a}\alpha_{b}\E\left[\sfK_{\nu_{a}}^{\circ}(X_{1},X_{2})\sfK_{\nu_{b}}^{\circ}(X_{1},X_{2})\right] . 
\end{align}
In the following lemma we compute the limit of each term inside the above summation. Using this lemma together with \eqref{eq:HNcirc12Kcirc12} completes the proof of Lemma \ref{lemma:asympEHNcirc}. 

\begin{lemma}\label{lemma:limofKN1circKN2circ} 
Fix $1 \leq a,b \leq r$. Suppose $\nu_{a} = \eta_{a}\lambda_{N}$ and $\nu_{b} = \eta_{b}\lambda_{N}$, where $\lambda_{N} = o\left(N^{4/d}\right)$ such that $\lambda_{N}\ra\infty$, in the asymptotic regime \eqref{eq:mn}, as in Assumption \ref{assmp:relate}. Then 
 \begin{align*}
   \lambda_{N}^{d/2}\left[\sfK_{\nu_a}^{\circ}(X_{1},X_{2})\sfK_{\nu_b}^{\circ}(X_{1},X_{2})\right]\ra \pi^{d/2}\|f_{P}\|_{2}^2 \frac{1}{(\eta_a + \eta_b)^{d/2}} . 
\end{align*} 

\end{lemma}
\begin{proof}
    By \eqref{eq:Kxycentered},
    \begin{align*}
        \E\left[\sfK_{\nu_a}^{\circ}(X_{1},X_{2})\sfK_{\nu_b}^{\circ}(X_{1},X_{2})\right] = T_{1} + T_{2} - 2T_{3} , 
    \end{align*}
where 
\begin{align*}
    T_{1} & = \E\left[\sfK_{\nu_a}(X_{1},X_{2})\sfK_{\nu_b}(X_{1},X_{2})\right] , \\  T_{2} & = \E\left[\sfK_{\nu_a}(X_{1},X_{2})\right]\E\left[\sfK_{\nu_b}(X_{1},X_{2})\right] , \\
    T_{3} & = \E\left[\E\left[\sfK_{\nu_a}(X_{1},X_{2})|X_{1}\right]\E\left[\sfK_{\nu_b}(X_{1},X_{3})|X_{1}\right]\right] . 
\end{align*}
Recalling the definition of the Gaussian kernel we get, 
\begin{align*}
    \E\left[\sfK_{\nu_a}(X_{1},X_{2})\sfK_{\nu_b}(X_{1},X_{2})\right] = \E\left[\sfK_{(\nu_a + \nu_b)}(X_{1},X_{2})\right] . 
\end{align*} 
Invoking Lemma 4 from \cite{kernelnonparametricsmoothalternatives} now gives, 
\begin{align*}
    (\nu_a + \nu_b)^{d/2}\E\left[\sfK_{(\nu_a + \nu_b)}(X_{1},X_{2})\right]\ra {\pi}^{d/2} \|f_{P}\|_{2}^2.
\end{align*}
This implies, from the definition of $\nu_a$ and $\nu_b$, that 
\begin{align}\label{eq:firsttermconvg}
   \lambda_{N}^{d/2}\E\left[\sfK_{\nu_a}(X_{1},X_{2})\sfK_{\nu_b}(X_{1},X_{2})\right]\ra \pi^{d/2} \|f_{P}\|_{2}^2 \frac{ 1 }{(\eta_a + \eta_b)^{d/2}} . 
\end{align}
Similarly we get,
\begin{align}\label{eq:secondtermconvg}
    \lambda_{N}^{d}\E\left[\sfK_{\nu_a}(X_{1},X_{2})\right]\E\left[\sfK_{\nu_b}(X_{1},X_{2})\right]\ra \pi^d \|f_{P}\|_{2}^4 \frac{ 1 }{(\eta_a \eta_b)^{d/2}}. 
\end{align}
Also, from the proof of Theorem 1 from \cite{kernelnonparametricsmoothalternatives} we get,
    \begin{align}\label{eq:Knu}
        \E\left[\E\left[\sfK_{\nu_a}(X_{1},X_{2})|X_{2}\right]^2\right]\lesssim_{d,P} \lambda_{N}^{-3d/4} \text{ and } 
    \E\left[\E\left[\sfK_{\nu_b}(X_{1},X_{2})|X_{2}\right]^2\right]\lesssim_{d,P}\lambda_{N}^{-3d/4}.
\end{align}
where $A\lesssim_{\theta}B$ implies there exists a constant $C_{\theta}>0$ such that $A\leq C_\theta B$. Then using Cauchy-Schwartz inequality along with \eqref{eq:Knu} gives, 
\begin{align}\label{eq:thirdtermconvg}
    |T_{3}| &\leq \left(\E\left[\E\left[\sfK_{\nu_a}(X_{1},X_{2})|X_{1}\right]^2\right]\right)^{1/2}\left(\E\left[\E\left[\sfK_{\nu_b}(X_{1},X_{3})|X_{1}\right]^2\right]\right)^{1/2}
    \lesssim_{d,P}\lambda_{N}^{-3d/4} . 
\end{align}
Combining \eqref{eq:firsttermconvg}, \eqref{eq:secondtermconvg} and \eqref{eq:thirdtermconvg} completes the proof of Lemma \ref{lemma:limofKN1circKN2circ}. 
\end{proof}

\subsubsection{Proof of Lemma \ref{lemma:touseHall}}\label{sec:proofoftouseHall}
By Lemma \ref{lemma:asympEHNcirc} it is enough to show,
\begin{align*}
    \lambda_{N}^{d}\left(\underbrace{\frac{\E\left[\sfH_{\bm \nu}^{\circ}(X_{1},X_{2})^4\right]}{N^2}}_{L_{1}} + \underbrace{\frac{\E\left[\sfH_{\bm \nu}^{\circ}(X_{1},X_{2})^2\sfH_{\bm \nu}^{\circ}(X_{1},X_{3})^2\right]}{N}}_{L_{2}} + \underbrace{\E\left[\sfG_{\bm \nu}^2(X_{1},X_{2})\right]}_{L_{3}}\right)\ra 0 . 
\end{align*}
We show convergence of terms $L_{1},L_{2}$ and $L_{3}$ separately. Let us start with $L_{1}$. Observe that 
\begin{align*}
    \E\left[\sfH_{\bm \nu}^{\circ}(X_{1},X_{2})^4\right]\lesssim_{r}\sum_{a=1}^{r}\alpha_{a}^{4}\E\left[\sfK_{\nu_{a}}^{\circ}(X_{1},X_{2})^4\right]. 
\end{align*}
From the proof of Theorem 1 in \cite{kernelnonparametricsmoothalternatives} we get $\E\left[\sfK_{\nu_{a}}^{\circ}(X_{1},X_{2})^4\right]\lesssim_{d}\lambda_{N}^{-d/2}$, for all $1\leq a\leq r$. 
Hence, 

\begin{align}\label{eq:L1convg}
\lambda_{N}^{d} L_1 =   \frac{\lambda_{N}^{d}}{N^{2}}\E\left[\sfH_{\bm \nu}^{\circ}(X_{1},X_{2})^4\right]\lesssim_{r,d,\bm{\alpha}}\frac{\lambda_{N}^{d/2}}{N^2} \ra 0 . 
\end{align} 

Next, we consider $L_{2}$. From \eqref{eq:HcircalphaK} it follows that 
\begin{align*}
    \sfH_{\bm \nu}^{\circ}(x,y)^2\lesssim_{r}\sum_{a=1}^{r}\alpha_{a}^2\sfK_{\nu_{a}}^{\circ}(x,y)^2. 
\end{align*}
Hence, 
\begin{align}\label{eq:bbdL2Kcirc}
    \E\left[\sfH_{\bm \nu}^{\circ}(X_{1},X_{2})^2\sfH_{\bm \nu}^{\circ}(X_{1},X_{3})^2\right] \lesssim_{r}\sum_{a=1}^{r}\sum_{b=1}^{r}\alpha_{a}^2\alpha_{b}^2\E\left[\sfK_{\nu_{a}}^{\circ}(X_{1},X_{2})^2\sfK_{\nu_{b}}^{\circ}(X_{1},X_{3})^2\right]. 
\end{align}
In the following lemma we bound the summands in \eqref{eq:bbdL2Kcirc}. This lemma together with \eqref{eq:bbdL2Kcirc} shows,
\begin{align}\label{eq:L2convg}
    \lambda_{N}^{d}L_{2}\lesssim_{r,d,P,\bm{\alpha}}\frac{\lambda_{N}^{d/4}}{N}\ra 0.
\end{align}

\begin{lemma}\label{lemma:asympbddKcirc2prod}
Fix $1 \leq a, b \leq r$. Then under the assumptions of Theorem \ref{thm:asymkernelonN} we have, 
\begin{align*}
    \E\left[\sfK_{\nu_a}^{\circ}(X_{1},X_{2})^2\sfK_{\nu_b}^{\circ}(X_{1},X_{3})^2\right]\lesssim_{d,P} \lambda_{N}^{-3d/4} . 
\end{align*}
\end{lemma} 

\begin{proof}
    Observe that,
    \begin{align}\label{eq:KzKpsicircsqprod}
        \E\left[\sfK_{\nu_a}^{\circ}(X_{1},X_{2})^2\sfK_{\nu_b}^{\circ}(X_{1},X_{3})^2\right] = \E\left[\E\left[\sfK_{\nu_a}^{\circ}(X_{1},X_{2})^2|X_{1}\right]\E\left[\sfK_{\nu_b}^{\circ}(X_{1},X_{3})^2|X_{1}\right]\right] . 
    \end{align}
    By \eqref{eq:Kxycentered} note that $\E\left[\sfK_{\nu_a}^{\circ}(X_{1},X_{2})|X_{1}\right] = 0$. Hence, 
    \begin{align}
        \E\left[\sfK_{\nu_a}^{\circ}(X_{1},X_{2})^2|X_{1}\right]
        & = \Var\left[\sfK_{\nu_a}^{\circ}(X_{1},X_{2})\middle|X_{1}\right] \nonumber\\
        & = \Var\left[\sfK_{\nu_a}(X_{1},X_{2}) - \E\left[\sfK_{\nu_a}(Z,X_{2})|X_{2}\right]|X_{1}\right] \nonumber\\
        &\leq \E\left[\left(\sfK_{\nu_a}(X_{1},X_{2}) - \E\left[\sfK_{\nu_a}(Z,X_{2})|X_{2}\right]\right)^2|X_{1}\right]\nonumber\\
        &\lesssim \E\left[\sfK_{\nu_a}(X_{1},X_{2})^2|X_{1}\right] + \E\left[\E\left[\sfK_{\nu_a}(Z,X_{2})|X_{2}\right]^2\right] , \label{eq:bdonEK02givenX1}
    \end{align}
    where $Z\sim P$ is independent of $X_{1},X_{2}$. From \eqref{eq:Knu} we have, 
    \begin{align*}
        \E\left[\E\left[\sfK_{\nu_a}(Z,X_{2})|X_{2}\right]^2\right]\lesssim_{d,P}\lambda_{N}^{-3d/4}.
    \end{align*}
    Then \eqref{eq:bdonEK02givenX1} implies,
    \begin{align*}
        \E\left[\sfK_{\nu_a}^{\circ}(X_{1},X_{2})^2|X_{1}\right]\lesssim_{d,P} \E\left[\sfK_{\nu_a}(X_{1},X_{2})^2|X_{1}\right] + \lambda_{N}^{-3d/4} , 
    \end{align*}
    and similarly,
    \begin{align*}
        \E\left[\sfK_{\nu_b}^{\circ}(X_{1},X_{2})^2|X_{1}\right]\lesssim_{d,P} \E\left[\sfK_{\nu_b}(X_{1},X_{2})^2|X_{1}\right] + \lambda_{N}^{-3d/4} . 
    \end{align*}
    Since $|\sfK_{\nu_a}|,|\sfK_{\nu_b}|\leq 1$ for the Gaussian kernel, by \eqref{eq:KzKpsicircsqprod} we have,
    \begin{align}
        \E\left[\sfK_{\nu_a}^{\circ}(X_{1},X_{2})^2\sfK_{\nu_b}^{\circ}(X_{1},X_{3})^2\right]
        & = \E\left[\sfK_{\nu_a}^2(X_{1},X_{2})\sfK_{\nu_b}^2(X_{1},X_{3})\right] + \lambda_{N}^{-3d/4} . 
        \label{eq:KzKpsibddKKrate}
    \end{align}
    Now, assume without loss of generality that $\nu_a \leq \nu_b$. This implies,             $\sfK_{\nu_b}(x,y)\leq \sfK_{\nu_a}(x,y)$  for all $x,y\in\R^d$, by the monotonicity of the exponential function. Hence, 
        \begin{align*}
        \E\left[\sfK_{\nu_a}^{\circ}(X_{1},X_{2})^2\sfK_{\nu_b}^{\circ}(X_{1},X_{3})^2\right]
        &\lesssim_{d,P}\E\left[\sfK_{\nu_a}^2(X_{1},X_{2})\sfK_{\nu_a}^2(X_{1},X_{3})\right] + \lambda_{N}^{-3d/4}\\
        &\lesssim_{d,P}\E\left[\E\left[\sfK_{\nu_a}(X_1,X_2)\middle|X_1\right]^2\right] + \lambda_{N}^{-3d/4}\\
        &\lesssim_{d,P}\lambda_{N}^{-3d/4} , 
    \end{align*} 
    where the last step uses \eqref{eq:Knu}. This completes the proof of Lemma \ref{lemma:asympbddKcirc2prod}.
\end{proof}

Finally, we turn our attention to $L_{3}$. Recalling \eqref{eq:Gkernel}, note that 
\begin{align*}
    L_{3} = \E\left[\E\left[\sfH_{\bm \nu}^{\circ}(X_{1},X_{3})\sfH_{\bm \nu}^{\circ}(X_{2},X_{3})\middle|X_{1},X_{2}\right]^2\right] . 
\end{align*}
Then from \eqref{eq:HcircalphaK} we have,
\begin{align}
    L_{3} 
    & = \E\left[\left[\sum_{a=1}^{r}\sum_{b=1}^{r}\alpha_{a}\alpha_{b}\E\left[\sfK_{\nu_{a}}^{\circ}(X_{1},X_{3})\sfK_{\nu_{N,b}}^{\circ}(X_{2},X_{3})|X_{1},X_{2}\right]\right]^2\right]\nonumber\\
    &\lesssim_{r}\sum_{a=1}^{r}\sum_{b=1}^{r}\alpha_{a}^2\alpha_{b}^2\E\left[\E\left[\sfK_{\nu_{a}}^{\circ}(X_{1},X_{3})\sfK_{\nu_{N,b}}^{\circ}(X_{2},X_{3})|X_{1},X_{2}\right]^2\right]\label{eq:L3bdd}
\end{align}
As in the previous cases, in the following lemma we compute the limit of each individual term in the above expansion. By Lemma \ref{lemma:L3termconvg} and \eqref{eq:L3bdd} it is follows that 
\begin{align}\label{eq:L3convg}
    \lambda_{N}^{d}L_{3}\ra 0 . 
\end{align}

\begin{lemma}\label{lemma:L3termconvg}
Fix $1 \leq a,b \leq r$.  Then under the assumptions of Theorem \ref{thm:asymkernelonN}, 
\begin{align*}
    \lambda_{N}^{d}\E\left[\E\left[\sfK_{\nu_a}^{\circ}(X_{1},X_{3})\sfK_{\nu_b}^{\circ}(X_{2},X_{3})|X_{1},X_{2}\right]^2\right]\ra 0 . 
\end{align*}
\end{lemma} 

\begin{proof}
By \eqref{eq:Kxycentered},
\begin{align}
    \E\Bigg[\E
    \Bigg[\sfK_{\nu_a}^{\circ}(X_{1},X_{3})\sfK_{\nu_b}^{\circ}(X_{2},X_{3})|X_{1},X_{2}\Bigg]^2\Bigg] 
    &\lesssim S_{1} + S_{2} + S_{3} + S_{4} , \label{eq:bddS1to4}
\end{align}
where $Z,Z_{1},Z_{2}$ are i.i.d. $P$ independently of $X_{1},X_{2},X_{3}$ and 
\begin{align*}
    S_{1} 
    & = \E\left[\E\left[\sfK_{\nu_a}^{\circ}(X_{1},X_{3})\sfK_{\nu_b}(X_{2},X_{3})|X_{1},X_{2}\right]^2\right],\\
    S_{2}
    & = \E\left[\E\left[\sfK_{\nu_a}^{\circ}(X_{1},X_{3})\E\left[\sfK_{\nu_b}(X_{2},Z)|X_{2}\right]|X_{1},X_{2}\right]^2\right],\\
    S_{3}
    & = \E\left[\E\left[\sfK_{\nu_a}^{\circ}(X_{1},X_{3})\E\left[\sfK_{\nu_b}(Z,X_{3})|X_{3}\right]|X_{1},X_{2}\right]^2\right] , \\
   S_{4}
    & = \E\left[\sfK_{\nu_b}(Z_{1},Z_{2})\right]^2\E\left[\E\left[\sfK_{\nu_a}^{\circ}(X_{1},X_{3})|X_{1},X_{2}\right]^2\right] . 
\end{align*}
Recalling that $\E\left[\sfK_{\nu_a}^{\circ}(X_{1},X_{3})|X_{1}\right] = 0$ we get $S_{4} = 0$. Therefore, it suffices to bound $S_1, S_2, S_3$. 

We begin with $S_1$. Observe that, 
\begin{align}
    S_{2} 
    & \leq \E\left[\E\left[\sfK_{\nu_a}^{\circ}(X_{1},X_{3})^2\E\left[\sfK_{\nu_b}(X_{2},Z)|X_{2}\right]^2|X_{1},X_{2}\right]\right]\nonumber\\
    & = \E\left[\sfK_{\nu_a}^{\circ}(X_{1},X_{3})^2\E\left[\sfK_{\nu_b}(X_{2},Z)|X_{2}\right]^2\right]\nonumber\\
    & = \E\left[\sfK_{\nu_a}^{\circ}(X_{1},X_{3})^2\right]\E\left[\E\left[\sfK_{\nu_b}(X_{2},Z)|X_{2}\right]^2\right]\lesssim_{d,P}\lambda_{n}^{-5d/4} , \label{eq:bddonS2}
\end{align}
where the last inequality follows from \eqref{eq:Knu} and Lemma \ref{lemma:limofKN1circKN2circ}. 

Now, we consider $S_3$. To begin with using the Cauchy-Schwarz inequality note that, 
\begin{align*}
    \bigg|\E\bigg[\sfK_{\nu_a}^{\circ}(X_{1},X_{3})
    &\E\left[\sfK_{\nu_b}(Z,X_{3})|X_{3}\right]|X_{1},X_{2}\bigg]\bigg|\\
    &\leq\E\left[\left|\sfK_{\nu_a}^{\circ}(X_{1},X_{3})\E\left[\sfK_{\nu_b}(Z,X_{3})|X_{3}\right]\right||X_{1},X_{2}\right]\\
    &\leq \left(\E\left[\sfK_{\nu_a}^{\circ}(X_{1},X_{3})^2|X_{1},X_{2}\right]\right)^{1/2}\left(\E\left[\E\left[\sfK_{\nu_b}(Z,X_{3})|X_{3}\right]^2|X_{1},X_{2}\right]\right)^{1/2}\\
    & = \left(\E\left[\sfK_{\nu_a}^{\circ}(X_{1},X_{3})^2|X_{1}\right]\right)^{1/2}\left(\E\left[\E\left[\sfK_{\nu_b}(Z,X_{3})|X_{3}\right]^2\right]\right)^{1/2}. 
\end{align*}
Hence, recalling the definition of $S_{3}$ we have, 
\begin{align}
    S_{3} 
    &=  \E\left[\E\left[\sfK_{\nu_a}^{\circ}(X_{1},X_{3})\E\left[\sfK_{\nu_b}(Z,X_{3})|X_{3}\right]|X_{1},X_{2}\right]^2\right]  \nonumber \\ 
    & \leq \E\left[\E\left[\sfK_{\nu_a}^{\circ}(X_{1},X_{3})^2|X_{1}\right]\E\left[\E\left[\sfK_{\nu_b}(Z,X_{3})|X_{3}\right]^2\right]\right]\nonumber\\
    & = \E\left[\sfK_{\nu_a}^{\circ}(X_{1},X_{3})^2\right]\E\left[\E\left[\sfK_{\nu_b}(Z,X_{3})|X_{3}\right]^2\right]\nonumber\\
    &\lesssim_{d,P}\lambda_{N}^{-5d/4} , \label{eq:bddonS3}
\end{align}
where once again the last inequality follows from \eqref{eq:Knu} and Lemma \ref{lemma:limofKN1circKN2circ}. 

Next, we consider $S_1$. Recalling \eqref{eq:Kxycentered} shows that,
\begin{align*}
    S_{1} & \lesssim S_{11} + S_{12} + S_{13} + S_{14} , 
\end{align*}
    where 
    \begin{align}
    S_{11} &:= \E\left[\sfK_{\nu_a}(X_{1},X_{3})\sfK_{\nu_b}(X_{2},X_{3})|X_{1},X_{2}\right]^2 , \nonumber \\ 
    S_{12} &:= \E\left[\E\left[\sfK_{\nu_a}(X_{1},Z)|X_{1}\right]\sfK_{\nu_b}(X_{2},X_{3})|X_{1},X_{2}\right]^2 , \nonumber\\
    S_{13} &:= \E\left[\E\left[\sfK_{\nu_a}(Z,X_{3})|X_{3}\right]\sfK_{\nu_b}(X_{2},X_{3})|X_{1},X_{2}\right]^2 , \nonumber \\ 
  S_{14} & := \E\left[\E\left[\sfK_{\nu_a}(Z_{1},Z_{2})\right]\sfK_{\nu_b}(X_{2},X_{3})|X_{1},X_{2}\right]^2 . \label{eq:structureS1}
\end{align}
Recalling \eqref{eq:Knu}, \eqref{eq:firsttermconvg} and similar to the proofs of 
\eqref{eq:bddonS2} and \eqref{eq:bddonS3} it can be shown that,
\begin{align}\label{eq:bddS12S13}
    \E S_{12}\lesssim_{d,P}\lambda_{N}^{-5d/4} \text{ and } \E S_{13}\lesssim_{d,P}\lambda_{N}^{-5d/4} , 
\end{align}
respectively. Also, observe that, 
 \begin{align}\label{eq:bddS14}
    \E S_{14} = \E\left[\sfK_{\nu_a}(Z_{1},Z_{2})\right]^2\E\left[\E\left[\sfK_{\nu_b}(X_{2},X_{3})|X_{2}\right]^2\right]\lesssim_{d,P}\lambda_{N}^{-5d/4} , 
\end{align}
where the last inequality follows from Lemma 4 of \cite{kernelnonparametricsmoothalternatives} and \eqref{eq:Knu}. Now, without loss of generality assume that $\nu_a \leq \nu_b$. Then following the proof of Theorem 1 from \cite{kernelnonparametricsmoothalternatives} we conclude that,
\begin{align}\label{eq:S11convg}
    \lambda_{N}^{d}\E S_{11}\leq \lambda_{N}^{d}\E\left[\E\left[\sfK_{\nu_a}(X_{1},X_{3})\sfK_{\nu_a}(X_{2},X_{3})|X_{1},X_{2}\right]^2\right]\ra 0 . 
\end{align}
The proof is now completed by recalling \eqref{eq:bddS1to4}, \eqref{eq:structureS1}, collecting the bounds from \eqref{eq:bddonS2}, \eqref{eq:bddonS3}, \eqref{eq:bddS12S13}, \eqref{eq:bddS14} and the convergence from \eqref{eq:S11convg}.
\end{proof}

The proof of Lemma \ref{lemma:touseHall} is now complete by collecting the convergences from \eqref{eq:L1convg}, \eqref{eq:L2convg} and \eqref{eq:L3convg}.

\subsection{Proof of Corollary \ref{corollary:KNH0} } 
\label{sec:corKNpf}

The following lemma gives a consistent estimate of $\|f_P\|_2^2$. The proof  follows directly from the proof of Theorem 4 in \cite{kernelnonparametricsmoothalternatives} and hence
is omitted.

\begin{lemma}\label{lemma:estnormfP}
Let $\sfK_{\lambda_{N}}$ and $\sfK_{2\lambda_{N}}$ be Gaussian kernels with scaling parameters $\lambda_{N}$ and $2\lambda_{N}$ respectively, satisfying Assumption \ref{assmp:relate}. Consider  
\begin{align*}
    \widetilde{s}_{\lambda_{N}}^2 :=
    & \frac{1}{N(N-1)}\sum_{1 \leq i\neq j \leq N} \sfK_{2\lambda_{N}}(Z_{i},Z_{j}) - \frac{2(N-3)!}{N!}\sum_{1\leq i \ne j_{1} \ne j_{2} \leq N }\sfK_{\lambda_{N}}(Z_{i}, Z_{j_{1}})\sfK_{\lambda_{N}}(Z_{i},Z_{j_{2}})\\
    & \quad \quad + \frac{(N-4)!}{N!}\sum_{1\leq i_{1} \ne i_{2} \ne j_{1} \ne j_{2}\leq N }\sfK_{\lambda_{N}}(Z_{i_{1}},Z_{j_{1}})\sfK_{\lambda_{N}}(Z_{i_{2}},Z_{j_{2}}) , 
\end{align*}
where $Z_{i} = X_{i}$, if $i\leq m$, and $Z_{i} = Y_{i-m}$, if $i>m$. Define $\widehat{s}_{\lambda_{N}}^2 = \max\{\widetilde{s}_{\lambda_{N}}^2, N^{-2}\}$. Then, under $H_{0}$ and Assumption \ref{assmp:PQdensitypq} we have,
\begin{align*}
   \|\hat f_{P}\|_{2}^{2} = \left(\frac{2\lambda_{N}}{\pi}\right)^{d/2}\widehat{s}_{\lambda_{N}}^2 \stackrel{P} \rightarrow \|f_{P}\|_{2}^{2}.
\end{align*}
\end{lemma}

Applying Theorem \ref{thm:asymkernelonN} and Lemma \ref{lemma:estnormfP}, the result in Corollary \ref{corollary:KNH0} follows. 

\color{black} 

\section{Invertibility of Kernel Matrices}
\label{sec:Sigma}

In this section we discuss the invertibility of matrix $\bm{\Sigma}_{H_0}$ (recall the definition from \eqref{eq:H0sigma}). Throughout this section we will assume that the underlying space  
$\mathcal{X} = \mathbb{R}^{d}$ and the distribution $P$ satisfy the following: 

\begin{assumption}\label{assumption:PQdensity}
Suppose $\mathcal{X} = \mathbb{R}^{d}$ and the distribution $P$ has a density with respect to the Lebesgue measure on $\mathbb{R}^{d}$ with full support. 
\end{assumption}

The following proposition gives a set of general conditions under which $\bm{\Sigma}_{H_0}$ is non-singular. In Corollary \ref{cor:Kvariance} we will show that these conditions are satisfied by the commonly used kernels, such as the Gaussian and Laplace kernels.

\begin{proposition}\label{ppn:Kvariance} Suppose Assumption \ref{assumption:PQdensity} holds and  $\cK=\{\sfK_1, \sfK_2, \ldots, \sfK_r\}$ be a collection of $r$ distinct characteristic kernels such that: 
\begin{itemize}
     \item
     For every $x, y \in \mathbb{R}^{d}$ and $1\leq a \leq r$, 
     \begin{align}\label{eq:limitKxz}
     \lim_{\|z\|\rightarrow\infty} \sfK_{a}(x, z) = 0 \text{ and } \lim_{\|z\|\rightarrow\infty} \sfK_{a}(z, y) = 0. 
     \end{align} 
         
    \item
    For every collection $\{\alpha_{a}: 1\leq a \leq r\}$ there exists a set $\Gamma\in\mathbb{R}^{2d}$ with $\mu(\Gamma)>0$ such that 
    \begin{align}\label{eq:alphaKlinear}
        \sum_{a=1}^{r}\alpha_{a} \sfK_{a} \neq 0 \quad \text{ for all $(x,y)\in \Gamma$}.
    \end{align}
\end{itemize} 
Then $\bm{\Sigma}_{H_0}$ is non-singular. 
\end{proposition}

\begin{proof} Throughout the proof we will use $\mu$ to denote the Lebesgue measure in appropriate dimensions. Recall from \eqref{eq:H0sigma} that $\bm{\Sigma}_{H_0} = ((\sigma_{ab}))_{1 \leq a, b \leq r}$, where 
$$\sigma_{ab} = \frac{2}{\rho^2(1-\rho)^2} \E \left[\sfK_a^\circ(X, X')\sfK_b^\circ(X, X')\right] = \frac{2}{\rho^2(1-\rho)^2} \Cov[\sfK_a^\circ(X, X'), \sfK_b^\circ(X, X')] , 
$$ 
for $X, X' \sim P$. 
Hence, $\bm{\Sigma}_{H_0}$ is singular if and only if there exists $\alpha_1, \alpha_2, \ldots, \alpha_r \in \R^r$ such that 
\begin{align*}
    \sum_{a=1}^{r} \alpha_{a} \sfK_a^\circ(X, X') = 0 \quad \text{ almost surely } P^2. 
\end{align*} 
Then by Assumption \ref{assumption:PQdensity}, there exists a set $A \subseteq \mathbb{R}^{2d}$ with $\mu(A^{c}) = 0$, such that 
\begin{align}\label{eq:linearK}
    \sum_{a=1}^{r} \alpha_{a} \sfK_a^\circ(x, y) = 0 , 
\end{align} 
for all $(x, y) \in A$. Then considering $h(x,y) = \sum_{a=1}^{r}\alpha_{a} \sfK_{a}(x,y)$ gives, for $(x, y) \in A$, 
\begin{align} 
h(x,y) & = \sum_{a=1}^{r} \alpha_{a} \left( \E[\sfK_a(x,X')] + \E[\sfK_a(X,y)] + \sfK_a^\circ(x, y) - \E[\sfK_a(X,X')] \right) \nonumber \\ 
& = \sum_{a=1}^{r} \alpha_{a} \left( \E[\sfK_a(x,X')] + \E[\sfK_a(X,y)]  - \E[\sfK_a(X,X')] \right) \tag*{(by \eqref{eq:linearK})} \nonumber \\ 
\label{eq:hxyfgL} & = \sum_{a=1}^{r} \alpha_{a} \left( f_{a}(x) + g_{a}(y) \right) + L, 
\end{align}
where $L:=-\sum_{a=1}^{r} \alpha_{a} \E[\sfK_a(X,X')] $, $f_{a}(x) := \E[\sfK_a(x,X')]$, and $g_{a}(y) := \E[\sfK_a(X,y)]$, for $1 \leq a \leq r$.
For any $x\in\mathbb{R}^{d}$ consider,
\begin{align*}
    A_{x} := \left\{y\in\mathbb{R}^{d}:(x,y)\in A \right\}. 
\end{align*}
Denote $B : =\{x\in\mathbb{R}^{d}:\mu\left(A_{x}^{c}\right) = 0\}$. Then by  Fubini's Theorem, $\mu(B) = 1$. Now, consider $x' \in B$ and a sequence $\{x_{N}\}_{N \geq 1}$ in $B$ such that $\|x_{N}\|\rightarrow\infty$. Define,
\begin{align*}
    A_{\infty} = A_{x'} \bigcap \left( \bigcap_{n\geq 1}A_{x_{N}} \right). 
\end{align*} 
By definition $\mu(A_\infty^{c}) = 0$. Also, if $y \in A_\infty$, then $(x_N, y) \in A$ and recalling \eqref{eq:hxyfgL} we have,
\begin{align}\label{eq:hxny}
    h(x_{N},y) = \sum_{a=1}^r  \alpha_a f_{a}(x_{N}) + \sum_{a=1}^{r} \alpha_{a} g_{a}(y) + L,
\end{align}
for all $N \geq 1$. Similarly, we also have $h(x',y) = \sum_{a=1}^r  \alpha_a f_{a}(x') + \sum_{a=1}^{r} \alpha_{a} g_{a}(y) + L$. This implies, 
\begin{align}\label{eq:hxnygs}
    \sum_{a=1}^{r} \alpha_{a} g_{a}(y) = h(x', y) - \sum_{a=1}^{r} \alpha_{a}f_{a}(x') - L = h(x',y) + L_{x'} , 
\end{align} 
where $L_{x'} := - \sum_{a=1}^{r} \alpha_{a}f_{a}(x') - L$ is a constant  depending on $x'$. Next, fixing $y'\in A_\infty$ we get,
\begin{align}\label{eq:hxnyfs}
    \sum_{a=1}^{r} \alpha_{a}f_{a}(x_{N}) = h(x_{N},y') - \sum_{a=1}^r b_a g_{a}(y') - L = h(x_{N},y') + L_{y'} , 
\end{align} 
where $L_{y'} := - \sum_{a=1}^r b_a g_{a}(y') - L $ is a constant depending on $y'$. Thus, combining \eqref{eq:hxny}, \eqref{eq:hxnygs}, and \eqref{eq:hxnyfs}, 
\begin{align}\label{eq:hLxy}
    h(x_{N},y) = h(x',y) + h(x_{N},y') + L_{x'} + L_{y'} + L , 
\end{align}
for all $y\in A_\infty$. Now, since $\mu(A_\infty)=1$ we can choose a sequence $\{y_{M}\}_{M \geq 1}$ in $A_\infty$ such that $\|y_{M}\|\rightarrow\infty$. Then observe that,
\begin{align*}
    h(x_{N},y_{M}) = h(x',y_{M}) + h(x_{N},y') + L_{x'} + L_{y'} +L.
\end{align*}
Taking limits as $M \rightarrow\infty$ and then $N \rightarrow\infty$ and using condition \eqref{eq:limitKxz} it follows that $L_{x'} + L_{y'} + L = 0$.
Thus, from \eqref{eq:hLxy}, for all $y\in A_\infty$, $h(x_{N},y) = h(x',y) + h(x_{N},y')$. Therefore, taking $N\rightarrow\infty$ and using \eqref{eq:limitKxz} gives, 
\begin{align*}
    h(x',y) = \sum_{a=1}^{r} \alpha_{a} \sfK_{a}(x', y) = 0 \quad \text{ for all } y\in A_\infty.
\end{align*}
This implies, since $x'$ is arbitrarily chosen from $B$ and $\mu(B)$ = 1, 
\begin{align*}
 \mu   \left\{x\in\mathbb{R}^{d}: \mu \left\{y\in\mathbb{R}^{d}: \sum_{a=1}^{r} \alpha_{a} \sfK_{a}(x, y) \neq 0\right\}  > 0\right\} = 0 . 
\end{align*}
Therefore, by Fubini's theorem, 
\begin{align*}
    \mu\left\{ (x, y) \in\mathbb{R}^{2 d}:  \sum_{a=1}^{r} \alpha_{a} \sfK_{a}(x, y) \neq 0 \right\} = 0 , 
\end{align*}
which contradicts \eqref{eq:alphaKlinear}. Thus, $\bm{\Sigma}_{H_0}$ is non-singular whenever \eqref{eq:limitKxz} and \eqref{eq:alphaKlinear} hold.
\end{proof}

Proposition \ref{ppn:Kvariance} gives general conditions under which the matrix $\bm \Sigma_{H_0}$ is invertible. 
Clearly, condition \eqref{eq:limitKxz} is satisfied by most kernels. We will now show that condition \eqref{eq:alphaKlinear} holds when $\cK=\{ \sfK_1, \sfK_2, \ldots, \sfK_r \}$ is a collection of $r$ distinct Gaussian or Laplace kernels.

\begin{corollary}\label{cor:Kvariance} Suppose Assumption \ref{assumption:PQdensity} holds and $\cK = \{ \sfK_1, \sfK_2, \ldots, \sfK_r \}$ 
is a collection of $r$ distinct Gaussian or Laplace kernels, that is, $\sfK_a$ is either a Gaussian kernel or a Laplace kernel with bandwidth $\sigma_a$, for $1 \leq a \leq r$, with $\sigma_1 \ne  \sigma_2 \ne \cdots \ne \sigma_r > 0$. Then the conditions of Proposition \ref{ppn:Kvariance} hold and, consequently, $\bm \Sigma_{H_0}$ is non-singular. 
\end{corollary}

\begin{proof} Clearly, condition \eqref{eq:limitKxz} holds for the Gaussian and the Laplace kernels. To show \eqref{eq:alphaKlinear} assume for contradiction that there exists $\alpha_{1}, \alpha_{2}, \ldots, \alpha_{r}$ such that 
\begin{align*}
    \mu \left\{(x,y)\in \mathbb{R}^{2d}: \sum_{a=1}^{r} \alpha_{a} \sfK_{a}(x, y)\neq 0\right\}  = 0 . 
\end{align*} 
Note that without loss of generality we can assume that $\alpha_{a}\neq 0, 1\leq a\leq r$. Denote $D:= \{x\in\mathbb{R}^{d}: \mu \{y\in\mathbb{R}^{d}:\sum_{a=1}^r  \alpha_a  \sfK_a \neq 0 \} = 0 \}$. By Fubini's theorem $\mu(D) = 1$. 
For any $x\in D$ denote, 
\begin{align*}
    D_{x}:=\left\{y\in\mathbb{R}^{d}: \sum_{a=1}^r  \alpha_a  \sfK_a = 0\right\} .
    \end{align*} 
 By definition, $\mu(D_x) = 1$. Then we can find a sequence $\{ y_{M} \}_{M \geq 1}$ in $D_{x}$ such that $\|y_{M}\| \rightarrow\infty$ such that \begin{align*}
    \sum_{a=1}^r  \alpha_a  \sfK_{a}(x, y_{M}) = 0.
\end{align*}  
Then,
\begin{align}\label{eq:GKernel-lin-comb-0}
    \sum_{a=2}^{r} \alpha_a \dfrac{\sfK_{a}(x, y_{M})}{\sfK_{1}(x, y_{M})} = -\alpha_{1} . 
\end{align}
Now, we have the following cases: 
\begin{itemize} 

\item {\it $\sfK_a$ is a Gaussian kernel for all $1 \leq a \leq r$}: Without loss of generality, suppose $\sigma_1 = \arg\max_{1 \leq a \leq r} \sigma_a$. Then by the definition of Gaussian kernel, it follows that,
\begin{align*}
    \sum_{a=2}^{r} \alpha_a  \dfrac{\sfK_{a}(x, y_{M})}{\sfK_{1}(x, y_{M})}\rightarrow 0 , 
    \end{align*} 
    as $M \rightarrow \infty$.  This contradicts \eqref{eq:GKernel-lin-comb-0}.

\item {\it $\sfK_a$ is a Laplace kernel for some $1 \leq a \leq r$}: Without loss of generality, suppose $\sigma_1$ be the largest bandwidth among the Laplace kernels. Then by the definition of Gaussian and Laplace kernels it follows that 
\begin{align*}
    \sum_{a=2}^{r} \alpha_a  \dfrac{\sfK_{a}(x, y_{M})}{\sfK_{1}(x, y_{M})}\rightarrow 0 , 
    \end{align*} 
    as $M \rightarrow \infty$. As in the previous case, this contradicts \eqref{eq:GKernel-lin-comb-0}. 
\end{itemize} 
This completes the proof of Corollary \ref{cor:Kvariance}. 
\end{proof}

\section{Technical Lemmas}
\label{sec:lmpf}

In this section we collect the proofs of various technical lemmas. We begin by showing the continuity of the characteristic function of the limiting distribution.

\begin{lemma}\label{lm:phicontinuity} 
Let $\Phi(\bm \eta)$ be as in \eqref{eq:GK}. Then $\Phi(\bm 0) = 1$ and $\Phi(\bm \eta)$ is continuous at $\bm 0 \in \R^r$. 
\end{lemma}

\begin{proof} Note that when $\bm \eta = \bm 0$, the only eigenvalue of the operator $\cH_{\cK, \eta}$ is zero and, hence, $\Phi(\bm 0)=1$. 

For showing continuity at $\bm \eta = \bm 0$, recall that for all $\bm \eta \in \R^r$, $\sum_{\lambda\in \Lambda(\bm \eta)}\lambda^2<\infty$, since $\cH_{\cK, \bm \eta}$ is a Hilbert-Schmidt operator. Then by Fubini's theorem,
\begin{align}\label{eq:Zlambda}
    \E\left[\left(\sum_{\lambda\in \Lambda(\bm \eta)}\lambda(Z_{\lambda}^2-1)\right)^2\right] 
    &= \E\left[\sum_{\lambda}\lambda^2(Z_{\lambda}^2-1) + \sum_{\lambda_{1}\neq \lambda_{2}}\lambda_{1}\lambda_{2}\left(Z_{\lambda_{1}}^2-1\right)\left(Z_{\lambda_{2}}^2-1\right)\right]\nonumber\\
    &= \sum_{\lambda}\lambda^2\E(Z_{\lambda}^2-1)^2\nonumber\\
    &= 2 \sum_{\lambda}\lambda^2 . 
\end{align}
By the spectral theorem (see \citet[Theorem 6.35]{renardy2006introduction}) and recalling \eqref{eq:Hxyetaxy}, $$\sum_{\lambda\in\Lambda(\bm \eta)}\lambda^2 = \left\| \mathsf{H}_{\bm \eta}^\circ \right\|^2.$$ 
Clearly, $\lim_{\bm \eta \rightarrow \bm 0} \| \mathsf{H}_{\bm \eta}^\circ \|^2 = 0$ and thus by \eqref{eq:ZH} and \eqref{eq:Zlambda} and we conclude, $Z(\mathsf{H}_{\bm \eta}) \stackrel{L^2}\rightarrow  0$, as $\bm \eta \rightarrow \bm 0$. Hence, by \eqref{eq:expZH} and the Dominated Convergence Theorem, 
$$  \lim_{\bm \eta \rightarrow 0} \Phi(\bm \eta) = \lim_{\bm \eta \rightarrow \bm 0} \E \left [ e^{\iota Z(\mathsf{H}_{\bm \eta})} \right] = 1. $$
This shows the continuity of $\Phi(\bm \eta)$ at $\bm \eta = \bm 0 \in \R^r$. 
\end{proof}

Next, we compute the MGF of the random variable $Z(\mathsf{H})$ as defined \eqref{eq:H0distribution}. 

\begin{lemma}\label{lm:ZHM} 
The MGF of $Z(\mathsf{H})$, as defined in \eqref{eq:H0distribution}, exists for all $|t|<\frac{\rho (1-\rho)}{8}\left\|\mathsf{H}^\circ \right\|^{-1}$ and is given by,
\begin{align*}
    \log M_{Z(\mathsf{H})}(t) := \log\E\left[ e^{tZ(\mathsf{H}) } \right] = \left(\frac{t}{\rho (1-\rho) }\right)^{2}\|\mathsf{H}^\circ \|^{2} + \frac{1}{2}\sum_{K=3}^{\infty}\sum_{s=1}^{\infty}\dfrac{\left(\frac{2}{\rho(1-\rho)}\lambda_{s}t\right)^{K}}{K}
\end{align*}
where $\{\lambda_{s}:s\geq 1\}$ are the eigenvalues of the operator $\mathcal{H}_{\mathsf{H}^\circ}$.
\end{lemma}

\begin{proof} Define $\gamma = \frac{1}{\rho(1-\rho)}$. Since $\sum_{s=1}^{\infty}\lambda_{s}^{2}<\infty$, by \citet[Proposition 7.1]{bbbpdsm} we have,
\begin{align}\label{eq:mgf-Zh}
    M_{Z(\mathsf{H})}(t) := \E\left[ e^{tZ(\mathsf{H}) } \right]  = \prod_{s=1}^{\infty}\dfrac{e^{-\gamma\lambda_{s}t}}{\sqrt{1-2\gamma\lambda t}} , 
\end{align}
for all $|t|<\frac{1}{8 \gamma }(\sum_{s=1}^{\infty}\lambda_{s}^{2})^{-\frac{1}{2}} = \frac{1}{8 \gamma}\left\|\mathsf{H}^\circ \right\|^{-1}$ (since $\left\| \mathsf{H}^\circ \right\|^{2} = \sum_{s=1}^{\infty}\lambda_{s}^2$). Fix $t$ such that $|t|< \frac{1}{8 \gamma }\left\|\mathsf{H}^\circ \right\|^{-1}$, then taking log on both sides of \eqref{eq:mgf-Zh} and expanding we get,
\begin{align}
    \log M_{Z(\mathsf{H})}(t) = \sum_{s=1}^{\infty} \left\{ -\gamma\lambda_{s} t + \frac{1}{2}\left(\sum_{k=1}^{\infty}\dfrac{(2\gamma\lambda_{s} t)^k}{k}\right)  \right\}   &= \frac{1}{2}\sum_{s=1}^{\infty}\sum_{k=2}^{\infty}\dfrac{(2\gamma\lambda_{s} t)^k}{k} . \label{eq:mgf-Zh-expansion}
\end{align}
Using $\left\| \mathsf{H}^\circ \right\|^{2} = \sum_{s=1}^{\infty}\lambda_{s}^2$ and  \eqref{eq:mgf-Zh-expansion} we get, 
\begin{align}\label{eq:mgfH}
    \log M_{Z(\mathsf{H})}(t) 
    &= \gamma^{2}\|\mathsf{H}^\circ \|^{2}t^{2} + \frac{1}{2}\sum_{s=1}^{\infty}\sum_{k=3}^{\infty}\dfrac{(2\gamma\lambda_{s} t)^k}{k} . 
\end{align} 
Using the bounds $|t|< \frac{1}{8 \gamma}\left\|\mathsf{H}^\circ \right\|^{-1}$ and $|\lambda_{s}| \leq \frac{\|\mathsf{H}^\circ \|}{\sqrt{s}}$ for all $s\in \mathbb{N}$, where $|\lambda_{1}|\geq |\lambda_{2}|\geq \cdots$ (again using $\left\| \mathsf{H}^\circ \right\|^{2} = \sum_{s=1}^{\infty}|\lambda_{s}|^2$) gives, 
\begin{align*}
    \sum_{s=1}^{\infty}\sum_{k=3}^{\infty}\dfrac{\left|2\gamma\lambda_{s} t\right|^k}{k}
    &\leq \sum_{s=1}^{\infty}\sum_{k=3}^{\infty}\dfrac{|\lambda_s|^k}{4^k \| \mathsf{H}^\circ \|^k k } \leq  \sum_{s=1}^{\infty}\sum_{k=3}^{\infty}\dfrac{1}{s^{\frac{3}{2}}4^{k} k}<\infty , 
\end{align*}
Therefore, by Fubini's Theorem we can interchange the order of the sum in \eqref{eq:mgfH} to get, 
\begin{align*}
    \log M_{Z(\mathsf{H})}(t)
    & = \gamma^{2}\| \mathsf{H}^\circ \|^{2}t^{2} + \frac{1}{2}\sum_{k=3}^{\infty}\sum_{s=1}^{\infty}\dfrac{(2\gamma\lambda_{s} t)^k}{k} , 
\end{align*}
for all $|t|<\frac{1}{8\gamma}\| \mathsf{H}^\circ \|^{-1}$, which completes the proof. 
\end{proof} 

In the next lemma we show that the row sums of a characteristic kernel is asymptotically close to its expected value in an $L_2$ sense. This is used in the proof of Proposition \ref{ppn:HZW}.

\begin{lemma}\label{lemma:kernel-avg-convg}
Suppose $\sfK \in L^{2}(\mathcal{X}^2, P^2)$ is a characteristic kernel satisfying $\E_{X\sim P}[ \sfK^{2}(X,X)]<\infty$. Then for $X_1, X_2, \ldots, X_m$ i.i.d. from the distribution $P$, 
\begin{align*}
    \lim_{m \rightarrow \infty} \frac{1}{m}\sum_{i=1}^{m}\left(\frac{1}{m}\sum_{j=1}^{m} \sfK(X_{i}, X_{j})- \E_{Z \sim P} [ \sfK(X_{i}, Z)] \right)^2 = 0 , 
\end{align*}
on a set $\mathcal{B}_{\sfK}\in \sB(\cX)$ such that $\P(\mathcal{B}_{\sfK})=1$.
\end{lemma} 

\begin{proof} 
Let $\psi$ be the feature map corresponding to $\sfK$.  Recalling the definition of mean embedding $\mu_{P}$ from \eqref{eq:XP} observe that, for $1 \leq i \leq m$, 
\begin{align}\label{eq:k-to-psi}
    \frac{1}{m}\sum_{j=1}^{m} \sfK(X_{i},X_{j})- \E_{Z \sim P} [\sfK(X_{i}, Z) = \left\langle \psi(X_{i}),\frac{1}{m}\sum_{j=1}^{m}\psi(X_{j})-\mu_{P}\right\rangle_{\mathcal{H}}
\end{align}
By Cauchy-Schwartz inequality, 
\begin{align}\label{eq:use-CS}
    \left\langle \psi(X_{i}),\frac{1}{m} \sum_{j=1}^{m}\psi(X_{j})-\mu_{P}\right\rangle\leq \|\psi(X_{i})\|_{\mathcal{H}}\left\|\frac{1}{m}\sum_{j=1}^{m}\psi(X_{j})-\mu_{P}\right\|_{\mathcal{H}}
\end{align}
By \eqref{eq:k-to-psi} and \eqref{eq:use-CS}, 
\begin{align}\label{eq:kernel-bdd-psi}
    & \frac{1}{m}\sum_{i=1}^{m}\left(\frac{1}{m}\sum_{j=1}^{m} \sfK(X_{i},X_{j})- \E_{Z \sim P} [\sfK (X_{i}, Z)] \right) \nonumber \\ 
    &\leq \left(\frac{1}{m}\sum_{i=1}^{m}\|\psi(X_{i})\|_{\mathcal{H}}^2\right)\left\|\frac{1}{m}\sum_{j=1}^{m}\psi(X_{j})-\mu_{P}\right\|_{\mathcal{H}}^2\nonumber\\
    & = \left(\frac{1}{m}\sum_{i=1}^{m} \sfK(X_{i},X_{i})^2\right)\left\|\frac{1}{m}\sum_{j=1}^{m}\psi(X_{j})-\mu_{P}\right\|_{\mathcal{H}}^2 . 
\end{align}
From \eqref{eq:XP} we have $\mu_{P}(t) = \E_{X \sim P} [ \sfK (t, X) ]$ and hence, by \eqref{eq:fXP}, 
\begin{align}\label{eq:norm-mup}
    \|\mu_{P}\|_{\mathcal{H}}^2 = \langle\mu_{P}, \mu_{P}\rangle_{\mathcal{H}} = \E_{X' \sim P} [ \mu_{P}(X') ] = \E_{X, X' \sim P} [ \sfK(X, X') ] . 
\end{align}
Once again by \eqref{eq:XP} we observe that,
\begin{align*}
    \left\|\frac{1}{m}\sum_{j=1}^{m}\psi(X_{j})-\mu_{P}\right\|_{\mathcal{H}}^2
    & = \left\langle\frac{1}{m}\sum_{j=1}^{m}\psi(X_{j})-\mu_{P},\frac{1}{m}\sum_{j=1}^{m}\psi(X_{j})-\mu_{P}\right\rangle_{\mathcal{H}}\\
    & = \frac{1}{m^{2}}\sum_{1\leq i, j\leq m}\sfK(X_{i},X_{j}) - \frac{2}{m}\sum_{j=1}^{m}\E_{Z \sim P} [ \sfK(X_{j},Z) ] + \|\mu_{P}\|_{\mathcal{H}}^2\\
    & = \frac{1}{m^{2}}\sum_{1\leq i, j \leq m}\sfK(X_{i},X_{j}) - \frac{2}{m}\sum_{j=1}^{m}\E_{Z \sim P}[ \sfK(X_{j},Z)] + \E[\sfK(X_{1},X_{2}) ] , 
\end{align*}
where the last equality follows from \eqref{eq:norm-mup}. Notice that 
$\E |\E_{Z \sim P}[\sfK(X_{1},Z)] |\leq \E|\sfK(X_{1},X_{2})|<\infty$. Hence, by the strong law of large numbers for $U$-statistics \citet[Theorem 5.4.A]{serfling} we conclude that, 
\begin{align}\label{eq:average-mup-convg}
 \lim_{m \rightarrow \infty}   \left\|\frac{1}{m}\sum_{j=1}^{m}\psi(X_{j})-\mu_{P}\right\|_{\mathcal{H}}^2 = 0 , 
\end{align}
on a set $\mathcal{B}_{\sfK}^{(1)} \in \sB(\cX)$ such that $\P(\mathcal{B}_{\sfK}^{(1)})=1$. Also, since $\E[\sfK(X_{1},X_{1})^2] <\infty$, by the strong law of large numbers, 
\begin{align*}
   \lim_{m \rightarrow \infty} \frac{1}{m}\sum_{i=1}^{m}\sfK(X_{i},X_{i})^{2} \rightarrow \E[\sfK^{2}(X_{1},X_{1})] 
\end{align*}
on a set $\cB_{\sfK}^{(2)}$ such that $\P(\cB_{\sfK}^{(2)})=1$. Using this and \eqref{eq:average-mup-convg} in \eqref{eq:kernel-bdd-psi} the proof is completed, by choosing $\mathcal{B}_{\sfK} = \cB_{\sfK}^{(1)} \bigcap \cB_{\sfK}^{(2)}$.
\end{proof}
    
    \end{document}